\documentclass[11pt]{article} % use larger type; default would be 10pt

\usepackage{geometry} % to change the page dim.
\usepackage{fullpage}
\usepackage{setspace}
\usepackage{graphicx, amsmath,amsfonts,cite,appendix, amssymb,xcolor,enumerate, mathrsfs, amsthm,hyperref, mathtools,bbm,standalone,tikz} 
\usetikzlibrary{arrows, patterns}
\usepackage{booktabs} % for much better looking tables
\usepackage{array} % for better arrays (eg matrices) in maths
\usepackage{paralist} % very flexible & customisable lists (eg. enumerate/itemize, etc.)
\usepackage{verbatim} % adds environment for commenting out blocks of text & for better verbatim
\usepackage{subfig} % make it possible to include more than one captioned figure/table in a single float
\theoremstyle{plain}
\newtheorem{thm}{Theorem}
\newcommand{\be}{\begin{equation}}
\newcommand{\ee}{\end{equation}}
\newcommand{\ben}{\begin{equation*}}
\newcommand{\een}{\end{equation*}}
\newcommand{\mc}{\mathcal}

\newcommand{\snr}{\textsf{snr}}
\newcommand{\reals}{\mathbb{R}}

\newtheorem{lem}{Lemma}[section]
\newtheorem{lemma}{Lemma}[section]
\newtheorem{applem}{Lemma}[section]
\newtheorem{defi}{Definition}[section]

\newtheorem{rem}{Remark}[section]
\newtheorem{prop}{Proposition}[section]
\newtheorem{corr}{Corollary}[section]
\newcommand{\e}{\epsilon}

\newcommand{\abs}[1]{\left \lvert #1\right \rvert}
\newcommand{\norm}[1]{\left \lVert #1\right \rVert}
\newcommand{\ol}[1]{\overline{#1}}

\newcommand{\expec}{\mathbb{E}}

\newcommand{\mscrs}{\mathscr{S}}
\newcommand{\sfp}{\mathsf{u}}
\newcommand{\sfq}{\mathsf{v}}

\newcommand{\delt}{\widetilde{\delta}}

\newcommand{\normal}{\mc{N}}

\DeclarePairedDelimiter\ceil{\lceil}{\rceil}

\newcommand{\M}{M}
\newcommand{\Lc}{\textsf{C}}
\newcommand{\Lr}{\textsf{R}}
\newcommand{\Mc}{N_{\textsf{C}}}
\newcommand{\Mr}{N_{\textsf{R}}}

\newcommand{\Wricj}{W_{\textsf{r}(i)\textsf{c}(j)}}
\newcommand{\sfr}{\textsf{r}}

\newcommand{\sfc}{\textsf{c}}

\newcommand{\stdnorm}{\mathcal{N}(0,1)}
\newcommand{\vth}{\vartheta}
\newcommand{\Qmat}{\breve{\boldsymbol{Q}}}
\newcommand{\Mmat}{\breve{\boldsymbol{M}}}
\newcommand{\Qadjmat}{\boldsymbol{Q}}
\newcommand{\Madjmat}{\boldsymbol{M}}

\newcommand{\iden}{\boldsymbol{I}}
\newcommand{\bzero}{\boldsymbol{0}}
\newcommand{\Z}{\boldsymbol{Z}}
\newcommand{\bZ}{\boldsymbol{Z}}
\newcommand{\Zprime}{\boldsymbol{Z}'}
\newcommand{\Zsupr}{\boldsymbol{Z}^{\sf{r}}}
\newcommand{\Zprimesupc}{\boldsymbol{Z}^{'\sf{c}}}

\newcommand{\qadj}{\boldsymbol{q}}

\newcommand{\madj}{\boldsymbol{m}}

\newcommand{\indic}{\mathbb{I}}

\newcommand{\sfv}{\textsf{v}}

\newcommand{\avec}{\boldsymbol{a}}
\newcommand{\hvec}{\boldsymbol{h}}
\newcommand{\qvec}{\boldsymbol{q}}
\newcommand{\bvec}{\boldsymbol{b}}
\newcommand{\mvec}{\boldsymbol{m}}
\newcommand{\zvec}{\boldsymbol{z}}
\newcommand{\yvec}{\boldsymbol{y}}
\newcommand{\xvec}{\boldsymbol{x}}
\newcommand{\svec}{\boldsymbol{s}}
\newcommand{\wvec}{\boldsymbol{w}}
\newcommand{\Amat}{\boldsymbol{A}}
\newcommand{\Ymat}{\boldsymbol{Y}}
\newcommand{\Xmat}{\boldsymbol{X}}
\newcommand{\Hmat}{\boldsymbol{H}}
\newcommand{\Bmat}{\boldsymbol{B}}
\newcommand{\sfAmat}{\boldsymbol{\mathsf{A}}}
\newcommand{\Wmat}{\boldsymbol{W}}
\newcommand{\Smat}{\boldsymbol{S}}
\newcommand{\Zvec}{\boldsymbol{Z}}
\newcommand{\Dvec}{\boldsymbol \Delta}
\newcommand{\betavec}{\boldsymbol \beta}
\newcommand{\upvec}{\boldsymbol \upsilon}
\newcommand{\brqvec}{\breve{\boldsymbol{q}}}
\newcommand{\brmvec}{\breve{\boldsymbol{m}}}
\newcommand{\alphavec}{\boldsymbol \alpha}
\newcommand{\gammavec}{\boldsymbol \gamma}
\newcommand{\tCmat}{\widetilde{\boldsymbol{C}}}
\newcommand{\brCmat}{\breve{\boldsymbol{C}}}
\newcommand{\proj}{\boldsymbol{\mathsf{P}}}
\newcommand{\Upmat}{\boldsymbol \Upsilon}

\newcommand{\innerM}{\pmb{\mathbb{M}}}
\newcommand{\innerQ}{\pmb{\mathbb{Q}}}

\newcommand{\tv}{\text{v}}
\newcommand{\PC}{\Pi}
\newcommand{\pc}{\pi}

\usepackage{sectsty}
\allsectionsfont{\sffamily\mdseries\upshape} % (See the fntguide.pdf for font help)
\title{\vspace{-35pt} Capacity-achieving Spatially Coupled Sparse Superposition Codes  with AMP Decoding } %\vspace{-5pt} }
\author{ Cynthia Rush\thanks{Department of Statistics, Columbia University, New York, NY 10027, USA. Email: {\tt cynthia.rush@columbia.edu}} 
 \and Kuan Hsieh\thanks{Department of Engineering, University of Cambridge, Cambridge CB2 1PZ, UK. Email: {\tt kh525@cam.ac.uk}} 
  \and Ramji Venkataramanan\thanks{Department of Engineering, University of Cambridge, Cambridge CB2 1PZ, UK. Email: {\tt rv285@cam.ac.uk}  This work was supported in part by an EPSRC Doctoral Training Award and  a Turing Fellowship from the Alan Turing Institute, the National Science Foundation (NSF CCF \#1849883), the Simons Institute for the Theory of Computing, and NTT Research.  This paper was presented in part at the 2018 IEEE International Symposium on Information Theory and at the 2018 IEEE Information Theory Workshop.} 
}

\begin{document}
\maketitle
\vspace{-10pt}
\begin{abstract}
Sparse superposition codes, also referred to as sparse regression codes (SPARCs), are a  class of codes for efficient communication over the AWGN channel at rates approaching the channel capacity. In a standard SPARC, codewords are sparse linear combinations of columns of an i.i.d.~Gaussian design matrix, while in a spatially coupled SPARC  the design matrix has a block-wise structure,  where the variance of the Gaussian entries can be varied across blocks. A well-designed spatial coupling structure can significantly enhance the error performance of iterative decoding algorithms such as Approximate Message Passing (AMP). 

In this paper, we obtain a non-asymptotic bound on the probability of error of  spatially coupled SPARCs with AMP decoding. Applying this bound to a simple band-diagonal design matrix, we prove that spatially coupled SPARCs with AMP decoding achieve the capacity of the AWGN channel. The bound also highlights how the  decay of error probability depends on each design parameter of the spatially coupled SPARC. 

An attractive feature of  AMP decoding is that its asymptotic mean squared error (MSE) can be predicted via a deterministic recursion called state evolution. Our result provides the first proof that the  MSE  concentrates on the state evolution prediction for spatially coupled designs. Combined with the state evolution prediction, this result implies that spatially coupled SPARCs with the proposed band-diagonal design are capacity-achieving. Using the proof technique used to establish the main result, we also obtain a concentration inequality for the MSE of AMP applied to compressed sensing with spatially coupled design matrices.  Finally, we provide numerical simulation results that demonstrate the finite length error performance of spatially coupled SPARCs. The performance is compared with coded modulation schemes that use LDPC codes from the DVB-S2 standard. 
\end{abstract}

\tableofcontents

\newpage
% !TEX root =  sc_sparc_journal_paper_full.tex

\section{Introduction} \label{sec:intro}

We consider communication over the memoryless additive white Gaussian noise (AWGN) channel, where the  output symbol $y$ is generated from input symbol $u$ as $y = u + w$. The noise  $w$ is Gaussian with zero mean and variance $\sigma^2$. The input has an average power constraint $P$: for a codeword $\boldsymbol{x}= x_1, x_2, \ldots ,x_n$  transmitted over $n$ uses of the channel,
\be\label{eq:average_power_constraint}
\frac{1}{n}\sum_{i=1}^n x_i^2 \leq P.
\ee
The Shannon capacity of the channel is  $\mc{C} = \frac{1}{2}\ln \left(1+\frac{P}{\sigma^2}\right)$ nats/transmission.

Sparse superposition codes, or sparse regression codes (SPARCs), were  introduced by Joseph and Barron\cite{joseph2012, joseph2014} for efficient communication over the AWGN channel.  These codes have been proven to be reliable at rates approaching $\mc{C}$ with various low complexity iterative decoders \cite{joseph2014,cho2013,rush2017}.  As shown in Fig.\ \ref{fig:sparc_code_construction}, a SPARC is defined by a design matrix $\Amat$ of dimension $n\times \M L$, where $n$ is the code length and $\M$, $L$ are integers such that $\Amat$ has $L$ sections with $\M$ columns each. Codewords are generated as linear combinations of $L$ columns of $\Amat$, with one column from each section. Thus a codeword can be represented as $\boldsymbol{x} = \Amat\betavec$, with $\betavec$ being an $\M L\times 1$ \emph{message vector} with exactly one non-zero entry in each of its $L$ sections.  The message is indexed by the locations of the non-zero entries in $\betavec$. The values of the non-zero entries are fixed a priori. 

Since there are $M$ choices for the location of the non-zero entry in each of the $L$ sections, there are $M^L$ codewords.  To achieve a communication rate of $R$ nats/transmission, we therefore require
\begin{equation}\label{eq:rate_eq}
\M^L = e^{nR} \quad \text{or} \quad nR = L\ln \M.
\end{equation}
The decoding task is to recover the message vector $\betavec$ from the received sequence $\yvec \in \reals^n$ given by
\begin{equation}\label{eq:linear_model}
\yvec = \Amat\betavec + \wvec.
\end{equation}

In the standard SPARC construction introduced in \cite{joseph2012, joseph2014}, the design matrix $\Amat$ is constructed with i.i.d.\ standard Gaussian entries. The values of the non-zero coefficients in the message vector $\betavec$ then define a \emph{power allocation} across sections.  With an appropriately chosen power allocation (e.g., one that is exponentially decaying across sections),  the feasible decoders proposed in \cite{joseph2014,cho2013,rush2017}  have been shown to be asymptotically capacity-achieving. The choice of power allocation has also been shown to be crucial for obtaining good finite length performance  with the standard SPARC construction  \cite{GreigV17}.  A detailed discussion of the error rates of optimal versus feasible decoders for standard SPARCs can be found in \cite{venkataramanan19monograph}.

In a \emph{spatially coupled SPARC}, the design matrix $\Amat$ is composed of blocks with different variances, i.e., the entries of $\Amat$ are still independent and Gaussian but not identically distributed. Spatially coupled SPARCs were introduced by Barbier and co-authors in \cite{barbier2015,barbier2016,
barbier2017,barbier2019universal}. In these works, an approximate message passing (AMP)  algorithm was used for decoding, whose performance can be predicted via a deterministic recursion called \emph{state evolution}.  Empirical results indicate that spatially coupled SPARCs can have better error performance than power allocated SPARCs at finite code lengths (see, e.g., \cite{barbier2017,HsiehRV18}). Moreover, both standard SPARCs and power allocated SPARCs can be viewed as special cases of spatially coupled SPARCs (see Section \ref{sec:sc_AMP}). It is therefore of interest to rigorously characterize the achievable rates and the decay of error probability for spatially coupled SPARCs.

Two key steps are required to prove that spatially coupled SPARCs achieve vanishingly small error probability   with AMP decoding  for  rates $R < \mc{C}$:
\begin{enumerate}
\item Prove that in a suitable limit (corresponding to increasing code length),  the fixed point of the  state evolution recursion is one that corresponds to vanishing probability of decoding error;
\item Prove that the error rate of the AMP decoder is closely tracked by the state evolution parameters for sufficiently large code length.
\end{enumerate}

\begin{figure}[t]
\centering
    \begin{tikzpicture}[scale=0.375, font=\LARGE]
%    \begin{tikzpicture}[scale=0.5, every node/.style={transform shape}, font=\LARGE]
	% Design matrix
    	%% Border
        \draw[thick] (0.5, 0) -- (0, 0) -- (0, 8) -- (0.5, 8);
        \draw[thick] (39.5, 0) -- (40, 0) -- (40, 8) -- (39.5, 8);
        %% Grey chosen columns        
        \draw[line width=5pt, color=gray!50] (6,0) -- (6,8);
        \draw[line width=5pt, color=gray!50] (11.5,0) -- (11.5,8);
        \draw[line width=5pt, color=gray!50] (34,0) -- (34,8);
        %% Dashed lines
        \draw[dashed] (10, 0) -- (10, 8);
        \draw[dashed] (20, 0) -- (20, 8);
        \draw[dashed] (30, 0) -- (30, 8);
		%% ...'s
        \draw (25, 4) node[font=\Huge] {. . .};
		%% Dimensions
        \draw (5, 8.75) node[anchor=south] {Section $1$};
        \draw (15, 8.75) node[anchor=south] {Section $2$};        
        \draw (35, 8.75) node[anchor=south] {Section $L$};
        %% Text
        \draw (-0.5, 4) node[anchor=east] {$\boldsymbol{A}:$};

        \draw[triangle 45-triangle 45] (41, 0) -- (41, 8);
        \draw (43, 4.5) node {$n$};       
        \draw (43, 3) node {rows};

		% Message vector
		%% Border
        \draw[thick] (0.5, -4.5) -- (0, -4.5) -- (0, -1.5) -- (0.5, -1.5);
        \draw[thick] (39.5, -4.5) -- (40, -4.5) -- (40, -1.5) -- (39.5, -1.5);
        %% Dashed lines
        \draw[dashed] (10, -4.5) -- (10, -1.5);
        \draw[dashed] (20, -4.5) -- (20, -1.5);
        \draw[dashed] (30, -4.5) -- (30, -1.5);
        %% Text       
        \draw (-0.5, -3) node[anchor=east] {$\boldsymbol{\beta}:$};
        \draw (25, -3) node[font=\Huge] {. . .};
%        \draw (0.2, -3) node[anchor=west] {$0,0,..$};
        \draw (5.0, -3) node {$...0,a_1,0$};
%        \draw (10, -3) node[anchor=east] {$$};
%        \draw (10, -3) node[anchor=west] {$0,0,...$};
        \draw (14, -3) node {$a_2,0,0,...$};
%        \draw (19.8, -3) node[anchor=east] {$...0,0,0$};
%        \draw (30, -3) node[anchor=west] {$$};
        \draw (34.9, -3) node {$0,a_L,0,...$};
%        \draw (39.8, -3) node[anchor=east] {$...0,0$};
        \draw (40, -2.25) node[anchor=south west] {$\intercal$};
        %% Dimensions
        \draw[triangle 45-triangle 45] (10, -5) -- (20, -5);
        \draw (15, -5.5) node[anchor=north] {$M$ columns};
    \end{tikzpicture}
\caption{\small $\Amat$ is an $n \times \M L$ design matrix and $\betavec$ is an $\M L\times 1$ message vector with one non-zero entry in each of its $L$ sections.  Codewords are of the form $\Amat\betavec$. The non-zero values  $a_1, \ldots, a_L$ are fixed a priori.}
\label{fig:sparc_code_construction}
\vspace{-6pt}
\end{figure}
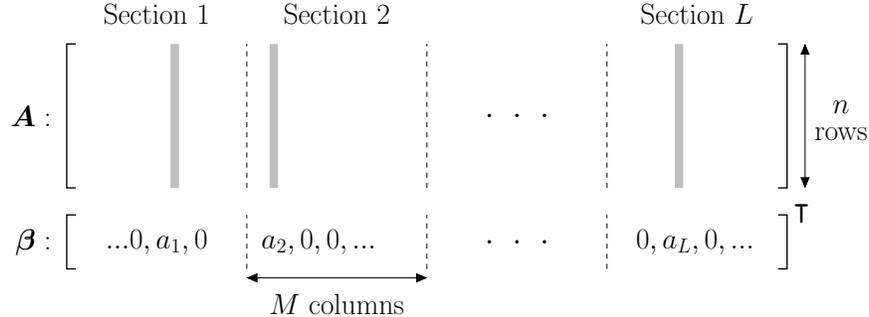

The first step was proved by   Barbier  et al.\  \cite{barbier2016}, using the potential function method \cite{yedla2014simple}. The result in \cite{barbier2016} shows `threshold saturation'  for a class of spatially coupled SPARCs with AMP decoding. For any fixed rate $R < \mc{C}$, this implies that  state evolution predicts vanishing probability of decoding error  in the large system limit. (Throughout, the terminology `large system limit' or `asymptotic limit' refers to $(L,M,n)$ all tending to infinity with $L \ln M = nR$.) 

In this paper, we rigorously prove the second step. We also  provide an alternative proof of the first step which gives insight into how the parameters defining the spatially coupled matrix influence the decoding progression.  These two steps together yield  the main result (Theorem \ref{thm:main}), which is a non-asymptotic bound on the probability of  decoding error of the AMP decoder. To our knowledge, this is the first complete proof that spatially coupled SPARCs are capacity-achieving with efficient decoding.

\textbf{Related work}: Approximate message passing (AMP) refers to a class of iterative algorithms for statistical  estimation in models corresponding to  dense factor graphs. AMP algorithms,  which are obtained via Gaussian and quadratic approximations of standard message passing techniques like belief propagation, have been successfully applied to compressed sensing \cite{donoho2009,bayati2011,krzakala2012,bayati2012} and its applications in communications such as  MIMO detection \cite{jeon2015optimality} and unsourced random access \cite{fengler2019sparcs}. Other applications of AMP include estimation in  generalized linear models \cite{rangan2011,schniter2014compressive}, robust estimation \cite{donohoMest16,bradic2016robustness}, and several variants of low-rank matrix estimation \cite{rangan2012iterative,kab16MF,Deshpande2014,parkerSch2014B,barbier2016mutual,montanari2017estimation}.

The idea of spatial coupling was introduced in the context of LDPC codes  \cite{felstrom1999,lentmaier2010iterative,kudekar2011, kudekar2013,kumar2014threshold,mitchell2015}, and first used for compressed sensing in \cite{kudekar2010}.   AMP algorithms for compressed sensing with spatially coupled matrices were first  proposed by Krzakala et al.\ in  \cite{krzakala2012X,krzakala2012}. 
Takeuchi et al. \cite{takeuchi2012phenomenological,Kawabata2015} proposed a method for analyzing the  state evolution recursion of spatially coupled systems using a potential function defined for the uncoupled system. This method characterizes the fixed points of the spatially coupled state evolution in terms of the stationary points of the potential function. The potential  function method was made rigorous by Yedla et al. in \cite{yedla2014simple}.
In \cite{donoho2013}, Donoho et al.\  proved that a class of spatially coupled Gaussian designs achieve the optimal information-theoretic limit for compressed sensing by analyzing the spatially coupled state evolution recursion in the continuum limit.
In a complementary work \cite{javanmard2013state},  Javanmard  and Montanari proved that the mean-squared error of AMP algorithm for spatially coupled compressed sensing converges (almost surely) to the state evolution prediction in each iteration. 

Though the SPARC model is similar to compressed sensing, the result of \cite{javanmard2013state} cannot be directly applied to AMP decoder since the SPARC message vector has a section-wise i.i.d.\ structure, with a growing section size $M$  in the large system limit. Moreover, our result is non-asymptotic and generalizes the proof technique developed for power-allocated SPARCs \cite{RushV19} to the present setting of spatially coupled designs.

\subsection{Structure of the paper and main contributions}

After describing the construction of spatially coupled SPARCs (SC-SPARCs) in 
Section~\ref{sec:sc_AMP}, we explain the AMP decoder and the associated state evolution recursion in Section \ref{sec:AMP}. In Section~\ref{sec:AMP_dec_prog}, we obtain upper and lower bounds on the state evolution parameters (Lemma~\ref{lem:psi_nonasymp}) which help explain the decoding progression of the AMP decoder for large  $M$ (Proposition~\ref{prop:nonasymp_se}). The main theoretical results of the paper are stated in Section~\ref{sec:main_result}:
\begin{compactitem}
\item    Theorem~\ref{thm:main} gives a non-asymptotic bound on the probability of excess section error rate of the AMP decoder for any fixed rate $R < \mc{C}$. The bound shows how the error performance depends on each parameter of the SC-SPARC, and highlights the tradeoffs involved in choosing these parameters.
\item Theorem~\ref{thm:main1} gives a concentration inequality for the mean-squared error (MSE) of the AMP decoder in each iteration. Theorem~\ref{thm:main} is a straightforward consequence of Theorem~\ref{thm:main1}.
\item With the same technique used for proving Theorem~\ref{thm:main1}, we can obtain a concentration inequality for the MSE of AMP for  compressed sensing with spatially coupled measurement matrices. This result, stated in Theorem~\ref{thm:sc_cs}, refines the asymptotic result for this model obtained in \cite{javanmard2013state}, and makes explicit how the  probability of deviation of the MSE  (from the state evolution prediction) depends on the problem dimension as well as the parameters defining the spatial coupling.
\end{compactitem}

Section~\ref{sec:emp_perf} presents numerical simulation results showing the finite length error performance of SC-SPARCs over the complex AWGN channel. The error performance is compared with coded modulation schemes which use Quadrature Amplitude Modulation (QAM) with LDPC codes from the DVB-S2 standard. We observe that at moderately high rates (around 1.5 bits/dimension), SC-SPARCs have better error performance than the coded modulation schemes considered.

In Section~\ref{sec:proof}, we prove  Theorem~\ref{thm:main1}. The proof has two key technical ingredients. The first is a conditional distribution lemma (Lemma \ref{lem:hb_cond}) which provides a non-asymptotic characterization of the conditional distribution of the AMP iterates. This characterization is then used to prove that various scalar products involving the AMP iterates concentrate around deterministic values (Lemma \ref{lem:main_lem}). In Section \ref{sec:proof}, we give an overview of the key ideas in the proof of Theorem~\ref{thm:main1}, then state the main technical lemmas and use them to prove Theorem~\ref{thm:main1}. The full proofs of the lemmas are deferred to Section \ref{sec:cond_conc_lemmas}.

Though our approach to proving Theorems~\ref{thm:main} and  \ref{thm:main1} is broadly similar to that used for power-allocated SPARCs in \cite{RushV19}, the block-wise structure of the spatially coupled design matrix introduces several new technical challenges. For example, we need to define block-adjusted versions of the AMP iterates (vectors) to obtain the appropriate linear constraints for the conditional distribution lemma (see \eqref{eq:tildem_def}-\eqref{eq:Alin_consts}). Similarly, in the main concentration lemma (Lemma~\ref{lem:main_lem}) we establish concentration results for  scalar products scaled block-wise by the entries of the base matrix.

\vspace{3pt}
\textbf{Notation}:
 For a positive integer $m$, we use $[m]$ to denote the set  $\{1, \dots, m \}$. For $x \in \reals$, we let $x_+ = \max\{ x, 0 \}$.  Throughout the paper, we use  plain font for scalars, bold font for vectors and matrices, and  subscripts to denote entries of a vector or matrix.   For example, if $\xvec$ is a vector,  we write $x_i$ for its $i^{th}$ component. Similarly,  if $\Xmat$ is a matrix, we write $X_{i j}$ for  its $(i,j)^{th}$ entry. 
The  transpose of $\Xmat$ is denoted by $\Xmat^*$. The Gaussian distribution with mean $\mu$ and variance $\sigma^2$ is denoted by $\normal(\mu, \sigma^2)$.

We write $\iden_t$ for the $t \times t$ identity matrix; the subscript is dropped when the dimension is clear from context. The indicator function of an event $\mc{A}$ is denoted by $\mathbb{I}\{ \mc{A}\}$. For deterministic  sequences $(s_n)_{n \geq 0}, (x_n)_{n \geq 0}$, we write $s_n = \Theta(x_n)$  if $s_n/x_n$ is bounded above and below by  strictly positive constants for all sufficiently large $n$.

% !TEX root =  sc_sparc_journal_paper_full.tex

\section{Spatially coupled SPARC construction} \label{sec:sc_AMP}

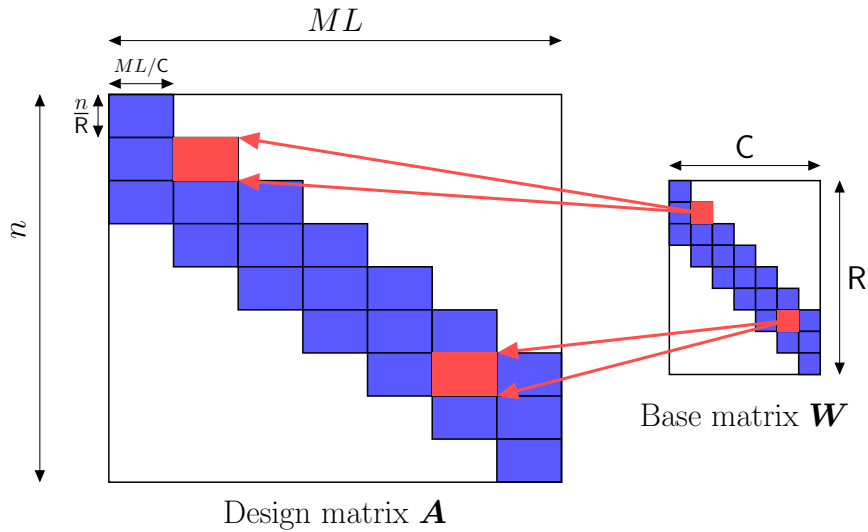
\begin{figure}[t]
\centering
    \begin{tikzpicture}[scale=0.4, font=\LARGE]
	
	% Design matrix
        %% Blocks (put blocks before border so they do not cover up)
	\foreach \x in {0,1, 2, 3, 4, 5, 6}{
		\draw[fill=blue!65, thick] (10+3*\x, {-2*\x}) rectangle (13+3*\x, {-2-2*\x});
		\draw[fill=blue!65, thick] (10+3*\x, {-2-2*\x}) rectangle (13+3*\x, {-4-2*\x});
		\draw[fill=blue!65, thick] (10+3*\x, {-4-2*\x}) rectangle (13+3*\x, {-6-2*\x});
	}

    	%% Border
        \draw[thick] (10, 0) -- (31, 0) -- (31, -18) -- (10, -18) -- cycle;
        %% Text
	\draw (20.5, -18.5) node[anchor=north] {Design matrix $\boldsymbol{A}$};
%	\draw (14, -14) node[font=\Huge] {$0$};
%	\draw (27, -4) node[font=\Huge] {$0$};
	
        %% Dimensions
        \draw[triangle 45-triangle 45] (9.5, 0) -- (9.5, -2);
        \draw (9.5, -1) node[anchor=east] {$\frac{n}{\Lr}$};
        \draw[triangle 45-triangle 45] (6.8, 0) -- (6.8, -18);
        \draw (5.8, -5.5) node[anchor=east, rotate=90] {$n$};
        \draw[triangle 45-triangle 45] (10, 0.5) -- (13, 0.5);
        \draw (11.5, 0.5) node[anchor=south] {\small{$ML/\Lc$}};
        \draw[triangle 45-triangle 45] (10, 2.5) -- (31, 2.5);
        \draw (20.5, 2.6) node[anchor=south] {$\M L$};  	
                     
  	% Base matrix W
	%% Border
        \draw[thick] (36, -4) -- (43, -4) -- (43, -13) -- (36, -13) -- cycle; 
        %% Text
        \draw (39.5, -14) node[anchor=north] {Base matrix $\boldsymbol{W}$};
                
        %% Dimensions
        \draw[triangle 45-triangle 45] (43.9, -4) -- (43.9, -13);
        \draw (43.9, -8.5) node[anchor=west] {$\Lr$};
        \draw[triangle 45-triangle 45] (36, -3.3) -- (43, -3.3);
        \draw (39.5, -3.3) node[anchor=south] {$\Lc$};
                
%        %% Draw gridlines
%	\foreach \x in {1,...,7}
%		\draw[thick, color=black!25] ({\x+6},-20.1) -- ({\x+6},-27.9);
%	\foreach \y in {-21,...,-27}
%		\draw[thick, color=black!25] (6.1,\y) -- (13.9,\y);
			
 	%% Fill in band-diagonal
	\foreach \x in {0,1, 2, 3, 4, 5, 6}{
		\draw[fill=blue!65, thick] (36+\x, {-4-\x}) rectangle (37+\x, {-5-\x});
		\draw[fill=blue!65, thick] (36+\x, {-5-\x}) rectangle (37+\x, {-6-\x});
		\draw[fill=blue!65, thick] (36+\x, {-6-\x}) rectangle (37+\x, {-7-\x});
	}
	
	% Lifting
        	\fill[red!70] (37,-5) rectangle (38,-6);
        	\fill[red!70] (13,-2) rectangle (16,-4);        
	\draw[-triangle 45, red!70, ultra thick] (37.5, -5.5) -- (16, -2);
	\draw[-triangle 45, red!70, ultra thick] (37.5, -5.5) -- (16, -4);
	
	\fill[red!70] (36+5, {-5-5}) rectangle (37+5, {-6-5});
        	\fill[red!70] (25,-12) rectangle (28,-14);
	\draw[-triangle 45, red!70, ultra thick] (41.5, -10.5) -- (28, -12);
	\draw[-triangle 45, red!70, ultra thick] (41.5, -10.5) -- (28, -14);
    \end{tikzpicture}
\caption{\small An $n \times ML$ spatially coupled design matrix $\Amat$ consists of $\Lr \times \Lc$ blocks, each of size $\frac{n}{\Lr} \times \frac{ML}{\Lc}$.
The  entries of $\Amat$ are independent and normally distributed with $A_{ij} \sim \mathcal{N}(0,\frac{1}{L}\Wricj)$, where $\Wmat$ is the base matrix. The base matrix shown here is an $(\omega, \Lambda, \rho)$ base matrix with parameters $\omega=3, \Lambda=7$ and $\rho=0$. The white parts of $\Amat$ and $\Wmat$ correspond to zeros.}
\label{fig:sparc_scmatrix}
\vspace{-6pt}
\end{figure}

As in the standard construction, a spatially coupled (SC) SPARC  is defined by a design matrix $\Amat$ of dimension  $n\times \M L$, where $n$ is the code length.  The codeword is $\xvec = \Amat\betavec$, where $\betavec$ has one non-zero entry in each of the $L$ sections.  In an SC-SPARC, since the variances of the entries of  different blocks of $\Amat$ can be varied, without loss of generality we will set the value of each non-zero entry of $\betavec$ to $1$. 

 In an SC-SPARC, the matrix $\Amat$ consists of independent zero-mean normally distributed entries whose variances are specified by a \emph{base matrix} $\Wmat$ of dimension $\Lr \times \Lc$. The design matrix $\Amat$ is obtained from the base matrix $\Wmat$ by replacing each entry $W_{\sfr \sfc}$ by an $(n/\Lr) \times (ML/\Lc)$ block with i.i.d.\ entries $\sim \mc{N}(0, W_{\sfr \sfc}/{L} )$, for $\sfr \in[\Lr]$, $\sfc \in [\Lc]$.  This is analogous to the ``graph lifting'' procedure for constructing spatially coupled LDPC codes from protographs \cite{mitchell2015}. See Fig.\ \ref{fig:sparc_scmatrix} for an example. 

From the construction, the design matrix $\Amat$ has independent normal entries
\be\label{eq:construct_Aij}
A_{ij} \sim \mc{N}\Big(0,\frac{1}{L} \Wricj \Big) \quad   i \in [n], \ j\in[\M L].
\ee
The operators $\sfr(\cdot):[n]\rightarrow[\Lr]$ and $\sfc(\cdot):[\M L]\rightarrow[\Lc]$ in \eqref{eq:construct_Aij} map a particular row or column index in $\Amat$ to its corresponding \emph{row block} or \emph{column block} index in 
$\Wmat$. We require  $\Lc$ to divide $L$, resulting in $\frac{L}{\Lc}$ sections per column block.
%\RV{Removed the operators $\sfR(\cdot)$ and $\sfC(\cdot)$}

The non-zero coefficients of $\betavec$  are all set to $1$.
% i.e.,  \be c_1 = c_2 = \ldots = c_L = 1.  \ee
Then it can be verified that  $\expec[ \norm{\xvec}^2] = nP$ (and the power constraint is  satisfied with high probability) if  the entries of the base matrix $\Wmat$ satisfy 
\be
\label{eq:W_power_constraint}
\frac{1}{\Lr \Lc}\sum_{\sfr=1}^{\Lr}  \sum_{\sfc=1}^{\Lc} W_{\sfr \sfc} = P
\ee

The trivial  base matrix with $\Lr=\Lc=1$  corresponds to a standard (non-coupled) SPARC  with uniform power allocation, while a base matrix consisting of a single row $\Lr=1$, $\Lc=L$   is equivalent to a standard SPARC with power allocation. In this paper, we will use the following base matrix  inspired by the coupling structure of SC-LDPC codes constructed from protographs \cite{mitchell2015}.
\begin{defi}
\label{def:ome_lamb}
An $(\omega , \Lambda, \rho)$ base matrix $\Wmat$ for SC-SPARCs is  described by three parameters: coupling width $\omega\geq1$ coupling length $\Lambda\geq 2\omega-1$, and $\rho \in [0,1)$ which determines the fraction of power allocated to the coupled entries in each column. The matrix has $\Lr=\Lambda+\omega-1$ rows,  $\Lc=\Lambda$ columns, with each column having $\omega$ identical non-zero entries. For an average power constraint $P$, the  $(\sfr,\sfc)$th entry of the base matrix, for  $\sfr \in [\Lr], \sfc\in[\Lc]$, is given by
\begin{equation}\label{eq:W_rc}
W_{\sfr \sfc} =
\begin{cases}
 	\ (1-\rho)P \cdot \frac{\Lambda+\omega-1}{\omega} \quad &\text{if} \ \sfc \leq \sfr \leq \sfc+\omega-1,\\
	\ \rho P \cdot \frac{\Lambda + \omega -1}{\Lambda-1} \quad &\text{otherwise}.
\end{cases}
\end{equation}
\end{defi}
It is easy to verify that this definition satisfies the power allocation constraint in \eqref{eq:W_power_constraint}. For example, the base matrix in Fig.\ \ref{fig:sparc_scmatrix} has parameters $\omega=3$ and $\Lambda=7$.
 For our simulations in Section \ref{sec:emp_perf}, we use $\rho= 0$, whereas for our main theoretical result (Theorem \ref{thm:main}) we choose $\rho$ to be a small positive value proportional to the rate gap from capacity. (Choosing $\rho =0$ causes some technical difficulties in the proof, which can be addressed by picking a suitable $\rho >0$.)
The $(\omega, \Lambda, \rho=0)$ base matrix construction was previously used for SC-SPARCs in \cite{liang2017}. Other base matrix constructions can be found in \cite{krzakala2012,donoho2013,barbier2016,barbier2017}.

Each non-zero entry in a base matrix $\Wmat$ corresponds to an $(n/\Lr) \times (\M L/\Lc)$ block in the design matrix $\Amat$. Each block can be viewed as a standard (non-coupled) SPARC with $\frac{L}{\Lc}$ sections (with $\M$ columns in each section), code length $n/\Lr$, and rate $R_\text{inner} = \frac{(L/\Lc) \ln \M}{(n/\Lr)}$ nats. Using \eqref{eq:rate_eq}, the overall rate of the SC-SPARC is related to $R_\text{inner}$ according to
\be\label{eq:R_Rsparc}
R  = \frac{\Lc}{\Lr} \, R_\text{inner}= \frac{\Lambda}{\Lambda + \omega -1} \, R_\text{inner},
\ee
where the last equality holds for an $(\omega, \Lambda, \rho)$  base matrix. 

With spatial coupling, $\omega$ is an integer greater than 1, so $R < R_\text{inner}$.  The difference $(R_{\text{inner}} -R)$ is sometimes referred to as the rate loss due to spatial coupling. From \eqref{eq:R_Rsparc}, we see that rate loss depends on the ratio $(\omega-1)/\Lambda$, which becomes negligible when $\Lambda$ is large w.r.t.\ $\omega$. For our theoretical results, we will be interested in the regime where $L \gg  \Lc = \Lambda \gg \omega$. Without loss of generality, we will henceforth assume that $\omega < \sqrt{\Lambda}$.

\begin{rem}
\normalfont
SC-SPARC constructions usually have a `seed' to jumpstart decoding. In \cite{barbier2016}, a small fraction of $\betavec$'s sections  are fixed a priori --- this pinning condition is used to analyze the state evolution equations via the potential function method. Analogously, the construction in \cite{barbier2017}  introduces additional rows in the design matrix for blocks corresponding to the first row of the base matrix. In an $(\omega, \Lambda, \rho)$ base matrix,  the fact that the number of rows in the base matrix exceeds the number of columns by $(\omega - 1)$ helps decoding start from both ends. The asymptotic state evolution equations  in Sec.\ \ref{sec:AMP_dec_prog}  describe  how AMP decoding progresses in an $(\omega, \Lambda, \rho)$ base matrix.
\end{rem}

In the remainder of the paper, we use subscripts in sans-serif font ($\sfr$ or $\sfc$) to denote row or column block indices. Thus, $\betavec_\sfc \in \reals^{ML/\Lc}$ denotes the $\sfc$-th column block of $\betavec \in \reals^{ML}$, for $\sfc \in [\Lc]$.

\section{The AMP decoder for spatially coupled SPARCs}\label{sec:AMP}

Recall that the decoding task is to recover the message vector $\betavec\in \reals^{\M L}$ from the channel output sequence $\yvec\in \reals^n$  produced according to \eqref{eq:linear_model}. An Approximate Message Passing (AMP) decoding algorithm can be derived using an approach similar to the one for standard SPARCs \cite[Appendix A]{rush2017}, with modifications to account for the different variances for the blocks of $\Amat$ specified by the base matrix.
The AMP decoder can also be derived from the Generalized AMP algorithm in  \cite{rangan2011}.

  The AMP decoder initializes  $\betavec^0$ to the all-zero vector, and for $t \geq 0$, iteratively computes:
\begin{align}
\zvec^t &= \yvec - \Amat\betavec^t + {\upvec}^t \odot \zvec^{t-1}, \label{eq:scamp_decoder_z} \\
\betavec^{t+1} &= \eta^t(\betavec^t + (\overline{\Smat}^t \odot \Amat)^* \zvec^t).  \label{eq:scamp_decoder_beta}
\end{align}
Here $\odot$ denotes the Hadamard (entry-wise) product. The vector  $\upvec^t \in \reals^n$, the matrix $\overline{\Smat}^t \in \reals^{n \times ML}$, and the  denoising function $\eta^t$  are defined below in terms of the state evolution parameters.  Quantities with negative time indices are set to zero. 

For any rate $R < \mathcal{C}$, the AMP decoder is run for a finite number of iterations $T$, where $T$ is specified later in Section~\ref{sec:AMP_dec_prog}.  After $T$ iterations, the maximum value in each section $\ell \in [L]$ of $\betavec^T$ is set to $1$ and remaining entries are set to $0$ to obtain the decoded message $\widehat{\betavec}$.

\subsection{State evolution} 
Given a base matrix $\Wmat$, state evolution (SE) iteratively defines a sequence of scalars $(\phi^t_\sfr)_{\sfr \in \Lr}$ and $(\psi^t_\sfc)_{\sfc \in \Lc}$, for $t \geq 0$. Initialize $\psi_\sfc^{0} = 1$ for  $\sfc \in[\Lc]$, and for $t=0,1,\ldots$, compute
\begin{align}
\sigma_{\sfr}^t & =  \frac{1}{\Lc}\sum_{\sfc=1}^{\Lc}W_{\sfr \sfc}\psi_\sfc^t \, , \qquad \phi_\sfr^t = \sigma^2 + \sigma_{\sfr}^t, \qquad \sfr \in [\Lr], \label{eq:se_phi} \\
\psi_\sfc^{t+1} &= 1 - \mathcal{E}(\tau_\sfc^t), \qquad  \qquad \sfc \in [\Lc],  \label{eq:se_psi}
\end{align}
%%%%%
where \be 
\tau_\sfc^t = \frac{R}{\ln{\M}}\left[\frac{1}{\Lr}\sum_{\sfr=1}^{\Lr} \frac{W_{\sfr \sfc}}{\phi_\sfr^t}\right]^{-1}, 
\label{eq:tau_ct_def}
\ee and
$\mathcal{E}(\tau_\sfc^t)$ is defined  with $U_1,\ldots,U_{\M} \stackrel{\text{i.i.d.}}{\sim} \stdnorm$ as 
\be
\label{eq:E_tau}
\mathcal{E}(\tau_\sfc^t) 
= \mathbb{E}\left[ \frac{e^{U_1/\sqrt{\tau_\sfc^t}}}{e^{U_1/\sqrt{\tau_\sfc^t}} + e^{ {-1}/{\tau_\sfc^t}}\sum_{j=2}^{\M} e^{U_j/\sqrt{\tau_\sfc^t}} }\right].
\ee
%
%%%%%%

For $t \geq 1$, the vector $\upvec^t \in\mathbb{R}^{n}$ in \eqref{eq:scamp_decoder_z} has a block-wise structure, with the $i$th entry defined as
\be
\upsilon^t_{i} = \frac{\sigma_{\sf{r}}^{t}}{\phi_{\sf{r}}^{t-1}}, \quad \text{ if } \sfr(i) = \sfr,
\label{eq:onsager_def}
\ee 
where we recall that $\sfr(i)$ denotes the row block index of the $i$th entry. (The vector $\upvec^0$ is defined to be all-zeros.)
Similarly,  $\overline{\Smat}^t \in \mathbb{R}^{n \times ML}$ in \eqref{eq:scamp_decoder_beta} has a block-wise structure, with entries defined as follows. For $i \in [n], j \in [ML]$,
\be
\overline{S}^t_{ij} = \frac{\tau^t_{\sf{c}}}{ \phi_{\sf{r}}^t}, \quad \text{ if } \sfr(i) = \sfr \text{ and } 
\sfc(j) = \sfc.
\label{eq:olS_def}
\ee

 The  function  $\eta^t= (\eta^t_1, \ldots, \eta^t_{ML}): \,  \mathbb{R}^{\M L}  \to \mathbb{R}^{\M L}$ in \eqref{eq:scamp_decoder_beta} is defined as follows, for $j \in [\M L]$.  For $j$ in section $\ell \in [L]$, with section $\ell$  in column block $\sfc \in [\Lc]$, 
\be\label{eq:eta_function}
\eta^t_j(\svec)
= \frac{e^{s_j/\tau^t_\sfc}}{\sum_{j'\in \text{sec}(\ell)}e^{s_{j'}/\tau^t_\sfc }},
\ee
where $\text{sec}(\ell) := \{ (\ell-1)M+1, \ldots, \ell M \}$ refers to the set of indices in section $\ell$.
We note that $\eta^t_j(\svec)$ depends on all the components of $\svec$ in the section containing $j$.
%%%
%%%%

\subsection{Interpretation of the AMP decoder} \label{subsec:AMP_int}
The input to  $\eta^t_j(\cdot)$ in \eqref{eq:eta_function} can be viewed as a noisy version of $\betavec$. In particular, consider an index $j$ in section $\ell \in [L]$ which belongs to column block $\sfc \in [\Lc]$.  Recall that $\betavec_{\ell} \in\mathbb{R}^{M}$ is section $\ell$ of the message vector, and let  $\svec^t_\ell = \svec_\ell$ denote section $\ell$ of the input vector to the function $\eta^t_j(\cdot)$. Then, $\svec^t_\ell$  is approximately distributed as 
$\betavec_\ell + \sqrt{\tau_\sfc^t}\Zvec_\ell$, where $\Zvec_\ell \in \reals^{M}$ is a standard normal random vector independent of $\betavec_\ell$. 
   Under the above distributional assumption, the denoising function $\eta^t_j$ in \eqref{eq:eta_function} is the  minimum mean squared error (MMSE) estimator  for  $\betavec_{j}$, i.e., 
\be
 \eta^t_j(\svec )=\mathbb{E}\left[\betavec_j \mid \betavec_\ell + \sqrt{\tau_\sfc^t} \, \Zvec_\ell = \svec_\ell \right], \qquad \text{ for } j \in [ML],
 \label{eq:MMSE_etaj}
 \ee
where the expectation is calculated over and  $\Zvec_\ell \sim \normal(\bzero, \iden_{M})$ and $\betavec_\ell$, which is uniformly distributed over the $M$ vectors with a single non-zero entry equal to 1.

The entries of the modified residual $\zvec^t$ in \eqref{eq:scamp_decoder_z} are approximately Gaussian and independent, with the variance determined by the block index. (A precise characterization of the distribution is given in Lemmas \ref{lem:hb_cond} and \ref{lem:ideal_cond_dist}.) For $\sfr \in [\Lr]$,  the SE parameter $\phi^t_\sfr$ approximates the variance of 
$\zvec^t_{\sfr}$, the $\sfr$th block of the residual.  The `Onsager' term $\upvec^t \odot \zvec^{t-1}$  in \eqref{eq:scamp_decoder_z} reflects the block-wise structure of $\zvec^t$.  To summarize, the key difference from the  state evolution parameters for standard SPARCs  is that here the variances of the effective observation $\svec^t$ and the modified residual $\zvec^t$ depend on their column- and row-block indices, respectively.   These variances are captured by $\{ \tau_\sfc^t \}_{\sfc \in [\Lc]}$ and 
$\{ \phi_\sfr^t \}_{\sfr \in [\Lr]}$.

%%%

\subsection{Measuring the performance of the AMP decoder} \label{subsec:MSE_SE}
The performance of a SPARC decoder is  measured by the \emph{section error rate},  defined as
\begin{equation}\label{eq:ser_def}
\mathcal{E}_{\text{sec}} := \frac{1}{L}\sum_{\ell=1}^L \indic\{\widehat{\betavec}_{\text{sec}(\ell)} \neq \betavec_{\text{sec}(\ell)} \}.
\end{equation}
 If the AMP decoder is run for $T$ iterations, the section error rate can be bounded in terms of the squared error $\| \betavec^T - \betavec \|^2$ as follows.  Since the unique non-zero entry in any section $\ell \in [L]$ of $\betavec$ equals $1$ and $\widehat{\betavec}_{\text{sec}(\ell)} \neq \betavec_{\text{sec}(\ell)}$ implies that the corresponding element of $\betavec^{T}_{\text{sec}(\ell)}$ is less than or equal to $1/2$,
\begin{equation}
\label{eq:SER_to_MSE}
\widehat{\betavec}_{\text{sec}(\ell)} \neq \betavec_{\text{sec}(\ell)} \quad \Rightarrow \quad \|\betavec^{T}_{\text{sec}(\ell)} - \betavec_{\text{sec}(\ell)} \|_2^2 \geq \frac{1}{4}.
\end{equation}
We recall that $\betavec_{\sfc}$ is the part of the message vector corresponding to column block $\sfc$ of the design matrix. There are $\frac{L}{\Lc}$ sections in $\betavec_{\sfc}$, with the non-zero entry in each section being equal to  $1$; we denote by $\betavec_{\sfc_\ell}$ the $\ell$th of these sections, for $\ell \in [L/\Lc]$. Then, \eqref{eq:SER_to_MSE} implies
\begin{align}
\mathcal{E}_{\text{sec}} & =   \frac{1}{L}\sum_{\ell=1}^L \indic\{\widehat{\betavec}_{\text{sec}(\ell)} \neq \betavec_{\text{sec}(\ell)} \} \nonumber 
  =   \frac{1}{L} \sum_{\sfc=1}^{\Lc} \sum_{\ell=1}^{L/\Lc} \indic\left\{\widehat{\betavec}_{\sfc_\ell} \neq \betavec_{\sfc_\ell} \right\}  \nonumber \\
&  \leq   \frac{4}{L} \sum_{\sfc=1}^{\Lc} \sum_{\ell=1}^{L/\Lc} \|\betavec^{T}_{\sfc_\ell} - \betavec_{\sfc_\ell}\|_2^2  
  =  4 \left[ \frac{1}{\Lc}\sum_{\sfc=1}^{\Lc} \frac{\|\betavec^{T}_{\sfc} - \betavec_{\sfc} \|_2^2}{L/\Lc}  \right]
  = 4 \left[ \frac{1}{L}\| \betavec^T - \betavec \|^2 \right].\label{eq:mse_LB_w_ser}
 \end{align}

We can therefore focus on bounding the bracketed term on the RHS of \eqref{eq:mse_LB_w_ser} which is the  overall  \emph{normalized mean square error} (NMSE). Fig.\  \ref{fig:ser_waveprop} shows that $\psi^t$ (\eqref{eq:se_psi}) closely tracks the NMSE of each block of the message vector, i.e.,  $\psi^t_\sfc \approx \frac{\|\betavec_{\sfc}^t-\betavec_{\sfc}\|_2^2}{L/\Lc}$ for $\sfc\in[\Lc]$.   
We additionally observe from the figure that as AMP iterates, the NMSE  reduction propagates from the ends towards the center blocks.  

\begin{figure}[t]
\centering
\includegraphics[width=4in]{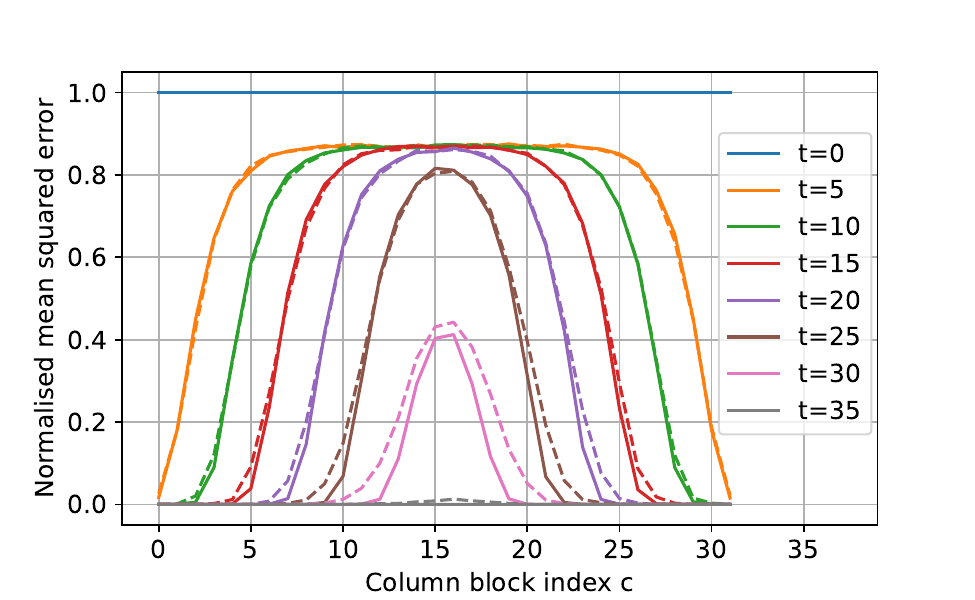}
\caption{ \small NMSE $\frac{\|\betavec_{\sfc}^t-\betavec_{\sfc}\|_2^2}{L/\Lc}$ vs.\ column block index $\sfc \in [\Lc]$ for several iteration numbers. The SC-SPARC with an $(\omega,\Lambda, \rho=0)$ base matrix has the following parameters: $R=1.5$ bits, $\mc{C}=2$ bits, $\omega=6$, $\Lambda=32$, $M=512$, $L=2048$ and $n=12284$. The solid lines are the state evolution predictions from \eqref{eq:se_psi}, and the dotted lines are the average NMSE over 100 instances of AMP decoding.}
\label{fig:ser_waveprop}
\vspace{-5pt}
\end{figure}

% !TEX root =  sc_sparc_journal_paper_full.tex

\section{Decoding progression according to state evolution} \label{sec:AMP_dec_prog}

In this section, we derive bounds for the state evolution parameters which help explain the decoding propagation illustrated in Fig.\ \ref{fig:ser_waveprop}.  These bounds lead to a succinct asymptotic characterization of state evolution (as $M \to \infty$). The non-asymptotic version of these bounds (for large but finite $M$) will be used to establish the main result in Theorems \ref{thm:main} and \ref{thm:main1}.

\begin{lemma}
Let $\Wmat \in \reals^{\Lr \times \Lc}$ be a base matrix having row and column averages that are bounded above and below by strictly positive constants. That is, there exist constants
$\kappa_{\rm L}, \kappa_{\rm U} >0$ such that 
\[ \kappa_{\rm L} \leq  \  \frac{1}{\Lc} \sum_{\sfc'} W_{\sfr \sfc'}, \ \frac{1}{\Lr} \sum_{\sfr '} W_{\sfr' \sfc}  \leq \kappa_{\rm U}, \qquad \sfr \in [\Lr], \ \sfc \in [\Lc].  \] 
Let
\be
\nu_\sfc^t := \frac{1}{\tau_\sfc^t \ln{\M}}  =  \frac{1}{R} \left[\frac{1}{\Lr}\sum_{\sfr=1}^{\Lr} \frac{W_{\sfr \sfc}}{\phi_\sfr^t}\right].
\label{eq:nu_c^t} 
\ee
For sufficiently large $\M$ and any $\delta\in(0,\frac{1}{2})$, $\tilde{\delta} \in (0, 1)$, 
\be
\left( 1 - M^{-k_1 \tilde{\delta}^2} \right)\indic\{\nu_\sfc^t < 2 - \tilde{\delta} \}  \,  \leq  \, \psi_\sfc^{t+1} \leq 1- \left(1- \frac{\M^{-k \delta^2}}{\delta\sqrt{\log{M}}}   \right) \indic\{\nu_\sfc^t > 2+\delta\},    \qquad \sfc \in [\Lc \, ],
\label{eq:psi_c_bound}
\ee
where $k, k_1$ are positive constants  depending only on $\kappa_{\rm L}$ and $\kappa_{\rm U}$.
\label{lem:psi_nonasymp}
\end{lemma}
The proof of the lemma is given in Appendix \ref{app:psi_nonasymp}.
%%%

Lemma \ref{lem:psi_nonasymp} implies the following asymptotic state evolution recursion as $M \to \infty$. Initialise $\bar{\psi}_\sfc^{0} = 1$, for $\sfc \in[\Lc]$, and for $t=0,1,2,\ldots$,
\begin{align}
\bar{\phi}_\sfr^t &= \sigma^2 + \frac{1}{\Lc}\sum_{\sfc=1}^{\Lc}W_{\sfr \sfc}\bar{\psi}_\sfc^t, \qquad \sfr \in [\Lr], \label{eq:se_asmyp_phi} \\
\bar{\psi}_\sfc^{t+1} &= 1 - \indic\left\{ \frac{1}{R\, \Lr}\sum_{\sfr=1}^{\Lr} \frac{W_{\sfr \sfc}}{\bar{\phi}_\sfr^t}>2 \right\},  \qquad \sfc \in [\Lc], \label{eq:se_asmyp_psi}
\end{align}
where $\bar{\phi}, \bar{\psi}$ indicate asymptotic values as $M \to \infty$.

The asymptotic SE recursion \eqref{eq:se_asmyp_phi}-\eqref{eq:se_asmyp_psi} is given for a general base matrix $\Wmat$. To get some insight into the decoding progression, we specialise the result to the $(\omega,\Lambda, \rho=0)$ base matrix introduced in Definition \ref{def:ome_lamb}. Recall that an $(\omega,\Lambda, \rho=0)$ base matrix has $\Lr=\Lambda+\omega-1$ rows and $\Lc=\Lambda$ columns, with each column having $\omega$ non-zero entries, all equal to $P \cdot \frac{\Lambda+\omega-1}{\omega}$. 

\begin{corr}\label{corr:se_asymp_wLbasematrix}
The asymptotic state evolution recursion \eqref{eq:se_asmyp_phi}-\eqref{eq:se_asmyp_psi} for an $(\omega,\Lambda, \rho=0)$ base matrix is as follows. Initialise $\bar{\psi}_\sfc^{0} = 1 \ \forall \ \sfc \in[\Lambda ]$, and for $t=0,1,2,\ldots$,
\begin{align}
\bar{\phi}_\sfr^t &= \sigma^2  + \frac{\vth \, P}{\omega} \sum_{\sfc= \underline{\sfc}_\sfr }^{\overline{\sfc}_\sfr} \bar{\psi}_\sfc^t, \quad \sfr \in [ \Lambda + \omega -1], \label{eq:se_asmyp_flat_phi} \\
\bar{\psi}_\sfc^{t+1} &= 1 - \indic\left\{ \frac{P}{ R \, \omega}\sum_{\sfr=\sfc}^{\sfc+\omega-1} \frac{1}{\bar{\phi}_\sfr^t} > 2 \right\}, \quad \sfc \in [ \Lambda ],  \label{eq:se_asmyp_flat_psi}
\end{align}
where $\vth = \frac{\Lambda+\omega-1}{\Lambda}$,  and
\be\label{eq:c_r}
(\underline{\sfc}_\sfr,\, \overline{\sfc}_\sfr)
	=\begin{cases}
  		(1,\, \sfr) \ & {\normalfont \text{if}}  \  \ 1\leq \sfr\leq\omega\\
		(\sfr-\omega+1,\, \sfr) \ & {\normalfont \text{if}}  \ \  \omega \leq \sfr \leq \Lambda\\
	        (\sfr-\omega+1,\, \Lambda) \ & {\normalfont \text{if}} \  \ \Lambda \leq \sfr \leq \Lambda + \omega - 1.
	\end{cases}
\ee
\end{corr}
\begin{proof}
Substitute the value of $W_{\sfr \sfc}$ from \eqref{eq:W_rc}, with $\rho=0$ and $\Lc=\Lambda$, $\Lr=\Lambda+\omega-1$ in \eqref{eq:se_asmyp_phi}-\eqref{eq:se_asmyp_psi}.
\end{proof}

Observe that the $\bar{\phi}_\sfr^t$'s and $\bar{\psi}_\sfc^t$'s are symmetric about the middle indices, i.e.\ $\bar{\phi}_\sfr^t = \bar{\phi}_{\Lr- \sfr+1}^t$ for $\sfr\leq \lfloor \frac{\Lr}{2} \rfloor$ and $\bar{\psi}_\sfc^t = \bar{\psi}_{\Lc-\sfc+1}^t$ for $\sfc\leq \lfloor \frac{\Lc}{2} \rfloor$.

Consider the initial step ($t=0$): from \eqref{eq:se_asmyp_flat_phi}  the value of $\bar{\phi}_\sfr^0$  for each $\sfr$ depends on the number of non-zero entries in row $\sfr$ of $\Wmat$, which is equal to $\overline{\sfc}_\sfr - \underline{\sfc}_\sfr + 1$, with $\overline{\sfc}_\sfr, \underline{\sfc}_\sfr$ given by \eqref{eq:c_r}. Therefore, $\bar{\phi}_\sfr^0$ increases from  $\sfr=1$ until $\sfr=\omega$, is constant for $\omega\leq \sfr \leq \Lambda$, and then starts decreasing again for $\Lambda < \sfr \leq \Lambda + \omega -1$. As a result,  $\bar{\psi}_\sfc^{1}$ is smallest for $\sfc$ at either end of the base matrix ($\sfc\in\{1,\Lambda\}$) and increases as $\sfc$ moves towards the middle, since the $\sum_{\sfr=\sfc}^{\sfc+\omega-1} (\bar{\phi}_\sfr^0)^{-1}$ term in \eqref{eq:se_asmyp_flat_psi} is largest  for $\sfc\in\{1,\Lambda\}$, followed by $\sfc\in\{2, \Lambda-1\}$, and so on. Therefore, we expect the blocks of the message vector corresponding to column index $\sfc\in\{1,\Lambda\}$ to be decoded most easily, followed by $\sfc\in\{2, \Lambda-1\}$, and so on. Fig.\ \ref{fig:ser_waveprop} shows that this is indeed the case. 

The decoding propagation phenomenon seen in Fig.\ \ref{fig:ser_waveprop} can also be explained using Corollary \ref{corr:se_asymp_wLbasematrix} by tracking the evolution of the $\bar{\phi}_\sfr^t$'s and $\bar{\psi}_\sfc^t$'s. In particular, one finds that if column $\sfc^*$ decodes in iteration $t$, i.e.\ $\bar{\psi}_{\sfc^*}^t=0$, then columns within a coupling width away, i.e.\ columns $\sfc\in\{\sfc^*-(\omega-1), \ldots, \sfc^*+ (\omega-1)\}$, will become easier to decode in iteration $(t+1)$.  This wave-like decoding propagation also occurs in spatially coupled LDPC codes decoded with belief propagation. The propagation of the LDPC decoding wave (in the large system limit) was studied in \cite{el2018velocity}.

\subsection{Decoding progression} \label{subsec:dec_prog}
We make the above discussion precise by characterizing the decoding progression for an $(\omega, \Lambda, \rho)$ base matrix (with $\rho >0$)  using Lemma \ref{lem:psi_nonasymp}.  Recalling that
 \be \vth = 1 + \frac{\omega-1}{\Lambda},  \label{eq:vartheta_def} 
 \ee 
 we will consider rates $R$ such that 
\be
R <  \frac{1}{2 \vth} \ln (1 + \vth \, \snr), \quad  \text{ where }  \quad \snr=\frac{P}{\sigma^2}.
\label{eq:kappa_ub}
\ee
Note that the RHS of \eqref{eq:kappa_ub} can be made arbitrarily close to the channel capacity $\mc{C}$ by making
$\frac{\omega-1}{\Lambda}$ small enough.  Indeed,  since the expression in \eqref{eq:kappa_ub} is decreasing in $\vth$ for  $\vth >1$, we have
\be
\label{eq:ckappa_bnd}
\begin{split}
& \mc{C} >   \frac{1}{2 \vth} \ln (1 + \vth \, \snr)  > \frac{\mc{C}}{\vth}.
\end{split}
\ee
%%%
\begin{prop}
Consider a rate $R$ SC-SPARC with an $n \times ML$ design matrix constructed using an $(\omega, \Lambda, \rho)$ base matrix
and a constant $\delta \in(0, \min\{ \frac{\Delta}{2R}, \frac{1}{2} \})$, 
where $0 \leq \rho \leq \min\{ \frac{\Delta}{3 \snr}, \, \frac{1}{2}\}$, and
\be 
\Delta := \frac{1}{2\vth}\ln(1 + \vth \snr) - R.
\label{eq:Delta}
\ee
If the rate satisfies $R < \frac{(1- \rho)\snr}{(2 + \delta)(1 + \vth \snr)}$, then all the column blocks of the message vector simultaneously decode in one iteration, i.e., for all $\sfc \in [\Lambda]$,
\be
\psi_\sfc^1 \leq f_{M,\delta}:= \frac{\M^{-k \delta^2}}{\delta\sqrt{\log{M}}}
\label{eq:fmdelta_def}
\ee
for sufficiently large $M$, where $k>0$ is a universal constant.

Otherwise,  if the rate satisfies $\frac{(1- \rho)\snr}{(2 + \delta)(1 + \vth \snr)} \leq R < \frac{1}{2 \vth} \ln (1 + \vth \, \snr)$,
the coupling width $\omega$ satisfies
\be\label{eq:omega_thresh}
\omega > \left( \frac{\vth \,  \snr^2}{1+ \vth\snr} \right)\frac{1}{\Delta}, 
\ee
and
\be
g = \frac{(1 + \vth \, \snr)\Delta }{\vth \, \snr^2} \,  \omega, \label{eq:g_def} 
\ee
then, for $t \geq 1$ and 
\be 
\sfc \leq \max \left\{tg , \ \left\lceil \frac{\Lambda}{2} \right\rceil \right\},
\ee  
we have 
\be 
\psi_{\sfc}^t  \, = \,  \psi_{\Lambda - \sfc + 1}^t  \,  \leq \, f_{M, \delta}
\label{eq:psic_prog}
\ee
for sufficiently large $M$.
\label{prop:nonasymp_se}
\end{prop}
The proof of the proposition is given in Appendix \ref{app:nonasymp_se}.

\begin{rem}[BP threshold]
\normalfont
For rates $R$ smaller than $\frac{\snr}{2(1+\snr)}$, one does not require spatial coupling (or power allocation) for reliable SPARC decoding.
Indeed, consider a standard non-coupled SPARC where the 1-by-1 base matrix is a single entry equal to $P$. Using Lemma \ref{lem:psi_nonasymp} in the state evolution recursion \eqref{eq:se_phi}-\eqref{eq:E_tau}, we see that if $R < \frac{\snr}{(2 + \delta)(1 + \snr)}$ for some $\delta \in (0,1)$, then the whole message vector decodes in one iteration, i.e., $\psi^1 \leq f_{M,\delta}$. The threshold $\frac{\snr}{2(1+\snr)}$ can be interpreted as the BP threshold in the $M \to \infty$ limit.
\label{rem:full_decoding}
\end{rem}

\begin{rem}[Choice of base matrix parameters]
\normalfont
For any fixed rate $\frac{\snr}{2(1+\snr)} \leq R < \mc{C} = \frac{1}{2} \ln(1+ \snr)$, Proposition \ref{prop:nonasymp_se}  requires the parameters $(\omega, \Lambda)$ to be chosen such that: {\normalfont i)} the ratio $(\omega/\Lambda)$ is small enough that the rate gap $\Delta$ in \eqref{eq:Delta} is positive,  and {\normalfont ii)} $\omega$ is large enough that \eqref{eq:omega_thresh} is satisfied. These two conditions can be satisfied for any fixed $R < \mc{C}$ by taking 
(for example) $\Lambda > \omega^2$, and $\omega$ sufficiently large.
\label{rem:choose_omega_lambda}
\end{rem}

%\begin{rem}
For rates larger than the threshold $\frac{\snr}{2(1+\snr)}$, the proposition says that if the coupling width $\omega$ is large enough (as specified by \eqref{eq:omega_thresh}),  then in iteration $t$ at least the first and last $\lfloor gt \rfloor$ column blocks from each end are expected to decode.  Furthermore, the proof shows that if $g_t \leq \Lambda/2$ is the exact number of column blocks such that $\psi_\sfc^t = \psi_{\Lambda -\sfc+1}^t \leq f_{M, \delta}$ for $\sfc \leq g_t$, then $g_{t+1} \geq \lfloor g_{t} + g \rfloor$, i.e., in each iteration at least $\lfloor g \rfloor  \geq 1$ additional column blocks of the message vector from each end are expected to  decode.

This decoding progression continues until iteration $T$ when all column blocks have been decoded, i.e., $\psi^T_\sfc \leq f_{M,\delta}$ for $\sfc \in [\Lc]$. 
More precisely, we run the AMP decoder for $T$ iterations where
\be
T := \min\{t: \, \psi^t_\sfc  \leq f_{M,\delta} \text{ for } \sfc \in [\Lc] \}.
\label{eq:T_def}
\ee
 Proposition \ref{prop:nonasymp_se} implies that for rates larger than the threshold
\be
T \leq  \left\lceil   \frac{\Lambda}{2 g} \right\rceil. 
\label{eq:Tdef}
\ee

Using the interpretation of the AMP decoder in Sections \ref{subsec:AMP_int} and \ref{subsec:MSE_SE}, after iteration $T$ we expect the  mean  squared error $\frac{1}{L} \| \betavec - \betavec^T \|^2$ to be small.  We note that $g$ is proportional to 
$\Delta$, which represents the rate gap from capacity (see \eqref{eq:ckappa_bnd}, \eqref{eq:Delta}). Therefore, from \eqref{eq:Tdef} the number of iterations $T$ grows as the rate approaches the channel capacity.   For a fixed $R$ the quantity $f_{M,\delta}$ tends to $0$ with growing $M$.

% !TEX root =  sc_sparc_journal_paper_full.tex

\section{Main Result} \label{sec:main_result}

The main result, stated in the theorem below is a bound on the probability of the section error rate of the AMP decoder exceeding a target level $\e$, for any $\e \in (0,1)$.

\begin{thm}
Consider an $(\omega, \Lambda, \rho)$ base matrix $\Wmat \in \reals^{\Lr \times \Lc}$ with $(\Lr, \Lc)=  (\Lambda +\omega-1, \, \Lambda)$ and $\rho = \min\{  \frac{\Delta}{3\snr}, \frac{1}{2} \}$, where  $\Delta$ is the rate gap defined in \eqref{eq:Delta}. Fix rate $\frac{\snr}{2(1+\snr)}  \leq R < \frac{1}{2\vartheta}\ln(1 + \vartheta \, \snr)$, where $\vartheta =1 + \frac{\omega-1}{\Lambda}$, and let $\omega$ be large enough that the condition in  \eqref{eq:omega_thresh} is satisfied.
Let $\mc{S}_n$ be a  SC-SPARC of rate $R$ defined via an $n \times ML$ design matrix constructed using the base matrix $\Wmat$.  The  parameters $(n, M, L)$ satisfy \eqref{eq:rate_eq}.

Fix $\e \in (0,1)$, and  for $f_{M, \frac{\Delta}{R}}$ defined in Proposition \ref{prop:nonasymp_se}, let $M$ be large enough that $f_{M,\delta} \leq \frac{\e}{8}$ for $\delta = \min\{ \frac{\Delta}{3R}, \frac{1}{3} \}$. Then after  the  AMP decoder is run for $T$ iterations, with $T$ defined in \eqref{eq:T_def}, the section error rate (defined in \eqref{eq:ser_def}) satisfies 
\be
\begin{split}
& P \left( \mc{E}_{\text{sec}}(\mc{S}_n)  >  \e \right) \leq K_{T-1} (\Lr \Lc)^{T} \exp  \Big\{ \frac{- \kappa_{T-1} n \e^2}{ 64(\log M)^{2T} (\Lr/\omega)^{2T-1}} \Big\}. 
 \label{eq:pezero}
\end{split}  
\ee
For $t \geq 0$, the constants $\kappa_t$ and $K_t$  are given by $\kappa_t = [\xi^{2t} (t!)^{24}]^{-1}$ and $K_t = \Xi^{2t} (t!)^{14}$ where $\xi, \Xi > 0$ are universal constants (not depending on the AMP parameters $(L, M, n, \Lr, \Lc)$ or $\e$), but not explicitly specified.  
\label{thm:main}
\end{thm}

\begin{rem}
The theorem is stated for rates
$\frac{\snr}{2(1+\snr)} \leq R <  \frac{1}{2\vartheta}\ln(1 + \vartheta \, \snr)$ as this is the region where spatial coupling is required. Indeed, for $R < \frac{\snr}{2(1+\snr)}$, Remark \ref{rem:full_decoding} and the proof of Theorem \ref{thm:main} imply that  the
 probability bound \eqref{eq:pezero} holds with $\Lr=\Lc=\omega=1$ and $T=1$. This result also follows from the analysis in \cite{RushV19}, applied with a uniform power allocation.
\end{rem}

The bound  \eqref{eq:pezero} on the probability of excess section error rate is obtained via a concentration inequality on the normalized MSE  of the AMP decoder.   Theorem \ref{thm:main1} below gives a concentration inequality bounding the probability  of deviation of the normalized MSE from the state evolution prediction in each iteration.  Recall, by \eqref{eq:mse_LB_w_ser}, the section error rate can be bounded in terms of the normalized MSE, $\| \betavec^T - \betavec \|^2/L$. This connection is used to prove Theorem \ref{thm:main} from Theorem \ref{thm:main1}.

\begin{thm}
With the same assumptions as Theorem \ref{thm:main}, 
for $1 \leq t \leq T$ and $\e >0$ we have
\be
P\Big( \Big \lvert  \frac{\| \betavec^t - \betavec \|^2}{L}  - \frac{1}{\Lc} \sum_{\sfc \in[\Lc]} \psi^{t}_{\sfc} \Big \lvert \geq \epsilon \Big)  \leq K_{t-1} (\Lr \Lc)^{t} \exp 
\Big\{ \frac{- \kappa_{t-1} n \e^2}{ (\log M)^{2t} (\Lr/\omega)^{2t-1}} \Big\}. 
\ee
Here $\{ {\psi}^{t}_{\sfc} \}_{\sfc \in [\Lc]}$ are the state evolution parameters  defined in \eqref{eq:se_psi}, and the constants $K_{t-1}$ and $\kappa_{t-1}$ are as defined in the statement of Theorem \ref{thm:main}.
\label{thm:main1}
\end{thm}

The dependence of the constants $K_t$ and $\kappa_t$ on  $t!$ arises due to the induction-based proof of the  concentration lemma (Lemma \ref{lem:main_lem}). These constants have not been optimized, but we believe that these constants will depend on 
$t!$ in any induction-based proof of the result. It is an open question whether the $t!$ factors  are fundamental to the problem or if a different analysis of the AMP can yield a different $t$ dependence in these constants.

Theorem \ref{thm:main1} is proved in Section \ref{sec:proof}. We now show how Theorem \ref{thm:main} follows from Theorem \ref{thm:main1}.

\begin{proof}[Proof of Theorem \ref{thm:main}]
Without loss of generality, we can assume that  rate gap $\Delta$ (defined in \eqref{eq:Delta}) satisfies  $\Delta < 2R$. Otherwise the arguments below hold with $(\Delta/R)$ replaced by $1$.

Taking $t=T$, and noting that $\psi^T_\sfc \leq f_{M, \Delta/(3R)}$ for  $\sfc \in [\Lc]$ (from  Proposition \ref{prop:nonasymp_se}), Theorem \ref{thm:main1} implies that for any 
$\tilde{\e} >0$,
\be
P\left(  \frac{\| \betavec^T - \betavec \|^2}{L} \geq \tilde{\epsilon}  +  f_{M, \Delta/(3R)} \right)  \leq K_{T-1} (\Lr \Lc)^{T} \exp 
\Big\{ \frac{-\kappa_{T-1} n \tilde{\e}^2}{ (\log M)^{2T} (\Lr/\omega)^{2T-1}} \Big\}.  
\label{eq:mse_etil}
\ee
Furthermore, from \eqref{eq:mse_LB_w_ser} we have
\be
P \left( \mc{E}_{\text{sec}}(\mc{S}_n)  >  \e \right) \leq P\left(  \frac{\| \betavec^T - \betavec \|^2}{L}  \geq \frac{\e}{4} \right).
\label{eq:pe_mse}
\ee
Combining \eqref{eq:mse_etil} and \eqref{eq:pe_mse}, and taking $\tilde{\e} = \frac{\e}{8}$ and $M$ large enough so that $f_{M, \Delta/(3R)} \leq \frac{\e}{8}$ (see \eqref{eq:fmdelta_def}) yields Theorem \ref{thm:main}.
\end{proof}

\subsection{Choosing the SC-SPARC parameters}

Theorems \ref{thm:main} and \ref{thm:main1} give guidance on how to systematically choose parameters of the spatially coupled SPARC for any fixed rate $R < \mc{C} = \frac{1}{2} \ln(1+ \snr)$ and a target section error rate $\e$. First, choose  $(\omega, \Lambda)$ so that the rate gap $\Delta$ in \eqref{eq:Delta} is positive and \eqref{eq:omega_thresh} is satisfied. As described in Remark \ref{rem:choose_omega_lambda}, this can be done by choosing $\Lambda > \omega^2$, and $\omega$ sufficiently large.  This determines  $\Lr = (\Lambda + \omega -1)$ the number of iterations $T$, which from \eqref{eq:Tdef}, is bounded by a value proportional to $\Lambda/(\omega \Delta)$.

Next choose $M$ large enough for $f_{M,\delta} \leq \frac{\e}{8}$ for $\delta = \min\{ \frac{\Delta}{3R}, \frac{1}{3} \}$.  For fixed values of $(\omega, \Lambda, M)$, \eqref{eq:pezero} shows that the probability that the section error rate exceeds  $\e$ decays exponentially in the block length $n$. Once $n$ is chosen to be a large multiple of $\Lr = (\Lambda + \omega -1)$, the number of sections is  $L = \frac{n R}{\ln M}$, which completes the specification of the SPARC.\footnote{The parameters $L$ and $\Mc= ML/\Lambda$ determined in this way need to be integer-valued, which can be ensured by picking suitable values for  $\omega, \Lambda, M$. }

\begin{rem}
\normalfont
Theorem \ref{thm:main} implies that for any fixed $R < \mc{C}$ and $\e \in (0,1)$, one can construct a sequence  of rate $R$ spatially coupled SPARCs $\{ \mathcal{S}_n\}$ (indexed by code length $n$) for which 
\be
\begin{split}
\lim_{n \to \infty} \mc{E}_{\text{sec}}(\mc{S}_n)  = 0 \quad \text{ almost surely}.
\label{eq:asymp_result}
\end{split}
\ee
Indeed, once $(\Lambda, \omega, M)$ are chosen to satisfy the conditions in Theorem \ref{thm:main},  the bound in \eqref{eq:pezero} decreases exponentially in $n$. The  Borel-Cantelli lemma then yields the asymptotic result in \eqref{eq:asymp_result}.
\end{rem}

\begin{rem}
\normalfont
As described in \cite{joseph2012}, to obtain a small probability of codeword error  $P(\widehat{\betavec} \neq \betavec)$,   one can use a concatenated code with the SPARC as the inner code and  an outer Reed-Solomon  code.    A suitably chosen Reed-Solomon code of rate $(1-2\e)$ ensures that $\widehat{\betavec} = \betavec$ whenever the section error rate $ \mc{E}_{sec} < \e$, for any $\e >0$. The overall rate for such a concatenated code  is $(1-2\e)R$ and the  $P(\widehat{\betavec} \neq \betavec)$ is bounded by the RHS of \eqref{eq:pezero}. The reader is referred to \cite[Sec.\ 2.2]{venkataramanan19monograph} for details of how to choose an appropriate  Reed-Solomon code. 
\end{rem}

\subsection{Tradeoffs in choosing the parameters}

We first consider the effect of the base matrix parameters $\omega, \Lambda$. The condition in \eqref{eq:vartheta_def} and the bounds in \eqref{eq:ckappa_bnd} together imply that the minimum gap from capacity $\mc{C}-R$ is of order $\frac{\omega}{\Lambda}$, or equivalently, of order $\frac{\omega}{\Lambda+ \omega -1} = \frac{\omega}{\Lr}$. Therefore, decreasing the ratio $\frac{\omega}{\Lr}$   allows rates closer to capacity, but weakens the probability bound in Theorem \ref{thm:main}. Indeed, the bound in \eqref{eq:pezero} is exponential in $(\omega/\Lr)^{2T-1}$, with $T$ also increasing with $(\Lr/\omega)$.

Next consider the effect of increasing $M$. From Theorem \ref{thm:main} increasing $M$ allows for  a smaller target section error rate $\e$, since we require $\e \geq 8 f_{M, \delta}$. Equivalently, from \eqref{eq:psic_prog} and Theorem \ref{thm:main1}, increasing $M$ allows for a smaller state evolution estimate of the normalized MSE, ${\| \betavec^T - \betavec \|^2}/{L} $. On the other hand, the probability bounds in Theorems \ref{thm:main} and \ref{thm:main1} worsen with increasing $M$. 

With $\Lambda, \omega$ and $M$ fixed, increasing the code length $n$ (or, equivalently $L$) exponentially improves the probability bound in Theorem \ref{thm:main}. 

The per-iteration computational complexity of the AMP decoder is determined by the complexity of matrix-vector multiplications with the design matrix $\Amat$. This complexity is $O(nML)$ for Gaussian design matrices. For our empirical results in Section \ref{sec:emp_perf}, we use DFT-based design matrices which reduce the per-iteration complexity to $O(ML \log (ML) )$.

\subsection{Compressed sensing with spatially coupled design matrices}

%\RV{Cite relevant papers, some intro }
In this section,  we establish a nonasymptotic result analogous to  Theorem \ref{thm:main1} for compressed sensing with a spatially coupled measurement matrix.
In compressed sensing, the goal is to estimate a vector $\betavec \in \reals^p$ from a linear measurement $\yvec = \Amat \betavec + \wvec$.  This model is similar to SPARC decoding, with the main difference being that the entries of the signal vector $\betavec$ are now assumed to be drawn from a generic prior $P_\beta$  rather than the section-wise structure of a SPARC message vector. 

 We consider a spatially coupled measurement matrix $\Amat \in \reals^{n \times p}$  defined via a base matrix $\Wmat \in \reals^{\Lr \times \Lc}$ by replacing each entry $W_{\sfr\sfc}$ by an $(n/\Lr) \times (p/\Lc)$ block with i.i.d.\  $\mc{N}(0, \frac{W_{\sfr \sfc}}{n/\Lr} )$ entries, for $\sfr \in[\Lr]$ and $\sfc \in [\Lc]$. 
 %Therefore, $n= \Lr  \Mr$ and $p =  \Lc \Mc$ 
  The AMP algorithm for spatially coupled compressed sensing has the same form as the one in \eqref{eq:scamp_decoder_z}-\eqref{eq:scamp_decoder_beta}, with the main difference being the denoising function  $\eta^t$ and the corresponding changes in the state evolution parameters.   Before describing these differences, we state the assumptions on the model:

\begin{enumerate}[(1)]
\item The components of the signal vector $\betavec \in \reals^p$ are i.i.d.\ with a sub-Gaussian distribution $P_\beta$. 

\item The denoising function $\eta^t: \reals^{p} \to \reals^{p}$ used in the AMP algorithm is separable, and its components $\eta_j^t: \reals \to \reals$ are Lipschitz continuous, for $j \in [p]$.

\item  As the signal dimension $p$ grows, the sampling ratio ${n}/{p}$ is constant and denoted by $\delta$.

\item The entries of the noise vector $\wvec \in \reals^n$ are i.i.d.\ Gaussian with zero mean and variance $\sigma^2$. 

\item The entries of the base matrix  $\Wmat \in \reals^{\Lr \times \Lc}$  are bounded below by a strictly positive constant, and for $\sfr \in [\Lr]$, the row sums satisfy $\frac{1}{2} \leq \sum_{\sfc=1}^{\Lc} W_{\sfr\sfc} \leq 2$. (The assumption on the row sums is made  to ensure that the definitions of the state evolution parameters are consistent with \cite{donoho2013}.)
\end{enumerate}

The state evolution recursion for the spatially coupled compressed sensing AMP is as follows.
Initialize $\psi_\sfc^{0} = 1$ for  $\sfc \in[\Lc]$, and for $t \geq 0$:
\begin{align}
\phi_\sfr^t &  = \sigma^2 + \frac{1}{\delta}\sum_{\sfc=1}^{\Lc}W_{\sfr \sfc}\psi_\sfc^t , \qquad \sfr \in [\Lr], \label{eq:cs_se_phi} \\
\psi_\sfc^{t+1} &= \expec \left\{ \Big[\beta - \eta^t\big( (\tau_\sfc^t)^{-1/2} \beta + G \big) \Big]^2 \right\}, \qquad \sfc \in [\Lc],  \label{eq:cs_psi_def}
%1 - \mathcal{E}(\tau_\sfc^t), \qquad  \qquad \sfc \in [\Lc],  \label{eq:se_psi}
\end{align}
%%%%%
where the expectation in \eqref{eq:cs_psi_def} is over the independent pair $\beta \sim P_\beta$ and  $G \sim \normal(0,1)$. Furthermore,  $ \tau_\sfc^t = \big( \sum_{\sfr=1}^{\Lr} W_{\sfr \sfc}/\phi_{\sfr}^t \big)^{-1}$ for $\sfc \in [\Lc]$. 

The entries of the vector $\upvec^t$ and the matrix $\Smat$ in the AMP algorithm are defined as in \eqref{eq:onsager_def} and \eqref{eq:olS_def}, respectively. If the prior $P_\beta$ is known,  the Bayes optimal choice for scalar denoising function $\eta^t$  is the MMSE estimator. Indeed, $\psi_\sfc^t$ is minimized by taking $\eta_j(s)= \expec[\beta \, | \, (\tau_\sfc^t)^{-1/2} \beta + G =s  ]$  for an index $j$ in column block $\sfc$.

\begin{thm}
Consider the spatially coupled compressed sensing model under the assumptions listed above.  For $t \geq 1$, the mean-squared error of the AMP satisfies 
\be
P\Big( \Big \lvert  \frac{\| \betavec^t - \betavec \|^2}{p}  - \frac{1}{\Lc} \sum_{\sfc \in[\Lc]} \psi^{t}_{\sfc} \Big \lvert \geq \epsilon \Big)  \leq K_t (\Lr \Lc)^{t} \exp 
\Big\{ \frac{- k_t n \e^2}{ [\Lr \, (\max_{\sfr, \sfc} W_{\sfr,\sfc}) ]^{2t}} \Big\}. 
\label{eq:cs_pr_bnd}
\ee
Here $K_{t}, k_t $ are positive constants that depend only on $t$. (These constants are not the same as the ones in Theorem \ref{thm:main1}).
\label{thm:sc_cs}
\end{thm}

The proof of Theorem~\ref{thm:sc_cs} is  similar to that of Theorem \ref{thm:main1}, with appropriate changes (along the lines of the result in \cite{rush2018finite}) to account for the fact that  $\betavec$ now has i.i.d.\ entries rather than the block-wise structure of a SPARC message vector.  A few remarks about Theorem \ref{thm:sc_cs} and the underlying assumptions:

\begin{compactenum}
\item Theorem \ref{thm:sc_cs}  refines the asymptotic result established in \cite{javanmard2013state} for spatially coupled design matrices, which showed that $\lim_{p \to \infty} \| \betavec^t - \betavec \|^2/p = \frac{1}{\Lc}\sum_{\sfc} \psi^{t}_{\sfc}$ almost surely. 

\item The assumption that the entries of the base matrix are lower bounded by a positive constant is due to a technical detail in the proof. For the same reason, we need $\rho$ to be strictly positive (rather than $0$)  in Theorems \ref{thm:main} and \ref{thm:main1}. 

\item The scaling of the base matrix entries (implied by the column sums condition in Assumption 5) differs from the one used for SPARCs in \eqref{eq:W_power_constraint} by a factor of order $\Lr$. Therefore,  the maximum entry in the equivalent of the $(\omega, \Lambda, \rho)$ base matrix for compressed sensing  would be  of order $1/\omega$, leading to an exponent of $ \frac{-k_t n \e^2}{(\Lr/\omega)^{2t}}$ in 
\eqref{eq:cs_pr_bnd}. Thus the probability bound in Theorem \ref{thm:sc_cs} is similar to that of Theorem \ref{thm:main1}.
\end{compactenum}

% !TEX root =  sc_sparc_journal_paper_full.tex

\section{Empirical performance of SC-SPARCs} \label{sec:emp_perf}

In this section, we investigate the finite length error performance of SC-SPARCs with AMP decoding via numerical simulations. We use the $(\omega, \Lambda, \rho=0)$ base matrix construction in all the simulations.

We would like to compare the performance of SC-SPARCs with that of standard coded modulation schemes such as LDPC codes with Quadrature Amplitude Modulation (QAM), which produce complex-valued symbols to be transmitted over the channel.
Therefore, in these simulations we consider the communication over the \emph{complex} AWGN channel, where the noise is circularly-symmetric complex Gaussian.  We use complex SC-SPARCs, which are defined as described in Section \ref{sec:sc_AMP}, except that the design matrix now has independent circularly-symmetric complex Gaussian entries instead of real-valued Gaussian entries. 
The AMP decoder for complex SC-SPARCs is similar to the one  in \eqref{eq:scamp_decoder_z}-\eqref{eq:scamp_decoder_beta}: we take $\Amat^*$ to be the conjugate transpose of $\Amat$, and  modify the definition of $\eta^t_j$ in \eqref{eq:eta_function} according to \eqref{eq:MMSE_etaj}. For additional details on complex SC-SPARCs and its AMP decoder, see \cite{hsieh2020modulated}.

\begin{figure*}[!t]
\centering
\includegraphics[width=3.1in]{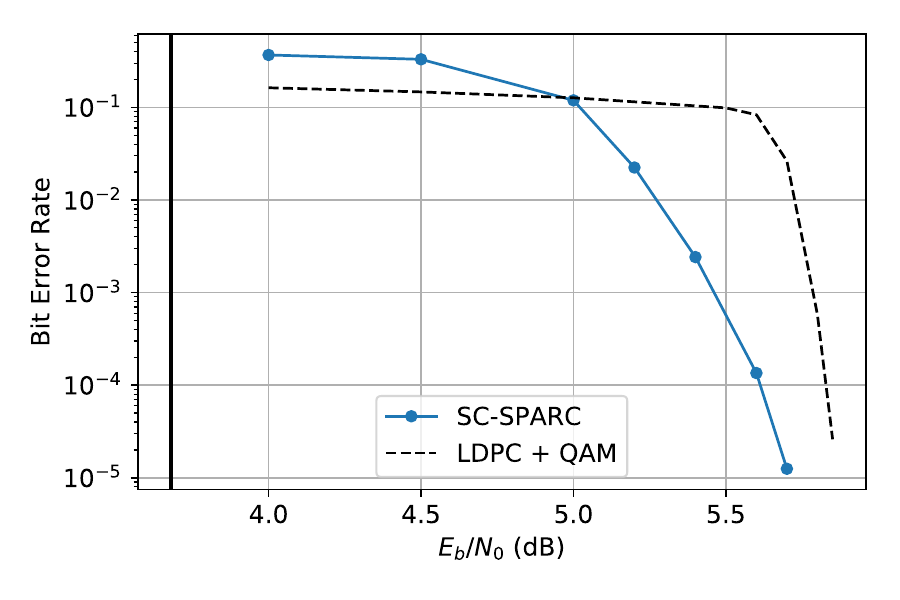}%
\hfil
\includegraphics[width=3.1in]{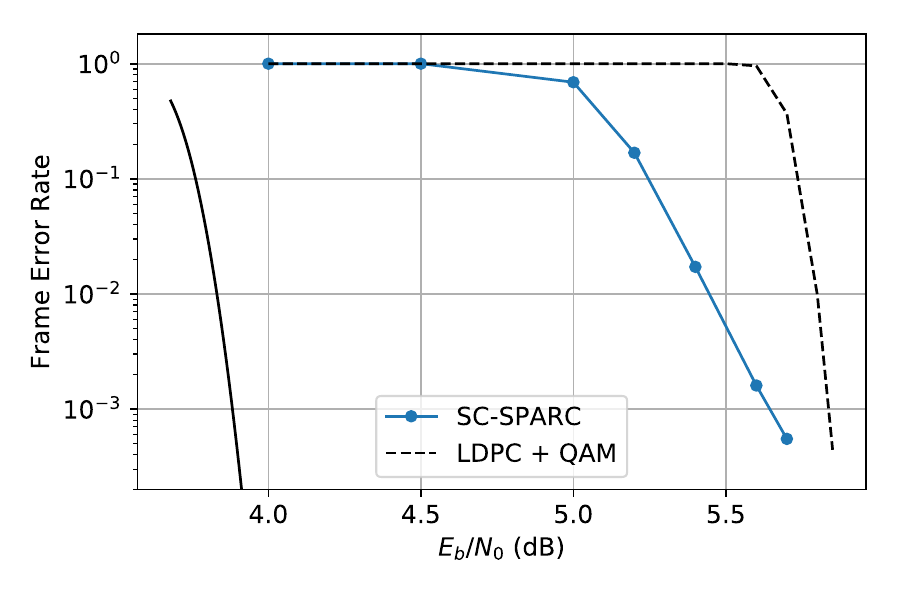}%
\caption{\small Error performance of complex SC-SPARCs defined via a $(\omega=4, \Lambda=32, \rho=0)$ base matrix. $R=1.5$ bits/dimension, $L=2944$, $\M=2048$, code length $n=10795$. The dashed lines show the performance of coded modulation: $(K=32400, N=64800)$ DVB-S2 LDPC + 64 QAM, frame length $=10800$ symbols. The solid black line in the BER plot is the AWGN Shannon limit for $R=1.5$ bits/dimension, and in the FER plot, it is the normal approximation to the AWGN finite length error probability bound in \cite{polyanskiy2010channel}.}
\label{fig:complex_R150}
\end{figure*}

\begin{figure*}[!t]
\centering
\includegraphics[width=3.0in]{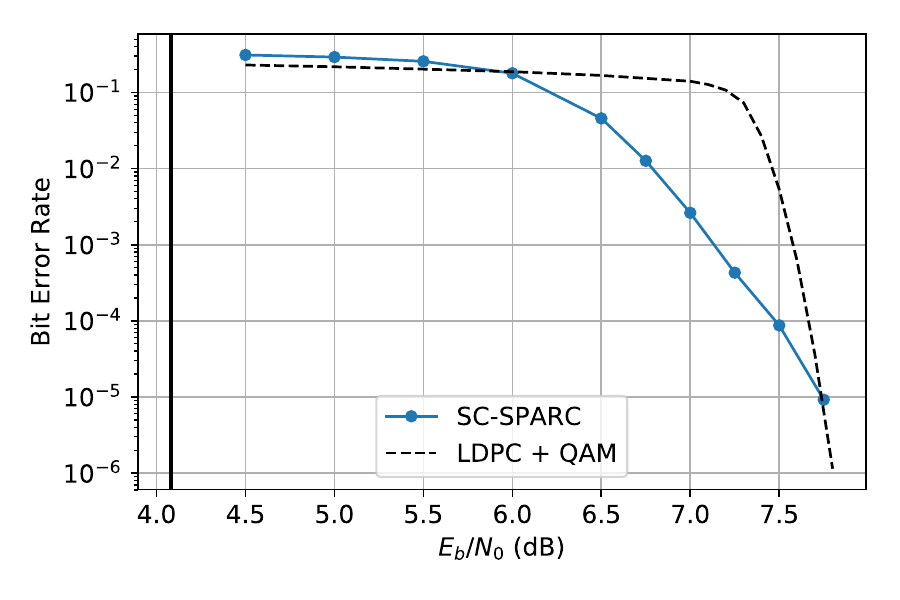}%
\hfil
\includegraphics[width=3.0in]{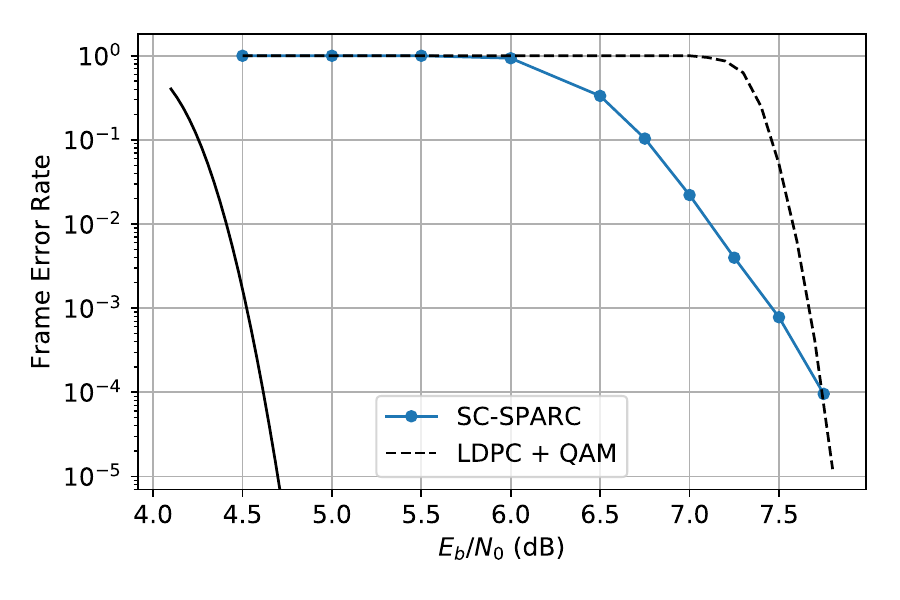}%
\caption{\small Error performance of complex SC-SPARCs defined via an $(\omega=6, \Lambda=32, \rho=0)$ base matrix. $R=1.6$ bits/dimension, $L=960$, $\M=128$, code length $n=2100$. The dashed lines show the performance of coded modulation: $(K=6480, N=16200)$ DVB-S2 LDPC + 256 QAM, frame length $=2025$ symbols. The solid black line in the BER plot is the AWGN Shannon limit for $R=1.6$ bits/dimension, and in the FER plot, it is the normal approximation to the AWGN finite length error probability bound in \cite{polyanskiy2010channel}.}
\label{fig:complex_R160}
\end{figure*}

\begin{figure*}[!t]
\centering
\includegraphics[width=3.0in]{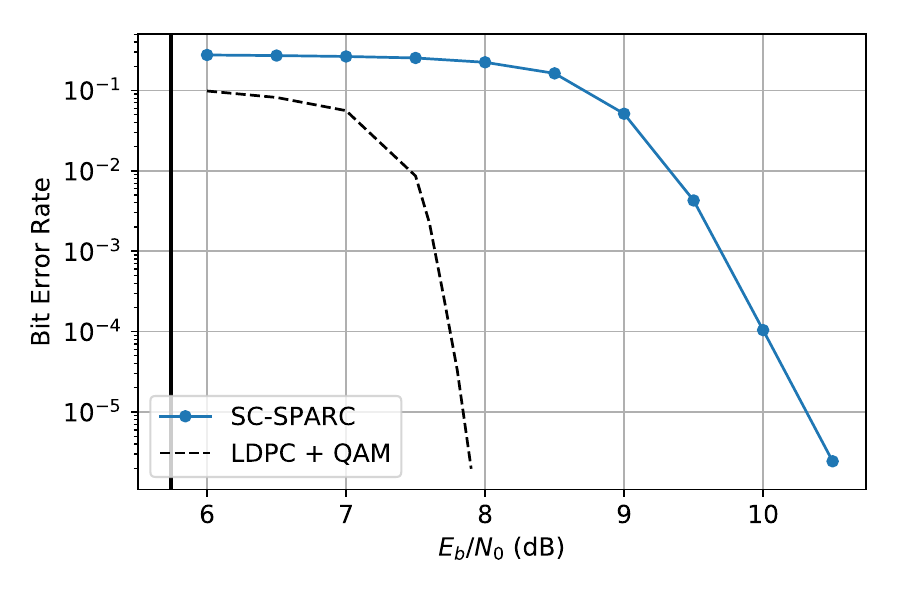}%
\hfil
\includegraphics[width=3.0in]{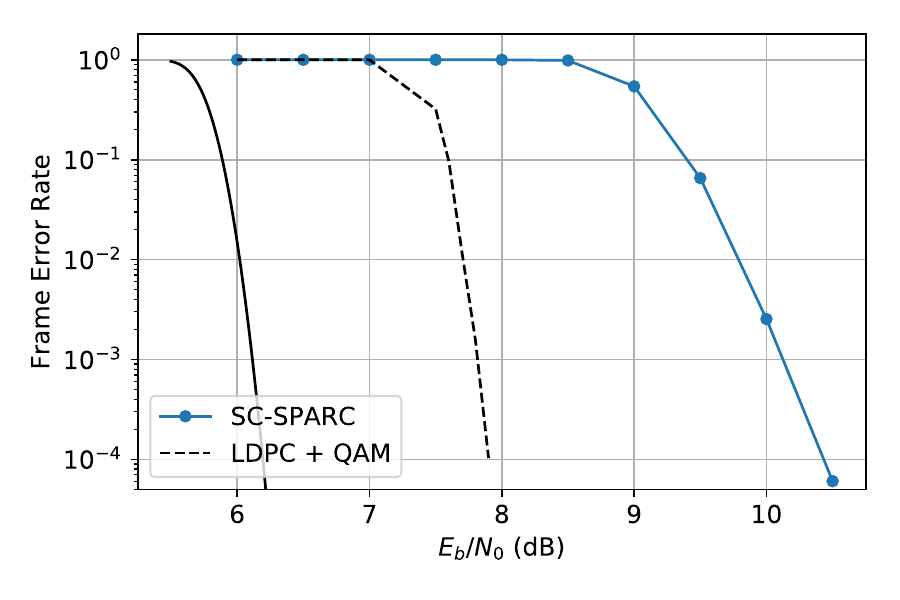}%
\caption{\small Error performance of complex SC-SPARCs defined via an $(\omega=6, \Lambda=32, \rho=0)$ base matrix. $R=2$ bits/dimension, $L=2688$, $\M=16$, code length $n=2688$. The dashed lines show the performance of coded modulation: $(K=10800, N=16200)$ DVB-S2 LDPC + 64 QAM, frame length $=2700$ symbols. The solid black line in the BER plot is the AWGN Shannon limit for $R=2$ bits/dimension, and in the FER plot, it is the normal approximation to the AWGN finite length error probability bound in \cite{polyanskiy2010channel}.}
\label{fig:complex_R200}
\end{figure*}

In  Figures \ref{fig:complex_R150}, \ref{fig:complex_R160}, and \ref{fig:complex_R200},
we provide numerical simulation results demonstrating the finite length error performance of complex SC-SPARCs with AMP decoding at different code rates and code lengths.
The error performance is evaluated using both the bit error rate (BER) and the frame error rate (FER). (The FER is the message/codeword error rate.)
We also simulate and plot the error performance of coded modulation schemes (LDPC + QAM) for reference using the AFF3CT toolbox \cite{cassagne2017fast}. The LDPC codes are chosen from the DVB-S2 standard and a belief propagation (BP) decoder is used which runs for 50 iterations. For fair comparison, in each figure, the frame length of the coded modulation scheme  is chosen to be close to the code length of the SC-SPARC.

Fig.\ \ref{fig:complex_R150} shows the performance of SC-SPARCs with rate $1.5$ bits/dimension and code length  $n = 10795$. The AMP decoder for the SC-SPARC is run for a maximum of 200 iterations (details in Sec.\ \ref{subsec:imp_details}). 
The coded modulation scheme uses a rate $\frac{1}{2}$ $(32400, 64800)$ DVB-S2 LDPC code with 64-QAM modulation, for the same overall rate of $1.5$ bits/dimension and a frame length of 10,800 symbols.
We observe that the SC-SPARC requires a smaller $E_b/N_0$ to achieve BERs in the range $10^{-1}$ to $10^{-5}$, and FERs down to $5\times 10^{-4}$ compared to the coded modulation scheme.
However, for much lower FERs, we expect the coded modulation scheme to require a smaller $E_b/N_0$ because its frame error rate drops faster as $E_b/N_0$ increases.

Fig.\ \ref{fig:complex_R160} shows the performance of an SC-SPARC with a shorter code length $n=2100$, and a rate of $1.6$ bits/dimension. The AMP decoder for the SC-SPARC is run for a maximum of 100 iterations.
The coded modulation scheme uses a rate $\frac{1}{2}$ (6480, 16200) DVB-S2 LDPC code with 256-QAM modulation,  for the same overall rate of $1.6$ bits/dimension and a frame length of 2025 symbols.
We observe that the SC-SPARC requires a smaller $E_b/N_0$ to achieve BERs in the range $10^{-1}$ to $10^{-5}$ and FERs down to $10^{-4}$ compared to the coded modulation scheme. However, for BERs and FERs lower than $10^{-5}$ and $10^{-4}$, respectively, we expect the coded modulation scheme to require a smaller $E_b/N_0$ because its error rate drops faster as $E_b/N_0$ increases.

In Fig.\ \ref{fig:complex_R200}, the rate of the SC-SPARC is $2$ bits/dimension and the code length is $n=2688$. The AMP decoder for the SC-SPARC is run for a maximum of 100 iterations.
The coded modulation scheme uses a rate $\frac{2}{3}$ (10800, 16200) DVB-S2 LDPC code with 64-QAM modulation, for the same overall rate of $2$ bits/dimension and a frame length of 2700 symbols.
We observe that the SC-SPARC has a higher BER and FER compared to the coded modulation scheme for all values of $E_b/N_0$, and its error rate also drops more slowly as $E_b/N_0$ increases. 

In the above plots, the SC-SPARC parameters $(\omega, \Lambda, L, M, n)$ have not been carefully optimized. An interesting direction for future work is to develop good finite length design guidelines for choosing these parameters as a function of rate and $\snr$. Another direction is to explore whether alternative base matrix designs could improve the finite length performance at higher rates like $2$ bits/dimension.

%%%%%
%%%%%

% Several modifications to SPARCs and AMP
\subsection{Implementation details} \label{subsec:imp_details}
The $(\omega, \Lambda, \rho=0)$ base matrix was used for all the simulations. Furthermore, to reduce the  decoding complexity and the memory required,  a few modifications were made to the SC-SPARC construction and the AMP decoder:

\textit{1) DFT based design matrices}:
We replaced the complex Gaussian design matrix with  a Discrete Fourier Transform (DFT) based design matrix. This enables the matrix-vector multiplications in the AMP decoder \eqref{eq:scamp_decoder_z}--\eqref{eq:scamp_decoder_beta} to be computed via the Fast Fourier Transform (FFT), which significantly lowers the decoding complexity and memory requirement.
Our approach is similar to that of \cite{barbier2017, rush2017} where Hadamard based design matrices were used for real-valued SPARCs.

The computational complexity of the AMP decoder is dominated by the two matrix-vector multiplications associated with the design matrix $\Amat$. These operations have  complexity $O(nL\M)$ when $\Amat$ has independent Gaussian entries. The memory requirements of the encoder and decoder are also proportional to $nL\M$ since the Gaussian design matrix has to be stored.
By constructing the design matrix using randomly sampled rows and columns of the (deterministic) DFT matrix, the complexity of the matrix-vector multiplications (replaced by FFTs) is reduced to $O(L\M \log(L\M))$, and the memory requirements of the encoder and decoder are proportional to $\omega L \M $. The error performance of DFT based design matrices was found to be similar to that of Gaussian matrices for large matrix sizes.

\textit{2) Online estimation of state evolution parameters}:
The AMP decoder in \eqref{eq:scamp_decoder_z}--\eqref{eq:scamp_decoder_beta} contains parameters computed using the state evolution (SE) recursion \eqref{eq:se_phi}--\eqref{eq:E_tau}. In particular, the vector $\upvec^t \in \reals^n$, and the matrix $\overline{\Smat}^t \in \reals^{n \times ML}$ are determined via SE parameters computed offline.
Instead of computing the SE parameters offline, the SE parameters can be estimated online (at runtime) using the outputs of the AMP decoder in each iteration.
The SE parameters $\{\sigma_{\sfr}^t\}_{\sfr\in[\Lr]}$, $\{\phi_{\sfr}^{t}\}_{\sfr\in[\Lr]}$ and $\{\tau_{\sfc}^t\}_{\sfc\in[\Lc]}$, which are needed to compute  $\upvec^t$ and $\overline{\Smat}^t$ (see \eqref{eq:onsager_def} and \eqref{eq:olS_def}) can be estimated online in the following way.  For $\sfr \in[\Lr]$ and $\sfc \in[\Lc]$,
\begin{align}
	\widehat{\sigma}_\sfr^t &= \frac{1}{\Lc} \sum_{\sfc=1}^{\Lc} W_{\sfr\sfc} \bigg(1 - \frac{\|\betavec^t_{\sfc}\|^2}{L/\Lc}\bigg), \label{eq:sigma_r_est}\\
	\widehat{\phi}_\sfr^t \, &= \,
	\begin{cases}
		\sigma^2 + \widehat{\sigma}_\sfr^t \quad &\text{if the decoder knows } \sigma^2,\\
		\frac{\|\boldsymbol{z}_{\sfr}^t\|^2}{n/\Lr} \quad &\text{otherwise},
	\end{cases}\\
	\widehat{\tau}_\sfc^t &= \frac{L}{n} \left[\frac{1}{\Lr}\sum_{\sfr=1}^{\Lr}\frac{W_{\sfr \sfc}}{\widehat{\phi}_\sfr^t}\right]^{-1} \label{eq:tau_c_est}.
\end{align}

The justification for these estimates comes from Lemma \ref{lem:main_lem}, which proves that the estimates
$\widehat{\sigma}_\sfr^t, \widehat{\phi}_\sfr^t$ concentrate on $\sigma_\sfr^t, \phi_\sfr^t$, respectively (eqs.\ \eqref{eq:Bc} and \eqref{eq:Hc_adj}), for large $(n,L)$.
We observe that using  online estimates of the SE parameters results in a better error performance than using deterministic SE parameters. A similar improvement was observed in \cite{GreigV17} for power allocated SPARCs.

\textit{3) Early stopping of AMP}:
Since the empirical estimates of SE parameters in \eqref{eq:sigma_r_est}-\eqref{eq:tau_c_est} are estimates of certain noise variances related to the decoding error in each iteration of the AMP, we chose to stop the AMP decoder early if the change in $\widehat{\boldsymbol{\sigma}}^t, \widehat{\boldsymbol{\phi}}^t$  or $\widehat{\boldsymbol{\tau}}^t$ fell below a prescribed threshold over consecutive iterations.
A similar stopping criterion was used in \cite{GreigV17} to terminate the AMP decoder for power allocated SPARCs.

A Python implementation of SPARCs (both power allocated and spatially coupled) with AMP decoding is available at \cite{KuanPyScripts}.

% !TEX root =  sc_sparc_journal_paper_full.tex

\section{Proof of Theorem \ref{thm:main1}} \label{sec:proof}

The main ingredients in the proof of Theorem \ref{thm:main1} are two technical lemmas (Lemmas \ref{lem:hb_cond} and \ref{lem:main_lem}).  After laying down some definitions and notation, we give a brief overview of the proof in Section \ref{subsec:proof_overview}.  We then state the key technical  lemmas, and use them to prove Theorem \ref{thm:main1}.  For consistency with earlier analyses of AMP, we use notation similar to  \cite{bayati2011, rush2017, RushV19}, with modifications to account for the row- and column-block dependence induced due to spatial coupling.

\subsection{Definitions and Preliminaries}

Recall that $\betavec_0 \in \reals^{ML}$ is the message vector chosen by the transmitter,  $\wvec \in \reals^n$ is the channel noise vector, and the  AMP decoder is intialized with $\betavec^0= \mathbf{0}$ and $\zvec^0=\yvec$.  Throughout the proof, we use the notation
\be
\Mr := {n}/{\Lr}, \quad \Mc := {ML}/{\Lc}.
\label{eq:MrMcdef}
\ee

Define the  column vectors $\hvec^{t+1}, \brqvec^{t+1} \in \mathbb{R}^{ML}$ and $\bvec^t, \brmvec^t \in \mathbb{R}^{n}$ for $t \geq 1$ recursively as follows. Starting with the initial conditions 
\begin{equation}
\label{eq:vec_init0}
\begin{split} 
\brqvec^0   = - \betavec_0, &  \qquad \hvec^1 = \betavec_{0} -(\ol{\Smat}^0 \odot \Amat)^* \zvec^0,  \\
\brmvec^0 = - \zvec^0, &  \qquad \bvec^0  = \wvec - \zvec^0,
\end{split}
\end{equation}
for $t\geq 1$, the vectors $ \brqvec^{t}, \hvec^{t+1},  \brmvec^t,  \bvec^t$ are defined as 
\begin{equation}
\begin{split}
 \brqvec^{t}  = \betavec^t - \betavec_{0}, &  \qquad  \hvec^{t+1} := \betavec_{0} - \Big([(\ol{\Smat}^t \odot \Amat)^* \zvec^t] + \betavec^t \Big), \\
 \brmvec^t = -\zvec^t,  & 
 \qquad   \bvec^t = \wvec - \zvec^t,
\end{split}
\label{eq:hqbm_def}
\end{equation}
where $\ol{\Smat}^t \in \mathbb{R}^{n \times ML}$ is the matrix with entries defined in \eqref{eq:olS_def}. For notational convenience, we define the matrix $\Smat^t \in \reals^{\Lr \times \Lc}$ with entries 
\be
S_{\sfr \sfc}^t = {\tau_\sfc^t}/{\phi^t_\sfr}, \quad \sfr \in [\Lr], \ \sfc \in [\Lc].
\label{eq:Src_def1}
\ee

We define a modified design matrix $\sfAmat \in \reals^{n \times ML}$ having entries given by
\be
\label{eq:sfA_def}
{\sf{A}}_{ij} = \frac{A_{ij}}{\sqrt{W_{\sfr(i)\sfc(j)}}}, \qquad i \in [n], \ j \in [ML].
\ee
Since $A_{ij} \stackrel{\text{i.i.d.}}{\sim} \mc{N}(0, \frac{1}{L} W_{\sfr(i) \sfc(j)})$, the modified matrix has entries  $ {\sf{A}}_{ij}  \stackrel{\text{i.i.d.}}{\sim} \mc{N}(0, \frac{1}{L})$. We note that $\Amat = \sqrt{\widetilde{\Wmat}} \odot \sfAmat$ where $\widetilde{\Wmat} \in \mathbb{R}^{n \times ML}$ is the matrix with entries $\widetilde{W}_{ij} = W_{\sfr(i)\sfc(j)}$.

Using the definitions  \eqref{eq:vec_init0}--\eqref{eq:sfA_def} in the AMP update equations \eqref{eq:scamp_decoder_z}-\eqref{eq:scamp_decoder_beta}, we find that the following block-wise relationships are satisfied for $t \geq 0$:
\be
\begin{split}
&\bvec^{t}_{\sfr} - \frac{\sigma_\sfr^t}{\phi_\sfr^{t-1}}  \brmvec^{t-1}_{\sfr}  =  [(\sqrt{\widetilde{\Wmat}} \odot \sfAmat) \, \brqvec^{t}]_{\sfr} = \sum_{\sfc \in [\Lc]} \sqrt{W_{\sfr\sfc}} \, \sfAmat_{\sfr \sfc} \, \brqvec^{t}_{\sfc}, \qquad \text{ for } \sfr \in [\Lr],\\
&\hvec^{t+1}_{\sfc} + \brqvec^{t}_{\sfc}=  [(\ol{\Smat}^t \odot \sqrt{\widetilde{\Wmat}} \odot \sfAmat)^* \brmvec^t]_{\sfr} =  \sum_{\sfr \in [\Lr]}  S^{t}_{\sfr \sfc} \, \sqrt{W_{\sfr\sfc}} \, (\sfAmat_{\sfr \sfc})^* \brmvec^t_{\sfr}, \qquad \text{ for } \sfc \in [\Lc].
\label{eq:bmq_rowcol}
\end{split}
\end{equation}

We define complementary quantities for the $\brmvec^t$ and $\brqvec^{t}$ vectors that will be useful in the conditional distribution lemma that follows.   For $t \geq 0$ and $\sfr \in [\Lr]$ and $\sfc \in [\Lc]$, let 
\be
\madj^{t,\sfc} = \begin{bmatrix} 
S^t_{1 \sfc} \, \sqrt{W_{1 \sfc}} \, \brmvec^{t}_{\sfr = 1} \\
S^t_{2 \sfc} \, \sqrt{W_{2 \sfc}} \, \brmvec^{t}_{\sfr = 2} \\
\vdots \\ 
S^t_{\Lr \sfc} \, \sqrt{W_{\Lr \sfc}} \, \brmvec^{t}_{ \sfr = \Lr}
\end{bmatrix}\in \mathbb{R}^{n \times 1} \quad \text{ and } \quad \qadj^{t,\sfr} =
\begin{bmatrix} 
\sqrt{W_{\sfr 1}} \, \brqvec^{t}_{\sfc = 1} \\
\sqrt{W_{\sfr 2}} \, \brqvec^{t}_{\sfc = 2} \\
\vdots \\ 
\sqrt{W_{\sfr \Lc}} \,  \brqvec^{t}_{\sfc = \Lc}
\end{bmatrix}\in \mathbb{R}^{ML \times 1}.
\label{eq:tildem_def}
\ee
(Here  $\brmvec^{t}_{\sfr=1} \in \reals^{\Mr}$ refers to the first row block of $\brmvec^{t}$, and $\brqvec^{t}_{\sfc = 1}  \in \reals^{\Mc}$ to the first column block of $\brqvec^{t}$.)
A word about the notation before we proceed: when a  row- or column-block index ($\sfr$ or $\sfc$) appears as a subscript of a vector or a matrix (e.g., as in $\sfAmat_{\sfr \sfc}$), it denotes the corresponding block of that vector/matrix, but when  a  row- or column-block index  appears as a \emph{superscript} of a vector/matrix (e.g., $\madj^{t,\sfc}$ and  $\qadj^{t,\sfr}$), it denotes the dependence of the vector/matrix on that index.

Using the vectors defined in \eqref{eq:tildem_def}, we can rewrite \eqref{eq:bmq_rowcol} as
\be
\begin{split}
&\bvec^{t}_{\sfr} - \frac{\sigma_\sfr^t}{\phi_\sfr^{t-1}}  \brmvec^{t-1}_{\sfr}  = 
[ \sfAmat \, \qadj^{t, \sfr}]_{\sfr} = \sum_{\sfc \in [\Lc]} \sfAmat_{\sfr \sfc} \, \qadj^{t, \sfr}_{\sfc} \ \text{ for }   \sfr \in [\Lr], \\
&\hvec^{t+1}_{\sfc} + \brqvec^{t}_{\sfc} = 
[ \sfAmat^* \madj^{t, \sfc}]_{\sfc} = \sum_{\sfr \in [\Lr]} [\sfAmat_{\sfr \sfc}]^* \madj^{t, \sfc}_{\sfr} \ \text{ for } \sfc \in  [\Lc].
\label{eq:bmq_rowcol_v2}
\end{split}
\end{equation}
It will be useful to  write the equations in \eqref{eq:bmq_rowcol_v2} in matrix form. For this, we define the following matrices for $t \geq 1$:
\be
\label{eq:XYMSt}
\begin{split}
\Qmat_{t} & :=  [\brqvec^0 \mid \ldots \mid \brqvec^{t-1}]  \in \mathbb{R}^{ML \times t}, \\
\Hmat_{t} & := [\hvec^1 | \ldots | \hvec^{t}]  \in \mathbb{R}^{ML \times t},  \\
\Xmat_{t}   &  := [\hvec^1 +  \brqvec^0 \mid \hvec^2 +  \brqvec^1 \mid \ldots \mid \hvec^t + \brqvec^{t-1}]  \in \mathbb{R}^{ML \times t},    \\  
\Mmat_{t}  & := [\brmvec^0 \mid \ldots \mid \brmvec^{t-1} ] \in \mathbb{R}^{n \times t}, \\
\Bmat_{t} & := [\bvec^0 | \ldots | \bvec^{t-1}]  \in \mathbb{R}^{n \times t} \\
\Ymat_{t}  &  := [\bvec^0 \mid \bvec^1 -  \upvec^1 \odot \brmvec^0 \mid \ldots \mid \bvec^{t-1} - \upvec^{t-1} \odot \brmvec^{t-2}]  \in \mathbb{R}^{n \times t}, \\
  \Upmat^{\sfr}_{t} &  := \text{diag}\left(0, \frac{\sigma_\sfr^1}{\phi_\sfr^0} , \ldots, \frac{\sigma_\sfr^{t-1}}{\phi_\sfr^{t-2}} \right)  \in \mathbb{R}^{t \times t}.
\end{split}
\ee
In the equations above, the notation $[\avec_1 \mid \avec_2 \mid \ldots \mid \avec_k]$ is used to denote a matrix with columns $\avec_1, \ldots, \avec_k$. We also recall from \eqref{eq:onsager_def} that  $\upsilon^t_i =  {\sigma_\sfr^t}/{\phi_\sfr^{t-1}} $ for $t \geq 1$ if $i \in [n]$ is  in row block $\sfr$. For $t=0$, the matrices above are all defined as all-zeros. For $\sfc \in [\Lc]$, we define $\Qmat_{t, \sfc}, \Xmat_{t, \sfc}, \Hmat_{t, \sfc} \in \mathbb{R}^{\Mc \times t}$ to correspond only to rows $\sfc$ of the corresponding matrix.   We similarly define 
  $\Mmat_{t, \sfr}, \Bmat_{t, \sfr}, \Ymat_{t, \sfr}  \in \mathbb{R}^{\Mr \times t}$ for   $\sfr \in [\Lr]$.    Using these definitions we have 
\be
\Ymat_{t, \sfr}  = \Bmat_{t,\sfr} - [ \bzero | \Mmat_{t-1, \sfr}]  \Upmat^{\sfr}_{t},  \qquad
\Xmat_{t,\sfc} = \Hmat_{t,\sfc} + \Qmat_{t, \sfc}.
\label{eq:Ytr_Xtc}
\ee

Let $\Qadjmat_0^{\sfr}$ and $\Madjmat_0^{\sfc}$ be all-zero vectors.  For $t \geq 1$ and $\sfc \in [\Lc]$, $\sfr \in [\Lr]$, let
\be
\Madjmat_t^{\sfc} = [\madj^{0,\sfc}| \ldots | \madj^{t-1,\sfc}] \in \mathbb{R}^{n \times t}  \quad \text{ and } \quad \Qadjmat_t^{\sfr} = [\qadj^{0,\sfr}| \ldots | \qadj^{t-1,\sfr}] \in \mathbb{R}^{ML \times t}.
\label{eq:tildeM_def}
\ee
With this notation, we can compactly write \eqref{eq:bmq_rowcol_v2} for $t \geq 1$ as
\begin{align}
\big[\sfAmat \Qadjmat_t^{\sfr}\big]_{(\sfr, \cdot)} = \Ymat_{t, \sfr}, \ \ \sfr \in [\Lr],  \qquad \text{ and } \qquad 
\big[\sfAmat^* \Madjmat_{t}^{\sfc}\big]_{(\sfc, \cdot)}  = \Xmat_{t, \sfc}, \ \ \sfc \in [\Lc].
\label{eq:Alin_consts}
\end{align}
Here the subscript ${(\sfr, \cdot)}$ on a matrix denotes the $\sfr$th block of rows of the matrix. 
%The subscript $(\sfc, .)$ is analogously defined. 

We use the notation $\madj^{t, \sfc}_{\|}$ and $\qadj^{t,\sfr}_{\|}$ to denote the projection of $\madj^{t, \sfc}$ and $\qadj^{t,\sfr}$ onto the column space of $\Madjmat^{\sfc}_{t}$ and $\Qadjmat^{\sfr}_{t}$, respectively. 
Let
 \be 
\alphavec^{t,\sfc} := (\alpha^{t,\sfc}_{0}, \ldots,  \alpha^{t,\sfc}_{t-1})^*, \qquad  \gammavec^{t,\sfr} :=  (\gamma^{t,\sfr}_{0}, \ldots, \gamma^{t,\sfr}_{t-1})^* 
 \label{eq:vec_alph_gam_conc}
 \ee 
 be the coefficient vectors of these projections, i.e.,
 \be
\madj^{t, \sfc}_{\|} := \sum_{i=0}^{t-1} \alpha^{t,\sfc}_{i} \madj^{i,\sfc}, \qquad  \qadj^{t,\sfr}_{\|} := \sum_{i=0}^{t-1} \gamma^{t,\sfr}_{i} \qadj^{i,\sfr}.
 \label{eq:mtqt_par}
 \ee
Writing $\proj^{\parallel}_{\Madjmat^{\sfc}_{t}}$ for the orthogonal projection matrix onto the column space of $\Madjmat^{\sfc}_{t}$, we have $\madj^{t, \sfc}_{\|} = \proj^{\parallel}_{\Madjmat^{\sfc}_{t}}\madj^{t, \sfc} = \Madjmat^{\sfc}_{t}((\Madjmat^{\sfc}_{t})^*\Madjmat^{\sfc}_{t})^{-1} (\Madjmat^{\sfc}_{t})^* \madj^{t, \sfc}$ and so $\alphavec^{t,\sfc} = ((\Madjmat^{\sfc}_{t})^*\Madjmat^{\sfc}_{t})^{-1} (\Madjmat^{\sfc}_{t})^* \madj^{t, \sfc}$.
 (If the columns of $\Madjmat^{\sfc}_{t}$ are linearly dependent,  $((\Madjmat^{\sfc}_{t})^*\Madjmat^{\sfc}_{t})^{-1}(\Madjmat^{\sfc}_{t})^*$ is interpreted as the pseudoinverse of $\Madjmat^{\sfc}_{t}$.)
 We can similarly write $\gammavec^{t,\sfr} = ((\Qadjmat^{\sfr}_{t})^*\Qadjmat^{\sfr}_{t})^{-1} (\Qadjmat^{\sfr}_{t})^* \qadj^{t,\sfr}$.   The projections of $\madj^{t, \sfc}$ and $\qadj^{t,\sfr}$ onto the orthogonal complements of $\Madjmat^{\sfc}_{t}$ and $\Qadjmat^{\sfr}_{t}$, respectively,  are denoted by
 \be
 \madj^{t, \sfc}_{\perp} := \madj^{t, \sfc} - \madj^{t, \sfc}_{\|}, \qquad   \qadj^{t,\sfr}_{ \perp} :=  \qadj^{t,\sfr} -  \qadj^{t,\sfr}_{\|}.
  \label{eq:mtqt_perp}
 \ee
 
In Lemma \ref{lem:main_lem}, we show that  the entries of $\alphavec^{t,\sfc}$ and $\gammavec^{t,\sfr}$ concentrate around constants. We now specify these constants. For $\sfc \in [\Lc]$, $\sfr \in [\Lr]$, define matrices $\tCmat^{t,\sfr}, \brCmat^{t,\sfc} \in \mathbb{R}^{t \times t}$ for $t \geq 1$ such that
\be
\widetilde{C}^{t,\sfr}_{i+1,j+1} = \sigma^{\max(i,j)}_{\sfr}, \quad \text{ and } \quad \breve{C}^{t,\sfc}_{i+1,j+1} = \frac{n}{L} \tau^{\max(i, j)}_{\sfc}, \quad 0\leq i,j \leq t-1. 
\label{eq:Ct_def}
\ee
The concentrating values for $\gammavec^{t,\sfr}$ and $\alphavec^{t,\sfc}$ are
\be
\begin{split}
\widehat{\gammavec}^{t,\sfr} &:= \sigma^{t}_{\sfr} (\tCmat^{t,\sfr})^{-1}  (1, \ldots, 1)^* \stackrel{(a)}{=} (0,\ldots, 0, \sigma^t_{\sfr}/\sigma^{t-1}_{\sfr})^* \in \mathbb{R}^t, \\
\widehat{\alphavec}^{t,\sfc} &:=  \frac{n}{L} \tau^{t}_{\sfc} (\brCmat^{t,\sfc})^{-1}  (1, \ldots, 1)^* \stackrel{(b)}{=} (0,\ldots, 0, \tau^t_{\sfc}/\tau^{t-1}_{\sfc})^*  \in \mathbb{R}^t.
\label{eq:hatalph_hatgam_def}
\end{split}
\ee
To see that $(a)$ holds,  we observe that $(\tCmat^{t,\sfr})^{-1} \tCmat^{t,\sfr} = \iden_{t}$ implies that $(\tCmat^{t,\sfr})^{-1}  (\sigma^{t-1}_{\sfr}, \ldots, \sigma^{t-1}_{\sfr})^*  = (0, \ldots, 0, 1)^* \in \mathbb{R}^t$. The equality $(b)$ is obtained similarly.

Let $\sigma^{0}_{\perp, \sfr} := \sigma^{0}_{\sfr}$ and $\tau^{0}_{\perp, \sfc} :=  \tau^{0}_{\sfc}$, and for $t \geq 1$ define 
\be
\begin{split}
& \sigma_{\perp, \sfr}^{t} :=  \sigma_{\sfr}^t \Big(1 - \frac{\sigma^t_{\sfr}}{\sigma^{t-1}_{\sfr}} \Big), \quad \text{ and } \quad \tau^{t}_{\perp, \sfc} :=  \tau^{t}_{\sfc} \Big(1 - \frac{\tau^{t}_{\sfc}}{\tau^{t-1}_{\sfc}} \Big).
\label{eq:sigperp_defs}
\end{split}
\ee

\begin{lem} \label{lem:sigmatperp}
Under the assumptions of Prop. \ref{prop:nonasymp_se}, for sufficiently large $M$, the constants $ \sigma_{\perp, \sfr}^{t}$ and $\frac{n}{L} \tau^{t}_{\perp, \sfc}$ are bounded below  for $0 \leq k < T$: 
\be
\sigma_{\perp, \sfr}^{t} \geq \ol{C}_1\left( \frac{\omega}{\Lambda}\right)^2,  \  \ \sfr \in [\Lr], \qquad 
\tau_{\perp, \sfc}^{t} \geq \ol{C}_2\left( \frac{\omega}{\Lambda}\right), \ \  \sfc \in [\Lc], 
\ee
where
\be
\ol{C}_1 = \left( 1+ \frac{1}{\vth \snr} \right)^2 \frac{P\rho^2}{\vth \snr^2} \, \Delta^2, \qquad 
\ol{C}_2 = \frac{\rho ( 1 + \vth \snr )}{\vth R}{\Delta},
\ee
where $\Delta$ is the rate gap defined in \eqref{eq:Delta}. 
\end{lem}

\begin{proof}
In Appendix \ref{app:sigmatperp}.
\end{proof}

\begin{lem} \label{lem:Ct_invert}
If the $ \sigma_{\perp, \sfr}^{t}$ and $\frac{n}{L} \tau^{t}_{\perp, \sfc}$ are bounded below by some positive constants for $0 \leq k < T$, then the matrices $\tCmat^{k,\sfr}$ and $\brCmat^{k,\sfc}$ defined in \eqref{eq:Ct_def} are invertible for $1 \leq k \leq T$. 
\end{lem}

\begin{proof}
The proof can be found in \cite[Lemma 2]{rush2018finite}.
\end{proof}

We will use the following notation.  Given two random vectors $\xvec_1, \xvec_2$ and a sigma-algebra $\mscrs$, $\xvec_1 |_\mscrs \stackrel{d}{=} \xvec_2$ denotes that the conditional distribution of $\xvec_1$  given $\mscrs$ equals the distribution of $\xvec_2$. 

\subsection{Outline of the proof of Theorem \ref{thm:main1}} \label{subsec:proof_overview}

Theorem  \ref{thm:main1} gives  a concentration inequality for the normalized MSE of the AMP decoder, which  for iteration $(t+1)$ can be written as
\be
 \frac{1}{L} \| \betavec^{t+1} - \betavec \|^2= \frac{1}{L} \sum_{\sfc \in [\Lc]} \| \eta^{t}_\sfc( \betavec_\sfc - \hvec^{t+1}_\sfc) -\betavec_\sfc \|^2,
 \label{eq:MSE_rewrite}
\ee
where we recall that $\hvec^{t+1}_\sfc$ is the $\sfc$th column block of $\hvec^{t+1}  \in \reals^{ML}$. (We choose iteration $(t+1)$ rather than $t$ for notational convenience.) 

The proof of the theorem is based on showing that for $t \geq 0$, the vector $\hvec^{t+1}_\sfc$ is approximately Gaussian, for $\sfc \in [\Lc]$.
In particular, we show that  $\hvec^{t+1}_{\sfc}$ is approximately distributed as  $\sqrt{\tau^{t}_\sfc} \widetilde{\Z}_{t, {\sfc}}$, where $\widetilde{\Z}_{t, \sfc}$ is standard Gaussian and independent across $\sfc \in [\Lc]$.  If we assume that $\hvec^{t+1}_{\sfc}$ is \emph{exactly} distributed as $\sqrt{\tau^{t}_\sfc} \widetilde{\Z}_{t, {\sfc}}$ for  $\sfc \in [\Lc]$, then obtaining a concentration inequality for the MSE in \eqref{eq:MSE_rewrite} is straightforward. Indeed, for a fixed $\betavec$, the MSE $\frac{1}{L} \sum_{\sfc} \| \eta^{t}_\sfc( \betavec_\sfc - \hvec^{t+1}_\sfc) -\betavec_\sfc \|^2$ is a bounded and Lipschitz function of $\hvec^{t+1}$.  Therefore, if $\hvec^{t+1}$ were Gaussian, one could obtain a concentration inequality for the MSE of the AMP decoder via standard Gaussian concentration results \cite{BLMconcbook}. The bulk of the technical work is in precisely quantifying and controlling the deviation from Gaussianity of the vectors $\hvec^{t+1}$, for $t \geq 0$.

%writing $\hvec^{t} = [\hvec^{t}_{1}, \ldots, \hvec^{t}_{\Lc}]^*$, with $\hvec^{t}_{\sfc}$ denoting the $\sfc$th column block,  we claim that 
To study the distribution of $\hvec^{t+1}$, we use the recursion in \eqref{eq:bmq_rowcol_v2}, or equivalently, the matrix version in \eqref{eq:Alin_consts}. Note that  \eqref{eq:bmq_rowcol_v2} is a restatement of the dynamics of the AMP algorithm, although AMP cannot be run this way in practice (since  it is initialized with $\brqvec^{0}= -\betavec_0$ which is unknown).  The first ingredient in the proof is Lemma \ref{lem:hb_cond}, which specifies the conditional distribution of $\bvec^t_\sfr$ and $\hvec^{t+1}_\sfc$ given the past iterates of the algorithm in \eqref{eq:bmq_rowcol_v2}, for $\sfr \in [\Lr], \sfc \in [\Lc]$. More precisely, for $t\geq 0$, the lemma specifies the conditional distribution of    $\bvec^t_\sfr \lvert_{\mscrs_{t, t}}$ and $\hvec^{t+1}_\sfc \lvert_{\mscrs_{t+1, t}}$, where $\mscrs_{t_a, t}$ is the sigma-algebra generated by the collection of vectors
\be 
\bvec^0, ..., \bvec^{t_a -1}, \brmvec^0, ..., \brmvec^{t_a - 1}, \hvec^1, ..., \hvec^{t}, \brqvec^0, ..., \brqvec^{t}, \text{ and }  \betavec_0, \wvec. \label{eq:sig_tat} 
\ee
Lemmas \ref{lem:hb_cond} and \ref{lem:ideal_cond_dist} together show that the conditional distributions have the following form:
\begin{align}
 \bvec^t_{\sfr} \lvert_{\mscrs_{t, t}}& \stackrel{d}{=} 
\sum_{i=0}^{t} \frac{ \sigma^{t}_{\sfr} }{\sigma^{i}_{\sfr}} \left( \sqrt{\sigma^i_{\perp, \sfr}} \Zprime_{i, \sfr}  +  \Dvec_{i, i, \sfr} \right),  \quad \sfr \in [\Lr], \label{eq:bvec_dist} \\
\hvec^{t+1}_{\sfc} \lvert_{\mscrs_{t+1, t}} & \stackrel{d}{=} \sum_{i=0}^{t} \frac{\tau^{t}_{\sfc} }{\tau^{i}_\sfc} \left( \sqrt{\tau^i_{\perp, \sfc}}  \Z_{i, \sfc} + \Dvec_{i+1, i, \sfc}  \right), \quad \sfc \in [\Lc]. \label{eq:hvec_dist}
\end{align}
Here,  $\Zprime_{i} = [\Z_{i, \sf{1}}, \ldots,  \Z_{i, \sf{\Lr}}]^* \sim \mathcal{N}(0, \iden_{n})$ is independent of the sigma algebra $\mscrs_{i,i}$, and  $\Z_{i} = [\Z_{i, \sf{1}}, \ldots,  \Z_{i, \sf{\Lc}}]^* \sim \mathcal{N}(0, \iden_{ML})$ is independent of  $\mscrs_{i+1,i}$,  for $0 \leq i \leq t$. The deviation vectors $\Dvec_{i,i} = [\Dvec_{i,i,1}, \ldots, \Dvec_{i,i, \Lr}]^*$ and $\Dvec_{i+1,i} = [\Dvec_{i+1,i,1}, \ldots, \Dvec_{i+1,i, \Lc}]^*$  are measurable with respect to the sigma algebras $\mscrs_{i, i}$ and $\mscrs_{i+1, i}$, respectively. Their precise definitions are given in Lemma \ref{lem:hb_cond}.

If we ignore the deviation terms in \eqref{eq:hvec_dist}, then  $\hvec^{t+1}$ would be an i.i.d.\ Gaussian vector with  the variance of the entries equal to $\tau^t_\sfc \sum_{i=1}^t {\tau^i_{\perp, \sfc}}/{(\tau^i_\sfc)^2} = \tau_\sfc^t$. (The equality can be seen by using the definition of $\tau^i_{\perp, \sfc}$ in \eqref{eq:sigperp_defs}.) In this case, a concentration inequality for the MSE could be directly obtained using standard concentration results, as described above. 

The deviation terms in \eqref{eq:bvec_dist} and \eqref{eq:hvec_dist} are controlled via results in  Lemma \ref{lem:main_lem}, specifically the concentration results in  \eqref{eq:Ba1}, \eqref{eq:Ha} and \eqref{eq:Hb}.  The  definitions of the terms $ \Dvec_{t, t, \sfr}$ and $ \Dvec_{t+1, t, \sfc}$ (see \eqref{eq:Dtt} and \eqref{eq:Dt1t}) involve a combination of vectors that  are measurable with respect to $\mscrs_{t, t}$ and $\mscrs_{t+1, t}$, respectively.  We need several concentration results for scalar products involving these vectors in order to show that the deviation terms are negligible. Lemma \ref{lem:main_lem} lists all the required  concentration  results, which are proved  using an induction based argument in Section \ref{subsec:main_lem_proof}.

\subsection{Conditional distribution lemma} \label{sec:cond_dist_lemma}

For $t \geq 1$ and $t_a \in \{ t, t+1\}$, we recall that $\mathscr{S}_{t_a, t}$ be the sigma-algebra generated by the collection of vectors in \eqref{eq:sig_tat}. Furthermore, let $\mathscr{S}_{0, 0}$ and $\mathscr{S}_{1, 0}$ be the sigma-algebras generated by $\{\brqvec^0, \betavec_0, \wvec\}$ and $\{\bvec^0, \mvec^0, \brqvec^0, \betavec_0, \wvec\}$, respectively.  Given the vectors in \eqref{eq:sig_tat},  the vectors $\madj^{0, \sfc}, ..., \madj^{t_a - 1, \sfc}, \qadj^{0, \sfr}, ..., \qadj^{t, \sfr}$ are determined via \eqref{eq:tildem_def} for $\sfr \in [\Lr]$, $\sfc \in [\Lc]$. 

The conditional distribution on $\sfAmat$ given $\mathscr{S}_{t_a, t}$ is the same as the conditional distribution given the following linear constraints:
\be
\begin{split}
[ \sfAmat \Qadjmat^{\sfr}_{t_a} ]_{(\sfr, \cdot)} &= \Ymat_{t_a, \sfr},  \ \  \sfr \in [\Lr], \ \qquad  \ 
[ \sfAmat^* \Madjmat^{\sfc}_t ]_{(\sfc, \cdot)} =\Xmat_{t, \sfc}, \ \   \sfc \in [\Lc].
 \label{eq:lin_const_cond}
\end{split}
\ee
where $\Ymat_{t_a, \sfr}, \Xmat_{t, \sfc}$ are defined in \eqref{eq:Ytr_Xtc}, and $\Qadjmat^{\sfr}_{t_a}, \, \Madjmat^{\sfc}_{t}$  in \eqref{eq:tildeM_def}.  When conditioning on the linear constraints in \eqref{eq:lin_const_cond}, we emphasize that only $\sfAmat$ is treated as random. In the following lemma, we characterize the conditional distributions of the vectors $[\sfAmat^* \madj^{t, \sfc}]_{\sfc} |_{\mscrs_{t+1,t}}$ and $[\sfAmat \, \qadj^{t, \sfr}]_{\sfr} |_{\mscrs_{t,t}}$.  This result is then used in Lemma \ref{lem:hb_cond} to compute the conditional distributions of $\bvec^t |_{ \mscrs_{t, t}}$ and $\hvec^{t+1} |_{ \mscrs_{t+1, t}}$.

We write $\sfAmat_{\sfr\sfc}$ for the $(\sfr,\sfc)$$^{th}$ block of $\sfAmat$,  $\sfAmat_{(\sfr, \cdot)}$ for the $\sfr^{th}$ block of rows, and  $\sfAmat_{(\cdot, \sfc)}$ for the $\sfc$$^{th}$ block of columns of $\sfAmat$. We also recall that when $\sfr$ or $\sfc$ are used as a subscript, it refers to a row or column block of a larger vector, whereas when $\sfr$ or $\sfc$ are used as a superscript it represents a `full' vector whose entires depend on row block $\sfr$ or column index $\sfc$.
For a projection matrix $\proj^{\parallel} \in \mathbb{R}^{ML \times ML}$ we let $[\proj^{\parallel}]_{\sfc\sfc'} \in \mathbb{R}^{\Mc \times \Mc}$ be the sub-matrix of of 
$\proj^{\parallel}$ composed of the $\sfc$$^{th}$ block of rows and the $\sfc'$$^{th}$ block of columns.  The  sub-matrix $[\proj^{\parallel}]_{(\sfc, \cdot)} \in \mathbb{R}^{\Mc \times ML}$  is composed of the $\sfc$$^{th}$ block of rows and all columns and a similar definition is given to $[\proj^{\parallel}]_{(\cdot, \sfc)} \in \mathbb{R}^{ML \times \Mc}$.  For a projection matrix ${\proj}^{\parallel} \in \mathbb{R}^{n \times n}$ , the sub-matrices $[{\proj}^{\parallel}]_{\sfr \sfr'} \in \reals^{\Mr \times \Mr}$ and $[{\proj}^{\parallel}]_{(\sfr, \cdot)} \in \reals^{\Mr \times n}$ are similarly defined.

%%%% First part of Lemma 

\begin{lem}
The conditional distributions of the vectors in \eqref{eq:bmq_rowcol_v2} satisfy the following for  $\sfc \in [\Lc]$ and $\sfr \in [\Lr]$, assuming $n >t$ and $\Qadjmat^{\sfr}_{t+1}$ and $\Madjmat^{\sfc}_{t}$ have full column rank.
\be
 \big[ \sfAmat \, \qadj^{0, \sfr} \big]_{\sfr} \,  \big |_{\mscrs_{0,0}} \stackrel{d}{=} \bvec^0_{\sfr} \quad \text{ and } \quad  
  \big[ \sfAmat^* \madj^{0, \sfc} \big]_\sfc \, \big |_{\mscrs_{1,0}} \stackrel{d}{=}  
\sum_{\sfr \in [\Lr]} \left[  \big[   \proj^{\perp}_{\qvec^{0, \sfr}} \,  \sfAmat'^* \big]_{(\sfc, \sfr)} \,  + \,  \qvec^{0, \sfr}_{\sfc}  \norm{\qvec^{0, \sfr}}^{-2} (\bvec^0_{\sfr})^* \right] \madj^{0, \sfc}_{\sfr} ,
\label{eq:A_cond_0case}
 \ee
and for $t \geq 1$,
\begin{align}
&[\sfAmat\, \qadj^{t, \sfr}]_{\sfr} |_{\mscrs_{t,t}} \stackrel{d}{=} \Ymat_{t, \sfr} \, \gammavec^{t, \sfr} + \sum_{\sfc \in [\Lc]} \Big(  [\proj^{\perp}_{\Madjmat_t^{\sfc'}} \,  \widehat{\sfAmat}]_{(\sfr, \sfc)}  +    \Madjmat_{t,\sfr}^{\sfc} ((\Madjmat_t^{\sfc})^* \Madjmat_t^{\sfc})^{-1}(\Xmat_{t, \sfc})^*\Big) \qadj_{\perp, \sfc}^{t, \sfr}, \label{eq:Aq_cond} \\
&[\sfAmat^* \madj^{t, \sfc}]_{\sfc} |_{\mscrs_{t+1,t}} \stackrel{d}{=}  \Xmat_{t,  \sfc} \, \alphavec^{t, \sfc} + \sum_{\sfr \in [\Lr]}  \Big( [\proj^{\perp}_{\Qadjmat^{\sfr'}_{t+1}}\,  \sfAmat'^*]_{(\sfc, \sfr)}  +  \Qadjmat^{\sfr}_{t+1, \sfc}  ((\Qadjmat^{\sfr}_{t+1})^* \Qadjmat^{\sfr}_{t+1})^{-1} (\Ymat_{t+1, \sfr} )^*\Big) \madj^{t, \sfc}_{\perp, \sfr}. \label{eq:Astarm_cond} 
\end{align}
Here $\sfAmat' \stackrel{d}{=} \sfAmat$ and $\widehat{\sfAmat} \stackrel{d}{=} \sfAmat$ are random matrices independent of $\mscrs_{t+1,t}$ and $\mscrs_{t,t}$, and  $\proj^\perp_{ \Qadjmat}$ denotes the projection matrix onto the orthogonal complement of the space spanned by the columns of $ \Qadjmat$. 
\label{lem:Et1t}
\end{lem}
The proof of the lemma is given in Section \ref{subsec:lemEt1t_proof}.

%%%%%%%%%%%%%%%%%%%%%%%%%%%%%%%%%%%%%%%%%%%%%%%%%%%%%%%%%%%%%%%%%%%%%%%%%%%%

\begin{lem}[Conditional Distribution Lemma]
For the vectors $\hvec^{t+1}$ and $\bvec^t$ defined in \eqref{eq:hqbm_def}, the following hold for $t \geq 1$, provided $n >t$ and $\Madjmat_t^{\sfc},  \Qadjmat_t^{\sfr}$ have full column rank.
\begin{align}
\hvec^1_{\sfc} \lvert_{\mscrs_{1, 0}} \, \stackrel{d}{=} \, \sqrt{ \tau^{0}_{\perp, \sfc}} \, \Z_{0, \sfc}  + \Dvec_{1,0, \sfc}, \quad &\text{ and }  \quad \hvec^{t+1}_{\sfc} \lvert_{\mscrs_{t+1, t}} \, \stackrel{d}{=} \, \frac{\tau^t_{\sfc}}{\tau^{t-1}_{\sfc}} \hvec^t_{\sfc}  + \sqrt{ \tau^{t}_{\perp, \sfc}} \, \Z_{t,\sfc}   + \Dvec_{t+1,t, \sfc}, \label{eq:Ha_dist} \\ 
\bvec^{0}_{\sfr} \lvert_{\mscrs_{0, 0}} \, \stackrel{d}{=} \,   \sqrt{\sigma^{0}_{\perp, \sfr}} \, \Zprime_{0, \sfr}, \quad &\text{ and } \quad \bvec^t_{\sfr} \lvert_{\mscrs_{t, t}}
\, \stackrel{d}{=} \, \frac{\sigma^t_{\sfr}}{\sigma^{t-1}_{\sfr}}  \bvec^{t-1}_{\sfr}+ \sqrt{\sigma^{t}_{\perp, r}} \, \Zprime_{t, \sfr}  + \Dvec_{t,t, \sfr}.\label{eq:Ba_dist}
\end{align}
For  each $t \geq 0$, the vectors $\Z_{t} = [\Z_{t, \sf{1}}, \ldots,  \Z_{t, \sf{\Lc}}]^* \sim \mathcal{N}(0, \iden_{ML})$ and $\Zprime_{t} = [\Z_{t, \sf{1}}, \ldots,  \Z_{t, \sf{\Lr}}]^* \sim \mathcal{N}(0, \iden_{n})$
are  independent of the corresponding conditioning sigma algebras. The terms $\widehat{\gamma}^{t, \sfr}_{i}$ and $\widehat{\alpha}^{t,\sfc}_{i}$ for $i \in [t-1]$ are defined in \eqref{eq:hatalph_hatgam_def} and  $ \sigma^{\perp}_{t,\sfr}$ and $ \tau^{\perp}_{t,\sfc}$ are defined in \eqref{eq:sigperp_defs}.
The deviation terms are 
\begin{align}
\Dvec_{1,0, \sfc} &= \left[\frac{1}{\sqrt{L}} \norm{ \madj^{0, \sfc}}  - \sqrt{ \tau^{0}_{\sfc}} \right]  \Z_{0,\sfc} -  \sum_{\sfr \in [\Lr]} \frac{1 }{\sqrt{ L}} \norm{ \madj^{0, \sfc}_{\sfr} }   [\proj^{\parallel}_{\Qadjmat^{\sfr}_{1}}  \Zsupr_{0}]_{\sfc} + \brqvec^{0}_{\sfc}\left(\sum_{\sfr \in [\Lr]}  \frac{\sqrt{W_{\sfr\sfc}} (\bvec^0_{\sfr})^*\madj^{0, \sfc}_{\sfr}}{L \sigma^0_{\sfr}} -1\right), 
\label{eq:D10}
\end{align}
and for $t >0$,
\begin{align}
&\Dvec_{t,t, \sfr} = \sum_{i=0}^{t-2} \bvec^i_{\sfr}  \gamma^{t,\sfr}_{i}  \,+\, \bvec^{t-1}_{\sfr}   \left[\gamma^{t,\sfr}_{t-1} - \frac{\sigma^t_{\sfr}}{\sigma^{t-1}_{\sfr}} \right] \, + \,  \left[\frac{ 1}{\sqrt{L}}\norm{\qadj^{t,\sfr}_{\perp}} - \sqrt{\sigma_{\perp, \sfr}^{t}}\right] \Zprime_{t,\sfr} \, - \, \sum_{\sfc \in [\Lc]}  \frac{1}{\sqrt{L}}  \norm{\qadj^{t, \sfr}_{\perp, \sfc} }  \, [\proj^{\parallel}_{\Madjmat_t^{\sfc}} \, \Zprimesupc_{t}]_{\sfr} \nonumber \\
& \qquad  + \,  \sum_{\sfc \in [\Lc]} \Madjmat_{t,\sfr}^{\sfc} ((\Madjmat_t^{\sfc})^* \Madjmat_t^{\sfc})^{-1}(\Xmat_{t,\sfc})^*\, \qadj_{\perp, \sfc}^{t, \sfr} \,- \,  \sum_{i=1}^{t-1}\gamma^{t,\sfr}_{i} \upsilon^i_{\sfr} \brmvec^{i-1}_{\sfr}   \,  + \, \upsilon^t_{\sfr}  \brmvec^{t-1}_{\sfr} ,\label{eq:Dtt} \\
&\Dvec_{t+1,t, \sfc} =   \sum_{i=0}^{t-2} \hvec^{i+1}_{\sfc} \alpha^{t,\sfc}_{i}  + \hvec^t_{\sfc}  \left[\alpha^{t,\sfc}_{t-1} - \frac{\tau^t_{\sfc}}{\tau^{t-1}_{\sfc}} \right]  + \left[\frac{1}{\sqrt{L}} \norm{ \madj^{t, \sfc}_{\perp} }   - \sqrt{\tau^{t}_{\perp, \sfc}} \right]  \Z_{t,\sfc} -  \sum_{\sfr \in [\Lr]} \frac{1 }{\sqrt{L}}    \norm{ \madj^{t, \sfc}_{\perp, \sfr} }   [\proj^{\parallel}_{\Qadjmat^{\sfr}_{t+1}}  \Zsupr_{t}]_{\sfc}  \nonumber \\
&\qquad  + \, \sum_{\sfr \in [\Lr]} \Qadjmat^{\sfr}_{t+1, \sfc}  ((\Qadjmat^{\sfr}_{t+1})^* \Qadjmat^{\sfr}_{t+1})^{-1} (\Ymat_{t+1, \sfr} )^* \madj^{t, \sfc}_{\perp, \sfr}\, + \,  \sum_{i=0}^{t-1}\alpha^{t,\sfc}_{i} \brqvec^i_{\sfc}    \,- \, \brqvec^t_{\sfc}.\label{eq:Dt1t}  
\end{align} 
In \eqref{eq:Dtt}, we recall that $\upsilon^i_{\sfr}  = {\sigma_\sfr^i}/{\phi_\sfr^{i-1}}$ for $i \in [t]$. The vectors $\Zprimesupc_{t} \sim \mathcal{N}(0, \iden_{n})$ are  i.i.d.\ for  $\sfc \in [\Lc]$, and  independent of $\mscrs_{t,t}$. Similarly,  in \eqref{eq:D10} and \eqref{eq:Dt1t},  $\Zsupr_{t} \sim \mathcal{N}(0, \iden_{ML})$ are i.i.d.\ for  $\sfr \in [\Lr]$ and  independent of $\mscrs_{t+1,t}$. Furthermore,  $\Zprime_{t} = \frac{1}{\sqrt{\Lc}} \sum_{\sfc \in [\Lc]} \Zprimesupc_{t}$ and $\Z_{t}  = \frac{1}{\sqrt{\Lr}} \sum_{\sfr \in [\Lr]} \Zsupr_{t}$. 
\label{lem:hb_cond}
\end{lem}

The proof of the lemma is given in Section \ref{subsec:lemhb_cond_proof}.

The next lemma uses the representation in Lemma \ref{lem:hb_cond} to show that $\hvec^{t+1}_{\sfc}$ is the sum of a $\mc{N}(0, \tau^t_{\sfc} \, \iden_{\Mc})$ random vector and a deviation term for each $t \geq 0$.  Similarly $\bvec^t_{\sfr}$ is the sum of a  $\mc{N}(0, \sigma^t_{\sfr} \, \iden_{\Mr})$ random vector and a deviation term.  
\begin{lem}
For $t \geq 0$, the conditional distributions in Lemma \ref{lem:hb_cond} can be expressed as 
\be
\hvec^{t+1}_{\sfc} \lvert_{\mscrs_{t+1, t}} \stackrel{d}{=} \widetilde{\hvec}^t_{\sfc} + \widetilde{\Dvec}_{t+1, {\sfc}} , \qquad  \bvec^t_{\sfr} \lvert_{\mscrs_{t, t}}\stackrel{d}{=} \breve\bvec^t_{\sfr} + \breve{\Dvec}_{t, {\sfr}},
\label{eq:htil_rep}
\ee
where
\begin{align}
& \widetilde{\hvec}^{t+1}_{\sfc} := \tau^{t}_{\sfc} \sum_{i=0}^{t} \sqrt{\tau^i_{\perp, \sfc}} \left(\frac{1}{\tau^{i}_{\sfc}}\right) \Z_{i, {\sfc}}, 
\qquad \widetilde{\Dvec}_{t+1, {\sfc}} := \tau^{t}_{\sfc} \sum_{i=0}^{t} \left(\frac{1}{\tau^{i}_{\sfc}}\right) \Dvec_{i+1, i, {\sfc}},  \label{eq:htilde_def} \\
& \breve\bvec^t_{\sfr}:= \sigma^{t}_{\sfr} \sum_{i=0}^{t} \sqrt{\sigma^i_{\perp, \sfr}} \left(\frac{1}{\sigma^{i}_{\sfr}}\right) \Zprime_{i, {\sfr}},
\qquad \breve{\Dvec}_{t, {\sfr}} := \sigma^{t}_{\sfr} \sum_{i=0}^{t} \left(\frac{1}{\sigma^{i}_{\sfr}}\right) \Dvec_{i, i, {\sfr}}.
\label{eq:btilde_def}
\end{align}
Here, for each $0 \leq i \leq t$,  the standard Gaussian vectors $\Z_{i} = [\Z_{i, \sf{1}}, \ldots,  \Z_{i, \sf{\Lc}}]^* \sim \mathcal{N}(0, \iden_{ML})$ and $\Zprime_{i} = [\Z_{i, \sf{1}}, \ldots,  \Z_{i, \sf{\Lr}}]^* \sim \mathcal{N}(0, \iden_{n})$
are the ones used in Lemma \ref{lem:hb_cond}, Eqs.\ \eqref{eq:Ha_dist} and \eqref{eq:Ba_dist}.

Consequently, $\widetilde{\hvec}^{t+1}_{\sfc}  \stackrel{d}{=} \sqrt{ \tau^t_{\sfc} }\widetilde{\Z}_{t, {\sfc}}$, and $\breve\bvec^t_{\sfr}  \stackrel{d}{=} \sqrt{\sigma^t_{\sfr}} \breve{\Z}_{t, {\sfr}}$, where $\widetilde{\Z}_{t} = [\widetilde{\Z}_{t, \sf{1}}, \ldots,  \widetilde{\Z}_{t, \sf{\Lc}}]^* \sim \mathcal{N}(0, \iden_{ML})$ and $  \breve{\Z}_{t} = [\breve{\Z}_{t, \sf{1}}, \ldots,  \breve{\Z}_{t, \sf{\Lr}}]^* \sim \mathcal{N}(0, \iden_{n})$ such that for any $j  \in [\Mc]$ and  $i \in [\Mr]$,  the length-$t$ vectors $([\widetilde{Z}_{0, {\sfc}}]_j, \ldots, [\widetilde{Z}_{t, {\sfc}}]_j)$ and  $( [\breve{Z}_{0, {\sfr}}]_i , \ldots, [\breve{Z}_{t, {\sfr}}]_i)$ are each jointly Gaussian with 
\be \expec\{[\widetilde{Z}_{\widetilde{s}, {\sfc}}]_j [\widetilde{Z}_{s, {\sfc}}]_j\} = \sqrt{{\tau^{\widetilde{s}}_{\sfc}}/{\tau^s_{\sfc}}}, \qquad  \expec\{ [\breve{Z}_{\widetilde{s}, {\sfr}}]_i [\breve{Z}_{s, {\sfr}}]_i\} = \sqrt{{\sigma^{\widetilde{s}}_{\sfr}}/{\sigma^s_{\sfr}}} \qquad \text{ for }0 \leq s \leq \widetilde{s} \leq t. \ee
Hence for any $s \leq t$ we can write 
\be
\widetilde{\Z}_{t, \sfc} \overset{d}{=} \widetilde{\Z}_{s, \sfc}  \sqrt{\frac{\tau^{t}_{\sfc}}{\tau^{s}_{\sfc}}}  +\widetilde{\mathbf{U}}_{t, \sfc}\sqrt{1  -  \frac{\tau^{t}_{\sfc}}{\tau^{s}_{\sfc}}}, \quad \text { and } \quad \breve{\Z}_{t, \sfr} \overset{d}{=} \breve{\Z}_{s, \sfr}  \sqrt{\frac{\sigma^{t}_{\sfr}}{\sigma^{s}_{\sfr}}}  +\breve{\mathbf{U}}_{t, \sfr}\sqrt{1  -  \frac{\sigma^{t}_{\sfr}}{\sigma^{s}_{\sfr}}},
\ee
where $\widetilde{\Z}_{s}$ is independent of $\widetilde{\mathbf{U}}_{t} = [\widetilde{U}_{t, \sf{1}}, \ldots,  \widetilde{U}_{t, \sf{\Lc}}]^* \sim \mathcal{N}(0, \iden_{ML})$ and $ \breve{\Z}_{s}$ is independent of $\breve{\mathbf{U}}_{t} = [\breve{U}_{t, \sf{1}}, \ldots,  \breve{U}_{t, \sf{\Lr}}]^* \sim \mathcal{N}(0, \iden_{n})$.
\label{lem:ideal_cond_dist}
\end{lem}

\begin{proof}
The proof is similar to that of \cite[Lemma 6]{RushV19} and is omitted.
\end{proof}

\subsection{Order of SPARC parameters and state evolution constants}

We recall a few facts about the SC-SPARCs construction from Section \ref{sec:sc_AMP} that will be used throughout the proof.   There are $L/\Lc$ sections per column block of $\betavec$, with the non-zero coefficient in each section equal to $1$. Each block in the design matrix $\Amat$ can be viewed as a standard (non-SC) SPARC with $L/\Lc$ sections (with $M$ columns in each section), code length $\Mr = {n}/{\Lr}$, and rate $$R_{\text{inner}} = \frac{(L/\Lc) \ln M}{(n/\Lr)} = R \cdot \frac{\Lr}{\Lc}.$$ For an $(\omega, \Lambda, \rho)$ base matrix, $\Lc= \Lambda$ and  $\Lr = (\Lambda + \omega - 1)$. Since $L \gg  \Lc = \Lambda \gg \omega$, we have $\frac{\Lr}{\Lc} = 1 + \frac{\omega -1}{\Lambda} < 2$, hence $R < R_{\text{inner}} < 2R$.
%we observe that $R_{\text{inner}}$ is slightly larger than $R$.   In particular,  

From \eqref{eq:W_rc}, we have $\max_{\sfr, \sfc} \, W_{\sfr, \sfc} \leq P  \frac{\Lr}{\omega}$.
 From \eqref{eq:W_rc}, it can also be verified that for 
$\sfr \in [\Lr]$, 
\begin{align}
\frac{1}{\Lc} \sum_{\sfc \in [\Lc]} W_{\sfr \sfc} =  \kappa_1,  \qquad 
\frac{1}{\Lc} \sum_{\sfc \in [\Lc]} W^2_{\sfr \sfc} = \kappa_2 \Big( \frac{\Lr}{\omega} \Big),
\qquad
\frac{1}{\Lc} \sum_{\sfc \in [\Lc]} W^4_{\sfr \sfc} = \kappa_3 \Big( \frac{\Lr}{\omega} \Big)^3,
\label{eq:Wrc_avgs}
\end{align}
where $\kappa_1, \kappa_2, \kappa_3$ are absolute positive constants. A similar statement holds when the summations in \eqref{eq:Wrc_avgs} are over $\sfr \in [\Lr]$, with $\sfc$ fixed.

 Studying the state evolution equations \eqref{eq:se_phi}-\eqref{eq:se_psi}, we see that $\psi^t_{\sfc}, \sigma^t_{\sfr},$ and $\phi^t_{\sfr}$ are all $\Theta(1)$ for all $t \geq 0, \sfr \in [\Lr]$, and $\sfc \in [\Lc],$ while $\tau^t_{\sfc} = \Theta(1/\ln M) = \Theta(L/n)$ for all $t \geq 0$ and $\sfc \in [\Lc]$. This implies that $S^t_{\sfr \sfc} = \tau^t_{\sfc} (\phi^t_{\sfr})^{-1} = \Theta(1/\ln M) = \Theta(L/n)$.  
We will use the following facts implicitly.  For $t \geq 0,$
$
\min_{\sfr \in [\Lr]} \phi^t_{\sfr} =  \phi^t_1$ and $ \min_{\sfr \in [\Lr]} \sigma^t_{\sfr} = \sigma^t_1,
%\label{eq:phi_sigma_min}
$
additionally,
$
\max_{\sfr \in [\Lr]} \phi^t_{\sfr} =  \phi^t_{\lfloor \Lr/2\rfloor}$ and $\max_{\sfr \in [\Lr]} \sigma^t_{\sfr} = \sigma^t_{\lfloor \Lr/2\rfloor}.
%\label{eq:phi_sigma_max}
$
Furthermore, 
$
\max_\sfr \, \sigma^t_\sfr \leq \max_\sfr \frac{1}{\Lc} \sum_{\sfc = 1}^{\Lc} W_{\sfr \sfc} \leq 2 P.
$
It follows that for $\sfc \in [\Lc]$,
\be
S^t_{\sfr \sfc} = \tau^t_{\sfc} (\phi^t_{\sfr})^{-1} \leq \tau^t_{\sfc} (\phi^t_1)^{-1} = S^t_{1, \sfc}.
\label{eq:S_max}
\ee

%%%%%%%
%%%%%%%

\subsection{Concentration lemma}
 
The next lemma uses the conditional distribution given by Lemma \ref{lem:hb_cond} to  prove concentration results for various inner products and functions involving $\{\bvec^t, \madj^{t, \sfc}, \hvec^{t+1}, \qadj^{t, \sfr} \}$.  The concentration lemma is stated in two parts. The first part gives concentration inequalities for inner products involving the vectors 
$\{ \bvec^s, \madj^{s, \sfc} \}_{0 \leq s \leq t}$ (Eqs.\ \eqref{eq:Bb1}--\eqref{eq:Bg}). The second part gives concentration inequalities for inner products involving the vectors 
$\{ \hvec^{s+1}, \qadj^{s, \sfr} \}_{0 \leq s \leq t}$ (Eqs.\ \eqref{eq:Hb}--\eqref{eq:Hg}).  These results are proved using an induction argument that includes two other concentration inequalities (Eqs.\ \eqref{eq:Ba1} and \eqref{eq:Ha}) showing that that the deviation terms in Lemma \ref{lem:hb_cond} are small. 

The proof of Theorem \ref{thm:main1} requires only one of the results in the concentration lemma, namely, Eq.\ \eqref{eq:Hc_adj}. However, the other results in the lemma are required for the induction argument. 

%, in the sense that their block-wise maximum absolute value and norm concentrate around $0$.  Lemma \ref{lem:main_lem} also provides  

To keep the notation compact, we use $K, K',  \kappa,$ and $\kappa'$ to denote generic positive universal  constants whose values may change throughout the lemma statement and proof. 

Let  $\xi, \Xi >0$ be universal constants not depending on $n$, $\e$, or $t$.  
For $t \geq 0$, let
\be K_{t} = \Xi^{2t} (t!)^{14} , \quad \kappa_{t} = \frac{1}{\xi^{2t} (t!)^{24}},  \quad K_{t}' = \Xi(t+1)^{7} K_{t} , \quad \kappa_{t}' = \frac{\kappa_{t}}{\xi (t+1)^{12}}. \label{eq:Ktkapt_defs} \ee
We also define the following iteration-dependent quantities that summarize the problem parameters:
\be \PC_{t} = (\Lr \Lc)^{t+1}, \quad \pc_{t} = \frac{ (\Mr \omega)(\omega/\Lr)^{2(t+1)}  }{(\log M)^{2(t+1)}},  \quad \PC_{t}' =  \Lc \PC_{t} , \quad \pc_{t}'  = \pc_{t}. \label{eq:param_const_defs} \ee
%

%\RV{check these after concentration results are finalized}
\begin{lem}
The following results hold for $\e \in (0,1)$ and $1 \leq t < T$, where $T$ is defined in \eqref{eq:Tdef}.

$\mathbf{1}$.  Let $\sfp$ be an integer with $\sfp \in \{ 0,1,2 \}$.  Let $X_n \doteq \textsf{const}$ be shorthand for $$P(\abs{X_n - \textsf{const}} \geq \e) \leq   t^4   K K_{t-1}  \PC_{t-1}    \exp 
\Big\{\frac{- \kappa \kappa_{t-1}  (\omega/\Lr)^{2(\sfp-1)_{+}}  \pc_{t-1}  \e^2}{t^8  }  \Big\}.$$

\textbf{(a)} For $\sfc \in [\Lc]$
\begin{align}
& P \Big(\frac{1}{n}\sum_{\sfr \in [\Lr]} W^{\sfp}_{\sfr \sfc} \norm{\Dvec_{t,t, \sfr}}^2 \geq \epsilon \Big)  \leq t^3   K K_{t-1}   \PC_{t-1}   \exp\Big\{\frac{- \kappa \kappa_{t-1}   (\omega/\Lr)^{(\sfp-1)_{+}}  \pc_{t-1}  \e}{t^6  }  \Big\},  \label{eq:Ba1}
\end{align}
%

%-----------------------------
\textbf{(b)}  For all $\sfc \in [\Lc]$,
\begin{align}
 \frac{1}{n} \sum_{\sfr \in [\Lr]} W^{\sfp}_{\sfr \sfc} (\bvec^t_{\sfr})^* \wvec_{\sfr} &\doteq 0 ,  \label{eq:Bb1} 
\end{align}
%------------------------

\textbf{(c)}  For all $0 \leq s \leq t$ and $\sfc \in [\Lc]$,
\begin{align}
\frac{1}{n} \sum_{\sfr \in [\Lr]} W^{\sfp}_{\sfr \sfc} (\bvec^s_{\sfr})^*\bvec^t_{\sfr} &\doteq \frac{1}{\Lr} \sum_{\sfr \in [\Lr]} W^{\sfp}_{\sfr \sfc} \sigma^t_{\sfr}, \label{eq:Bc}
\end{align}
%-------------------------

\textbf{(d)}  For all $0 \leq s, \widetilde{s} \leq t$  and $\sfc \in [\Lc]$,
\be
\frac{1}{L}  \sum_{\sfr \in [\Lr]} W^{\sfp - 1/2}_{\sfr \sfc}(\bvec^{\widetilde{s}}_{\sfr})^*\madj^{s, \sfc}_{\sfr} \doteq  \frac{\Mr}{L}  \sum_{\sfr \in [\Lr]} S^s_{\sfr \sfc} W^{\sfp}_{\sfr \sfc} \, \sigma^{\max(\widetilde{s},s)}_{\sfr} , \label{eq:Bd}
\ee
%-----------------------------

\textbf{(e)} 
For all $0 \leq s \leq t$ and $\sfc \in [\Lc]$,
\begin{align}
\frac{n }{L^2} \sum_{\sfr \in [\Lr]} W^{\sfp-1}_{\sfr \sfc} (\madj^{s, \sfc}_{\sfr})^* \madj^{t, \sfc}_{\sfr} &\doteq \frac{n^2 }{L^2 \Lr}  \sum_{\sfr \in [\Lr]}  S^s_{\sfr \sfc} \, S^t_{\sfr \sfc} \, W_{\sfr \sfc}^{\sfp} \, \phi^t_{\sfr}   \quad (= \frac{n}{L}\tau_{\sfc}^t \text{ when } \sfp = 1), \label{eq:Be}
\end{align}
%

%------------------------

\textbf{(f)} For $\tau^t_{\perp, \sfc}$ defined in \eqref{eq:sigperp_defs} and shown to be positive in Lemma \ref{lem:sigmatperp} and  $\widehat{\alphavec}^{t, \sfc}$  defined in \eqref{eq:vec_alph_gam_conc}, for all $\sfc \in [\Lc]$ and $0 < i \leq t$,
\begin{align}
&P\Big(\lvert \alpha^{t, \sfc}_{i-1} - \widehat{\alpha}^{t, \sfc}_{i-1} \lvert \geq \epsilon\Big) \leq   t^5  K K_{t-1}  \PC_{t-1}    \exp\Big\{ \frac{-\kappa \kappa_{t-1}  \pc_{t-1} \e^2}{t^{10}  }  \Big\}, \quad t \geq 1, \label{eq:Bf} \\
&P\Big(\frac{n}{L}\Big\lvert \frac{1}{L}\| \madj^{t, \sfc}_{\perp} \|^2 -  \tau^{t}_{\perp, \sfc} \Big \lvert \geq \epsilon \Big) \leq  t^6  K K_{t-1}  \PC_{t-1}    \exp\Big\{ \frac{- \kappa \kappa_{t-1} \pc_{t-1} \e^2}{t^{12}}   \Big\}. \label{eq:Bf1}
\end{align}

%-----------------------------

\textbf{(g)}  Let $\innerM^{\sfc}_{t+1} := \frac{n}{L^2}(\Madjmat^{\sfc}_{t+1})^* \Madjmat^{\sfc}_{t+1}$.  Then for  $\sfc \in [\Lc]$,
\be
P\Big(\innerM^{\sfc}_{t+1} \text{ is singular}\Big) \leq  t^7  K K_{t-1}  \PC_{t-1}    \exp
\Big\{\frac{- \kappa \kappa_{t-1} \pc_{t-1}}{t^{12} }  \Big\}. \label{eq:Msing}
\ee
For matrix $\brCmat^{t+1, \sfc}$ defined in \eqref{eq:Ct_def}, when the inverses of $\innerM^{\sfc}_{t+1} $ exist,
\begin{align}
&P\Big(\Big \lvert [ (\innerM^{\sfc}_{t+1} )^{-1} - (\brCmat^{t+1, \sfc})^{-1}]_{i j} \Big \lvert \geq \epsilon\Big) \leq t^6   K K_{t-1}  \PC_{t-1}    \exp\Big\{ \frac{- \kappa \kappa_{t-1} \pc_{t-1} \e^2 }{t^{12}}
\Big\},  \label{eq:Bg}
\end{align}
%
%Moreover,  all non-zero elements $[(\brCmat^{t+1, \sfc})^{-1}]_{i j}$ for $1 \leq i, j \leq (t+1)$ have $[(\brCmat^{t+1, \sfc})^{-1}]_{i j}\in \Theta(1)$.
%\RV{may need to rephrase}

\vspace{10pt}

$\mathbf{2}$.  For $\sfq \in \{ 0,1\}$, we have the following:

\textbf{(a)} For all $\sfr \in [\Lr]$,
\begin{align}
P\Big(\frac{1}{L} \sum_{\sfc \in [\Lc]} W^{2 \sfq}_{\sfr \sfc} \sum_{\ell \in \sfc}  \max_{j \in sec(\ell)}    \abs{[\Dvec_{t+1,t, \sfc}]_{j}}^2 \geq \epsilon \Big) 
& \leq  t^3 K K'_{t-1}   \PC_{t-1}'  \exp\Big\{\frac{-\kappa \kappa'_{t-1} (\omega/\Lr)^{2\sfq}  \pc_{t-1}'  \e}{t^6 } \Big\}.  \label{eq:Ha} 
\end{align}
%

%-----------------------------

\textbf{(b)} For all $\sfr \in [\Lr]$,
\be
P\Big( \Big \lvert \frac{1}{L}   \sum_{\sfc \in [\Lc]}  \sqrt{W_{\sfr \sfc}}(\hvec^{t+1}_{\sfc})^* \qadj^{0, \sfr}_{\sfc} \Big \lvert \geq \e \Big) \leq t^3 K K'_{t-1}  \PC_{t-1}'   \exp \Big\{
\frac{-\kappa \kappa'_{t-1} (\omega/\Lr)^{2}   \pc'_{t-1} \e^2}{t^6}  \Big\}. \quad \label{eq:Hb} 
\ee
%

%------------------------

\textbf{(c)} For all $0 \leq s \leq t+1$ and $\sfr \in [\Lr]$,
\begin{align}
P\Big( \Big \lvert \sum_{\sfc \in [\Lc]} W^{\sfq}_{\sfr \sfc}\Big[ \frac{1}{L }(\brqvec^{s}_{\sfc})^* \brqvec^{t+1}_{\sfc}   -   \frac{1}{\Lc} \psi^{t+1}_{\sfc}\Big] \Big \lvert \geq \e \Big) \leq  t^4 K K'_{t-1}  \PC_{t-1}'   \exp\Big\{\frac{-\kappa \kappa'_{t-1} (\omega/\Lr)^{ 2\sfq} \pc_{t-1}'\e^2}{t^8  (\log M)^{2} }\Big\}, \label{eq:Hc_adj}
\end{align}
where for $\sfq=1$, we note that $\frac{1}{\Lc} \sum_{\sfc \in [\Lc]} W_{\sfr \sfc} \psi^{t+1}_{\sfc} = \sigma_{\sfr}^{t+1}$.

%-------------------------

\textbf{(d)} For all $0 \leq s, \widetilde{s} \leq t$ and $\sfr \in [\Lr]$, 
\begin{align}
&P\Big( \Big \lvert \sum_{\sfc \in [\Lc]}    \sqrt{W_{\sfr \sfc}} \Big[   \frac{(\hvec^{\widetilde{s}+1}_{\sfc})^*\qadj^{s+1, \sfr}_{\sfc}}{L} + \frac{\sqrt{W_{\sfr \sfc}} \psi^{s+1}_{\sfc} \tau^{\max(\widetilde{s}, s)}_{\sfc}}{\Lc \tau^{s}_{\sfc}}\Big] \Big \lvert \geq \e \Big) \leq t^4 K K'_{t-1}  \PC_{t-1}'   \exp\Big\{\frac{-\kappa \kappa'_{t-1} (\omega/\Lr)^{2}  \pc_{t-1}' \e^2}{t^8  (\log M)^{2} }\Big\} . \label{eq:Hd} 
\end{align}
%
%-----------------------------

\textbf{(e)}
\begin{align}
& P\Big(  \sum_{\sfc \in [\Lc]}  \sum_{\ell \in \sfc}  \frac{W^{2\sfq}_{\sfr \sfc}}{L} \max_{j \in sec(\ell)}  ([\hvec^{t+1}_{\sfc}]_{j})^2 \geq 6 (\max_{\sfc} W^{2\sfq}_{\sfr \sfc})\tau^t_{\sfc} \log M + \e \Big)  \nonumber \\
&  \leq  t^4 K K'_{t-1}   \PC_{t-1}'    \exp\Big\{ \frac{-  \kappa \kappa'_{t-1}(\omega/\Lr)^{2\sfq} \pc'_{t-1} \e }{t^8}\Big\}.    \label{eq:He}
\end{align}
%\RV{check -- may not need an $\e$ in the bound}
%------------------------

\textbf{(f)} For $\sigma^{t+1}_{\perp, \sfr}$ defined in \eqref{eq:sigperp_defs} and shown to be positive in Lemma \ref{lem:sigmatperp}, and $\widehat{\gamma}^{t+1, \sfr}$ is defined in \eqref{eq:vec_alph_gam_conc}, for all $\sfr \in [\Lr]$ and $0 \leq i \leq (t+1)$,
\begin{align}
&P\Big(\lvert \gamma^{t+1, \sfr}_{i-1} - \widehat{\gamma}^{t+1, \sfr}_{i-1} \lvert \geq \epsilon \Big) \leq  t^5 K K'_{t-1}   \PC_{t-1}'   \exp\Big\{\frac{-\kappa \kappa'_{t-1} (\omega/\Lr)^{2} \pc'_{t-1}  \e^2}{t^{10}  (\log M)^{2} }\Big\}. \label{eq:Hf}  \\
&P\Big(\Big\lvert \frac{1}{L}\| \qadj^{t+1, \sfr}_{\perp} \|^2 - \sigma^{t+1}_{\perp, \sfr} \Big \lvert \geq \epsilon \Big) \leq t^6 K K'_{t-1} \PC_{t-1}'   \exp\Big\{\frac{-\kappa \kappa'_{t-1}(\omega/\Lr)^{2}  \pc'_{t-1} \e^2}{t^{12} (\log M)^{2} }\Big\}.\label{eq:Hf1}
\end{align}
%-----------------------------

\textbf{(g)}  Let $\innerQ^{\sfr}_{t+2} := \frac{1}{L} (\Qadjmat^{\sfr}_{t+2})^* \Qadjmat^{\sfr}_{t+2}$.  Then for all $\sfr \in [\Lr]$,
\be
P(\innerQ^{\sfr}_{t+2} \text{ is singular}) \leq t^7 K K'_{t-1}   \PC_{t-1}'   \exp\Big\{\frac{-\kappa \kappa'_{t-1} 
  (\omega/\Lr)^{2}  \pc'_{t-1} \e^2}{t^{12} (\log M)^{2} }\Big\}, \label{eq:Qsing}
\ee
For matrix $\tCmat^{t+2, \sfr}$ defined in \eqref{eq:Ct_def}, when the inverses of $\innerQ^{\sfr}_{t+2}$ exist, 
\begin{align}
&P\Big(\Big \lvert [(\innerQ^{\sfr}_{t+2})^{-1} - (\tCmat^{t+2, \sfr})^{-1}]_{i j} \Big \lvert \geq \epsilon\Big) \leq   t^6 K K'_{t-1}   \PC_{t-1}'    \exp\Big\{\frac{-\kappa \kappa'_{t-1}(\omega/\Lr)^{2} \pc'_{t-1}  \e^2}{t^{12} (\log M)^{2} }\Big\}. \label{eq:Hg}
\end{align}
%
%Moreover,  all non-zero elements $[(\tCmat^{t+2, \sfr})^{-1}]_{i j}$ for $1 \leq i, j \leq (t+2)$ have $[(\tCmat^{t+2, \sfr})^{-1}]_{i j}\in \Theta(1)$. \RV{may need to rephrase, make sure $\Theta(1)$ has been properly defined.}

\label{lem:main_lem}
\end{lem}

%As we describe below, Theorem \ref{thm:main1} is an immediate consequence of   \eqref{\eqref{eq:Hc_adj}, applied with $\sfq=0$ and $s=t$. The other concentration results in the above lemma are required for the induction argument to prove that the deviation term , \eqref{eq:Ha}

\subsection{Proof of Theorem \ref{thm:main1}} \label{subsec:thm_main1_proof}

Recall from \eqref{eq:hqbm_def} that $ \brqvec^{s}  = \betavec^s - \betavec_{0}$. With $\sfq=0$, Eq.\ \eqref{eq:Hc_adj} implies that  for $0 \leq s \leq (T-1)$, 
\be
\begin{split}
P\Big(  \frac{1}{L } \| \betavec^{s+1} - \betavec_0 \|^2 - \frac{1}{\Lc}  \sum_{\sfc \in [\Lc]}   \psi^{s+1}_{\sfc} \Big \lvert \geq \e \Big) 
& \leq s^4 K K'_{s-1}  (\Lr \Lc)^{s} \Lc   \exp\Big\{\frac{-\kappa \kappa'_{s-1} (\Mr \omega) \, (\omega/\Lr)^{2s} \, \e^2}{t^8  (\log M)^{2s+2} }\Big\} \\
& \leq K_s (\Lr \Lc)^{s+1} \exp 
\Big\{ \frac{- \kappa_s (\Mr \omega) \e^2}{ (\log M)^{2s+2} (\Lr/\omega)^{2s}} \Big\}.
\end{split}
\ee
Setting $s=(t-1)$ and recalling that $\Mr = n /\Lr $  yields the statement of Theorem \ref{thm:main1}. \qed

\begin{rem}
\normalfont
Although Theorem \ref{thm:main1} is an immediate consequence of  \eqref{eq:Hc_adj}, the other concentration results in Lemma \ref{lem:main_lem} are required for the induction based proof of  \eqref{eq:Hc_adj}. Indeed, we write
\be
\frac{1}{L } \|  \brqvec^{s+1}_{\sfc} \|^2 =  \frac{1}{L} \sum_{\sfc \in [\Lc]} \| \eta^{t}_\sfc( \betavec_\sfc - \hvec^{s+1}_\sfc) -\betavec_\sfc \|^2,
\label{eq:MSE_exp2}
\ee
and use the representation for $\hvec^{s+1}_\sfc$ from Lemma \ref{lem:hb_cond}.  The concentration results in the  Lemma \ref{lem:main_lem} are used in an induction argument to show that the deviation term is negligible in the sense of \eqref{eq:Ha}. Consequently, $\hvec^{s+1}_\sfc$ is approximately Gaussian, and hence $\frac{1}{L} \sum_{\sfc \in [\Lc]} \| \eta^{t}_\sfc( \betavec_\sfc - \hvec^{s+1}_\sfc) -\betavec_\sfc \|^2$ concentrates on a deterministic value, as described in Section \ref{subsec:proof_overview}. The detailed proof of Lemma \ref{lem:main_lem} is given in Section \ref{subsec:lemhb_cond_proof}.
\end{rem}
%\RV{fixed: check after finalising statement of concentration lemma.}

% !TEX root =  sc_sparc_journal_paper_full.tex

\section{Proofs of conditional distribution and concentration lemmas} \label{sec:cond_conc_lemmas}

\subsection{Proof of  Lemma \ref{lem:Et1t}} \label{subsec:lemEt1t_proof}

For $t_a \geq 1$, let $\proj^{\parallel}_{\Qadjmat^{\sfr}_{t_a}}$ and $\proj^{\perp}_{\Qadjmat^{\sfr}_{t_a}}$ denote the  orthogonal projectors onto the column space of $\Qadjmat^{\sfr}_{t_a}$ and its orthogonal complement, respectively. Given the constraints in \eqref{eq:lin_const_cond}, for $\sfr \in [\Lr]$, we can  write
\be
\label{eq:Ardot}
\begin{split}
\sfAmat_{(\sfr, \cdot)} 
 = \sfAmat_{(\sfr, \cdot)}\left(\proj^{\perp}_{\Qadjmat^{\sfr}_{t_a}} + \proj^{\parallel}_{\Qadjmat^{\sfr}_{t_a}}\right) & = \sfAmat_{(\sfr, \cdot)}\left(\proj^{\perp}_{\Qadjmat^{\sfr}_{t_a}} +\Qadjmat^{\sfr}_{t_a} ((\Qadjmat^{\sfr}_{t_a})^* \Qadjmat^{\sfr}_{t_a})^{-1}(\Qadjmat^{\sfr}_{t_a})^* \right) \\
&= \sfAmat_{(\sfr, \cdot)}\proj^{\perp}_{\Qadjmat^{\sfr}_{t_a}} + \Ymat_{t_a, \sfr}  ((\Qadjmat^{\sfr}_{t_a})^* \Qadjmat^{\sfr}_{t_a})^{-1} (\Qadjmat^{\sfr}_{t_a})^*.
\end{split}
\ee
 Similarly, for $\sfc \in [\Lc]$ and $t \geq 0$, we can write
\be
\label{eq:Adotc}
\begin{split}
\sfAmat_{(\cdot, \sfc)} = \left(\proj^{\perp}_{\Madjmat_t^{\sfc}} + \proj^{\parallel}_{\Madjmat_t^{\sfc}}\right) \sfAmat_{(\cdot, \sfc)} = \proj^{\perp}_{\Madjmat_t^{\sfc}} \, \sfAmat_{(\cdot, \sfc)}  + \Madjmat_t^{\sfc} ((\Madjmat_t^{\sfc})^* \Madjmat_t^{\sfc})^{-1}(\Xmat_{t, \sfc})^*,
\end{split}
\ee
where we interpret the $t=0$ case as follows: $\proj^{\perp}_{\Madjmat_0^{\sfc}} = \proj^{\perp}_{\mathbf{0}} = \iden$ and $\proj^{\parallel}_{\Madjmat_0^{\sfc}} = \proj^{\parallel}_{\mathbf{0}} = \mathbf{0}.$ From \eqref{eq:Ardot} and \eqref{eq:Adotc}, we have two equivalent representations for the submatrix $\sfAmat_{(\sfr, \sfc)} \in \reals^{\Mr \times \Mc}$:
\begin{align}
\sfAmat_{\sfr \sfc}  & =  %\sfAmat_{(\sfr, \cdot)} [\proj^{\perp}_{\Qadjmat^{\sfr}_{t_a}}]_{(\cdot, \sfc)} 
\sum_{\sfc' \in [\Lc]} \sfAmat_{\sfr \sfc'} [\proj^{\perp}_{\Qadjmat^{\sfr}_{t_a}}]_{\sfc' \sfc}
+ \Ymat_{t_a, \sfr}  ((\Qadjmat^{\sfr}_{t_a})^* \Qadjmat^{\sfr}_{t_a})^{-1} (\Qadjmat^{\sfr}_{t_a, \sfc})^* 
\label{eq:cond1_rep}
\\
&=  %[\proj^{\perp}_{\Madjmat_t^{\sfc}}]_{(\sfr, \cdot)} \,  \sfAmat_{(\cdot, \sfc)}   
\sum_{\sfr' \in [\Lr]} [\proj^{\perp}_{\Madjmat_t^{\sfc}}]_{\sfr \sfr'} \,  \sfAmat_{\sfr' \sfc}  
+ \Madjmat_{t, \sfr}^{\sfc} ((\Madjmat_t^{\sfc})^* \Madjmat_t^{\sfc})^{-1}(\Xmat_{t, \sfc})^*. \label{eq:A_rep2}
\end{align}

Using the representation \eqref{eq:A_rep2} in \eqref{eq:cond1_rep}, we obtain
\begin{align}
\sfAmat_{\sfr \sfc}  & = \sum_{\sfc' \in [\Lc]} \sum_{\sfr' \in [\Lr]} [\proj^{\perp}_{\Madjmat_t^{\sfc'}}]_{\sfr \sfr'} \,  \sfAmat_{\sfr' \sfc'}  [\proj^{\perp}_{\Qadjmat^{\sfr}_{t_a}}]_{\sfc' \sfc}
 + \sum_{\sfc' \in [\Lc]} \Madjmat_{t, \sfr}^{\sfc'} ((\Madjmat_t^{\sfc'})^* \Madjmat_t^{\sfc'})^{-1}(\Xmat_{t, \sfc'})^* [\proj^{\perp}_{\Qadjmat^{\sfr}_{t_a}}]_{\sfc' \sfc} \nonumber \\
 & \qquad +  \Ymat_{t_a, \sfr}  ((\Qadjmat^{\sfr}_{t_a})^* \Qadjmat^{\sfr}_{t_a})^{-1} (\Qadjmat^{\sfr}_{t_a, \sfc})^*, 
 \label{eq:Arc_rep1}
%
% & = \sum_{\sfc' \in [\Lc]} [\proj^{\perp}_{\Madjmat_t^{\sfc'}} \sfAmat]_{(\sfr, \sfc')}  [\proj^{\perp}_{\Qadjmat^{\sfr}_{t_a}}]_{\sfc' \sfc}
%%
% + \sum_{\sfc' \in [\Lc]} \Madjmat_{t, \sfr}^{\sfc'} ((\Madjmat_t^{\sfc'})^* \Madjmat_t^{\sfc'})^{-1}(\Xmat_{t, \sfc'})^* [\proj^{\perp}_{\Qadjmat^{\sfr}_{t_a}}]_{\sfc' \sfc} \nonumber \\
% %
% & \qquad +  \Ymat_{t_a, \sfr}  ((\Qadjmat^{\sfr}_{t_a})^* \Qadjmat^{\sfr}_{t_a})^{-1} (\Qadjmat^{\sfr}_{t_a, \sfc})^* 
\end{align}
and using \eqref{eq:cond1_rep} in \eqref{eq:A_rep2}, we obtain
\begin{align}
\sfAmat_{\sfr \sfc}  & =  \sum_{\sfr' \in [\Lr]} \sum_{\sfc' \in [\Lc]} [\proj^{\perp}_{\Madjmat_t^{\sfc}}]_{\sfr \sfr'} \,  \sfAmat_{\sfr' \sfc'}  [\proj^{\perp}_{\Qadjmat^{\sfr'}_{t_a}}]_{\sfc' \sfc}  +  \sum_{\sfr' \in [\Lr]}  [\proj^{\perp}_{\Madjmat_t^{\sfc}}]_{\sfr \sfr'} \, \Ymat_{t_a, \sfr'}  ((\Qadjmat^{\sfr'}_{t_a})^* \Qadjmat^{\sfr'}_{t_a})^{-1} (\Qadjmat^{\sfr'}_{t_a, \sfc})^*  \nonumber \\
& \qquad +  \Madjmat_{t, \sfr}^{\sfc} ((\Madjmat_t^{\sfc})^* \Madjmat_t^{\sfc})^{-1}(\Xmat_{t, \sfc})^*.
\label{eq:Arc_rep2}
\end{align}

The first term in the RHS of \eqref{eq:Arc_rep1} can be written as $\sum_{\sfc' \in [\Lc]} [\proj^{\perp}_{\Madjmat_t^{\sfc'}} \sfAmat]_{(\sfr, \sfc')}  [\proj^{\perp}_{\Qadjmat^{\sfr}_{t_a}}]_{\sfc' \sfc}$.  Recall that conditioning on the sigma-algebra $\mathscr{S}_{t, t}$ is equivalent to conditioning on  the linear constraints in \eqref{eq:lin_const_cond} with $t_a=t$.  
We note that the constraint corresponding to $\sfr$ in \eqref{eq:lin_const_cond} involves only the $\sfr$th row block of 
$\sfAmat$, and the constraint corresponding to $\sfc'$ in \eqref{eq:lin_const_cond}  involves only the $\sfc'$th column block of $\sfAmat$, for each $\sfr \in [\Lr], \sfc' \in [\Lc]$.  Therefore, since the entries of $\sfAmat$ are i.i.d.\ Gaussian,  the conditional distribution of the above term given $\mathscr{S}_{t, t}$  satisfies \cite[Lemmas 10, 12]{bayati2011},
\be
\sum_{\sfc' \in [\Lc]} [\proj^{\perp}_{\Madjmat_t^{\sfc'}} \sfAmat]_{(\sfr, \sfc')}  [\proj^{\perp}_{\Qadjmat^{\sfr}_{t_a}}]_{\sfc' \sfc} \, \Big \vert_{\mathscr{S}_{t, t}} \, \stackrel{d}{=} \, \sum_{\sfc' \in [\Lc]} [\proj^{\perp}_{\Madjmat_t^{\sfc'}} \hat{\sfAmat} ]_{(\sfr, \sfc')}  [\proj^{\perp}_{\Qadjmat^{\sfr}_{t_a}}]_{\sfc' \sfc},
\label{eq:AAhat_rep}
\ee
where $\hat{\sfAmat} \stackrel{d}{=} \sfAmat$ is independent of $\mathscr{S}_{t, t}$.  

We first prove the two results in \eqref{eq:A_cond_0case}.  We note that $\sfAmat$ is independent of $\mscrs_{0,0}$ and $\bvec^0_{\sfr} = \sfAmat \qadj^{0, \sfr}$ by definition.  When $t_a = 1$ and $t=0$,  that the result in  \eqref{eq:Arc_rep2} gives
\begin{align*}
\sfAmat_{\sfr \sfc}  & =  \sum_{\sfr' \in [\Lr]} \sum_{\sfc' \in [\Lc]} \iden_{\sfr \sfr'} \,  \sfAmat_{\sfr' \sfc'}  [\proj^{\perp}_{\qvec^{0, \sfr'}}]_{\sfc' \sfc}  +  \sum_{\sfr' \in [\Lr]}  \iden_{\sfr \sfr'} \, \bvec^0_{\sfr'} \norm{\qvec^{0, \sfr'}}^{-2} (\qvec^{0, \sfr'}_{\sfc})^*  = [ \sfAmat  \proj^{\perp}_{\qvec^{0, \sfr}}]_{(\sfr, \sfc)}  + \bvec^0_{\sfr} \norm{\qvec^{0, \sfr}}^{-2} (\qvec^{0, \sfr}_{\sfc})^*.
\end{align*}
Using the above and  noting that $\sfAmat  \proj^{\perp}_{\qvec^{0, \sfr}} \lvert_{\mscrs_{1,0}} \stackrel{d}{=} \sfAmat'  \proj^{\perp}_{\qvec^{0, \sfr}}$ (by a similar argument to \eqref{eq:AAhat_rep}), we have
\begin{align*}
&  \left[ \sfAmat^* \madj^{0, \sfc} \right]_\sfc \, \big |_{\mscrs_{1,0}}  =   \,  \sum_{\sfr \in [\Lr]} 
(\sfAmat_{(\sfr, \sfc)})^* \madj^{0, \sfc}_{\sfr} \  \big \vert_{\mathscr{S}_{1,0}}  \stackrel{d}{=} \sum_{\sfr \in [\Lr]}\left[  \left[   \proj^{\perp}_{\qvec^{0, \sfr}} \,  \sfAmat'^* \right]_{(\sfc, \sfr)} \,  + \,  \qvec^{0, \sfr}_{\sfc}  \norm{\qvec^{0, \sfr}}^{-2} (\bvec^0_{\sfr})^* \right] \madj^{0, \sfc}_{\sfr}.
\end{align*}

Next, using \eqref{eq:Arc_rep1} and \eqref{eq:AAhat_rep} and multiplying by $\qvec^{t, \sfr}$, we obtain
\begin{align}
& \left[ \sfAmat \qvec^{t, \sfr} \right]_\sfr \, \big \vert_{\mathscr{S}_{t, t}} \, =  \,  \sum_{\sfc \in [\Lc]} 
\sfAmat_{(\sfr, \sfc)} \qvec^{t, \sfr}_{\sfc} \  \big \vert_{\mathscr{S}_{t, t}}  \nonumber \\
& \stackrel{d}{=} 
\sum_{\sfc' \in [\Lc]}  [\proj^{\perp}_{\Madjmat_t^{\sfc'}} \hat{\sfAmat} ]_{(\sfr, \sfc')}  \, \qadj^{t, \sfr}_{\perp, \sfc'}  + \sum_{\sfc' \in [\Lc]} \Madjmat_{t, \sfr}^{\sfc'} ((\Madjmat_t^{\sfc'})^* \Madjmat_t^{\sfc'})^{-1}(\Xmat_{t, \sfc'})^* \qvec^{t, \sfr}_{\perp, \sfc'}  +  \Ymat_{t, \sfr}  ((\Qadjmat^{\sfr}_{t})^* \Qadjmat^{\sfr}_{t})^{-1} (\Qadjmat^{\sfr}_{t, \sfc})^*\qvec^{t, \sfr} \nonumber .
%
%& = \sum_{\sfc = \underline{\sfc}_{\sfr}}^{\overline{\sfc}_{\sfr}} \,   [\proj^{\perp}_{\Madjmat_t^{\sfc}} \hat{\sfAmat} ]_{(\sfr, \sfc)}  \, \qadj^{t, \sfr}_{\perp, \sfc}  \, +  \,   \sum_{\sfc = \underline{\sfc}_{\sfr}}^{\overline{\sfc}_{\sfr}} \Madjmat_{t, \sfr}^{\sfc} ((\Madjmat_t^{\sfc})^* \Madjmat_t^{\sfc})^{-1}(\Xmat_{t, \sfc})^*\, \qadj_{\perp, \sfc}^{t, \sfr}
%\, + \, \Ymat_{t, \sfr} \gammavec^{t, \sfr},
\end{align}
%where the last equality follows by noting that for a given $\sfr$, if $\sfc \notin [\underline{\sfc}_{\sfr}, \overline{\sfc}_{\sfr}]$, then $\Madjmat_t^{\sfc}$ is an all-zeros matrix.  
Noting that $((\Qadjmat^{\sfr}_{t})^* \Qadjmat^{\sfr}_{t})^{-1} (\Qadjmat^{\sfr}_{t, \sfc})^*\qvec^{t, \sfr} =  \gammavec^{t, \sfr}$ completes the proof of \eqref{eq:Aq_cond}. The result  \eqref{eq:Astarm_cond} can be similarly  obtained, by using the representation in \eqref{eq:Arc_rep2} to express $\left[ \sfAmat^* \mvec^{t, \sfc} \right]_\sfc$.

%%%%%%
%%%%%%
%%%%%%

\subsection{Proof of Lemma \ref{lem:hb_cond}}  \label{subsec:lemhb_cond_proof}

\begin{proof}
We begin by demonstrating  \eqref{eq:Ba_dist}.  By \eqref{eq:hqbm_def} it follows that
\ben
\bvec^{0}_{\sfr} \lvert_{\mscrs_{0,0}} = [\sfAmat\qadj^{0, \sfr}]_{\sfr} = \sum_{\sfc \in [\Lc]} \sfAmat_{\sfr \sfc} \, \qadj^{0, \sfr}_{\sfc} \overset{d}{=} \Zprime_{0,\sfr} \left(\frac{1}{\sqrt{L}}\norm{\qadj^{0, \sfr}} \right),
\een
where $\Zprime_{0,\sfr} \sim \mathcal{N}(0, \iden_{\Mr})$ independent of $\mscrs_{0,0}$.  Furthermore, we have
\be
 \frac{1}{L} \norm{\qadj^{0, \sfr}}^2 = \frac{1}{L} \sum_{\sfc \in [\Lc]} W_{\sfr\sfc} \norm{\brqvec^{0}_{\sfc}}^2 = \frac{1}{L} \sum_{\sfc \in [\Lc]} W_{\sfr\sfc} \norm{\betavec_{0, \sfc}}^2 =  \frac{1}{\Lc} \sum_{\sfc \in [\Lc]} W_{\sfr\sfc} \psi^0_{\sfc} =  \sigma^0_{\sfr}.
 \label{eq:breveq_simp}
 \ee

For the case $t \geq 1$, we use \eqref{eq:bmq_rowcol_v2} and \eqref{eq:Aq_cond} to write
\be
\begin{split}
&\bvec^t_{\sfr}  \lvert_{\mscrs_{t, t}} =  ([\sfAmat\qadj^{t,\sfr}]_{\sfr} + \frac{\sigma^t_{\sfr}}{\phi^{t-1}_{\sfr}} \brmvec^{t-1}_{\sfr}) \lvert_{\mscrs_{t, t}} \\
&\overset{d}{=}  \sum_{\sfc \in [\Lc]}  [\proj^{\perp}_{\Madjmat_t^{\sfc}}]_{(\sfr, \cdot)} \,  \hat{\sfAmat}_{(\cdot, \sfc)}  \, \qadj^{t, \sfr}_{\perp, \sfc} +\Ymat_{t, \sfr} \gammavec^{t, \sfr}  +   \sum_{\sfc \in [\Lc]} \Madjmat_{t,\sfr}^{\sfc} ((\Madjmat_t^{\sfc})^* \Madjmat_t^{\sfc})^{-1}(\Xmat_{t, \sfc})^*\, \qadj_{\perp, \sfc}^{t, \sfr}   + \frac{\sigma^t_{\sfr}}{\phi^{t-1}_{\sfr}}   \brmvec^{t-1}_{\sfr}.
\label{eq:lemma13a_1}
\end{split}
\ee
%
%In the above, the first equality is obtained from   and the second using the Lemma \ref{lem:Et1t} result.  
Since $\hat{\sfAmat}$ is independent of $\mscrs_{t, t}$, with  i.i.d.\ $\sim \mc{N}(0, \frac{1}{L})$ entries, the first term on the right side of \eqref{eq:lemma13a_1} can be written as:
\begin{align}
 \sum_{\sfc \in [\Lc]}  [\proj^{\perp}_{\Madjmat_t^{\sfc}}]_{(\sfr, \cdot)} \,  \hat{\sfAmat}_{(\cdot, \sfc)}  \, \qadj^{t, \sfr}_{\perp, \sfc} &\overset{d}{=}  \sum_{\sfc \in [\Lc]}   \frac{\| \qadj^{t, \sfr}_{\perp, \sfc} \| }{\sqrt{L}}    \, [\proj^{\perp}_{\Madjmat_t^{\sfc}} \, \Zprimesupc_{t}]_{\sfr} \overset{d}{=} \frac{\| \qadj^{t, \sfr}_{\perp}\| }{\sqrt{L}}  \, \Zprime_{t,\sfr} -  \sum_{\sfc \in [\Lc]}  \frac{\| \qadj^{t, \sfr}_{\perp, \sfc}\| }{\sqrt{L}}    \, [\proj^{\parallel}_{\Madjmat_t^{\sfc}} \, \Zprimesupc_{t}]_{\sfr},
\label{eq:bt_dist_eq1}
\end{align}
where in the first step, the vectors $\Zprimesupc_{t} \sim \mathcal{N}(0, \iden_{n})$  are independent for $\sfc \in [\Lc]$. For the second equality we use $\proj^{\perp}_{\Madjmat_t^{\sfc}} = \iden_{n} - \proj^{\parallel}_{\Madjmat_t^{\sfc}}$, and writing
$\Zprime_{t,\sfr}  = \sum_{\sfc \in [\Lc]}    \frac{\| \qadj^{t, \sfr}_{\perp, \sfc} \| }{\| \qadj^{t, \sfr}_{\perp}\|}     \Zprimesupc_{t, \sfr} $, we note that 
$\Zprime_{t, \sfr}  \sim \mathcal{N}(0, \iden_{\Mr})$.
%
%$\sum_{\sfc \in [\Lc]}   \frac{\| \qadj^{t, \sfr}_{\perp, \sfc} \| }{\sqrt{L}}     \Zprimesupc_{t, \sfr}  =  \frac{\| \qadj^{t, \sfr}_{\perp}\| }{\sqrt{L}}  \, \Zprime_{t,\sfr}$
%
%hence $\Zprime_{t, \sfr} =  \sum_{\sfc \in [\Lc]} \Zprimesupc_{t, \sfr}$, which  implies $\Zprime_{t, \sfr} \sim \mathcal{N}(0, \iden_{\Mr})$.  \CR{I think there is a mistake here: we use 
%\[ \sum_{\sfc \in [\Lc]}   \frac{\| \qadj^{t, \sfr}_{\perp, \sfc} \| }{\sqrt{L}}     \Zprimesupc_{t, \sfr}  =  \frac{\| \qadj^{t, \sfr}_{\perp}\| }{\sqrt{L}}  \, \Zprime_{t,\sfr} \]
%so we have $\Zprime_{t,\sfr}  = \sum_{\sfc \in [\Lc]}    \frac{\| \qadj^{t, \sfr}_{\perp, \sfc} \| }{\| \qadj^{t, \sfr}_{\perp}\|}     \Zprimesupc_{t, \sfr}  $, which is still $\Zprime_{t, \sfr} \sim \mathcal{N}(0, \iden_{\Mr})$.}
%
Using \eqref{eq:bt_dist_eq1} and $\Ymat_{t,\sfr} = \Bmat_{t, \sfr} - [0| \Mmat_{t-1, \sfr}] \Upmat^{\sfr}_{t}$, we simplify \eqref{eq:lemma13a_1} to the following:
\be
\begin{split}
\bvec^t_{\sfr}  \lvert_{\mscrs_{t, t}} &\stackrel{d}{=} \Bmat_{t, \sfr} \gammavec^{t,\sfr} \, + \,  \frac{1}{\sqrt{L}} \, \| \qadj^{t, \sfr}_{\perp} \|   \, \Zprime_{t,\sfr}  \, - \, \frac{1}{\sqrt{L}}  \sum_{\sfc \in [\Lc]}\| \qadj^{t, \sfr}_{\perp, \sfc} \|    \, [\proj^{\parallel}_{\Madjmat_t^{\sfc}} \, \Zprimesupc_{t} ]_{\sfr}\\
&  + \,  \sum_{\sfc \in [\Lc]} \Madjmat_{t,\sfr}^{\sfc} ((\Madjmat_t^{\sfc})^* \Madjmat_t^{\sfc})^{-1}(\Xmat_{t, \sfc})^*\, \qadj_{\perp, \sfc}^{t, \sfr} \, - \,  [ \bzero \, | \, \Mmat_{t-1, \sfr}] \Upmat^{\sfr}_{t} \gammavec^{t,\sfr}   \,  + \, \frac{\sigma^t_{\sfr}}{\phi^{t-1}_{\sfr}}   \brmvec^{t-1}_{\sfr} .
\end{split}
\label{eq:lemma13a}
\ee
All the quantities in the RHS of \eqref{eq:lemma13a} except the vectors $\Zprime_{t, \sfr}$ and $\Zprimesupc_{t}$ are in the conditioning sigma-field.  We can rewrite \eqref{eq:lemma13a_1} with the following pair of values:
\begin{align*}
\bvec^t_{\sfr} \lvert_{\mscrs_{t, t}} &\overset{d}{=} \bvec^{t-1}_{\sfr} \frac{\sigma^t_{\sfr}}{\sigma^{t-1}_{\sfr}}+ \sqrt{\sigma_{\perp, \sfr}^{t}} \Zprime_{t,\sfr} + \Dvec_{t,t, \sfr}, \\
\Dvec_{t,t, \sfr} &=  \sum_{i=0}^{t-2} \bvec^i_{\sfr}  \gamma^{t,\sfr}_{i}  \,+\, \bvec^{t-1}_{\sfr}  \left[\gamma^{t,\sfr}_{t-1} - \frac{\sigma^t_{\sfr}}{\sigma^{t-1}_{\sfr}}\right]  \, + \,  \left[ \frac{1}{\sqrt{L}}  
\|\qadj^{t,\sfr}_{\perp}\| - \sqrt{\sigma_{\perp, \sfr}^{t}}\right] \Zprime_{t,\sfr}  \,  - \,  \frac{1}{\sqrt{L}}  \sum_{\sfc \in [\Lc]} \| \qadj^{t, \sfr}_{\perp, \sfc} \|    \, [\proj^{\parallel}_{\Madjmat_t^{\sfc}} \, \Zprimesupc_{t}]_{\sfr} \\
& \qquad   + \,  \sum_{\sfc \in [\Lc]} \Madjmat_{t,\sfr}^{\sfc} ((\Madjmat_t^{\sfc})^* \Madjmat_t^{\sfc})^{-1}(\Xmat_{t,\sfc})^*\, \qadj_{\perp, \sfc}^{t, \sfr} \,- \, [ \bzero \mid \Mmat_{t-1, \sfr}] \Upmat^{\sfr}_{t} \gammavec^{t,\sfr}   \,  + \, \frac{\sigma^t_{\sfr}}{\phi^{t-1}_{\sfr}}   \brmvec^{t-1}_{\sfr} . 
\end{align*}
To complete the proof of \eqref{eq:Ba_dist}, we note that $[ 0 | \Mmat_{t-1, \sfr}] \Upmat^{\sfr}_{t} \gammavec^{t,\sfr} =  \sum_{i=1}^{t-1}\gamma^{t,\sfr}_{i} \frac{\sigma^i_{\sfr}}{\phi^{i-1}_{\sfr}}  \brmvec^{i-1}_{\sfr}$.

For the result  \eqref{eq:Ha_dist}  we start by writing $\hvec^{t+1}_{\sfc} =  [\sfAmat^* \madj^{t, \sfc}]_{\sfc}- \brqvec^{t}_{\sfc}$, and using \eqref{eq:Astarm_cond} for the conditional distribution of  $[\sfAmat^* \madj^{t, \sfc}]_{\sfc}$ given  $\mscrs_{t+1,t}$. We omit the proof as the argument is along similar lines as the above.
 \end{proof}

\subsection{Proof of Lemma  \ref{lem:main_lem}} \label{subsec:main_lem_proof}

The proof proceeds by induction on $t$.  We label the results  in \eqref{eq:Ba1} through \eqref{eq:Bg} as step $\mathcal{B}_t$, and those in \eqref{eq:Ha} through \eqref{eq:Hg} as  step $\mathcal{H}_{t+1}$.  The proof consists of four steps, inductively showing that: (1) $\mathcal{B}_0$ holds, (2) $\mathcal{H}_1$ holds, (3) if $\mathcal{B}_{\tilde{s}}, \mathcal{H}_s$ hold for all $\tilde{s} < t $ and $s \leq t $, then $\mathcal{B}_t$ holds, and (4) if $\mathcal{B}_{\tilde{s}}, \mathcal{H}_s$ hold for all $\tilde{s} \leq t $ and $s \leq t$, then $\mathcal{H}_{t+1}$ holds.
Appendix \ref{app:conc_lemmas} lists a few basic concentration inequalities and other lemmas that are used in the proof. 

\subsubsection{Step 1: Showing $\mathcal{B}_0$ holds} \label{subsub:step1}

\noindent \textbf{(a)} $\Dvec_{0,0} =\mathbf{0}$, so there is nothing to prove. 

\noindent \textbf{(b)} First, $\bvec^0_{\sfr} \stackrel{d}{=} \sqrt{\sigma^0_{\sfr}} \,  \Zprime_{0, \sfr}$ where $ \Zprime_{0, \sfr} \sim \mathcal{N}(0, \mathbb{I}_{\Mr})$ by Lemma \ref{lem:hb_cond}.  Further, $\wvec_{\sfr} \overset{d}{=} \sigma \mathbf{U}_{\sfr}$ where $\mathbf{U}_{\sfr} \sim \mathcal{N}(0, \mathbb{I}_{\Mr})$ independent of $\Zprime_{0, \sfr}$. 
Hence,
\be
\begin{split}
& P\Big(\frac{1}{n} \Big \lvert  \sum_{\sfr \in [\Lr]} W^{\sfp}_{\sfr \sfc} (\bvec^0_{\sfr})^*\wvec_{\sfr} \Big \lvert \geq \epsilon  \Big)  \leq P \Big(\frac{1}{n} \Big \lvert  \sum_{\sfr \in [\Lr]}   \sigma \sqrt{\sigma^0_{\sfr}} W^{\sfp}_{\sfr \sfc}(\Zprime_{0, \sfr})^*\mathbf{U}_{\sfr} \Big \lvert \geq \epsilon \Big)  \\
&\overset{(a)}{\leq}  K\exp\Big\{\frac{-\kappa n \e^2}{\max\{ \frac{1}{\Lr} \, \sum_{\sfr} W^{2\sfp}_{\sfr \sfc}, \,  \max_\sfr W^{\sfp}_{\sfr  \sfc}, \, 1 \}}\Big \}\overset{(b)}{\leq}   K\exp\{-\kappa n   ( {\omega}/{\Lr})^{\max\{2\sfp -1, 0\}} \e^2 \}.
\label{eq:Bb2_1}
\end{split}
\ee
In  step $(a)$, we use Lemma \ref{lem:max_abs_normals} \eqref{res3} and that $\sigma_\sfr^0 \leq 2P$ for $\sfr \in [\Lr]$. In step $(b)$, we use $\max_{\sfr, \sfc} W^{\sfp}_{\sfr \sfc} \leq (P \Lr/\omega)^{\sfp}$, and $\frac{1}{\Lr} \, \sum_{\sfr} W^{2\sfp}_{\sfr  \sfc} \leq P (P \Lr/\omega )^{2\sfp -1}$ for $\sfp \in \{1,2 \}$ (see \eqref{eq:Wrc_avgs}).  Notice that this is consistent with the $\PC, \pc$ notation defined in \eqref{eq:param_const_defs} and used  in the stated concentration results since $\PC_{-1} = 1$ and $\pc_{-1} = \Mr \omega$, and in particular, $n ( {\omega}/{\Lr})^{\max\{2\sfp -1, 0\}} = \Mr \omega   ( {\omega}/{\Lr})^{\max\{2\sfp -1, 0\}-1} \geq  \Mr \omega   ( {\omega}/{\Lr})^{\max\{2(\sfp -1), 0\}}$ (the only inequality is when $\sfp = 0$).

\noindent  \textbf{(c)}  
Proving $\mathcal{B}_0(c)$ is similar to $\mathcal{B}_0(b)$ and results in the same bound: $\| \bvec^0_{\sfr} \|^2 \overset{d}{=} \sigma_{\sfr}^0 \|\Zprime_{0, \sfr} \|^2$ by Lemma \ref{lem:hb_cond}, therefore the result follows from  Lemma \ref{lem:max_abs_normals} \eqref{res2} and arguments as used in step $(b)$ of \eqref{eq:Bb2_1}.

\noindent  \textbf{(d)} 
Recall, $\madj_{\sfr}^{0,\sfc} = S^0_{\sfr \sfc} \sqrt{W_{\sfr \sfc}} (\bvec^0_{\sfr} - \wvec_{\sfr})$, therefore
$ \sqrt{W_{\sfr \sfc}}  (\bvec^0_{\sfr})^*\madj^{0,\sfc}_{\sfr} 
 =  S^0_{\sfr \sfc} W_{\sfr \sfc}  (\norm{\bvec^0_{\sfr}}^2 - ( \bvec^0_{\sfr})^*\wvec_{\sfr}.$
Then,%using  below,
\be
\begin{split}
P&\Big( \frac{\Mr}{L} \Big \lvert \frac{1}{ \Mr } \sum_{\sfr \in [\Lr]} W_{\sfr \sfc}^{\sfp- \frac{1}{2}}   (\bvec^0_{\sfr})^*\madj^{0,\sfc}_{\sfr} - \sum_{\sfr \in [\Lr]}   S^0_{\sfr \sfc} W^{\sfp}_{\sfr \sfc} \, \sigma^{0}_{\sfr} \Big \lvert \geq \epsilon \Big) \\
&= P\Big(\Big \lvert \frac{n}{L \Lr}   \sum_{\sfr \in [\Lr]} S^0_{\sfr \sfc}  W^{\sfp}_{\sfr \sfc}  \Big(\frac{\norm{\bvec^0_{\sfr}}^2}{\Mr} - \frac{( \bvec^0_{\sfr})^*\wvec_{\sfr}}{\Mr} - \sigma^{0}_{\sfr}\Big) \Big \lvert \geq \epsilon \Big) \\
&\overset{(a)}{\leq}  P\Big(\Big \lvert \frac{1}{\Lr}   \sum_{\sfr \in [\Lr]}  \Big(\frac{n}{L} S^0_{\sfr \sfc}\Big)  W^{\sfp}_{\sfr \sfc}  \Big[\frac{\norm{\bvec^0_{\sfr}}^2}{\Mr} - \sigma^{0}_{\sfr}\Big] \Big \lvert \geq \frac{\epsilon}{2} \Big) +P\Big(\Big \lvert \frac{1}{\Lr}   \sum_{\sfr \in [\Lr]} \Big(\frac{n}{L} S^0_{\sfr \sfc}\Big)  W^{\sfp}_{\sfr \sfc}  \Big[ \frac{( \bvec^0_{\sfr})^*\wvec_{\sfr}}{\Mr} \Big] \Big \lvert \geq \frac{\epsilon}{2} \Big) \\
&\overset{(b)}{\leq} K \exp\{-\kappa n \e^2  ( {\omega}/{\Lr})^{\max\{2\sfp -1, 0\}} \},
\label{eq:B0d_eq1}
\end{split}
\ee
where we have used  Lemma~\ref{sums} in step $(a)$, and  step $(b)$ follows by $\mathcal{B}_0(b)$ and $\mathcal{B}_0(c)$, using that fact that  $\frac{n}{L} S^0_{\sfr \sfc}$ is bounded above and below by positive constants for all $(\sfr, \sfc)$.

\noindent  \textbf{(e)} %
Proving $\mathcal{B}_0(e)$ is similar to $\mathcal{B}_0(d)$ and results in the same bound.  We sketch the detais. First, $\madj_{\sfr}^{0,\sfc} = S^0_{\sfr \sfc} \sqrt{W_{\sfr \sfc}} (\bvec^0_{\sfr} - \wvec_{\sfr})$, therefore,
\ben
\begin{split}
\sum_{\sfr \in [\Lr]}  \|\madj^{0,\sfc}_{\sfr} \|^2 &= \sum_{\sfr \in [\Lr]} (S^0_{\sfr \sfc})^2 W_{\sfr \sfc} \norm{\bvec^0_{\sfr} - \wvec_{\sfr}}^2  =  \sum_{\sfr \in [\Lr]} (S^0_{\sfr \sfc})^2 W_{\sfr \sfc} ( \| \bvec^0_{\sfr} \|^2 +  \| \wvec_{\sfr} \|^2 - 2  (\bvec^0_{\sfr})^*\wvec_{\sfr}),
%
%\label{eq:m0_norm}
\end{split}
\een
Next, $  \phi^0_{\sfr} =   (\sigma^2 + \sigma^0_{\sfr})$, and letting $\wvec_{\sfr} \overset{d}{=} \sigma \mathbf{U}_{\sfr}$ where $\mathbf{U}_{\sfr} \sim \mathcal{N}(0, \mathbb{I}_{\Mr})$, the result follows by Lemma \ref{sums}, $\mathcal{B}_0(c)$, $\mathcal{B}_1(b)$, and Lemma \ref{lem:max_abs_normals} as well as an argument as in step $(b)$ of \eqref{eq:Bb2_1}.

\noindent  \textbf{(f)} There is nothing to prove for \eqref{eq:Bf} since $ \| \madj^{0,\sfc}_{\perp} \|^2 = \| \madj^{0,\sfc} \|^2$ and $\tau^0_{\perp, \sfc} = \tau^0_{\sfc}$ and result \eqref{eq:Bf1} follows from $\mathcal{B}_0 (e)$ with $\sfp = 1$.

\noindent  \textbf{(g)} First, $\innerM^{\sfc}_{1} := \frac{n}{L^2}(\Madjmat^{\sfc}_{1})^* \Madjmat^{\sfc}_{1} = \frac{n}{L^2}\| \madj^{0,\sfc} \|^2$  concentrates to $\frac{n }{L}\tau^{0}_{\sfc} $ by $\mathcal{B}_0 (e)$ with $\sfp = 1$.  Result \eqref{eq:Bg} then follows from $\mathcal{B}_0 (e)$ and Lemma \ref{inverses}.  By Fact \ref{fact:eig_proj}, if $\frac{n}{L^2}
\| \madj^{0, \sfc}_{\perp} \|^2 \geq c >0$, then $\innerM^{\sfc}_1$ is invertible.  Note from  $\mc{B}_{0}(f)$ that $\frac{n}{L^2}\| \madj^{0, \sfc}_{\perp} \|^2$ concentrates on $\frac{n}{L}\tau^{0}_{\perp, \sfc} = \frac{n}{L} \tau^{0}_{\sfc} $. From  \eqref{eq:tau_ct_def}, since $\phi^0_\sfr \geq \sigma^2$, we have  $  \frac{n}{L} \tau^{0}_{\sfc}  \geq \frac{\sigma^2}{P}$ for all $\sfc$. Choosing $\kappa =  \min \{1, \frac{n}{2L}\tau^{0}_{\perp, \sfc}\}$, we therefore have from $\mc{B}_0 (f)$ that
\ben
P\Big(\innerM^{\sfc}_1 \text{ singular}\Big)   \leq  P\Big(\frac{n}{L}\Big \lvert \frac{1}{L} \|  \madj^{0, \sfc}_{\perp} \|^2 - \tau^0_{\perp, \sfc} \Big \lvert \geq  \kappa \Big) \leq K\exp\{ -c \Mr \omega\}.
\een
%

%%%%%%%%%%%%%%%%%%%%%%%%
%%%%%%%%%%%%%%%%%%%%%%%%
%%%%%%%%%%%%%%%%%%%%%%%%
%%%%%%%%%%%%%%%%%%%%%%%%
%%%%%%%%%%%%%%%%%%%%%%%%
%%%%%%%%%%%%%%%%%%%%%%%%

\subsubsection{Step 2: Showing $\mathcal{H}_1$ holds} \label{subsub:step2}

\noindent  \textbf{(a)} We use the expression for $\Dvec_{1,0, \sfc}$ from Lemma \ref{lem:hb_cond}, and write the second term in \eqref{eq:D10} as 
\begin{align*}
\sum_{\sfr \in [\Lr]}  \frac{\|  \madj^{0, \sfc}_{\sfr} \|  }{\sqrt{L}}    \, [\proj^{\parallel}_{\Qadjmat^{\sfr}_{1}} \, \Z^{\sfr}_{0}]_{\sfc} &= \sum_{\sfr \in [\Lr]}   \frac{ \| \madj^{0, \sfc}_{\sfr} \|  }{\sqrt{L}}   \Big[\frac{\qadj^{0, \sfr}_{\sfc} \, (\qadj^{0, \sfr})^* \, \Z^{\sfr}_{0}}{ \norm{\qadj^{0, \sfr}}^2 }\Big]  \overset{d}{=} \sum_{\sfr \in [\Lr]}  \frac{ \|  \madj^{0, \sfc}_{\sfr} \|  }{\sqrt{L}}   \Big[\frac{\qadj^{0, \sfr}_{\sfc} \, Z^{\sfr}}{\norm{\qadj^{0, \sfr}}}\Big],
\end{align*}
where $Z^{\sfr}  \sim \mathcal{N}(0, 1)$ are i.i.d.\  for $\sfr \in [\Lr]$.  Recall, $\qadj^{0, \sfr}_\sfc =  \sqrt{W_{\sfr \sfc}} \brqvec^{0}_{\sfc}$ and  $\norm{\qadj^{0, \sfr}}^2 = L \sigma^0_{\sfr}$ from \eqref{eq:breveq_simp}, so
\begin{align*}
 \sum_{\sfr \in [\Lr]}  \frac{ \| \madj^{0, \sfc}_{\sfr} \|  }{\sqrt{L}}   \Big[\frac{\qadj^{0, \sfr}_{\sfc} \, Z^{\sfr}}{\| \qadj^{0, \sfr} \| }\Big] =  \brqvec^{0}_{\sfc}  \sum_{\sfr \in [\Lr]}  
 \frac{ \sqrt{W_{\sfr \sfc}} \, \| \madj^{0, \sfc}_{\sfr} \| Z^{\sfr} }{\sqrt{L} \| \qadj^{0, \sfr} \|}  =   Z \, \brqvec^{0}_{\sfc}  \Big({\sum_{\sfr \in [\Lr]}  \frac{W_{\sfr \sfc} \, \| \madj^{0, \sfc}_{\sfr} \|^2   }{L^2 \, \sigma^0_{\sfr}} }\Big)^{1/2},
\end{align*}
where $Z\sim \normal(0,1)$. Using $  \max_{j \in sec(\ell)} \abs{[\brqvec^0_{\sfc}]_j}^2 =\max_{j \in sec(\ell)} \abs{[\betavec_{0, \sfc}]_j}^2 = 1$, for any section $\ell$ in $\sfc$ and the above, we have 
\be
\begin{split}
\max_{j \in sec(\ell)} \abs{[\Dvec_{1,0, \sfc}]_j}^2 &\leq  \Big\lvert \frac{\norm{\madj^{0, \sfc}}}{\sqrt{L}}  - \sqrt{ \tau^{0}_{\sfc}}\Big \lvert^2  \max_{j \in sec(\ell)} \abs{ [\Z_{0, \sfc}]_j}^2 +    \abs{Z}^2\sum_{\sfr \in [\Lr]}  \frac{W_{\sfr \sfc} \, \| \madj^{0, \sfc}_{\sfr}\|^2   }{L^2 \, \sigma^0_{\sfr}}  \\
& \qquad + \Big \lvert \sum_{\sfr \in [\Lr]}\frac{\sqrt{W_{\sfr \sfc}}}{L \sigma^0_{\sfr}} (\bvec^0_{\sfr})^*\madj^{0, \sfc}_{\sfr} - 1\Big \lvert^2 .
\label{eq:normDel10}
\end{split}
\ee
We therefore have the following bound for $\sfr' \in [\Lr]$, and  $\sfq = 0$ or $\sfq =1$ using Lemma~\ref{sums}.
\begin{align}
& P  \Big(\frac{1}{L}  \sum_{\sfc \in [\Lc]}\sum_{\ell \in \sfc} W^{2 \sfq}_{\sfr' \sfc} \max_{j \in sec(\ell)} \abs{[\Dvec_{1,0, \sfc}]_j}^2 \geq \epsilon \Big) \nonumber \\
&  \leq P \Big(\frac{1}{L} \sum_{\sfc \in [\Lc]}  \frac{L  W^{2 \sfq}_{\sfr' \sfc} }{n} \Big \lvert \sqrt{\frac{n \| \madj^{0, \sfc} \|^2 }{ L^2}}  - \sqrt{\frac{n \tau^0_{\sfc}}{L}} \Big \lvert^2 \sum_{\ell \in \sfc} \max_{j \in sec(\ell)} \abs{ [\Z_{0, \sfc}]_j}^2  \geq \frac{\epsilon }{3} \Big) \label{eq:sum_del2_bnd1}  \\
& +P \Big(\frac{1}{\Lc}  \sum_{\sfc \in [\Lc]}  W^{2 \sfq}_{\sfr' \sfc}  \Big \lvert\sum_{\sfr \in [\Lr]}\frac{\sqrt{W_{\sfr \sfc}} }{L \sigma^0_{\sfr}}(\bvec^0_{\sfr})^*\madj^{0, \sfc}_{\sfr} - 1\Big \lvert^2 \geq \frac{ \epsilon}{3} \Big) +   P \Big( \frac{\abs{Z}^2}{L \Lc}  \sum_{\sfc \in [\Lc]} \hspace{-3pt} W^{2 \sfq}_{\sfr' \sfc} \sum_{\sfr \in [\Lr]}  \hspace{-3pt}  \frac{W_{\sfr \sfc} \, \| \madj^{0, \sfc}_{\sfr} \|^2   }{L \, \sigma^0_{\sfr}} \geq \frac{ \epsilon }{3} \Big). \nonumber
\end{align}
Label the terms on the RHS of \eqref{eq:sum_del2_bnd1} as $T_1, T_2, T_3$ and we bound each separately.

Consider term $T_1$ and let $\Pi_0$ be the event under consideration, so that $T_1 = P(\Pi_0)$, and define an event $\mathcal{F}$ as 
\be
\mathcal{F} =  \Big\{ \max_{\sfc  \in [\Lc]}   \frac{L }{n}\Big \lvert \sqrt{\frac{n}{L} \norm{\madj^{0, \sfc}}^2 }  - \sqrt{\frac{n}{L} \tau^0_{\sfc}} \Big \lvert^2 \geq \frac{\epsilon}{(\max_{\sfc' \in [\Lc]} W^{2 \sfq}_{\sfr' \sfc'})9 \log M} \Big \}.
\label{eq:H1F_def}
\ee
With this definition, 
$T_1 = P(\Pi_0) \leq P(\mathcal{F}) + P(\Pi_0 | \mathcal{F}^{c}).$
First,
\be
\begin{split}
P(\mathcal{F})  &\leq \sum_{\sfc  \in [\Lc]} P\Big(\Big \lvert \sqrt{\frac{n}{L} \norm{\madj^{0, \sfc}}^2 }  - \sqrt{\frac{n}{L} \tau^0_{\sfc}}  \Big \lvert \geq \sqrt{\frac{n\epsilon}{(\max_{\sfc' \in [\Lc]} W^{2 \sfq}_{\sfr' \sfc'}) 9  L \log M}} \Big) \\
& \overset{(a)}{\leq}   \sum_{\sfc  \in [\Lc]} K\exp\Big\{\frac{-\kappa \Mr \omega (n \tau^0_{\sfc}/L) \e}{\max_{\sfc' \in [\Lc]} W^{2 \sfq}_{\sfr' \sfc'} } \Big\}   \overset{(b)}{\leq} K \Lc \exp\Big\{\frac{-\kappa \Mr \omega \e}{\max_{\sfc' \in [\Lc]} W^{2 \sfq}_{\sfr' \sfc'}}\Big\},
\label{eq:F_bound}
\end{split}
\ee
where step $(a)$ uses $\mathcal{B}_0(f)$ and Lemma \ref{sqroots}, and step $(b)$ the fact that $n \tau^0_{ \sfc}/L \in \Theta(1)$ and $n R = L \log M$. %, and $W_{\sfr, \sfc} \leq P  \frac{\Lr}{\omega}$. 
Next, by Lemma \ref{lem:max_abs_normals}, using  that $\Z_{0, \sfc}$ is independent of the event $\mathcal{F}$ for $\sfc \in [\Lc]$,
\be
\begin{split}
P(\Pi_0 | \mathcal{F}^{c}) &\leq  P \Big(\frac{1}{L } \sum_{\sfc \in [\Lc]}  \sum_{\ell \in \sfc}   \max_{j \in sec(\ell)} \abs{ [\Z_{0, \sfc}]_j}^2  \geq  3 \log M \Big) \leq K \exp\{- \kappa  L \log M \}.
\label{eq:Fcond_bound}
\end{split}
\ee
From \eqref{eq:F_bound} and \eqref{eq:Fcond_bound},  we find
\be
T_1 \leq K \Lc \exp\Big\{\frac{-\kappa \Mr \omega \e}{\max_{\sfc' \in [\Lc]} W^{2 \sfq}_{\sfr' \sfc'}}\Big\}
 \leq K \Lc \exp\{-\kappa' n (\omega/\Lr)^{1+2\sfq} \e\} .
\label{eq:H1T1_bound}
\ee
The final inequality uses $\max_{\sfr, \sfc} W^{2 \sfq}_{\sfr \sfc} \leq (P \Lr/\omega)^{2 \sfq}$, thus $\Mr \omega/(\max_{\sfr, \sfc} W^{2 \sfq}_{\sfr \sfc}) \geq \kappa \Mr \omega (\omega/\Lr)^{2 \sfq}.$

For term $T_2$, we use result $\mathcal{B}_0 (d)$.  To see this,
\ben
\begin{split}
&T_2 \overset{(a)}{=} P \Big(\frac{1}{\Lc}  \sum_{\sfc \in [\Lc]} W^{2 \sfq}_{\sfr' \sfc}  \Big \lvert\sum_{\sfr \in [\Lr]}  \Big(\frac{\sqrt{W_{\sfr \sfc}} }{L \sigma^0_{\sfr}}(\bvec^0_{\sfr})^*\madj^{0, \sfc}_{\sfr} - \frac{\Mr}{L} S^0_{\sfr \sfc} W_{\sfr \sfc}  \Big) \Big \lvert^2 \geq \frac{ \epsilon}{3} \Big)\\ 
%
%&= P \Big(\frac{1}{\Lc}  \sum_{\sfc \in [\Lc]} W^{2 \sfq}_{\sfr' \sfc}  \Big \lvert\sum_{\sfr \in [\Lr]} \frac{1}{\sigma^0_{\sfr}} \Big[ \frac{\sqrt{W_{\sfr \sfc}} }{L}(\bvec^0_{\sfr})^*\madj^{0, \sfc}_{\sfr} - \frac{\Mr}{L} S^0_{\sfr \sfc} W_{\sfr \sfc} \sigma^0_{\sfr} \Big]\Big \lvert^2 \geq \frac{ \epsilon}{3} \Big) \\
%
&\overset{(b)}{\leq}  \sum_{\sfc \in [\Lc]}   P \Big( \Big \lvert\sum_{\sfr \in [\Lr]} \Big(\frac{\sqrt{W_{\sfr \sfc}} }{L}(\bvec^0_{\sfr})^*\madj^{0, \sfc}_{\sfr} - \frac{\Mr}{L} S^0_{\sfr \sfc} W_{\sfr \sfc} \sigma^0_{\sfr}\Big) \Big \lvert^2 \geq \frac{\kappa \e}{\sum_{\sfc' \in [\Lc]} W^{2 \sfq}_{\sfr' \sfc'}/\Lc} \Big) \overset{(c)}{\leq}  K \Lc e^{ - \kappa n (\omega/\Lr)^{1+ \sfq} \e}.
\end{split}
\een
In the above, step $(a)$ holds because $\frac{\Mr}{L}  \sum_{\sfr \in [\Lr]} S^0_{\sfr \sfc} W_{\sfr \sfc} = 1$, which can be verified using $S^0_{\sfr \sfc} = \tau^0_\sfc/\phi^0_\sfr$, and the state evolution equations in \eqref{eq:se_phi}--\eqref{eq:tau_ct_def}.  Step $(b)$ uses that  $ \sigma^0_{\sfr} \in \Theta(1)$ and Lemma \ref{sums}.  Step $(c)$ follows from $\mathcal{B}_0 (d)$ (Eq.\ \eqref{eq:B0d_eq1}) with $\sfp = 1$, noting from \eqref{eq:Wrc_avgs} that $\sum_{\sfc' \in [\Lc]} W^{2 \sfq}_{\sfr' \sfc'}/\Lc = \Theta( (\Lr/\omega)^\sfq)$.

Finally we bound term $T_3$.  Note that for any $\sfr, \sfr' \in [\Lr]$, we have $\frac{1}{\Lc}  \sum_{\sfc \in [\Lc]}  W^{2 \sfq}_{\sfr' \sfc} = \Theta((\Lr /\omega)^\sfq)$ since $\sfq \in \{0,1\}$ and $\sigma_\sfr^0 = \Theta(1)$. Furthermore, from $\mathcal{B}_0 (e)$ with $\sfp = 2$,  for $\sfc \in [\Lc]$ the term 
$\frac{n}{L^2} \sum_{\sfr \in [\Lr]}  \hspace{-3pt}  W_{\sfr \sfc} \, \| \madj^{0, \sfc}_{\sfr} \|^2$
concentrates on $\frac{n\Mr}{L^2} \sum_{\sfr \in [\Lr]}   (S^0_{\sfr \sfc})^2 \, W^{2}_{\sfr \sfc} \,\phi^0_{\sfr} = \Theta(\Lr/\omega)$, with the deviation probability bounded by $\mathcal{B}_0 (e)$. We therefore have
%%
%\begin{align}
%T_3 &  \leq P \left( \cup_{\sfc\in [\Lc]} \,  \Big\{ \frac{|Z|^2}{n} \kappa \Big(\frac{\Lr}{\omega}\Big)^\sfq  \cdot \frac{n}{L^2} \sum_{\sfr \in [\Lr]}  \hspace{-3pt}  W_{\sfr \sfc} \, \| \madj^{0, \sfc}_{\sfr} \|^2 \geq \frac{\e}{3}  \Big\}\right)  \nonumber \\
%%
% & \leq \sum_{\sfc \in [\Lc]} P\left( \frac{n}{L^2} \sum_{\sfr \in [\Lr]}  \hspace{-3pt}  W_{\sfr \sfc} \, \| \madj^{0, \sfc}_{\sfr} \|^2 \leq 2 \cdot  \frac{n\Mr}{L^2} \sum_{\sfr \in [\Lr]}   (S^0_{\sfr \sfc})^2 \, W^{2}_{\sfr \sfc} \,\phi^0_{\sfr}   \right) + 
%P\left( |Z|^2 \geq \frac{n \e}{ \kappa (\Lr/\omega)^{2 \sfq}}\right) \nonumber \\
%%
%& \stackrel{(a)}{\leq} \Lc K  \exp\{-\kappa n ( {\omega}/{\Lr})^3 \} + \exp\{-n (\omega/\Lr)^{2 \sfq} \e /\kappa \},
%\end{align}
%where step $(a)$ follows from $\mathcal{B}_0 (e)$ and  Lemma \ref{lem:normalconc}.
%%, noting that 
%%$\frac{n\Mr}{L^2} \sum_{\sfr \in [\Lr]}   (S^0_{\sfr \sfc})^2 \, W^{2}_{\sfr \sfc} \,\phi^0_{\sfr} = \Theta(\Lr/\omega)$. 
%%%%%%%%
%%%%%%%%
%\CR{In the second row of the above, I think the term in the $|Z|$ concentration should be $(\Lr/\omega)^{ 1+ \sfq}$ instead.}
%\RV{Need to be careful about what the overall dominant exponent is-- can take it to be$ -n (\omega/\Lr)^{3} \e$} \CR{I think we can write it in the following way to get a tighter bound than $ -n (\omega/\Lr)^{3} \e$:
%
\begin{align}
T_3 &  \leq P \Big( \cup_{\sfc\in [\Lc]} \,  \Big\{ \frac{|Z|^2}{\Mr \omega} \kappa \Big(\frac{\Lr}{\omega}\Big)^\sfq  \cdot \frac{n \omega}{L^2 \Lr} \sum_{\sfr \in [\Lr]}  \hspace{-3pt}  W_{\sfr \sfc} \, \| \madj^{0, \sfc}_{\sfr} \|^2 \geq \frac{\e}{3}  \Big\}\Big)  \nonumber \\
 & \leq \sum_{\sfc \in [\Lc]} P\Big( \frac{n \omega}{L^2\Lr} \sum_{\sfr \in [\Lr]}  \hspace{-3pt}  W_{\sfr \sfc} \, \| \madj^{0, \sfc}_{\sfr} \|^2 \leq 2 \cdot  \frac{n^2 \omega}{L^2 \Lr^2} \sum_{\sfr \in [\Lr]}   (S^0_{\sfr \sfc})^2 \, W^{2}_{\sfr \sfc} \,\phi^0_{\sfr}   \Big) + 
P\Big( |Z|^2 \geq \frac{\Mr \omega \e}{ \kappa (\Lr/\omega)^{\sfq}}\Big) \nonumber \\
& \stackrel{(a)}{\leq} \Lc K  \exp\{-\kappa n ( {\omega}/{\Lr}) \} + \exp\{-  n (\omega/\Lr)^{\sfq+1} \e /\kappa \},
\end{align}
where step $(a)$ follows from $\mathcal{B}_0 (e)$ and  Lemma \ref{lem:normalconc}.  The overall upper bound is $\exp\{- \kappa  n (\omega/\Lr)^{2\sfq+1} \e\}$, which is again consistent with the $\PC, \pc$ notation defined in \eqref{eq:param_const_defs} for $t=0$.

\noindent  \textbf{(b)} From Lemma \ref{lem:hb_cond}, we have $\hvec^1_{\sfc} \lvert_{\mscrs_{1, 0}} \stackrel{d}{=} \sqrt{ \tau^{0}_{ \sfc}} \, \Z_{0, \sfc}  + \Dvec_{1,0, \sfc}$.  Using this, we obtain
\begin{align}
&P\Big(  \Big \lvert \sum_{\sfc \in [\Lc]} \sqrt{W_{\sfr \sfc}}  \frac{(\hvec^1_{\sfc})^*\qadj^{0, \sfr}_{\sfc}}{L} \Big \lvert \geq \epsilon \Big)  \leq P\Big(   \Big \lvert \sum_{\sfc \in [\Lc]}  \sqrt{W_{\sfr \sfc} \tau^0_{\sfc}}   \frac{\Z_{0, \sfc}^*  \qadj^{0, \sfr}_{\sfc}}{L} \Big \lvert\geq \frac{ \epsilon}{2} \Big) + P\Big( \sum_{\sfc \in [\Lc]}  \sqrt{W_{\sfr \sfc}} \Big \lvert \frac{\Dvec_{1,0, \sfc}^* \qadj^{0, \sfr}_{\sfc}}{L} \Big\lvert \geq \frac{ \epsilon}{2}\Big). \label{eq:H1b_split}
\end{align}
Label the terms on the right side as $T_1$ and $T_2$.  

Consider term $T_1$.  
Since $\qadj^{0, \sfr}_{\sfc}$ is independent of $\Z_{0, \sfc}$, we have $\Z_{0, \sfc}^*\,\qadj^{0, \sfr}_{\sfc} \stackrel{d}{=} \| \qadj^{0, \sfr}_{\sfc} \| Z^{\sfc}$, where $Z^{\sfc} \sim \mc{N}(0,1)$ and 
$\| \qadj^{0, \sfr}_{\sfc} \|^2 = W_{\sfr\sfc} \| \betavec_{0, \sfc} \|^2 =  W_{\sfr\sfc}  ({L}/{\Lc}).$
Therefore, 
\begin{align*}
\Big \lvert \sum_{\sfc \in [\Lc]}   \sqrt{W_{\sfr \sfc} \tau^0_{\sfc}} \,   \Z_{0, \sfc}^* \, \qadj^{0, \sfr}_{\sfc} \Big \lvert &\overset{d}{=} \Big \lvert \sqrt{\frac{L}{\Lc}} \sum_{\sfc \in [\Lc]}   \sqrt{\tau^0_{\sfc}  } \, W_{\sfr\sfc} \, Z^{\sfc} \Big \lvert \overset{d}{=} \Big \lvert Z \Big(\frac{L}{\Lc} \sum_{\sfc \in [\Lc]}  \tau^0_{\sfc} \, W^2_{\sfr\sfc}  \Big)^{\frac{1}{2}} \Big \lvert \leq  
\kappa \abs{Z} \sqrt{\frac{L^2 \Lr}{n \omega}},
\end{align*}
where $Z \sim \normal(0,1)$. Recall, $n = \Mr \Lr$, therefore $T_1 \leq  P( \kappa \abs{Z} \sqrt{{\Lr}/{(n \omega)}} \geq \frac{\e}{2}) \leq \exp\{ - \kappa \Mr \omega \e^2 \}.$

For term $T_2$, observe that $\sum_{i \in \sec(\ell)}  [\qadj^{0, \sfr}_{\sfc}]_{i}  = -\sqrt{W_{\sfr\sfc}}$ for each $\sfc \in [\Lc]$ and section $\ell \in \sfc$,
\begin{align*}
 \sum_{\sfc \in [\Lc]}   \sqrt{W_{\sfr \sfc}} \abs{\Dvec_{1,0, \sfc}^* \, \qadj^{0, \sfr}_{\sfc}}  & =  \sum_{\sfc \in [\Lc]}   \sum_{\ell \in \sfc}  \sqrt{W_{\sfr \sfc}} \Big \lvert ([\Dvec_{1,0, \sfc}]_{\ell})^* \, [\qadj^{0, \sfr}_{\sfc}]_{\ell} \Big \lvert 
\leq   \sum_{\sfc \in [\Lc]}  \sum_{\ell \in \sfc}  W_{\sfr \sfc}  \max_{i \in sec(\ell)} \abs{ [\Dvec_{1,0, \sfc}]_{i}}. 
\end{align*}
Now $ (\frac{1}{L} \sum_{\sfc \in [\Lc]}    \sum_{\ell \in \sfc} W_{\sfr \sfc} \max_{i \in sec(\ell)} \abs{ [\Dvec_{1,0, \sfc}]_{i}} )^2 \leq  \frac{1}{L} \sum_{\sfc \in [\Lc]}    \sum_{\ell \in \sfc} W_{\sfr \sfc}^2 \max_{i \in sec(\ell)} \abs{ [\Dvec_{1,0, \sfc}]_{i}}^2$ by Lemma \ref{lem:squaredsums} and therefore by $\mathcal{H}_1 (a)$ with $\sfq = 1$,
\begin{align}
T_2 &\leq %P\Big(  \frac{1}{L} \sum_{\sfc \in [\Lc]}    \sum_{\ell \in \sfc} W_{\sfr \sfc} \max_{i \in sec(\ell)} \abs{ [\Dvec_{1,0, \sfc}]_{i}} \geq \kappa \epsilon \Big) \nonumber \\
%
%&\leq 
P\Big(\frac{1}{L}  \sum_{\sfc \in [\Lc]}   \sum_{\ell \in \sfc} W^2_{\sfr \sfc} \max_{i \in sec(\ell)} \abs{ [\Dvec_{1,0, \sfc}]_{i}}^2 \geq \kappa  \epsilon^2 \Big) \leq  K \Lc \exp\{-\kappa n (\omega/\Lr)^3 \e^2 \}.  \label{eq:H1b_T2}
\end{align}

%%%%%

\noindent  \textbf{(c)}  We begin by showing the result \eqref{eq:Hc_adj} for $s=1$.    Recall that  
$\eta^0_{\sfc}(\betavec_{0, \sfc} - \hvec^1_{\sfc}) - \betavec_{0, \sfc}$ and $\hvec^1_{\sfc} \lvert_{\mscrs_{1, 0}} \stackrel{d}{=} \sqrt{ \tau^{0}_{ \sfc}} \, \Z_{0, \sfc}  + \Dvec_{1,0, \sfc}$ from Lemma \ref{lem:hb_cond}. Therefore, by Lemma~\ref{sums},
\begin{align*}
&P\Big( \Big \lvert  \sum_{\sfc \in [\Lc]}  W^{\sfq}_{\sfr \sfc} \Big(\frac{\norm{\brqvec^{1}_{\sfc}}^2}{L}   - \frac{ \psi^{1}_{\sfc}}{\Lc} \Big) \Big \lvert \geq \epsilon \Big) \\
&= P\Big(  \frac{1}{\Lc} \Big \lvert \sum_{\sfc \in [\Lc]}  W^{\sfq}_{\sfr \sfc}  \Big( \frac{\Lc}{L} \norm{\eta^0_{\sfc}(\betavec_{0, \sfc} - \sqrt{ \tau^0_{\sfc}} \Z_{0, \sfc} - \Dvec_{1,0, \sfc}) - \betavec_{0, \sfc}}^2 -    \psi^{1}_{\sfc} \Big) \Big \lvert \geq \epsilon \Big)\\
& \leq P\Big(  \frac{1}{\Lc} \Big \lvert \sum_{\sfc \in [\Lc]} W^{\sfq}_{\sfr \sfc}  \Big(  \frac{\Lc}{L} \norm{\eta^0_{\sfc}(\betavec_{0, \sfc} - \sqrt{\tau^0_{\sfc}} \Z_{0, \sfc}) - \betavec_{0, \sfc}}^2 -  \psi^{1}_{\sfc}  \Big)\Big \lvert \geq \frac{\epsilon}{2} \Big) \nonumber \\
&+ P\Big(  \frac{1}{L} \Big \lvert \sum_{\sfc \in [\Lc]} W^{\sfq}_{\sfr \sfc}  \Big( \norm{\eta^0_{\sfc}(\betavec_{0, \sfc} - \sqrt{\tau^0_{\sfc}} \Z_{0, \sfc} - \Dvec_{1,0, \sfc}) - \betavec_{0, \sfc}}^2  -   \norm{\eta^0_{\sfc}(\betavec_{0, \sfc} - \sqrt{\tau^0_{\sfc}} \Z_{0, \sfc}) - \betavec_{0, \sfc}}^2 \Big) \Big\lvert \geq \frac{\epsilon}{2} \Big).
\end{align*}
Label the two terms on the RHS as $T_1$ and $T_2$. 

To bound $T_1$, we write 
\begin{align}
T_1 &=  P\Big(\frac{1}{L}  \Big \lvert \sum_{\sfc \in [\Lc]}  W^{\sfq}_{\sfr \sfc}   \sum_{\ell \in \sfc} \Big(  
\| \eta^0_{\ell}(\betavec_{0, \sfc} - \sqrt{\tau^0_{\sfc}} \Z_{0, \sfc}) - [\betavec_{0, \sfc}]_{\ell} \|^2 - \psi_{\sfc}^1 \Big)\Big \lvert \geq \frac{\epsilon}{2} \Big), 
\label{eq:Hoeffdings1}
\end{align}
and apply Hoeffding's inequality (Lemma \ref{lem:hoeff_lem}). To do this, notice that 
$\big( \| \eta^0_{\ell}(\betavec_{0, \sfc} - \sqrt{\tau^0_{\sfc}} \Z_{0, \sfc}) - [\betavec_{0, \sfc}]_{\ell} \|^2 - \psi_{\sfc}^1 \big)$ is bounded in absolute value by 1, and has zero mean. Indeed, 
\be
\mathbb{E}_{\Z_{0}}\| \eta^0_{\ell}(\betavec_{0, \sfc} - \sqrt{\tau^0_{\sfc}} \Z_{0, \sfc}) - [\betavec_{0, \sfc}]_{\ell} \|^2 = \mathbb{E}_{\betavec, \Z_{0}}\| \eta^0_{\ell}(\betavec_{\sfc} - \sqrt{\tau^0_{\sfc}} \Z_{0, \sfc}) - [\betavec_{\sfc}]_{\ell} \|^2 = \psi^1_{\sfc},
\label{eq:Hoeffdings2}
\ee
where the first equality is true for each $\betavec_{0} \in \mathcal{B}_{M,L}$ because of the uniform distribution of the non-zero entry in each section of $\betavec_{0}$ over the $M$ possible locations and the i.i.d.\ distribution of $\Z_{0, \sfc}$. The second equality follows by Lemma \ref{lem:expect_etar_etas}. Applying Hoeffding's inequality to  \eqref{eq:Hoeffdings1}, we obtain
\be
T_1 \leq 2 \exp\Big\{\frac{-\kappa L^2 \e^2}{(\sum_{\sfc \in [\Lc]}   \sum_{\ell \in \sfc}   W^{2\sfq}_{\sfr \sfc} )}\Big\}  = 2 \exp\Big\{\frac{-\kappa L \e^2}{(\sum_{\sfc \in [\Lc]}    W^{2\sfq}_{\sfr \sfc} /\Lc)}\Big\} = 2 \exp\{-\kappa L  (\omega/\Lr)^{\sfq}\e^2\},
\ee
where for the last equality we have used \eqref{eq:Wrc_avgs}. 

Next, we bound term $T_2$. To save space,  we write $\eta^0_{\sfc}(\betavec_{0} - \sqrt{\tau^0_{\sfc}} \Z_{0} -   \Dvec_{1,0})$ to mean $\eta^0_{\sfc}(\betavec_{0, \sfc} - \sqrt{\tau^0_{\sfc}} \Z_{0, \sfc} - \Dvec_{1,0, \sfc})$. First, 
\be
\begin{split}
& \| \eta^0_{\sfc}(\betavec_{0} - \sqrt{\tau^0_{\sfc}} \Z_{0} - \Dvec_{1,0}) - \betavec_{0, \sfc} \|^2 - 
\| \eta^0_{\sfc}(\betavec_{0} - \sqrt{\tau^0_{\sfc}} \Z_{0}) - \betavec_{0, \sfc} \|^2 \\
&= [\eta^0_{\sfc}(\betavec_{0} - \sqrt{\tau^0_{\sfc}} \Z_{0} - \Dvec_{1,0}) - \betavec_{0, \sfc}]^*[\eta^0_{\sfc}(\betavec_{0} - \sqrt{\tau^0_{\sfc}} \Z_{0} -   \Dvec_{1,0}) - \eta^0_{\sfc}(\betavec_{0} - \sqrt{\tau^0_{\sfc}} \Z_{0})] \\
&\qquad +  [ \eta^0_{\sfc}(\betavec_{0} - \sqrt{\tau^0_{\sfc}} \Z_{0}) - \betavec_{0, \sfc}]^*[\eta^0_{\sfc}(\betavec_{0} - \sqrt{\tau^0_{\sfc}} \Z_{0} -  \Dvec_{1,0}) - \eta^0_{\sfc}(\betavec_{0} - \sqrt{\tau^0_{\sfc}} \Z_{0})].
\label{eq:HCterm2_1}
\end{split}
\ee
Using this, we have
\ben 
\begin{split}
&T_2 \leq P\Big(\frac{1}{L} \sum_{\sfc \in [\Lc]}  W^{\sfq}_{\sfr \sfc}   \Big \lvert   [ \eta^0_{\sfc}(\betavec_{0} - \sqrt{\tau^0_{\sfc}} \Z_{0}) - \betavec_{0, \sfc}]^*[\eta^0_{\sfc}(\betavec_{0} - \sqrt{\tau^0_{\sfc}} \Z_{0} -   \Dvec_{1,0}) - \eta^0_{\sfc}(\betavec_{0} - \sqrt{\tau^0_{\sfc}} \Z_{0})] \Big \lvert \geq \frac{\kappa \epsilon}{4} \Big)  \\
&+ P\Big(\frac{1}{L} \sum_{\sfc \in [\Lc]}  W^{\sfq}_{\sfr \sfc}   \Big \lvert  [\eta^0_{\sfc}(\betavec_{0} - \sqrt{\tau^0_{\sfc}} \Z_{0} -  \Dvec_{1,0}) - \beta_{0,\sfc}]^*[\eta^0_{\sfc}(\betavec_{0} - \sqrt{\tau^0_{\sfc}} \Z_{0} - \Dvec_{1,0}) - \eta^0_{\sfc}(\betavec_{0} - \sqrt{\tau^0_{\sfc}} \Z_{0})] \Big \lvert \geq \frac{\kappa \epsilon}{4} \Big).
\end{split} 
\een
Label the terms on the right side  above as $T_{2,a}$ and $T_{2,b}$.  Then,
\be 
\begin{split}
&T_{2,a} 
 \stackrel{(a)}{\leq}  P\Big(\frac{1}{L}\sum_{\sfc \in [\Lc]} \sum_{\ell \in \sfc} \sum_{j \in sec(\ell)}  W^{\sfq}_{\sfr \sfc}    \Big\lvert \eta^0_{j}(\betavec_{0} -  \sqrt{\tau^0_{\sfc}} \Z_{0} - \Dvec_{1,0}) - \eta^0_{j}(\betavec_{0} -  \sqrt{\tau^0_{\sfc}} \Z_{0}) \Big \lvert \geq \frac{\kappa \epsilon}{4} \Big) \\
&\overset{(b)}{\leq}  P\Big(\sum_{\sfc \in [\Lc]}   \sum_{\ell \in \sfc} \frac{ W^{\sfq}_{\sfr \sfc}   }{L \tau^0_{\sfc}} \max_{j \in sec(\ell)}   \abs{[\Dvec_{1,0, \sfc}]_j} \geq \frac{\kappa \epsilon}{8} \Big)  \\
&\overset{(c)}{\leq} P\Big(\frac{1}{L} \sum_{\sfc \in [\Lc]}  \sum_{\ell \in \sfc} W^{2 \sfq}_{\sfr \sfc}    \max_{j \in sec(\ell)} \abs{[\Dvec_{1,0,\sfc}]_j}^2 \geq \frac{\kappa^2  \epsilon^2 }{(\log M)^2} \Big) \overset{(d)}{\leq} K \Lc \exp\Big\{\frac{-\kappa n (\omega/\Lr)^{2 \sfq + 1} \e^2}{(\log M)^2} \Big\}.
\label{eq:HCterm2_2}
\end{split} 
\ee
Step $(a)$ holds since $|\eta^0_{j}(\betavec_{0} - \sqrt{\tau^0_{\sfc}} \Z_{0}) - \beta_{0, j}| \leq 1$ for $j  \in [ML]$, step $(b)$  by Lemma \ref{lem:BC9}, and step $(c)$ by Lemma~\ref{lem:squaredsums} and $\tau^0_{\sfc}= \Theta(1/\log M) $. Finally, step $(d)$ follows from $\mathcal{H}_1 (a)$.  

Using Lemma \ref{lem:BC9}, it can be shown that term $T_{2,b}$ also has the same upper bound. This proves the concentration result \eqref{eq:Hc_adj} for $s=1$. Proving the  result for $s=0$ is similar: we use Lemma \ref{sums} followed by Hoeffding's inequality and Lemma \ref{lem:BC9}. 

%%%%%%%
\vspace{5pt}

\noindent  \textbf{(d)} Recalling $\qadj^{1, \sfr}_{\sfc} = \sqrt{W_{\sfr \sfc}} [\eta^0_{\sfc}(\betavec_{0} - \hvec^1) - \betavec_{0, \sfc}]$, we write 
$\sum_{\sfc \in [\Lc]}  (\hvec^1_{\sfc})^* \qadj_{\sfc}^{1, \sfr} =  \sum_{\sfc \in [\Lc]}  \sqrt{W_{\sfr \sfc}}  (\hvec^1_{\sfc})^* \eta^0_{\sfc}(\betavec_{0} - \hvec^1) + (\hvec^1_{\sfc})^*\qadj^{0, \sfr}_{\sfc},$ and therefore, by Lemma \ref{sums},
\begin{align}
&P\Big(\Big \lvert \sum_{\sfc \in [\Lc]} \Big(\frac{ \sqrt{W_{\sfr \sfc}} }{L} (\hvec^1_{\sfc})^*\qadj^{1, \sfr}_{\sfc}  +    \frac{ W_{\sfr \sfc}}{\Lc}   \,  \psi^{1}_{\sfc} \Big) \Big \lvert \geq \epsilon \Big) \nonumber \\
 &\leq   P\Big( \frac{1}{\Lc} \Big \lvert \sum_{\sfc \in [\Lc]}  W_{\sfr \sfc}   \Big[ \frac{\Lc}{L }  (\hvec^1_{\sfc})^* \eta^0_{\sfc}(\betavec_{0} - \hvec^1)  +   \psi^{1}_{\sfc} \Big] \Big \lvert \geq \frac{\e}{2} \Big) + P\Big(  \frac{1}{L}   \Big \lvert  \sum_{\sfc \in [\Lc]}   \sqrt{W_{\sfr \sfc}} (\hvec^1_{\sfc})^*\qadj^{0, \sfr}_{\sfc} \Big \lvert \geq \frac{\e}{2} \Big). \label{eq:h1q1_split}
\end{align}
By $\mathcal{H}_1 (b)$, the second term in \eqref{eq:h1q1_split} is  bounded by 
$K \Lc \exp\{-\kappa  n (\omega/\Lr)^3  \e^2  \}$.  Using the conditional distribution of $\hvec^1$ stated in Lemma \ref{lem:hb_cond} and Lemma \ref{sums}, for the first term of \eqref{eq:h1q1_split} we write 
\begin{align}
&P\Big(  \frac{1}{\Lc}  \Big \lvert \sum_{\sfc \in [\Lc]} W_{\sfr \sfc} \Big ( \frac{\Lc}{L }(\hvec^1_{\sfc})^* \eta^0_{\sfc}(\betavec_{0} - \hvec^1)  +  \psi^{1}_{\sfc}  \Big)\Big \lvert \geq \frac{\e}{2} \Big)   \nonumber \\
&\leq P\Big( \frac{1}{\Lc}\Big \lvert \sum_{\sfc \in [\Lc]} W_{\sfr \sfc}  \Big ( \frac{\Lc}{L}(\sqrt{\tau^0_{\sfc}} \Z_{0, \sfc} + \Dvec_{1,0, \sfc})^* \eta^0_{\sfc}(\betavec_{0, \sfc} - \sqrt{\tau^0_{\sfc}} \Z_{0, \sfc} -  \Dvec_{1,0, \sfc})  +  \psi^{1}_{\sfc} \Big) \Big \lvert \geq \epsilon \Big) \nonumber  \\ 
&\leq P\Big( \frac{1}{L} \sum_{\sfc \in [\Lc]} W_{\sfr \sfc}   \Big \lvert (\Dvec_{1,0, \sfc})^* \eta^0_{\sfc}(\betavec_{0, \sfc} - \sqrt{\tau^0_{\sfc}} \Z_{0, \sfc} -  \Dvec_{1,0, \sfc})\Big \lvert \geq \kappa \epsilon \Big) \nonumber \\
&\ + P\Big(\frac{1}{\Lc} \Big \lvert  \sum_{\sfc \in [\Lc]} W_{\sfr \sfc} \Big( \frac{\Lc}{L} \sqrt{\tau^0_{\sfc}} \Z_{0, \sfc}^* \eta^0_{\sfc}(\betavec_{0, \sfc} - \sqrt{\tau^0_{\sfc}} \Z_{0, \sfc})  + \psi^{1}_{\sfc} \Big) \Big\lvert \geq  \kappa  \epsilon \Big)  \label{eq:H1dT1T2T3} \\
&\ + P\Big( \frac{1}{L} \sum_{\sfc \in [\Lc]} \sqrt{\tau^0_{\sfc}  }  W_{\sfr \sfc}  \Big \lvert \Z_{0, \sfc}^* \Big[\eta^0_{\sfc}(\betavec_{0, \sfc} - \sqrt{\tau^0_{\sfc}} \Z_{0, \sfc} - \Dvec_{1,0, \sfc}) - \eta_{\sfc}(\betavec_{0, \sfc} - \sqrt{\tau^0_{\sfc}} \Z_{0, \sfc})\Big] \Big \lvert \geq  \kappa  \epsilon  \Big). \nonumber
\end{align}
 Label the terms of the above as $T_1 , T_2, T_3$ and we bound each individually.

 For term $T_1$ notice that $T_{1} \leq P(\frac{1}{L} \sum_{\sfc \in [\Lc]} \sum_{\ell \in \sfc}  W_{\sfr \sfc}  \max_{i \in sec(\ell)} \abs{[\Dvec_{1,0, \sfc}]_i} \geq \kappa\epsilon)$ and therefore the term can be bounded as in \eqref{eq:H1b_T2} using $\mathcal{H}_1 (a)$ with $\sfq =1$.
% 
%
% We bound $T_1$ using $\mathcal{H}_1 (a)$ with $\sfq =1$:
% \begin{align*}
%%
%T_{1} &\leq P\Big(\frac{1}{L} \sum_{\sfc \in [\Lc]} \sum_{\ell \in \sfc}  W_{\sfr \sfc}  \max_{i \in sec(\ell)} \abs{[\Dvec_{1,0, \sfc}]_i} \geq \kappa\epsilon \Big) \\ 
%%
%&\overset{(a)}{\leq} P\Big( \frac{1}{L} \sum_{\sfc \in [\Lc]}  \sum_{\ell \in \sfc}   W^2_{\sfr \sfc}   \max_{i \in sec(\ell)} \abs{[\Dvec_{1,0, \sfc}]_i}^2 \geq  \kappa \epsilon^2 \Big) \overset{(b)}{ \leq} K  \Lc \exp\{-\kappa n (\omega/\Lr)^3 \e^2  \}.
%%
%\end{align*}

Next consider term $T_2$ of \eqref{eq:H1dT1T2T3}.  Because of the uniform distribution of the non-zero entry in each section of $\betavec$ over the $M$ possible locations and the i.i.d.\ distribution of $\Z_{0, \sfc}$,  for any $\betavec_{0} \in \mathcal{B}_{M,L}$, we have $\mathbb{E}_{\Z_{0}}\{\Z_{0, \sfc}^* \eta^0_{\sfc}(\betavec_{0, \sfc} - \sqrt{\tau^0_{\sfc}} \Z_{0, \sfc})\}  = \mathbb{E}_{\Z_{0}, \betavec}\{\Z_{0, \sfc}^* \eta^0_{\sfc}(\betavec_{0, \sfc} - \sqrt{\tau^0_{\sfc}} \Z_{0, \sfc})\} $.  The expectation equals 
\be
\begin{split} 
& \sqrt{\tau^0_{\sfc}} \mathbb{E}\{\Z_{0, \sfc}^* \eta^0_{\sfc}(\betavec_{0, \sfc} - \sqrt{\tau^0_{\sfc}} \Z_{0, \sfc})\} \stackrel{(a)}{=} 
\mathbb{E} \| \eta^0_{ \sfc}(\betavec_{0, \sfc} - \sqrt{\tau^0_{\sfc}} \Z_{0, \sfc}) \|^2 - \frac{L}{\Lc} \stackrel{(b)}{=} -\psi^1_{\sfc} \frac{L}{\Lc},
\label{eq:Hd12a}
\end{split}
\ee
where equality $(a)$ is obtained using Stein's lemma, Lemma \ref{lem:stein} (see  \cite[p.1491, Eqs.\ (102) -- (104)]{rush2017} for details) and equality $(b)$ from Lemma \ref{lem:expect_etar_etas}.  Now,  using \eqref{eq:Hd12a} and  the fact that $\tau^0_{\sfc}= \Theta(1/\log M)$,  the concentration result from Lemma \ref{lem:Hd_convergence} yields
\ben
\begin{split}
T_2 &= P\Big( \frac{1}{L}  \Big \lvert \sum_{\sfc \in [\Lc]} \sqrt{ \tau^0_{\sfc} } W_{\sfr \sfc}   \sum_{\ell \in \sfc} \Big[ ([\Z_{0, \sfc}]_{\ell})^* \eta^0_{\ell}(\betavec_{0, \sfc} -\sqrt{\tau^0_{\sfc}} \Z_{0, \sfc}) - \mathbb{E}\{([\Z_{0, \sfc}]_{\ell})^* \eta^0_{\ell}(\betavec_{0, \sfc} - \sqrt{\tau^0_{\sfc}} \Z_{0, \sfc})\} \Big] \Big \lvert \geq  \kappa \epsilon \Big) \\
&\leq \exp\{-\kappa L  (\omega/\Lr ) \e^2 \}. 
%= \exp\{-\kappa \Mr \omega \Lc \e^2/\Lr \log M\}.
%
\end{split}
\een

Finally consider term $T_3$ in \eqref{eq:H1dT1T2T3}.  First, 
\begin{align}
&T_{3}  
%= P\Big(\frac{1}{L}  \sum_{\sfc \in [\Lc]} \sqrt{\tau^0_{\sfc}} W_{\sfr \sfc}    \sum_{j \in \sfc} \Big \lvert [\Z_{0, \sfc}]_j [\eta^0_{j}(\betavec_{0} - \sqrt{\tau^0_{\sfc}} \Z_{0} - \Dvec_{1,0}) - \eta^0_{j}(\betavec_{0} - \sqrt{\tau^0_{\sfc}} \Z_{0})]\Big \lvert \geq  \kappa \e \Big) \\
%
\leq P\ \Big(\frac{1}{L}  \sum_{\sfc \in [\Lc]} \sqrt{\tau^0_{\sfc}} W_{\sfr \sfc}  \sum_{\ell \in \sfc} \max_{k  \in sec(\ell)} \abs{[\Z_{0}]_k} \hspace{-5pt} \sum_{j  \in  sec(\ell)} \lvert \eta^0_{j}(\betavec_{0} - \sqrt{\tau^0_{\sfc}} \Z_{0} -  \Dvec_{1,0}) - \eta^0_{j}(\betavec_{0} - \sqrt{\tau^0_{\sfc}} \Z_{0}) \lvert \geq \kappa 
\e \Big) \nonumber \\
 & \overset{(a)}{\leq}  P\Big( \frac{1}{L}  \sum_{\sfc \in [\Lc]} \frac{W_{\sfr \sfc}  }{\sqrt{\tau^0_{\sfc}}}  \sum_{\ell \in \sfc}  \max_{k \in sec(\ell)} \abs{[\Z_{0}]_k} 
 \max_{j \in sec(\ell)} \abs{[\Dvec_{1,0}]_{j}}  \geq  \kappa  \epsilon  \Big) \overset{(b)}{\leq} K \Lc \exp\Big\{\frac{-\kappa  n (\omega/\Lr)^3 \e^2 }{ (\log M)^2} \Big\}.  \label{eq:H1c_T3_eq1}
\end{align}
Step $(a)$ follows from Lemma \ref{lem:BC9} and step $(b)$  is obtained as follows.
\begin{align}
&P\Big( \frac{1}{L}  \sum_{\sfc \in [\Lc]}     \sum_{\ell \in \sfc}   \frac{W_{\sfr \sfc}   }{\sqrt{\tau^0_{\sfc}}}  \max_{j \in sec(\ell)} \abs{[\Dvec_{1,0, \sfc}]_{j}} \max_{k \in sec(\ell)} \abs{[\Z_{0, \sfc}]_k}   \geq  \kappa  \epsilon  \Big) \nonumber \\
&\overset{(c)}{\leq} P\Big(\frac{1}{L} \sum_{\sfc \in [\Lc]}   \sum_{\ell \in \sfc} W_{\sfr \sfc}  \max_{j \in sec(\ell)} \abs{[\Dvec_{1,0, \sfc}]_{j}} \max_{k \in sec(\ell)} \abs{[\Z_{0, \sfc}]_k}   \geq  \kappa \epsilon  \sqrt{L/n}  \Big) \nonumber \\
&\overset{(d)}{\leq} P\Big(\frac{1}{L}  \sum_{\sfc \in [\Lc]} \sum_{\ell \in \sfc}  \max_{k \in sec({\ell})} \abs{[Z_{0, \sfc}]_k}^2 \cdot \frac{1}{L} \sum_{\sfc\in [\Lc]} \sum_{\ell \in \sfc} W^2_{\sfr \sfc}   \max_{j \in sec(\ell)} \abs{[\Dvec_{1,0, \sfc}]_{j}}^2 \geq \frac{\kappa \epsilon^2 L}{n}  \Big)  \nonumber \\
&\leq P \Big( \frac{1}{L}  \sum_{\sfc \in [\Lc]} \sum_{\ell \in \sfc}   \max_{k \in sec({\ell})}  \abs{[Z_{0, \sfc}]_k}^2 \geq 3 \log M \Big) + P \Big( \frac{1}{L}  \sum_{\sfc \in [\Lc]}  \sum_{\ell \in \sfc} W^2_{\sfr \sfc}  \max_{j \in sec(\ell)} \abs{[\Dvec_{1,0, \sfc}]_{j}}^2   \geq\frac{\kappa \e^2L }{3n \log M}\Big) \nonumber  \\
&\overset{(e)}{\leq}  e^{-\kappa n } +K \Lc \exp\Big\{\frac{-\kappa  n (\omega/\Lr)^3   \e^2}{ (\log M)^2} \Big\}.  \label{eq:H1c_T3_eq2}
\end{align}
Step $(c)$ follows by using $\tau^0_{\sfc}= \Theta(L/n)$ and step $(d)$ by Cauchy-Schwarz.  Finally step $(e)$ follows from Lemma \ref{lem:max_abs_normals}, $\mathcal{H}_1 (a)$ with $\sfq =1$ along with  $n R = L \log M$ and $n = \Lr\Mr$.

%\RV{This seems to be the dominant term so far.}

\vspace{5pt}

\noindent  \textbf{(e)}  From the conditional distribution of $\hvec^1_{\sfc}$ stated in Lemma \ref{lem:hb_cond} and Lemma \ref{lem:squaredsums}, it follows that $([\hvec^1_{\sfc}]_j)^2 \leq 2\tau^0_{\sfc}([\Z_{0, \sfc}]_j)^2 + 2([ \Dvec_{1,0,\sfc}]_{j})^2$  for $j \in sec(\ell)$.  We therefore have the following  bound:
\begin{align}
& P\Big(\frac{1}{L}  \sum_{\sfc \in [\Lc]}  \sum_{\ell \in \sfc}  W^{2 \sfq}_{\sfr\sfc} \max_{j \in sec(\ell)}  ([\hvec^1_{\sfc}]_j)^2 \geq 6 (\max_{\sfc} W^{2\sfq}_{\sfr \sfc}) \tau^0_{\sfc} \log M +  \e \Big)  \nonumber \\
&\leq P\Big(\frac{1}{L}  \sum_{\sfc \in [\Lc]}  \sum_{\ell \in \sfc}  W^{2 \sfq}_{\sfr\sfc}  \Big[ \tau^0_{\sfc} \max_{j \in sec(\ell)}  ([\Z_{0, \sfc}]_j)^2 + \max_{j \in sec(\ell)}([ \Dvec_{1,0,\sfc}]_{j})^2 \Big]  \geq 3 (\max_{\sfc} W^{2\sfq}_{\sfr \sfc}) \tau^0_{\sfc} \log M  + \e/2 \Big)  \nonumber \\
& \leq P\Big(\frac{1}{L}  \sum_{\sfc \in [\Lc]}  \sum_{\ell \in \sfc}  W^{2 \sfq}_{\sfr\sfc}  \hspace{-1pt} \max_{j \in sec(\ell)} \hspace{-1pt}  ([\Z_{0, \sfc}]_j)^2 \geq 3 (\max_{\sfc} W^{2\sfq}_{\sfr \sfc})  \log M\Big)  + P\Big(\frac{1}{L}  \sum_{\sfc \in [\Lc]}  \sum_{\ell \in \sfc}  \hspace{-1pt} \max_{j \in sec(\ell)}  \hspace{-1pt}  W^{2 \sfq}_{\sfr\sfc} ([ \Dvec_{1,0,\sfc}]_{j})^2 \geq \frac{ \e}{2} \Big) \nonumber \\
&\overset{(a)}{\leq} 2 \exp\{-\kappa L \log M\} +  K \Lc \exp\{-\kappa n (\omega/\Lr)^{2 \sfq + 1} \e  \}. \label{eq:H1e_eq1}
\end{align}
The inequality $(a)$ follows from Lemma \ref{lem:max_abs_normals} and $\mc{H}_1 (a)$.

%%%%%%%%

\vspace{5pt}

\noindent  \textbf{(f)} We first prove \eqref{eq:Hf}, then \eqref{eq:Hf1}. Recall,
$\gamma^{1,\sfr}_0 = \frac{(\qadj^{0, \sfr})^*\qadj^{1,\sfr}}{\norm{\qadj^{0, \sfr}}^2} = \frac{(\qadj^{0, \sfr})^*\qadj^{1,\sfr}}{L \sigma^0_{\sfr}}$. Then since 
$\gamma^{1,\sfr}_0 = 
{\sigma^1_{\sfr}}/{ \sigma^0_{\sfr}}$, result \eqref{eq:Hf} follows directly from $\mathcal{H}_1 (c)$ with $\sfq = 1$ as $ \frac{1}{\Lc} \sum_{\sfc \in [\Lc]}   W_{\sfr \sfc} \psi^{1}_{\sfc} = \sigma^1_{\sfr} $ and $\sigma^0_{\sfr} \in \Theta(1)$:
\be \begin{split}
P \Big( \Big \lvert \gamma^{1,\sfr}_0 - \frac{\sigma^1_{\sfr}}{ \sigma^0_{\sfr}}  \Big \lvert \geq \epsilon  \Big) &= P \Big( \Big \lvert  \frac{1}{L}(\qadj^{0, \sfr})^*\qadj^{1,\sfr} -  \sigma^1_{\sfr} \Big \lvert \geq \epsilon \sigma^0_{\sfr}  \Big)  \leq K \Lc \exp\Big\{\frac{-\kappa n (\omega/\Lr)^3 \e^2}{(\log M)^2} \Big\}.
\end{split} 
\label{eq:eq:H1f1}
\ee
For \eqref{eq:Hf1}, by definition, $\| \qadj^{1, \sfr}_{\perp} \|^2 = \norm{\qadj^{1,  \sfr}}^2 - (\gamma^{1,  \sfr}_0)^2 \norm{\qadj^{0, \sfr}}^2 = \norm{\qadj^{1,  \sfr}}^2 - (\gamma^{1,  \sfr}_0)^2 L \sigma^0_{\sfr} .$  Using this and the fact that $\sigma^1_{\perp,  \sfr} = \sigma^1_{ \sfr}[1 - (\sigma^1_{ \sfr}/\sigma^0_{\sfr})]$, we obtain the following upper bound:
\begin{align*}
P&\Big(\Big \lvert \frac{ \| \qadj^{1,  \sfr}_{\perp} \|^2}{L} - \sigma^1_{\perp,  \sfr} \Big \lvert \geq \epsilon \Big) = P\Big(\Big \lvert \frac{\| \qadj^{1,  \sfr} \|^2 }{L} - (\gamma^{1,  \sfr}_0)^2\sigma^0_{\sfr}  - \sigma^1_{ \sfr}\Big[1 - \frac{\sigma^1_{ \sfr}}{\sigma^0_{\sfr}} \Big] \Big \lvert \geq \epsilon \Big) \\
&\leq P\Big(\Big \lvert \frac{\| \qadj^{1,  \sfr} \|^2 }{ L} - \sigma^1_{ \sfr}\Big \lvert \geq \frac{\epsilon}{2} \Big) + P\Big(\Big \lvert  (\gamma^{1,  \sfr}_0)^2 - \frac{(\sigma^1_{ \sfr})^2}{(\sigma^0_{\sfr})^2} \Big \lvert \geq \frac{\epsilon}{2 \sigma^0_{\sfr}} \Big) \leq K \Lc \exp\Big\{\frac{-\kappa n (\omega/\Lr)^3 \e^2}{(\log M)^2} \Big\}.
\end{align*}
The second inequality follows from $\mathcal{H}_1 (c)$ with $\sfq = 1$ and \eqref{eq:eq:H1f1}  above, along with Lemma \ref{powers}.

%%%%%%%%%%%

\vspace{5pt}

\noindent  \textbf{(g)}  Note that $\norm{\qadj^{0, \sfr}}^2 = L  \sigma^0_{\sfr}$ by \eqref{eq:breveq_simp}
and, therefore, $\Qadjmat^{\sfr}_{1} := \frac{1}{L} \norm{\qadj^{0, \sfr}}^2 = \sigma^0_{\sfr}= \tilde{C}^1.$ So results \eqref{eq:Qsing} and \eqref{eq:Hg} are trivially true for $\Qadjmat^{\sfr}_{1}$.

We now show \eqref{eq:Qsing} for $\Qadjmat^{\sfr}_{2}$.  Recall, $[\Qadjmat^{\sfr}_{2}]_{\tilde{s}+1,s+1}= \frac{1}{L} (\qadj^{\tilde{s}, \sfr})^*\qadj^{s, \sfr}$ for $0 \leq \tilde{s},s \leq 1$ and, therefore, by $\mc{H}_{1}(c)$ with $\sfq =1$,  $[\Qadjmat^{\sfr}_{2}]_{\tilde{s}+1,s+1}$  concentrates on $\sigma^{\max\{\tilde{s}, s\}}_{\sfr}$.    By Lemma \ref{fact:eig_proj}, if $ \frac{1}{L}\| \qadj^{s, \sfr}_{\perp} \|^2 \geq \kappa >0$ for all $0 \leq s \leq 1$, then $\Qadjmat^{\sfr}_{2}$ is invertible.  By  $\mc{H}_{1}(f)$, $ \frac{1}{L}\norm{\qadj^{s, \sfr}_{\perp}}^2$ concentrates on $ \sigma^{s}_{\perp, \sfr} $, and $\sigma^{s}_{\perp, \sfr} > 0$ by Lemma \ref{lem:sigmatperp}. Choosing $\kappa = \frac{1}{2} \min \{2, \sigma^{0}_{\perp, \sfr}, \sigma^{1}_{\perp, \sfr}\}$, using $\mc{H}_1 (f)$, we therefore have
\ben
\begin{split}
P\Big(\Qadjmat^{\sfr}_{2} \text{ singular}\Big)   \leq  P\Big(\Big\lvert \frac{ \| \qadj^{0, \sfr}_{\perp} \|^2 }{L}- \sigma^{0}_{\perp, \sfr} \Big \lvert \geq  \kappa \Big) + P\Big(\Big\lvert \frac{\| \qadj^{1, \sfr}_{\perp} \|^2}{L} - \sigma^{1}_{\perp, \sfr} \Big \lvert \geq  \kappa \Big) &\leq K \Lc \exp\Big\{\frac{-\kappa n (\omega/\Lr)^3 \e^2}{(\log M)^2} \Big\}.
  \end{split}
\een
Now we show \eqref{eq:Hg}.  Since
\ben
(\Qadjmat^{\sfr}_{2})^{-1} = \frac{L}{\norm{\qadj^{0, \sfr}}^2\| \qadj^{1, \sfr} \|^2 - ((\qadj^{0, \sfr})^*\qadj^{1, \sfr})^2} \begin{bmatrix}
\| \qadj^{1, \sfr} \|^2& - (\qadj^{0, \sfr})^*\qadj^{1, \sfr} \\
- (\qadj^{1, \sfr})^*\qadj^{0, \sfr} & \| \qadj^{0, \sfr} \|^2
\end{bmatrix},
\een
and
\ben
(\tCmat^{2, \sfr})^{-1} = \frac{1}{\sigma^{0}_{\sfr} \sigma^{1}_{\sfr} - (\sigma^{1}_{\sfr})^2} \begin{bmatrix}
\sigma^{1}_{\sfr} & - \sigma^{1}_{\sfr} \\
- \sigma^{1}_{\sfr} &\sigma^{0}_{\sfr},
\end{bmatrix}.
\een
element-wise concentration of $(\Qadjmat^{\sfr}_{2})^{-1}$ to $(\tCmat^{2, \sfr})^{-1} $ follows from $\mc{H}_{1}(c)$ with $\sfq =1$ using Lemma \ref{sums}, Lemma \ref{products}, and Lemma \ref{inverses}.

%\RV{Dominant term in  step $\mc{H}_1$ seems to be $K \Lc \exp\Big\{\frac{-\kappa  n (\omega/\Lr)^3 }{ (\log M)^2} \e^2 \Big\}$. (from $\mc{H}_1(c)$)  }

%%%%%%%%%%%%%%%%%%%%%%%%
%%%%%%%%%%%%%%%%%%%%%%%%
%%%%%%%%%%%%%%%%%%%%%%%%
%%%%%%%%%%%%%%%%%%%%%%%%
%%%%%%%%%%%%%%%%%%%%%%%%
%%%%%%%%%%%%%%%%%%%%%%%%

\subsection{Step 3: Showing $\mathcal{B}_t$ holds} \label{subsub:step3}

We prove the statements in $\mc{B}_t$ assuming that $\mc{B}_{0}, \ldots, \mc{B}_{t-1}$, and $\mc{H}_1, \ldots, \mc{H}_t$ hold due to the induction hypothesis.  We begin with a lemma that is used to prove $\mc{B}_t (a)$.  The lemma as well as other parts of $\mc{B}_{t}$ assume the invertibility of $\innerM^{\sfc}_1, \ldots, \innerM^{\sfc}_t$ for all $\sfc \in [\Lc]$, but for the  sake of brevity, we do not explicitly specify the conditioning.
The induction hypothesis implies that for $0 \leq s \leq (t-1)$, the large deviation bound in \eqref{eq:Ba1} gives
$$ P \Big(\frac{1}{n}\sum_{\sfr \in [\Lr]} W^{\sfp}_{\sfr \sfc} \norm{\Dvec_{s,s, \sfr}}^2 \geq \epsilon \Big)  \leq    K K_{t-1}   \PC_{t-2}   \exp\{-\kappa \kappa_{t-1}   (\omega/\Lr)^{(\sfp-1)_{+}}  \pc_{t-2}  \e\}.$$
Similarly, for $0 \leq s \leq (t-1)$, we can use the values $K_{t-1}, \kappa_{t-1}$ in prefactor and exponent, respectively, of the bounds given in \eqref{eq:Bb1}-\eqref{eq:Hg}.

\begin{lem}
\label{lem:Mv_conc}
Let $\innerM^{\sfc}_t := \frac{n}{L^2}(\Madjmat^{\sfc}_{t})^* \Madjmat^{\sfc}_{t}$ and 
$\mathbf{v}^{\sfr, \sfc}= \frac{n}{L^2}(\Xmat_{t,\sfc})^*\, \qadj_{\perp, \sfc}^{t, \sfr}$. 
If $\innerM^{\sfc}_1, \ldots, \innerM^{\sfc}_t$ are invertible, we have for $j \in [t-1]$ and for all $\sfr \in [\Lr],$
\begin{align}
P\Big(\Big \lvert \sum_{\sfc \in [\Lc]} S^{j-1}_{\sfr \sfc} \sqrt{W_{\sfr \sfc}} [(\innerM^{\sfc}_t)^{-1}\mathbf{v}^{\sfr,\sfc}]_{j} - \gamma^{t,\sfr}_{j} \frac{\sigma^{j}_{\sfr}}{\phi^{j-1}_{\sfr}}  \Big \lvert \geq \e \Big) \leq  t^2 K K_{t-1}   \PC'_{t-2}   \exp\Big\{ \frac{-\kappa \kappa_{t-1}  (\omega/\Lr)^{2} \pc'_{t-2}  \e^2}{t^4(\log M)^2}\Big\},
%
%t^2 K K_{t-1}  (\Lr \Lc)^{t-1} \Lc   \exp\Big\{\frac{-\kappa \kappa_{t-1} \Mr  \omega^{2t+ 1} \Lc^{2t} \e^2}{t^4 \Lr^{4t} (\log M)^{2t} }\Big\},
\label{eq:conc_j} \\
P\Big( \Big \lvert \sum_{\sfc \in [\Lc]} S^{t-1}_{\sfr \sfc} \sqrt{W_{\sfr \sfc}} [(\innerM^{\sfc}_t)^{-1}\mathbf{v}^{\sfr,\sfc}]_{t} + \frac{\sigma^{t}_{\sfr}}{\phi^{t-1}_{\sfr}}   \Big \lvert \geq \e \Big)  
\leq t^2 K K_{t-1}   \PC'_{t-2}   \exp\Big\{ \frac{- \kappa \kappa_{t-1}  (\omega/\Lr)^{2} \pc'_{t-2}  \e^2}{t^4(\log M)^2}\Big\}.
%
%t^2 K K_{t-1}  (\Lr \Lc)^{t-1} \Lc   \exp\Big\{\frac{-\kappa \kappa_{t-1} \Mr  \omega^{2t+ 1} \Lc^{2t} \e^2}{t^4 \Lr^{4t} (\log M)^{2t} }\Big\}.
\label{eq:conc_t}
\end{align}
\end{lem}
%\RV{to be filled in later}

%-----------------
\begin{proof}
First note that $(\mathbf{\Madjmat}^{\sfc}_t)^{-1}$ concentrates element-wise to 
$(\brCmat^{t, \sfc})^{-1}$ by $\mathcal{B}_{t-1}(g)$.  So before proving  \eqref{eq:conc_j} and  \eqref{eq:conc_t}, we state two results providing concentration and boundedness guarantees for the elements of $\mathbf{v}^{\sfr,\sfc}$. For for $ i \in [t]$, define
\be
\mathbb{E}_{i, \sfc}  =  -\frac{ 1}{\Lc}  \hat{\gamma}^{t, \sfr}_{t-1}   \Psi^{t-1}_{\sfc} \sqrt{W_{\sfr \sfc}}\Big(1 - \frac{ \tau^{\max\{i-1, \, t-2\}}_{\sfc} }{\tau^{t-2}_{\sfc} }  \Big) .
\label{eq:Ek_def}
\ee
Notice $\mathbb{E}_{1, \sfc} = \ldots = \mathbb{E}_{t-1, \sfc} = 0$.  Then, we will prove that for $i \in [t]$ and $\textsf{B} > 0$, a universal constant,
\begin{align}
P\Big(\Big \lvert  \sum_{\sfc \in [\Lc] } \sqrt{W_{\sfr \sfc}} \Big( \frac{L}{n} \tv^{\sfr,\sfc}_i -  \mathbb{E}_{i, \sfc} \Big) \Big \lvert \geq \e \Big) &\leq  t K K_{t-1}   \PC'_{t-2}   \exp\Big\{\frac{-\kappa \kappa_{t-1}  (\omega/\Lr)^{2} \pc'_{t-2}  \e^2}{t^2(\log M)^2}\Big\},
%t K K_{t-1}  (\Lr \Lc)^{t-1} \Lc   \exp\Big\{\frac{-\kappa \kappa_{t-1} \Mr  \omega^{2t+ 1} \Lc^{2t} \e^2}{t^2 \Lr^{4t} (\log M)^{2t} }\Big\}. 
\label{eq:T2a} \\
P\Big(  \sum_{\sfc \in [\Lc]}   \sqrt{W_{\sfr \sfc}} \Big \lvert  \frac{L}{n}\tv^{\sfr,\sfc}_i \Big \lvert \geq  \frac{\textsf{B}\Lr}{\omega} \Big) &\leq   
t K K_{t-1}   \PC'_{t-2}   \exp\Big\{\frac{-\kappa \kappa_{t-1}  (\omega/\Lr)^{2} \pc'_{t-2} }{t^2(\log M)^2}\Big\}.
%t K K_{t-1}  (\Lr \Lc)^{t-1} \Lc   \exp\Big\{\frac{-\kappa \kappa_{t-1} \Mr  \omega^{2t +1 }  \Lc^{2t-2} \e^2}{t^2 \Lr^{4t -2} (\log M)^{2t} }\Big\}. 
\label{eq:T2b}
\end{align}
%
%where $\textsf{B} > 0$ is a universal constant. % \RV{fill in the above later}
%
We prove the main result using \eqref{eq:T2a} and \eqref{eq:T2b}, and then prove  \eqref{eq:T2a} and \eqref{eq:T2b}.  We first claim:
\begin{align}
&  \frac{-n}{L} \sum_{\sfc \in [\Lc]} \sum_{k=1}^t  S^{t-1}_{\sfr \sfc}  \sqrt{W_{\sfr \sfc}} \, [(\brCmat^{t, \sfc})^{-1}]_{t, k} \, \mathbb{E}_{k, \sfc} = \frac{\sigma^{t}_{\sfr}}{\phi^{t-1}_{\sfr}}, \label{eq:bt_r_rep}  \\
 & \frac{n}{L}  \sum_{\sfc \in [\Lc]} \sum_{k=1}^{t} S^{j-1}_{\sfr \sfc} \, \sqrt{W_{\sfr \sfc}}      [(\brCmat^{t, \sfc})^{-1}]_{j k} \, \mathbb{E}_{k, \sfc} =  \frac{\hat{\gamma}^{t,\sfr}_{j} \sigma^{j}_{\sfr}}{\phi^{j-1}_{\sfr}}, \qquad 1 \leq j \leq (t-1).
\label{eq:concj_equality}
\end{align}
The result \eqref{eq:bt_r_rep} is obtained as follows using \eqref{eq:Ek_def}.
\begin{align}
& \frac{-n}{L} \sum_{\sfc \in [\Lc]} \sum_{k=1}^t  S^{t-1}_{\sfr \sfc}  \sqrt{W_{\sfr \sfc}} \, [(\brCmat^{t, \sfc})^{-1}]_{t k} \, \mathbb{E}_{k, \sfc} =  \frac{-n}{L}\sum_{\sfc \in [\Lc]} S^{t-1}_{\sfr \sfc}  \sqrt{W_{\sfr \sfc}} \, [(\brCmat^{t, \sfc})^{-1}]_{t t} \, \mathbb{E}_{t, \sfc} 
\label{eq:bt_r_rep_proof} \\
&= \frac{n \hat{\gamma}^{t, \sfr}_{t-1} }{L\Lc}  \sum_{\sfc \in [\Lc]} S^{t-1}_{\sfr \sfc}  W_{\sfr \sfc} \, [(\brCmat^{t, \sfc})^{-1}]_{t t}  \Psi^{t-1}_{\sfc} \Big( 1-  \frac{ \tau^{t-1}_{\sfc} }{\tau^{t-2}_{\sfc} } \Big) \overset{(a)}{=} \frac{\hat{\gamma}^{t, \sfr}_{t-1}  }{\Lc \phi^{t-1}_{\sfr}} \sum_{\sfc \in [\Lc]} W_{\sfr \sfc}  \Psi^{t-1}_{\sfc} =  \frac{\hat{\gamma}^{t, \sfr}_{t-1}  \sigma^{t-1}_{\sfr} }{\phi^{t-1}_{\sfr}} =  \frac{\sigma^{t}_{\sfr}}{\phi^{t-1}_{\sfr}}, \nonumber
\end{align}
where step $(a)$ uses  $S^{t-1}_{\sfr \sfc} = \tau_\sfc^{t-1}/\phi^{t-1}_\sfr$ and  $ [(\brCmat^{t, \sfc})^{-1}]_{t t} = \frac{L}{n}(\tau^{t-1}_{\perp, \sfc})^{-1}$ which can be seen as follows. From the definition of $\brCmat^{t, \sfc}$ in \eqref{eq:Ct_def},  if $\brCmat^{t-1, \sfc}$ is invertible,  using the block inversion formula, 
\be
 (\brCmat^{t, \sfc})^{-1} = \Big( \begin{array}{cc}
(\brCmat^{t-1, \sfc})^{-1} + \frac{L}{n} (\tau^{t-1}_{\perp, \sfc})^{-1} \hat{\alpha}^{t-1,\sfc}   (\hat{\alpha}^{t-1, \sfc})^*  & - \frac{L}{n}(\tau^{t-1}_{\perp, \sfc})^{-1} \hat{\alpha}^{t-1,\sfc}  \\
-\frac{L}{n} (\tau^{t-1}_{\perp, \sfc})^{-1}   (\hat{\alpha}^{t-1, \sfc})^* &  \frac{L}{n}(\tau^{t-1}_{\perp, \sfc})^{-1}  \end{array} \Big),
\label{eq:Cinverse_mat}
\ee
where we have used $\hat{\alpha}^{t-1,\sfc} :=  \frac{n}{L}\tau^{t-1}_{\sfc} (\brCmat^{t-1,\sfc})^{-1}  (1, \ldots, 1)^*$ and $\tau^{t-1}_{\sfc} -  \tau^{t-1}_{\sfc}(1, \ldots, 1) \hat{\alpha}^{t-1,\sfc} =\tau^{t-1}_{\sfc} - (\tau^{t-1}_{\sfc})^2/\tau^{t-2}_{\sfc} = \tau^{t-1}_{\perp, \sfc}$. Result \eqref{eq:concj_equality} is obtained using steps similar to \eqref{eq:bt_r_rep_proof} to show that the LHS of \eqref{eq:concj_equality} equals 
\[ 
\frac{\hat{\gamma}^{t, \sfr}_{t-1}}{ \phi^{j-1}_{\sfr} \Lc} \sum_{\sfc \in [\Lc]}   W_{\sfr \sfc}  \,  \Psi^{t-1}_{\sfc} \hat{\alpha}^{t-1, \sfc}_{j-1}  \Big[\frac{\tau^{j-1}_{\sfc}}{\tau^{t-1}_{\sfc}}\Big]  =  \frac{  \hat{\gamma}^{t,\sfr}_{j} \sigma^{j}_{\sfr}}{\phi^{j-1}_{\sfr}},
\]
where the last equality follows from the fact that for $j \in [t-2]$ we have $\hat{\gamma}^{t,\sfr}_{j} =  \hat{\alpha}^{t-1, \sfc}_{j-1} = 0$, and if $j=t-1$ then since $\hat{\alpha}^{t-1, \sfc}_{t-1} = 
\tau_\sfc^{t-1}/\tau_\sfc^{t-2}$.

%%%%%
We now prove \eqref{eq:conc_t}.  Using the result in \eqref{eq:bt_r_rep}, the LHS of \eqref{eq:conc_t}  can be expressed as follows:
\be
\begin{split}
&P\Big(\Big \lvert  \sum_{\sfc \in [\Lc]} S^{t-1}_{\sfr \sfc}  \sqrt{W_{\sfr \sfc}} \, \sum_{k=1}^t \Big[ [(\mathbf{\Madjmat}^{\sfc}_t)^{-1}]_{t k} \, \tv^{\sfr,\sfc}_{k} - \frac{n}{L}[(\brCmat^{t \sfc})^{-1}]_{t k} \, \mathbb{E}_{k, \sfc} \Big] \Big \lvert \geq \e \Big) \\
&\overset{(a)}{\leq}   \sum_{k=1}^t P\Big(\frac{n}{L}  \Big \lvert  \sum_{\sfc \in [\Lc]} \tau^{t-1}_{\sfc} \sqrt{W_{\sfr \sfc}} \Big[[(\mathbf{\Madjmat}^{\sfc}_t)^{-1}]_{t k} \, \Big(\frac{L}{n}\tv^{\sfr,\sfc}_{k}\Big) - [(\brCmat^{t, \sfc})^{-1}]_{t k} \, \mathbb{E}_{k, \sfc}\Big]  \Big \lvert \geq \frac{\kappa \e}{t} \Big) \\
& \overset{(b)}{\leq}  \sum_{k=1}^t P\Big(  \Big \lvert  \sum_{\sfc \in [\Lc]}   [(\brCmat^{t, \sfc})^{-1}]_{t k} \, \Big( \frac{n}{L}\tau^{t-1}_{\sfc}\Big)  \sqrt{W_{\sfr \sfc}} \Big(\frac{L}{n} \tv^{\sfr,\sfc}_{k}-  \mathbb{E}_{k, \sfc}\Big)   \Big \lvert \geq \frac{\kappa \e}{2t} \Big) \\
& \qquad + \sum_{k=1}^t P\Big(\frac{n}{L}  \Big \lvert  \sum_{\sfc \in [\Lc]} \tau^{t-1}_{\sfc} \sqrt{W_{\sfr \sfc}}\Big[[(\mathbf{\Madjmat}^{\sfc}_t)^{-1}]_{t k} - [(\brCmat^{t, \sfc})^{-1}]_{t k}\Big] \, \Big(\frac{L}{n}\tv^{\sfr,\sfc}_{k}\Big)  \Big \lvert \geq \frac{\kappa \e}{2t} \Big).
\label{eq:toprove_num22}
\end{split}
\ee
In the above, step $(a)$ follows by  Lemma \ref{sums} and $S^{t-1}_{\sfr \sfc} = \tau^{t-1}_{\sfc} /\phi^{t-1}_{\sfr}$ with $\phi^{t-1}_{\sfr} \in \Theta(1)$.  Step $(b)$ follows by Lemma \ref{sums} again.  Label the terms on the RHS of \eqref{eq:toprove_num22} as $T_1$ and $T_2$.  Note that term $T_1$ has the desired upper bound using \eqref{eq:T2a} and that the non-zero elements of $ [(\brCmat^{t, \sfc})^{-1}]_{t k} \, ( \frac{n}{L}\tau^{t-1}_{\sfc})$ are  $\Theta(1)$.   Now consider term $T_2$.   Using that  $\frac{n}{L}  \tau^{t-1}_{\sfc}$ is $\Theta(1)$ for all $\sfc \in [\Lc]$, 
\be
\begin{split}
T_2 &\leq  \sum_{k=1}^t P\Big( \sum_{\sfc \in [\Lc]} \sqrt{W_{\sfr \sfc}} \Big \lvert [(\mathbf{\Madjmat}^{\sfc}_t)^{-1}]_{t k} - [(\brCmat^{t, \sfc})^{-1}]_{t k}\Big \lvert   \Big\lvert \frac{L}{n}\tv^{\sfr,\sfc}_{k}\Big\lvert  \geq \frac{\kappa \e}{2t} \Big) \\
&\leq  \sum_{k=1}^t \Big[P\Big( \sum_{\sfc \in [\Lc]}  \hspace{-5pt} \sqrt{W_{\sfr \sfc}} \Big\lvert \frac{L}{n} \tv^{\sfr,\sfc}_{k}\Big\lvert  \geq \frac{\textsf{B}\Lr}{\omega} \Big) +  P\Big( \cup_{\sfc \in [\Lc]} \Big\{  \Big \lvert [(\mathbf{\Madjmat}^{\sfc}_t)^{-1}]_{t k} - [(\brCmat^{t, \sfc})^{-1}]_{t k}\Big \lvert \geq \frac{\kappa \e \omega}{2t \textsf{B} \Lr} \Big\} \Big)  \Big] \\
&\leq  \sum_{k=1}^t  \Big[\sum_{\sfc \in[\Lc]} \hspace{-3pt} P\Big(  \Big \lvert [(\mathbf{\Madjmat}^{\sfc}_t)^{-1}]_{t k} - [(\brCmat^{t, \sfc})^{-1}]_{t k}\Big \lvert \geq \frac{\kappa \e\omega}{2t \textsf{B} \Lr} \Big) + P\Big( \sum_{\sfc \in [\Lc]}   \sqrt{W_{\sfr \sfc}} \Big\lvert \frac{L}{n}\tv^{\sfr,\sfc}_{k}\Big\lvert  \geq \frac{\textsf{B} \Lr}{\omega} \Big) \Big] \\
&\leq t \Lc K K_{t-1}   \PC_{t-2}  \exp\Big\{  \frac{-\kappa \kappa_{t-1} (\omega/\Lr)^2  \pc_{t-2}  \e^2 }{t^2 \textsf{B}^2 } \Big\}+ t^2 K K_{t-1}   \PC'_{t-2}   \exp\Big\{ \frac{-\kappa \kappa_{t-1}  (\omega/\Lr)^{2} \pc'_{t-2}  \e^2}{t^2(\log M)^2}\Big\} .
\label{eq:Mv_termT2_bound}
\end{split}
\ee
The last inequality  follows by induction hypothesis $\mathcal{B}_{t-1} (g)$ and \eqref{eq:T2b}. The final result follows since $\Lc  \PC_{t-2}  = \PC'_{t-2}$ and $  \pc_{t-2} =   \pc'_{t-2} $.

We next consider the bound \eqref{eq:conc_j}. First, using Lemma \ref{sums},
\be
\begin{split}
&P\Big(\Big \lvert \sum_{\sfc \in [\Lc]} S^{j-1}_{\sfr \sfc} \, \sqrt{W_{\sfr \sfc}} \, [(\mathbf{\Madjmat}^{\sfc}_t)^{-1} \, \mathbf{v}^{\sfr,\sfc}]_{j} - \gamma^{t,\sfr}_{j} \frac{\sigma^{j}_{\sfr}}{\phi^{j-1}_{\sfr}}\Big \lvert \geq \e \Big) \\
&\leq P\Big( \lvert  \gamma^{t,\sfr}_{j}  - \hat{\gamma}^{t,\sfr}_{j}    \lvert  
\frac{\sigma^{j}_{\sfr}}{\phi^{j-1}_{\sfr}} \geq \frac{\kappa \e}{2} \Big)  +P\Big( \Big \lvert \sum_{\sfc \in [\Lc]} S^{j-1}_{\sfr \sfc} \, \sqrt{W_{\sfr \sfc}} \, \sum_{k=1}^{t}   [(\mathbf{\Madjmat}^{\sfc}_t)^{-1}]_{j k} \, \tv^{\sfr,\sfc}_{k}  -  \frac{ \hat{\gamma}^{t,\sfr}_{j} \sigma^{j}_{\sfr}}{\phi^{j-1}_{\sfr}}    \Big \lvert \geq \frac{\kappa \e}{2} \Big).
\label{eq:concj_toprove2}
\end{split}
\ee
The first term on the RHS of \eqref{eq:concj_toprove2} is upper bounded by $K K_{t-1}   \PC_{t-1}'   \exp\{{-\kappa \kappa_{t-1} (\omega/\Lr)^{2} \pc'_{t-1}  \e^2}/{(\log M)^{2} }\}$ using  $\mathcal{H}_t (f)$.  Using result \eqref{eq:concj_equality}, the second term on the RHS of \eqref{eq:concj_toprove2} can be written as 
\ben
\begin{split}
%
%&P\Big( \Big \lvert \sum_{\sfc \in [\Lc]} S^{j-1}_{\sfr \sfc} \, \sqrt{W_{\sfr \sfc}} \, \sum_{k=1}^{t}   [(\mathbf{\Madjmat}^{\sfc}_t)^{-1}]_{j k} \, \tv^{\sfr,\sfc}_{k}  - \frac{ \hat{\gamma}^{t,\sfr}_{j} \sigma^{j}_{\sfr}}{\phi^{j-1}_{\sfr}}    \Big \lvert \geq \frac{\kappa \e}{2} \Big) \\
%
P\Big( \Big \lvert  \sum_{\sfc \in [\Lc]} S^{j-1}_{\sfr \sfc} \, \sqrt{W_{\sfr \sfc}} \sum_{k=1}^{t}  \Big[ [(\mathbf{\Madjmat}^{\sfc}_t)^{-1}]_{j k} \, \tv^{\sfr,\sfc}_{k}  - \frac{n}{L}    [(\brCmat^{t, \sfc})^{-1}]_{j k} \, \mathbb{E}_{k, \sfc}\Big]   \Big \lvert \geq \frac{\kappa \e}{2} \Big),
%
%\label{eq:toprove_num3}
\end{split}
\een
which can be  bounded using steps similar to that in \eqref{eq:toprove_num22}.

To complete the proof, it remains to prove the bounds in \eqref{eq:T2a} and \eqref{eq:T2b}.   Note that for $ i \in [t]$,
\be
 \frac{L}{n}    \sqrt{W_{\sfr \sfc}}  \tv^{\sfr,\sfc}_{i} = \frac{   \sqrt{W_{\sfr \sfc}}}{L}(\hvec^{i}_{\sfc} + \brqvec^{i-1}_{\sfc})^* \, \qadj_{\perp, \sfc}^{t, \sfr}  = \frac{1}{L}    \Big(\frac{\sqrt{W_{\sfr \sfc}}}{L} \hvec^{i }_{\sfc} - \qadj^{i-1, \sfr}_{\sfc} \Big)^* \, (\qadj_{\sfc}^{t, \sfr} -\sum_{j=0}^{ t-1} \gamma^{t,\sfr}_{j} \qadj^{j,\sfr}),
\label{eq:vk_def}
\ee
where we have used the fact that $\qadj_{\perp, \sfc}^{t, \sfr} = \qadj_{\sfc}^{t, \sfr} - \qadj_{\parallel, \sfc}^{t, \sfr} = \qadj_{\sfc}^{t, \sfr} - \sum_{j=0}^{ t-1} \gamma^{t,\sfr}_{j} \qadj^{j,\sfr}$.  We first prove result \eqref{eq:T2a}. Using \eqref{eq:Ek_def} and  \eqref{eq:vk_def}, we can bound for the probability in \eqref{eq:T2a} as:
\be
\begin{split}
&P\Big(\Big \lvert  \sum_{\sfc \in [\Lc]}    \sqrt{W_{\sfr \sfc}} \Big( \frac{L}{n} \tv^{\sfr,\sfc}_i -  \mathbb{E}_{i, \sfc} \Big) \geq \e \Big) \\ 
&{\leq}  P\Big ( \Big \lvert \sum_{\sfc \in [\Lc]}   \Big[  \frac{\sqrt{W_{\sfr \sfc}} }{L }  (\hvec^{i }_{\sfc})^*\,  (\qadj_{\sfc}^{t, \sfr} - \sum_{j=0}^{ t-1} \gamma^{t,\sfr}_{j} \qadj^{j,\sfr}) + \frac{W_{\sfr \sfc}}{\Lc}  \Big( \Psi^{t}_{\sfc} -  \hat{\gamma}^{t, \sfr}_{t-1}   \Psi^{t-1}_{\sfc} \frac{ \tau^{\max\{i-1, t-2\}}_{\sfc} }{\tau^{t-2}_{\sfc} }\Big) \Big]  \Big \lvert \geq \frac{\e}{2} \Big ) \\
& \quad + P\Big (  \Big  \lvert \sum_{\sfc \in [\Lc]}   \Big[ \frac{ 1 }{L } (\qadj^{i-1, \sfr}_{\sfc})^*\,  (\qadj_{\sfc}^{t, \sfr} - \sum_{j=0}^{ t-1} \gamma^{t,\sfr}_{j} \qadj^{j,\sfr}) - \frac{W_{\sfr \sfc}}{\Lc} \Big( \Psi^{t}_{\sfc} - \hat{\gamma}^{t, \sfr}_{t-1} \Psi^{t-1}_{\sfc}\Big) \Big] \Big \lvert \geq \frac{\e}{2} \Big ).
\label{eq:v_conc_eq1}
\end{split}
\ee
%
%In the above, step $(a)$ follows by Lemma \ref{sums}. 
%
Label the terms on the RHS of \eqref{eq:v_conc_eq1} as $T_1$ and $T_2$, and we bound both.   First, for any $i \in [t]$,
\begin{align}
T_1 &\overset{(a)}{\leq}  P\Big( \Big \lvert   \sum_{\sfc \in [\Lc]}      \sqrt{W_{\sfr \sfc}} \Big[ \frac{  (\hvec^{i}_{\sfc})^* \qadj^{t, \sfr}_{\sfc}}{L}  +\frac{\sqrt{W_{\sfr \sfc}}}{\Lc}  \Psi^{t}_{\sfc} \Big] \Big  \lvert \geq \frac{\e}{6} \Big ) + \sum_{j=0}^{t-2} P\Big (\Big  \lvert  \gamma^{t,\sfr}_j   \sum_{\sfc \in [\Lc]}   \frac{    \sqrt{W_{\sfr \sfc}} (\hvec^{i}_{\sfc})^* \qadj^{j, \sfr}_{\sfc}} {L}  \Big  \lvert \geq  \frac{ \e}{6(t-1)}\Big )  \nonumber \\
&+  P\Big ( \Big  \lvert  \sum_{\sfc \in [\Lc]}    \sqrt{W_{\sfr \sfc}} \Big[  \gamma^{t,\sfr}_{t-1}    \frac{ (\hvec^{i}_{\sfc})^*  \qadj_{\sfc}^{t-1, \sfr}}{L}  + \hat{\gamma}^{t, \sfr}_{t-1} \frac{ \sqrt{W_{\sfr \sfc}} \Psi^{t-1}_{\sfc}  \tau^{\max\{i-1, t-2\}}_{\sfc} }{\Lc  \tau^{t-2}_{\sfc} } \Big] \Big  \lvert \geq  \frac{\e}{6} \Big )  \label{eq:v_conc_eq1_bound1} \\
& \overset{(b)}{ \leq} %t ??? t K K_{t-1}  (\Lr \Lc)^{t-1} \Lc   \exp\Big\{\frac{-\kappa \kappa_{t-1} \Mr  \omega^{2t+ 1} \Lc^{2t} \e^2}{t^2 \Lr^{4t} (\log M)^{2t} }\Big\}  
t K K_{t-1}   \PC'_{t-2}   \exp\Big\{\frac{-\kappa \kappa_{t-1}  (\omega/\Lr)^{2} \pc'_{t-2}  \e^2}{t^2(\log M)^2}\Big\}.
\nonumber 
\end{align}
%
%\RV{fill in!}
Step $(a)$ follows from  Lemma \ref{sums}, and step $(b)$ uses $\mathcal{H}_t (d)$ and $\mathcal{H}_t (f)$: the bound for the first term follows directly from $\mathcal{H}_t (d)$, while the bound for the other two terms uses  $\mathcal{H}_t (d)$  and  $\mathcal{H}_t (f)$, along with Lemma \ref{products} and Lemma \ref{products_0}.  Noting that $\hat{\gamma}^{t, \sfr}_{t-1}$, $({ \tau^{\max\{i-1, t-2\}}_{\sfc} }/{ \tau^{t-2}_{\sfc} })$, and $ \Psi^{t-1}_{\sfc}$ are all $\Theta(1)$ terms,  we observe  from \eqref{eq:Wrc_avgs} that  
$\frac{1}{\Lc} \sum_{\sfc}    W_{\sfr \sfc} \Psi^{t-1}_{\sfc} ({ \tau^{\max\{i-1, t-2\}}_{\sfc} }/\tau^{t-2}_{\sfc})$ is also $\Theta(1)$.

The second term  $T_2$ on the RHS of \eqref{eq:v_conc_eq1} can be bounded similarly using Lemma~\ref{sums}, $\mathcal{H}_t (c)$ with $\sfv = 1$,$\mathcal{H}_t (f)$, Lemma \ref{products}, and Lemma \ref{products_0}.
%
%bounded as follows.  Using Lemma~\ref{sums},
%%
%\be
%\begin{split}
%%
%T_2 & \leq P\Big(  \Big \lvert \sum_{\sfc \in [\Lc]}       \Big[ \frac{ (\qadj^{i-1, \sfr}_{\sfc})^* \qadj^{t, \sfr}_{\sfc}}{L}  -  \frac{W_{\sfr \sfc}}{\Lc}   \Psi^{t}_{\sfc}  \Big]\Big  \lvert \geq  \frac{\e}{6}  \Big ) + \sum_{j=0}^{t-2} P\Big( \Big \lvert   \gamma^{t,\sfr}_j  \sum_{\sfc \in [\Lc]}   \frac{  (\qadj^{i-1, \sfr}_{\sfc})^*\, \qadj^{j, \sfr}_{\sfc}}{L} \Big  \lvert \geq  \frac{\e}{6(t-1)} \Big ) \\
%%
%&+ P\Big(   \Big \lvert  \sum_{\sfc \in [\Lc]}    \Big[ \gamma^{t,\sfr}_{t-1}   \frac{(\qadj^{i-1, \sfr}_{\sfc})^*\, \qadj^{t-1, \sfr}_{\sfc}}{L} -    \hat{\gamma}^{t, \sfr}_{t-1} \frac{ W_{\sfr \sfc}}{\Lc}   \Psi^{t-1}_{\sfc}  \Big] \Big  \lvert \geq  \frac{\e}{6} \Big ) \\
%%
%&\overset{(a)}{\leq} t K K_{t-1}   \PC'_{t-2}   \exp\Big\{\frac{-\kappa \kappa_{t-1}  (\omega/\Lr)^{2} \pc'_{t-2}  \e^2}{t^2(\log M)^2}\Big\},
%%K K_{t-1}  (\Lr \Lc)^{t-1} \Lc   \exp\Big\{\frac{-\kappa \kappa_{t-1} \Mr  \omega^{2t+ 1} \Lc^{2t} \e^2}{t^2 \Lr^{4t} (\log M)^{2t} }\Big\},
%\label{eq:v_conc_eq1_bound2}
%%
%\end{split}
%\ee
%%
%where $(a)$ uses $\mathcal{H}_t (c)$ with $\sfv = 1$ and $\mathcal{H}_t (f)$, along with Lemma \ref{products} and Lemma \ref{products_0}. %\RV{fill in!}

Next we prove result \eqref{eq:T2b}.  Using \eqref{eq:vk_def} we have for $k \in [t]$,
\begin{align}
&P\Big( \sum_{\sfc \in [\Lc]}  \sqrt{W_{\sfr \sfc}} \Big \lvert  \frac{L}{n} \tv^{\sfr,\sfc}_k \Big \lvert \geq \frac{\textsf{B}\Lr}{\omega} \Big) =  P\Big( \sum_{\sfc \in [\Lc]}  \frac{1}{L} \Big \lvert  \Big( \sqrt{W_{\sfr \sfc}}\hvec^k_{\sfc} - \qadj^{k-1, \sfr}_{\sfc} \Big)^* \, (\qadj_{\sfc}^{t, \sfr} - \sum_{j=0}^{t-1} \gamma^{t,\sfr}_j \qadj^{j, \sfr}_{\sfc}) \Big \lvert \geq \frac{\textsf{B}\Lr}{\omega} \Big) \nonumber \\
& \leq  P\Big (\frac{ 1  }{L } \sum_{\sfc \in [\Lc]} \Big \lvert    \sqrt{W_{\sfr \sfc}} (\hvec^k_{\sfc})^*\, \qadj^{t, \sfr}_{\sfc} \Big \lvert \geq \frac{\textsf{B}\Lr}{4\omega} \Big ) +   P\Big (\frac{ 1  }{L }  \sum_{j=0}^{t-1} \abs{  \gamma^{t,\sfr}_j } \sum_{\sfc \in [\Lc]}   \Big \lvert \sqrt{W_{\sfr \sfc}} (\hvec^k_{\sfc})^*\, \qadj^{j, \sfr}_{\sfc}   \Big \lvert \geq  \frac{\textsf{B}\Lr}{4\omega} \Big ) \nonumber  \\
& \quad + P\Big (\frac{ 1  }{L } \sum_{\sfc \in [\Lc]}  \Big  \lvert   (\qadj^{k-1, \sfr}_{\sfc})^*\, \qadj^{t, \sfr}_{\sfc}  \Big \lvert \geq  \frac{\textsf{B}\Lr}{4\omega}\Big ) + P\Big ( \frac{ 1  }{L } \sum_{j=0}^{t-1}  \abs{\gamma^{t,\sfr}_j}  \sum_{\sfc \in [\Lc]}  \Big  \lvert  (\qadj^{k-1, \sfr}_{\sfc})^*\, \qadj^{j, \sfr}_{\sfc} \Big \lvert \geq  \frac{\textsf{B}\Lr}{4\omega} \Big ).
\label{eq:v_conc_eqA}
\end{align}
We now provide upper bounds for each term of \eqref{eq:v_conc_eqA}, labeled as $A_1 - A_4$.  We note that these results don't follow directly from the induction hypothesis as they involve the  sums  of the \emph{absolute values} of the inner products over the column blocks $\sfc$.  
%What we show in the concentration results of the induction are for sums $\sfc \in [\Lc]$ of the segments of the inner product, but \emph{without} the absolute value.      
%\CR{Will you check the work that follows... made some changes from what you had written.} \RV{I have checked, made a small change in \eqref{eq:BR4w_bnd} to make the argument clearer.}
First notice that by \eqref{eq:Wrc_avgs}, for any $0 \leq s \leq t$,
\be
\begin{split}
&   \frac{\norm{\qadj^{s, \sfr}}^2}{L} = \sum_{\sfc \in [\Lc]} \frac{W_{\sfr \sfc} \norm{\brqvec^{s}_{\sfc}}^2}{L} \leq   \sum_{\sfc \in [\Lc]} \frac{4 W_{\sfr \sfc}}{\Lc}  \leq 4 \kappa_1, \qquad  \sum_{i \in \text{sec}(\ell)} \abs{[\qadj^{s, \sfr}_{\sfc}]_i} =  \sqrt{W_{\sfr \sfc}} \sum_{i \in \text{sec}(\ell)} \abs{[\brqvec^{s}_{\sfc}]_i} \leq 2\sqrt{W_{\sfr \sfc}}.
\label{eq:qbound_facts}
\end{split}
\ee
%
%\RV{note that the first bound above is slightly improved from new-scaling.pdf}
For reasons that will become clear in a few steps, we take 
\be
\textsf{B} = 4\Big(1+ (\hat{\gamma}^{t,\sfr}_{t-1}+1) \max\Big\{ 4 \kappa_1,  \, \sqrt{24 \tau^0_{\sfc} \log M +1}\Big\}\Big).
\label{eq:sfB_def}
\ee 
(Note that $\tau^0_{\sfc} \log M \in \Theta(1)$). First, we claim that  third term, $A_3$, in \eqref{eq:v_conc_eqA} equals $0$. Indeed, using the Cauchy-Schwarz inequality and \eqref{eq:qbound_facts}, 
\begin{align}
 \frac{1 }{L }   \sum_{\sfc \in [\Lc]}   \lvert(\qadj^{k-1, \sfr}_{\sfc})^*\, \qadj^{j, \sfr}_{\sfc} \lvert \leq  \frac{1 }{L }   \sum_{\sfc \in [\Lc]}   \| \qadj^{k-1, \sfr}_{\sfc} \|  
 \| \qadj^{j, \sfr}_{\sfc} \| \leq   \frac{1 }{L }  \| \qadj^{k-1, \sfr} \|  \| \qadj^{j, \sfr} \| \leq 4\kappa_1.
\label{eq:q_inner_prod_bnd}
\end{align}
Since $\frac{\textsf{B}\Lr}{4\omega} > \frac{4 \kappa_1\Lr}{\omega} \geq 4 \kappa_1$, it follows that $A_3=0$. To bound $A_4$, from \eqref{eq:sfB_def} we  note that  $\textsf{B} \geq  4+ 16 \kappa_1 |\hat{\gamma}^{t,\sfr}_{t-1}| =  4+ 16  \kappa_1  \sum_{j=0}^{t-1} \hat{\gamma}^{t,\sfr}_{j}$ since $\hat{\gamma}^{t,\sfr}_{0} = \ldots = \hat{\gamma}^{t,\sfr}_{t-2} = 0$. Then, using \eqref{eq:q_inner_prod_bnd}, we have %
\begin{align*}
A_4 \leq  P\Big ( \sum_{j=0}^{t-1} \abs{\gamma^{t,\sfr}_j} &\geq  \frac{\textsf{B} \Lr}{16\kappa_1 \omega} \Big ) \leq   P\Big ( \sum_{j=0}^{t-1} \abs{\gamma^{t,\sfr}_j} \geq \sum_{j=0}^{t-1}\abs{\hat{\gamma}^{t,\sfr}_{j }} + \frac{1}{4\kappa_1} \Big ) \\
&\leq  \sum_{j=0}^{t-1} P\Big ( \lvert \gamma^{t,\sfr}_j - \hat{\gamma}^{t,\sfr}_{j } \lvert \geq  \frac{1}{t(4\kappa_1)} \Big ) \leq t K K_{t-1}   \PC'_{t-2}   \exp\Big\{\frac{-\kappa \kappa_{t-1}  (\omega/\Lr)^{2} \pc'_{t-2}  \e^2}{t^2(\log M)^2}\Big\}.
%K K_{t-1}  (\Lr \Lc)^{t-1} \Lc   \exp\Big\{\frac{-\kappa \kappa_{t-1} \Mr  \omega^{2t +1 }  \Lc^{2t-2} \e^2}{t^2 \Lr^{4t -2} (\log M)^{2t} }\Big\}.
%
\end{align*}
The last inequality uses inductive hypothesis $\mathcal{H}_t(f)$.  Now consider $A_1$. Using \eqref{eq:qbound_facts}, for  $0 \leq j,k \leq t$,
\be
\begin{split}
\sum_{\sfc \in [\Lc]} \Big \lvert   \frac{\sqrt{W_{\sfr \sfc}}}{L }  (\hvec^k_{\sfc})^*\, \qadj^{j, \sfr}_{\sfc} \Big \lvert  
&\leq \sum_{\sfc \in [\Lc]}  \sum_{\ell \in \sfc}  \frac{\sqrt{W_{\sfr \sfc}} }{L } \max_{i \in \sec(\ell)} \abs{[\hvec^k_{\sfc}]_{i}} \sum_{i \in \sec(\ell)} \abs{ [\qadj^{j, \sfr}_{\sfc}]_{i}}   \leq \sum_{\sfc \in [\Lc]}   \sum_{\ell \in \sfc}  \frac{2W_{\sfr \sfc} }{L } \max_{i \in \sec(\ell)} \abs{[\hvec^k_{\sfc}]_{i}}. 
\label{eq:abs_h_upper}
\end{split}
\ee
From \eqref{eq:sfB_def}, we have 
\be
 \frac{\textsf{B}\Lr}{4 \omega } \geq  \frac{\Lr}{\omega} (\hat{\gamma}^{t,\sfr}_{t-1}+1)  \max\Big\{ 4 \kappa_1,  \, \sqrt{24 \tau^0_{\sfc} \log M +1}\Big\}\Big) > (\hat{\gamma}^{t,\sfr}_{t-1}+1) \sqrt{24 (\max_{\sfc} W^2_{\sfr \sfc}) \tau^0_{\sfc} \log M + 1},
\label{eq:BR4w_bnd}
\ee
where the last inequality holds when the constant $\kappa_1$ is chosen to be large enough since $\max_{\sfc} W_{\sfr \sfc} = \Theta(\Lr/\omega)$. Therefore, using Lemma \ref{lem:squaredsums} and $\mathcal{H}_t (e)$ with $\sfq =1$ since $k \in [t]$, we obtain 
%
%\RV{I don't think you can have an unspecified constant $\kappa$ in the above bound for 
%${\textsf{B}\Lr}/{(4 \omega)}$ as $\mathcal{H}_t (e)$ requires $6 (\max_{\sfc} W^2_{\sfr \sfc}) \tau^0_{\sfc} \log M + 1$ on the RHS. Can be fixed by adjusting $\kappa_1$.}
\begin{align*}
&A_1 \leq  P\Big (\sum_{\sfc \in [\Lc]}   \sum_{\ell \in \sfc}  \frac{2W_{\sfr \sfc} }{L } \max_{i \in \sec(\ell)} \abs{[\hvec^k_{\sfc}]_{i}}  \geq \frac{\textsf{B}\Lr}{4 \omega} \Big ) \\
&\leq   P\Big ( \frac{1}{L} \sum_{\sfc \in [\Lc]}   \sum_{\ell \in \sfc}  4W^2_{\sfr \sfc} \max_{i \in \sec(\ell)} \abs{[\hvec^k_{\sfc}]_{i}}^2 \geq 24 (\max_{\sfc} W^2_{\sfr \sfc}) \tau^0_{\sfc} \log M + 1 \Big )    \leq     K K_{t-1}   \PC'_{t-2}  e^{-\kappa \kappa_{t-1}  (\omega/\Lr)^{2} \pc'_{t-2}  \e^2}.
%K K_{t-1}  (\Lr \Lc)^{t-1} \Lc  \exp\Big\{\frac{-\kappa \kappa_{t-1} \Mr  \omega^{2 t + 1}   \Lc^{2t- 2}  }{\Lr^{4t -2}(\log M)^{2t-2}}\Big\}.
%
\end{align*}
Finally, using \eqref{eq:abs_h_upper}, Lemma \ref{lem:squaredsums}, and \eqref{eq:BR4w_bnd}, we bound term $A_2$:
\begin{align*}
&A_2\leq P\Big ( \sum_{\sfc \in [\Lc]}   \sum_{\ell \in \sfc}  \frac{2W_{\sfr \sfc} }{L } \max_{i \in \sec(\ell)} \abs{[\hvec^k_{\sfc}]_{i}}  \sum_{j=0}^{t-1} \abs{  \gamma^{t,\sfr}_j }  \geq \frac{\textsf{B}\Lr}{4\omega} \Big ) \\
& \leq P\Big ( \sum_{\sfc \in [\Lc]}   \sum_{\ell \in \sfc}  \frac{4W^2_{\sfr \sfc} }{L } \max_{i \in \sec(\ell)} \abs{[\hvec^k_{\sfc}]_{i}}^2     \geq 24 (\max_{\sfc}W^2_{\sfr \sfc})  \tau^0_{\sfc} \log M + 1\Big )  + P\Big (  \sum_{j=0}^{t-1} \abs{  \gamma^{t,\sfr}_j }  \geq  \sum_{j=0}^{t-1}\abs{\hat{\gamma}^{t,\sfr}_{j }} + 1\Big ). 
\end{align*}
Then the upper bound follows along the same lines as that of $A_1$ and $A_4$.  This completes the proof of \eqref{eq:T2b}, and the lemma.
\end{proof}
%--------------------

%%%%%%%%%

\noindent \textbf{(a)}  Recall the definition of $\Dvec_{t,t, \sfr}$ from Lemma \ref{lem:hb_cond} Eq.\ \eqref{eq:Dtt}.
Consider the term $\frac{1}{\sqrt{L}}\sum_{\sfc } \|  \qadj^{t, \sfr}_{\perp, \sfc} \|  \, [\proj^{\parallel}_{\Madjmat_t^{\sfc}} \,   \Zprimesupc_{t} ]_{\sfr}$, where $\Zprimesupc_{t} \sim \mathcal{N}(0, \mathbb{I}_n)$.
Using Lemma \ref{lem:gauss_p0}, $\proj^{\parallel}_{\Madjmat_t^{\sfc}}\Zprimesupc_t \, \overset{d}{=} \, \sum_{j=0}^{t-1} \frac{\mvec^{j, \sfc}_{\perp, \sfr}}{\| \mvec^{j, \sfc}_\perp\|} Z^{\sfc}_j$ for i.i.d.\ $Z^{\sfc}_j \sim \normal(0,1)$.
 Then 
\be
\sum_{\sfc }  \frac{ \|\qadj^{t, \sfr}_{\perp, \sfc} \|}{\sqrt{L}}  \, [\proj^{\parallel}_{\Madjmat_t^{\sfc}} \,   \Zprimesupc_{t} ]_{\sfr} 
\overset{d}{=} 
\sum_{\sfc \in [\Lc]}  \sum_{j=0}^{t-1} \frac{\| \qadj^{t, \sfr}_{\perp, \sfc} \|}{\sqrt{L}} 
\frac{\mvec^{j, \sfc}_{\perp, \sfr}}{\| \mvec^{j, \sfc}_\perp\|} Z^{\sfc}_j, \qquad   
\{ Z^{\sfc}_j \}_{\sfc \in [\Lc], \, 0 \leq j \leq (t-1)}  \sim^{\text{i.i.d.}} \normal(0,1).
\label{eq:finalrep_deltaterm1}
\ee
Noe we simplify the final three terms of $\Dvec_{t,t, \sfr}$ in \eqref{eq:Dtt}.  Using the Lemma \ref{lem:Mv_conc} notation 
$\innerM^{\sfc}_t, \mathbf{v}^{\sfr,\sfc}$,
\ben
 \sum_{\sfc \in [\Lc]} \Madjmat_{t,\sfr}^{\sfc} ((\Madjmat_t^{\sfc})^* \Madjmat_t^{\sfc})^{-1}(\Xmat_{t,\sfc})^* \qadj_{\perp, \sfc}^{t, \sfr}
 = \sum_{\sfc \in [\Lc]} \Madjmat_{t,\sfr}^{\sfc}  (\innerM^{\sfc}_t)^{-1} \mathbf{v}^{\sfr, \sfc}
 = \sum_{j=1}^{t} \sum_{\sfc \in [\Lc]}  \madj_{\sfr}^{j-1, \sfc}  [(\innerM^{\sfc}_t)^{-1} \mathbf{v}^{\sfr,\sfc}]_{j}.
 \een
 Therefore,
 \begin{align}
 & \sum_{\sfc \in [\Lc]} \Madjmat_{t,\sfr}^{\sfc} ((\Madjmat_t^{\sfc})^* \Madjmat_t^{\sfc})^{-1}(\Xmat_{t,\sfc})^* \qadj_{\perp, \sfc}^{t, \sfr} -   \sum_{i=1}^{t-1}\gamma^{t,\sfr}_{i} \upsilon^{i}_{\sfr} \brmvec^{i-1}_{\sfr}    +  \upsilon^t_{\sfr}  \brmvec^{t-1}_{\sfr}  \label{eq:final_three_tt} \\
&= \sum_{j=1}^{t-1} \brmvec_{\sfr}^{j-1} [ \sum_{\sfc \in [\Lc]}S^{j-1}_{\sfr \sfc} \, \sqrt{W_{\sfr \sfc}}  [(\innerM^{\sfc}_t)^{-1} \mathbf{v}^{\sfr,\sfc}]_{j} - \gamma^{t,\sfr}_{j}  \upsilon^{j}_{\sfr} ]   +  \brmvec_{\sfr}^{t-1} [ \sum_{\sfc \in [\Lc]} S^{t-1}_{\sfr \sfc}  \sqrt{W_{\sfr \sfc}}  [(\innerM^{\sfc}_t)^{-1} \mathbf{v}^{\sfr,\sfc}]_{t}  + \upsilon^t_{\sfr}]. \nonumber
\end{align}

Using the expressions in \eqref{eq:finalrep_deltaterm1} and \eqref{eq:final_three_tt} in the definition of $\Dvec_{t,t, \sfr}$ in \eqref{eq:Dtt}, by Lemma \ref{lem:squaredsums},
\be
\begin{split}
&\frac{\norm{\Dvec_{t,t, \sfr}}^2}{2(t+1)} \leq \sum_{i=0}^{t-1}  (\gamma^{t,\sfr}_{i} - \hat{\gamma}^{t,\sfr}_{i})^2\, \norm{\bvec^i_{\sfr} }^2  \, +\, \norm{\Zprime_{t,\sfr} }^2 \Big[\frac{ \| \qadj^{t,\sfr}_{\perp} \|}{\sqrt{L}} - \sqrt{\sigma_{\perp, \sfr}^{t}} \Big]^2 \, + \, Z^2 \sum_{j=0}^{t-1} \sum_{\sfc \in [\Lc]}   \frac{
\| \qadj^{t, \sfr}_{\perp, \sfc}\|^2 \| \madj_{\perp, \sfr}^{j, \sfc} \|^2}{L \| \madj_{\perp}^{j, \sfc} \|^2}   \\
& +  \sum_{j=1}^{t-1} \| \brmvec_{\sfr}^{j-1}\|^2 \Big[ \hspace{-5pt} \sum_{\sfc \in [\Lc]} \hspace{-5pt} S^{j-1}_{\sfr \sfc}\sqrt{W_{\sfr \sfc}}  [(\innerM^{\sfc}_t)^{-1} \mathbf{v}^{\sfr,\sfc}]_{j} - \gamma^{t,\sfr}_{j} \upsilon^{j}_{\sfr}  \Big]^2  +  \|\brmvec_{\sfr}^{t-1}\|^2 \Big[ \hspace{-3pt} \sum_{\sfc \in [\Lc]} S^{t-1}_{\sfr \sfc}\sqrt{W_{\sfr \sfc}}  [(\innerM^{\sfc}_t)^{-1} \mathbf{v}^{\sfr,\sfc}]_{t} + \upsilon^{t}_{\sfr}  \Big]^2
\label{eq:delttt_sq_bound}
\end{split}
\ee

Applying Lemma \ref{sums}, with $\tilde{\e}_t = \frac{\e}{4(t+1)^2}$, we obtain
\begin{align}
& P\Big(\frac{1}{n}\sum_{\sfr \in [\Lr]} W^{\sfp}_{\sfr \sfc'} \norm{\Dvec_{t,t, \sfr}}^2  \geq \epsilon \Big) \leq \sum_{i=0}^{t-1} P\Big(\frac{1}{\Lr}\sum_{\sfr \in [\Lr]}  W^{\sfp}_{\sfr \sfc'}  \lvert\gamma^{t,\sfr}_i - \hat{\gamma}^{t,\sfr}_i \lvert^2 \frac{ \norm{\bvec^i_{\sfr} }^2}{\Mr} \geq  \tilde{\e}_t \Big)  \nonumber \\
& \ + P\Big(\frac{1}{\Lr}\sum_{\sfr \in [\Lr]} W^{\sfp}_{\sfr \sfc'} \Big \vert 
\frac{ \| \qadj^{t,\sfr}_{\perp} \|}{\sqrt{L}} - \sqrt{\sigma_{\perp, \sfr}^{t}} \Big \lvert^2 \frac{\| \Zprime_{t,\sfr} \|^2}{\Mr} \geq  \tilde{\e}_t \Big) + P\Big( \frac{Z^2}{n}\sum_{\sfr \in [\Lr]}  \sum_{j=0}^{t-1} \sum_{\sfc \in [\Lc]}   \frac{W^{\sfp}_{\sfr \sfc'}  
\|\qadj^{t, \sfr}_{\perp, \sfc} \|^2 \| \madj_{\perp, \sfr}^{j, \sfc} \|^2}{L \| \madj_{\perp}^{j, \sfc} \|^2}  \geq \tilde{\e}_t\Big)  \nonumber \\
& \ + \sum_{j=1}^{t-1} P\Big(\frac{1}{\Lr}\sum_{\sfr \in [\Lr]} \frac{W^{\sfp}_{\sfr \sfc'} \|\brmvec_{\sfr}^{j-1} \|^2}{\Mr} \Big \lvert \sum_{\sfc \in [\Lc]} S^{j-1}_{\sfr \sfc}\sqrt{W_{\sfr \sfc}}  [(\innerM^{\sfc}_t)^{-1} \mathbf{v}^{\sfr,\sfc}]_{j} - \gamma^{t,\sfr}_{j} \upsilon^{j}_{\sfr}  \Big \lvert^2 \geq \tilde{\e}_t \Big) \nonumber\\
&\ + P\Big(\frac{1}{\Lr}\sum_{\sfr \in [\Lr]} \frac{W^{\sfp}_{\sfr \sfc'}  \|\brmvec_{\sfr}^{t-1}\|^2}{\Mr} \Big \lvert \sum_{\sfc \in [\Lc]} S^{t-1}_{\sfr \sfc}\sqrt{W_{\sfr \sfc}}  [(\innerM^{\sfc}_t)^{-1} \mathbf{v}^{\sfr,\sfc}]_{t} + \upsilon^{t}_{\sfr}  \Big \lvert^2 \geq \tilde{\e}_t \Big).
\label{eq:delttt_sq_conc}
\end{align}
We  label the terms in \eqref{eq:delttt_sq_conc} $T_1 - T_5$, and show that each  has the desired upper bound.  

%____________________________________
First, using Lemma \ref{products_0} and induction hypotheses $\mathcal{B}_{0} (c)-\mathcal{B}_{t-1} (c)$ and $\mathcal{H}_t (f)$:
\be
\begin{split}
&T_1 \leq \sum_{i=0}^{t-1} P\Big(\max_{\sfr' \in [\Lr]}  \lvert\gamma^{t,\sfr'}_i - \hat{\gamma}^{t,\sfr'}_i \lvert^2 \cdot\sum_{\sfr \in [\Lr]} W^{\sfp}_{\sfr \sfc'} \frac{ 
\|\bvec^i_{\sfr} \|^2}{n} \geq  \tilde{\e}_t \Big) \\
&\leq \sum_{i=0}^{t-1} \Big[P\Big( \Big \lvert  \sum_{\sfr \in [\Lr]} W^{\sfp}_{\sfr \sfc'} \Big[\frac{ \| \bvec^i_{\sfr} \|^2}{n} -  \frac{\sigma^i_{\sfr}}{\Lr} \Big] \Big \lvert  \geq  \sqrt{\tilde{\e}_t} \Big)
+  \sum_{\sfr \in [\Lr]} P\Big(  \lvert \gamma^{t,\sfr}_i - \hat{\gamma}^{t,\sfr}_i \lvert^2  \geq  \tilde{\e}_t/ \Big( 2\max \Big\{1,  \sum_{\sfr \in [\Lr]}   \frac{W^{\sfp}_{\sfr \sfc'}\sigma^i_{\sfr}}{\Lr}  \Big\} \Big) \Big) \Big] \\
&\leq t K K_{t-1}  \PC_{t-2} \exp\Big\{ \frac{-\kappa \kappa_{t-1}  (\omega/\Lr)^{2 (\sfp-1)_+} \pc_{t-2} \e }{t^2} \Big\}  +  t \Lr K K_{t-1}  \PC'_{t-2}     \exp\Big\{\frac{-\kappa \kappa_{t-1}  (\omega/\Lr)^{2+(\sfp-1)_+} \pc'_{t-2}  \e}{t^2 (\log M)^2} \Big\}.
\label{eq:termT1}
\end{split}
\ee
For the second term we have used that  $\sum_{\sfr \in [\Lr]} W^{\sfp}_{\sfr \sfc}/\Lr = 1$ if $\sfp = 0$ and for $\sfp \in \{1,2\}$ we have $\sum_{\sfr \in [\Lr]} W^{\sfp}_{\sfr \sfc}/\Lr  \leq \kappa (\Lr/\omega)^{\sfp-1}$ (see \eqref{eq:Wrc_avgs}). We have also used the fact that  $\sigma^i_{\sfr}$ is bounded above and below by positive constants for $0 \leq i \leq t-1$ and $\sfr \in [\Lr]$.
%____________________________________

The second term in \eqref{eq:delttt_sq_conc} is bounded using Lemma \ref{products_0} along with induction hypothesis $\mathcal{H}_t (f)$, Lemma \ref{sqroots}, and Lemma \ref{lem:max_abs_normals},
\begin{align}
& T_2  \leq P\Big( \frac{1}{n} \sum_{\sfr \in [\Lr]}  W^{\sfp}_{\sfr \sfc'} \| \Zprime_{t, \sfr}\|^2 \max_{\sfr' \in [\Lr]}  \Big \vert \frac{ \| \qadj^{t,\sfr'}_{\perp} \| }{\sqrt{L}} - \sqrt{\sigma_{\perp, \sfr'}^{t}} \Big \lvert^2   \geq  \tilde{\e}_t \Big) \nonumber  \\
 &\leq \sum_{\sfr' \in [\Lr]} P\Big( \Big \vert \frac{ \|\qadj^{t,\sfr'}_{\perp} \|}{\sqrt{L}} - \sqrt{\sigma_{\perp, \sfr'}^{t}} \Big \lvert^2 \geq  {\tilde{\e}_t}/\big( {2} \max\{1,\sum_{\sfr \in [\Lr]}  \frac{W^{\sfp}_{\sfr \sfc'}}{\Lr}\}  \big)    \Big)  + P\Big(\Big \lvert \frac{1}{\Lr}\sum_{\sfr \in [\Lr]} W^{\sfp}_{\sfr \sfc'} \Big[\frac{\| \Zprime_{t, \sfr} \|^2}{\Mr} -1\Big] \Big \lvert \geq \sqrt{\tilde{\e}_t}\Big) \nonumber \\
& \leq   \Lr   K K_{t-1}  \PC'_{t-2}   \exp\Big\{   \frac{- \kappa \kappa_{t-1} (\omega/\Lr)^{2 + (\sfp-1)_+} \pc'_{t-2} \e}{t^2 (\log M)^2}  \Big\} + 2 \exp\Big\{ \frac{-\kappa n \epsilon (\omega/\Lr)^{2 \sfp-1}}{t^2} \Big\}.\label{eq:termT3}
\end{align}
We  note that $\sigma_{\perp, \sfr}^{t}$ is bounded below for all $\sfr$ and we have used that  $\sum_{\sfr \in [\Lr]} W^{\sfp}_{\sfr \sfc}/\Lr = 1$ if $\sfp = 0$ and for $\sfp \in \{1,2\}$ we have $\sum_{\sfr \in [\Lr]} W^{\sfp}_{\sfr \sfc}/\Lr  \leq \kappa (\Lr/\omega)^{\sfp-1}$ (see \eqref{eq:Wrc_avgs}).

%____________________________________

%____________________________________
The third term in \eqref{eq:delttt_sq_conc} is bounded as follows.  For all $\sfr \in [\Lr]$ and $\sfc \in [\Lc]$, using Lemma \ref{lem:squaredsums},
\[ \frac{\| \qadj^{t, \sfr}_{\perp, \sfc} \|^2}{L} \hspace{-1pt} \leq  \hspace{-1pt} \frac{W_{\sfr \sfc} }{L} \Big\| \brqvec^{t}_{\sfc} - \sum_{i=0}^{t-1} \gamma^{t, \sfr}_i \brqvec^{i}_{\sfc} \Big\|^2  \hspace{-1pt} \leq  \hspace{-1pt} \frac{W_{\sfr \sfc} (t+1) }{L} \Big[ \hspace{-1pt} \norm{\brqvec^{t}_{\sfc}}^2 + \sum_{i=0}^{t-1} (\gamma^{t, \sfr}_i)^2 \norm{\brqvec^{i}_{\sfc} }^2  \hspace{-1pt} \Big]  \hspace{-1pt} \leq  \hspace{-1pt} \frac{4 W_{\sfr \sfc} (t+1) }{\Lc} \Big[1 + \sum_{i=0}^{t-1} (\gamma^{t, \sfr}_i)^2\Big],\]
where the final inequality follows as $ \| \brqvec^{i}_{\sfc} \|^2 \leq 4L/\Lc$ for $0 \leq i \leq t$.   Therefore,
\be
\begin{split}
&T_3 \leq P\Big( \frac{Z^2}{n}\sum_{\sfr \in [\Lr]}  \sum_{j=0}^{t-1} \sum_{\sfc \in [\Lc]}  \frac{4   W^{\sfp}_{\sfr \sfc'} W_{\sfr \sfc} (t+1) }{\Lc} \Big[1 + \sum_{i=0}^{t-1} (\gamma^{t, \sfr}_i)^2\Big] \frac{\| \madj_{\perp, \sfr}^{j, \sfc} \|^2}{\| \madj_{\perp}^{j, \sfc} \|^2}  \geq \tilde{\e}_t\Big).
\label{eq:T3_bound1}
\end{split}
\ee
Now considering the RHS of \eqref{eq:T3_bound1}, note that if $\Big\{ (\gamma^{t, \sfr}_i)^2 \leq (\hat{\gamma}^{t, \sfr}_i)^2 + 1 \Big\}$ for  $\sfr \in [\Lr]$ and $0 \leq i \leq (t-1)$, 
then since $ (\hat{\gamma}^{t, \sfr}_i)^2 + 1 \in \Theta(1)$ for $\sfr \in [\Lr]$ (meaning $1 + \sum_{i=0}^{t-1} (\gamma^{t, \sfr}_i)^2 \leq \kappa t$ for a constant $\kappa > 0$),
\begin{align*}
&\frac{Z^2}{n}\sum_{\sfr \in [\Lr]}  \sum_{j=0}^{t-1} \sum_{\sfc \in [\Lc]}  \frac{4  W^{\sfp}_{\sfr \sfc'} W_{\sfr \sfc}  (t+1) }{\Lc} \Big[1 + \sum_{i=0}^{t-1} (\gamma^{t, \sfr}_i)^2\Big] \frac{\| \madj_{\perp, \sfr}^{j, \sfc} \|^2}{\| \madj_{\perp}^{j, \sfc} \|^2} \\
&\leq \frac{\kappa Z^2 t^2 \max_{\sfr \sfc}\{  W^{\sfp}_{\sfr \sfc'} W_{\sfr \sfc} \}}{n}\sum_{\sfr \in [\Lr]}  \sum_{j=0}^{t-1} \sum_{\sfc \in [\Lc]}  \frac{1 }{\Lc} \frac{\| \madj_{\perp, \sfr}^{j, \sfc} \|^2}{\| \madj_{\perp}^{j, \sfc} \|^2} = \frac{\kappa Z^2 t^3 \max_{\sfr \sfc}\{ W^{\sfp}_{\sfr \sfc'} W_{\sfr \sfc}  \}}{n}.
%\leq   \kappa'' \tilde{\e}_t .
%
\end{align*}
Using this argument in \eqref{eq:T3_bound1} and noting that $\max_{\sfr \sfc}\{ W^{\sfp}_{\sfr \sfc'} W_{\sfr \sfc} \} \leq  (P \Lr/\omega)^{\sfp+1}$ we obtain
\be
\begin{split}
T_3 &\leq  \sum_{i=0}^{t-1}  \sum_{\sfr \in [\Lr]} P\Big((\gamma^{t, \sfr}_i)^2 \geq (\hat{\gamma}^{t, \sfr}_i)^2 + 1\Big) + P\Big(Z^2 \geq  \kappa {n \tilde{\e}_t (\omega/\Lr)^{\sfp+1}}/ t^3 \Big) \\
&\leq   t \Lr  K K_{t-1}  \PC'_{t-2}    \exp\Big\{-\frac{\kappa \kappa_{t-1} (\omega/\Lr)^{2} \pc'_{t-2}}{ (\log M)^2}  \Big\}+ 2 \exp\Big\{  - \frac{1}{t^5}\kappa \Mr \omega (\omega/\Lr)^{\sfp} \e\Big\}.
\end{split}
\label{eq:DeltaB_T3_bound}
\ee
The final inequality in \eqref{eq:DeltaB_T3_bound} is obtained using $\mathcal{H}_{t} (f)$ and Lemma \ref{powers} for the first term, and Lemma \ref{lem:max_abs_normals} for the second term. 

Now consider the fourth term of \eqref{eq:delttt_sq_conc}:
\begin{align}
& T_4  \leq \sum_{j=1}^{t-1} P\Big( \sum_{\sfr' \in [\Lr]} \frac{ W^{\sfp}_{\sfr' \sfc'}  \| \brmvec_{\sfr'}^{j-1} \|^2}{n} \max_{\sfr \in [\Lr]}  \Big\lvert \sum_{\sfc \in [\Lc]} S^{j-1}_{\sfr \sfc} W_{\sfr \sfc}  [(\innerM^{\sfc}_t)^{-1} \mathbf{v}^{\sfr,\sfc}]_{j} - \gamma^{t,\sfr}_{j} \upsilon^{j}_{\sfr} \Big \lvert^2\geq  \tilde{\e}_t \Big) \nonumber \\
&\overset{(a)}\leq  \sum_{j=1}^{t-1}  P\Big(\Big\lvert \frac{1}{\Lr} \sum_{\sfr \in [ \Lr ]} W^{\sfp}_{\sfr \sfc'}  
\Big( \frac{ \| \brmvec_{\sfr}^{j-1} \|^2}{\Mr} - \phi_{\sfr}^{j-1} \Big) \Big\lvert \geq  \sqrt{\tilde{\e}_t} \Big) \label{eq:termT4} \\
&\quad + \sum_{j=1}^{t-1}  \sum_{\sfr \in [\Lr]}  P\Big(\Big\lvert \sum_{\sfc \in [\Lc]} S^{j-1}_{\sfr \sfc} W_{\sfr \sfc}  [(\innerM^{\sfc}_t)^{-1} \mathbf{v}^{\sfr,\sfc}]_{j} - \gamma^{t,\sfr}_{j} \upsilon^{j}_{\sfr} \Big \lvert^2 \geq  \tilde{\e}_t/\Big( 2 \max\{ 1, \big( \frac{1}{\Lr} \sum_{\sfr \in [\Lr]} W^{\sfp}_{\sfr \sfc'} \phi_{\sfr}^{j-1} \big) \} \Big) \Big) \nonumber \\
 &\overset{(b)}\leq  t   K K_{t-1}  \PC_{t-2}\hspace{-1pt} \exp\hspace{-1pt}\Big\{ \frac{- \kappa  \kappa_{t-1}(\omega/\Lr)^{2(\sfp-1)_+}  \pc_{t-2}   \e}{t^2} \Big\} +  t^3 \Lr K K_{t-1}   \PC'_{t-2} \hspace{-1pt}  \exp\hspace{-1pt}\Big\{\frac{-\kappa \kappa_{t-1}  (\omega/\Lr)^{2 + (\sfp-1)_+}  \pc'_{t-2}  \e}{t^6(\log M)^2}\Big\}. \nonumber
\end{align}
Step $(a)$ follows by Lemma \ref{products_0} and $(b)$ by the induction hypothesis $\mc{B}_{t-1}(e)$ and Lemma \ref{lem:Mv_conc}, noting $\sum_{\sfr \in [\Lr]} W^{\sfp}_{\sfr \sfc}/\Lr = 1$ if $\sfp = 0$ and $\sum_{\sfr \in [\Lr]} W^{\sfp}_{\sfr \sfc}/\Lr  \leq P(\Lr/\omega)^{\sfp-1}$ for $\sfp \in \{1,2\}$. In particular, $\mc{B}_{t-1}(e)$ is used as follows.   First, recall $\|  \madj^{j-1, \sfc}_{\sfr} \|^2  = (S^{j-1}_{\sfr \sfc})^2 W_{\sfr \sfc} \| \brmvec^{j-1}_{\sfr}\|^2$ and $\frac{n}{L} S^{j-1}_{\sfr \sfc} \in \Theta(1)$.  Therefore,
\begin{align*}
& P\Big(\Big \lvert \frac{1}{\Lr} \sum_{\sfr \in [\Lr]} W^{\sfp}_{\sfr \sfc'}  \Big[ \frac{ \| \brmvec_{\sfr}^{j-1} \|^2}{\Mr} - \phi_{\sfr}^{j-1} \Big] \Big \lvert \geq  \sqrt{\tilde{\e}_t} \Big) \leq   P\Big(\Big \lvert \frac{1}{\Lr} \sum_{\sfr \in [\Lr]} W^{\sfp}_{\sfr \sfc'} \Big(\frac{n}{L} S^{j-1}_{\sfr \sfc}\Big)^2 \Big[ \frac{ \| \brmvec_{\sfr}^{j-1} \|^2}{\Mr} -  \phi_{\sfr}^{j-1} \Big] \Big \lvert \geq  \kappa \sqrt{\tilde{\e}_t} \Big) \\
 & =  P\Big(\Big \lvert \frac{n}{L^2} \sum_{\sfr \in [\Lr]} W^{\sfp-1}_{\sfr \sfc'}   \| \madj_{\sfr}^{j-1, \sfc} \|^2 -  \frac{n^2}{\Lr L^2} \sum_{\sfr \in [\Lr]} W^{\sfp}_{\sfr \sfc'} (S^{j-1}_{\sfr \sfc})^2   \phi_{\sfr}^{j-1}\Big \lvert \geq  \kappa \sqrt{\tilde{\e}_t} \Big).
\end{align*}
Now the result from  $\mc{B}_{t-1}(e)$ can be applied directly.

Finally, the last term in \eqref{eq:delttt_sq_conc} can be bounded using the same arguments as for the fourth term.  

To see the overall concentration result, notice that $\PC_{t-2} = (\Lr \Lc)^{t-1} \leq \Lr \PC'_{t-2} = \Lr \Lc \PC_{t-2} = \PC_{t-1}$ and $(\omega/\Lr)^{2(\sfp-1)_+}  \pc_{t-2} \geq  (\omega/\Lr)^{2 + (\sfp-1)_+}  \pc'_{t-2} /(\log M)^2$ since $  \pc_{t-2} =   \pc'_{t-2}$ and $\sfp \in \{0, 1, 2\}$. Finally, $(\omega/\Lr)^{2 + (\sfp-1)_+}  \pc'_{t-2} /(\log M)^2 = (\omega/\Lr)^{(\sfp-1)_+}  \pc_{t-1}$.  Therefore we have the bound in \eqref{eq:Ba1}.

%%%--------------------------------------------------------------------------------------------------------------------------

\noindent  \textbf{(b)}  Using $\bvec^{t}_{\sfr} \lvert_{\mscrs_{t, t}}\stackrel{d}{=} \sqrt{\sigma^t_{\sfr}} \breve{\Z}_{t, {\sfr}}+ \breve{\Dvec}_{t, {\sfr}},$ by Lemma \ref{lem:ideal_cond_dist}, where $  \breve{\Z}_{t, {\sfr}} \sim \mathcal{N}(0, \mathbb{I}_{\Mr})$, and $\wvec_{\sfr} \overset{d}{=} \sigma \mathbf{U}_{\sfr}$ for $\mathbf{U}_{\sfr} \sim \mathcal{N}(0, \mathbb{I}_{\Mr})$ independent of  $ \breve{\Z}_{t, {\sfr}} $, we have by Lemma~\ref{sums},
\begin{align}
& P\Big(\Big \lvert \frac{1}{\Lr}\sum_{\sfr \in [\Lr]} W^{\sfp}_{\sfr \sfc} \Big[\frac{(\bvec^t_{\sfr})^*\wvec_{\sfr}}{\Mr} \Big]\Big \lvert \geq \epsilon \Big) = P\Big(\Big \lvert \frac{1}{\Lr}\sum_{\sfr \in [\Lr]} W^{\sfp}_{\sfr \sfc}\Big(\sqrt{\sigma^t_{\sfr}}  \Big[\frac{ (\breve{\Z}_{t, {\sfr}})^* \mathbf{U}_{\sfr}}{\Mr}  \Big]+ \frac{ \breve{\Dvec}_{t, {\sfr}}^*\mathbf{U}_{\sfr}}{\Mr} \Big)\Big \lvert \geq \frac{\epsilon}{\sigma} \Big) \nonumber \\
&\leq P\Big(\Big \lvert \frac{1}{\Lr}\sum_{\sfr \in [\Lr]} W^{\sfp}_{\sfr \sfc} \sqrt{\sigma^t_{\sfr}}  \Big[\frac{ (\breve{\Z}_{t, {\sfr}})^* \mathbf{U}_{\sfr}}{\Mr}  \Big] \Big \lvert \geq \frac{\epsilon}{2\sigma } \Big) + P\Big(\Big \lvert \frac{1}{\Lr}\sum_{\sfr \in [\Lr]} W^{\sfp}_{\sfr \sfc}   \Big[ \frac{ \breve{\Dvec}_{t, {\sfr}}^*\mathbf{U}_{\sfr}}{\Mr}  \Big] \Big \lvert \geq \frac{\epsilon}{2\sigma} \Big). \label{eq:breveDelUr}
\end{align}
To bound the first term in the above, we recall that $\breve{\Z}_{t, {\sfr}}$ is independent of $\mathbf{U}_{\sfr}$. Hence,  using Lemma \ref{lem:max_abs_normals} and the same argument as in \eqref{eq:Bb2_1}, this term is  bounded by $2\exp\{-\kappa n \e^2(\omega/\Lr)^{\max\{2\sfp - 1, 0\}}\}$.  

For the second term in \eqref{eq:breveDelUr}, we first obtain a concentration result on the norm of $\breve{\Dvec}_{t} = \sum_{i=0}^{t} ({\sigma^{t}_{\sfr}}/{\sigma^{i}_{\sfr}}) \Dvec_{i, i, {\sfr}}$, as defined in Lemma \ref{lem:ideal_cond_dist}. We have
\begin{align}
&P\Big( \frac{1}{n} \sum_{\sfr \in [\Lr]} W^{\sfp}_{\sfr \sfc}   \|  \breve{\Dvec}_{t, \sfr} \|^2 \geq \e^2 \Big) 
\overset{(a)}{\leq}  P\Big( \frac{1}{n}\sum_{\sfr \in [\Lr]}   \sum_{i=0}^{t} \Big(\frac{\sigma^{t}_{\sfr}}{\sigma^{i}_{\sfr}}\Big)^2 W^{\sfp}_{\sfr \sfc} \norm{\Dvec_{i, i, {\sfr}}}^2 \geq \frac{\e^2}{t+1} \Big)  \label{eq:breve_delta_bound}\\
&\overset{(b)}{\leq} % 
\sum_{i=0}^{t}  P\Big( \frac{1}{n} \sum_{\sfr \in [\Lr]} W^{\sfp}_{\sfr \sfc}   \norm{\Dvec_{i, i, {\sfr}}}^2 \geq \frac{\kappa \e^2}{(t+1)^2} \Big) \overset{(c)}{\leq} 
 t^4   K K_{t-1}   \PC_{t-1}   \exp\Big\{-\frac{1}{t^8  } \kappa \kappa_{t-1}   (\omega/\Lr)^{(\sfp-1)_{+}}  \pc_{t-1}  \e^2\Big\}. \nonumber
%t^4 K K_{t-1}  (\Lr \Lc)^{t}    \exp\Big\{\frac{-\kappa \kappa_{t-1} \Mr  \omega^{2t+ 1 + \max\{0, \sfp-1\}} \Lc^{2t} \e^2}{t^8 \Lr^{4t + \max\{0, \sfp-1\}} (\log M)^{2t} }\Big\}.
%
\end{align}
Step $(a)$ follows from Lemma  \ref{lem:squaredsums} and step $(b)$ by Lemma \ref{sums}, using that $\sigma^{i}_{\sfr} \in \Theta(1)$ for all $0 \leq i \leq t$. Finally step $(c)$ uses the result from $\mathcal{B}_t (a)$ above. %\RV{needs to be filled in}

Next, using $\abs{\breve{\Dvec}_{t, {\sfr}}^*\mathbf{U}_{\sfr}} \leq \| \breve{\Dvec}_{t, {\sfr}} \| \| \mathbf{U}_{\sfr} \|$, the second term can be bounded as
\begin{align}
&P\Big(\Big \lvert \frac{1}{\Lr}\sum_{\sfr \in [\Lr]}W^{\sfp}_{\sfr \sfc}  \Big[ \frac{ \breve{\Dvec}_{t, {\sfr}}^*\mathbf{U}_{\sfr}}{\Mr}  \Big] \Big \lvert \geq \frac{\epsilon}{2\sigma} \Big) \leq P\Big(\sum_{\sfr \in [\Lr]}W^{\sfp}_{\sfr \sfc} \frac{\| \breve{\Dvec}_{t, {\sfr}} \| \| \mathbf{U}_{\sfr} \|}{n} \geq \kappa \epsilon \Big)  \label{eq:stepb_T3bound} \\
& \overset{(a)}{\leq} P\Big(\frac{1}{n}\sum_{\sfr \in [\Lr]}W^{\sfp}_{\sfr \sfc} \| \breve{\Dvec}_{t, {\sfr}} \|^2  \cdot \frac{1}{n}\sum_{\sfr \in [\Lr]}W^{\sfp}_{\sfr \sfc} \| \mathbf{U}_{\sfr} \|^2 \geq \kappa \epsilon^2 \Big) \nonumber \\
&\overset{(b)}{\leq} P\Big(\frac{1}{n}\sum_{\sfr \in [\Lr]}W^{\sfp}_{\sfr \sfc} \| \breve{\Dvec}_{t, {\sfr}} \|^2 \geq {\kappa \epsilon^2}/ \Big( 2 \max\big\{1, \, \frac{1}{\Lr}\sum_{\sfr \in [\Lr]} W^{\sfp}_{\sfr \sfc}   \big\}\Big)  \Big) + P\Big( \Big \lvert \frac{1}{\Lr}\sum_{\sfr \in [\Lr]} W^{\sfp}_{\sfr \sfc}  \Big[\frac{\| \mathbf{U}_{\sfr} \|^2}{\Mr} - 1\Big] \Big \lvert \geq \e \Big)  \nonumber  \\
&\overset{(c)}{\leq}  
 t^4   K K_{t-1}   \PC_{t-1}   \exp\{-\frac{1}{t^8 } \kappa \kappa_{t-1}   (\omega/\Lr)^{2(\sfp-1)_{+}}  \pc_{t-1}  \e^2\}. + 2\exp\{-\kappa n \e^2(\omega/\Lr)^{\max\{2\sfp - 1, 0\}}\}. \nonumber 
\end{align}
Step $(a)$ follows by Cauchy-Schwarz, $(b)$ by Lemma \ref{products_0}, and $(c)$ by Lemma \ref{lem:max_abs_normals} and \eqref{eq:breve_delta_bound}. %\RV{needs to be filled in}

%--------------------------------

\noindent  \textbf{(c)} Let $0 \leq s \leq t$.  By Lemma \ref{lem:ideal_cond_dist}, $\bvec^{t}_{\sfr} \lvert_{\mscrs_{t, t}}\stackrel{d}{=} \sqrt{\sigma^t_{\sfr}} \breve{\Z}_{t, {\sfr}}+ \breve{\Dvec}_{t, {\sfr}}$ and $\bvec^{s}_{\sfr} \lvert_{\mscrs_{t, t}}\stackrel{d}{=} \sqrt{\sigma^s_{\sfr}} \breve{\Z}_{s, {\sfr}}+ \breve{\Dvec}_{s, {\sfr}},$ where $  \breve{\Z}_{t, {\sfr}} \sim \mathcal{N}(0, \mathbb{I}_{\Mr})$ and $ \breve{\Z}_{s, {\sfr}} \sim \mathcal{N}(0, \mathbb{I}_{\Mr})$, such that for any $i \in [\Mr]$, the entries 
$[\breve{Z}_{s, {\sfr}}]_i, \, [\breve{Z}_{t, {\sfr}}]_i$ are jointly Gaussian with covariance $\expec\{ [\breve{Z}_{s, {\sfr}}]_i [\breve{Z}_{t, {\sfr}}]_i \} = \sqrt{{\sigma^t_{\sfr}}/{\sigma^{s}_{\sfr}}}$. Now we use Lemma \ref{sums} and the fact that $\sigma^s_{\sfr} \in \Theta(1)$ to write
\be
\begin{split}
&P\Big(\Big \lvert\frac{1}{n} \sum_{\sfr \in [\Lr]} W^{\sfp}_{\sfr \sfc}  (\bvec^s_{\sfr})^*\bvec^t_{\sfr}  - \frac{1}{\Lr} \sum_{\sfr \in [\Lr]} W^{\sfp}_{\sfr \sfc}   \sigma^t_{\sfr} \Big \lvert \geq \epsilon \Big) \\
&= P\Big(\Big \lvert  \frac{1}{\Lr} \sum_{\sfr \in [\Lr]}W^{\sfp}_{\sfr \sfc}  \Big[ \frac{(\sqrt{\sigma^t_{\sfr}} \breve{\Z}_{t, {\sfr}}+ \breve{\Dvec}_{t, {\sfr}})^* (\sqrt{\sigma^s_{\sfr}} \breve{\Z}_{s, {\sfr}}+ \breve{\Dvec}_{s, {\sfr}})}{\Mr}  -   \sigma^t_{\sfr} \Big] \Big \lvert \geq \epsilon \Big)  \\
&\leq P\Big(\Big \lvert  \frac{1}{\Lr}\sum_{\sfr \in [\Lr]} W^{\sfp}_{\sfr \sfc}  \sqrt{\sigma^t_{\sfr} \sigma^s_{\sfr}} \Big[   \frac{ \breve{\Z}_{t, {\sfr}}^*\breve{\Z}_{s, {\sfr}}}{\Mr} -  \sqrt{\sigma^t_{\sfr}/\sigma^s_{\sfr}} \Big] \Big \lvert  \geq \frac{\epsilon}{4} \Big) + P\Big(\sum_{\sfr \in [\Lr]}  W^{\sfp}_{\sfr \sfc}  \frac{ \vert \breve{\Dvec}_{t, \sfr}^*\breve{\Dvec}_{s, \sfr} \vert}{n} \geq \frac{ \epsilon}{4} \Big)  \\\
&\qquad + P\Big(\frac{1}{\Lr} \sum_{\sfr \in [\Lr]} W^{\sfp}_{\sfr \sfc}  \frac{ \vert  \breve{\Dvec}_{s, {\sfr}}^* \,  \breve{\Z}_{t, {\sfr}} \vert }{\Mr} 
  \geq \frac{ \epsilon}{4} \Big) + P\Big( \frac{1}{\Lr} \sum_{\sfr \in [\Lr]} W^{\sfp}_{\sfr \sfc}  \frac{\vert  \breve{\Dvec}_{t, {\sfr}}^* \,  \breve{\Z}_{s, {\sfr}} \vert  }{\Mr} 
  \geq \frac{ \epsilon}{4} \Big).
\label{eq:Bct8}
\end{split}
\ee
Label the terms of the above as $T_1 - T_4$.  

For $T_1$ recall from Lemma \ref{lem:ideal_cond_dist} that $\breve{\Z}_{t, \sfc} \overset{d}{=} \breve{\Z}_{s, \sfr}  \sqrt{{\sigma^{t}_{\sfr}}/{\sigma^{s}_{\sfr}}}  +\breve{\mathbf{U}}_{t, \sfr}\sqrt{1  -  {\sigma^{t}_{\sfr}}/{\sigma^{s}_{\sfr}}}$ where $\breve{\mathbf{U}}_{t, \sfr}$ and $ \breve{\Z}_{s, \sfr}$ are independent. Therefore, using Lemma \ref{sums},
\be
\begin{split}
&T_1 \leq  P\Big(\Big \lvert  \frac{1}{\Lr}\sum_{\sfr \in [\Lr]} W^{\sfp}_{\sfr \sfc} \sigma^t_{\sfr} \Big[  \frac{\|\breve{\Z}_{s, {\sfr}} \|^2}{\Mr} - 1\Big] \Big \lvert  \geq \frac{\epsilon}{8} \Big) + P\Big(\Big \lvert  \frac{1}{\Lr}\sum_{\sfr \in [\Lr]} W^{\sfp}_{\sfr \sfc}  \sqrt{\sigma^t_{\sfr} ( \sigma^s_{\sfr}  - 1)} \Big[   \frac{\breve{\mathbf{U}}_{t, \sfr}^*\breve{\Z}_{s, {\sfr}}}{\Mr}  \Big] \Big \lvert  \geq \frac{\epsilon}{8} \Big) \\ 
& \leq K \exp\{-\kappa n \e^2 (\omega/\Lr)^{\max\{2\sfp - 1, 0\}} \},
\end{split}
\label{eq:Bct8_T1}
\ee
where the last inequality is obtained by using Lemma \ref{lem:max_abs_normals} to bound each of the two probabilities. We use \eqref{eq:Wrc_avgs} and the fact that $\sigma^s_{\sfr} \in \Theta(1)$.
%using a bound on $g(\mathbf{W}^{\sfp}_{(\cdot,  \sfc)})$  given in \eqref{eq:g_Wbound} and \eqref{eq:W2_bound}, The terms on the RHS of \eqref{eq:Bct8_T1} are upper bounded by $K \exp\{-\kappa n \e^2/g( \mathbf{W}^{\sfp}_{(\cdot,  \sfc)})\} \leq K \exp\{-\kappa n \e^2 (\omega/\Lr)^{\max\{2\sfp - 1, 0\}} \}$.

 Next for $T_2$, by two applications of Cauchy-Schwarz,
\begin{align}
&T_2 \leq P\Big(\frac{1}{n}\sum_{\sfr \in [\Lr]} W^{\sfp}_{\sfr \sfc}  \| \breve{\Dvec}_{t, \sfr} \| \| \breve{\Dvec}_{s, \sfr} \| \geq \frac{\kappa \epsilon}{4} \Big)  \leq P\Big( \frac{1}{n} \sum_{\sfr \in [\Lr]}   W^{\sfp}_{\sfr \sfc} \| \breve{\Dvec}_{t, \sfr} \|^2 \cdot \frac{1}{n} \sum_{\sfr \in [\Lr]} W^{\sfp}_{\sfr \sfc} \| \breve{\Dvec}_{s, \sfr} \|^2  \geq \kappa \epsilon^2 \Big) \label{eq:Bct8_eq2} \\
&\leq\Big(\frac{1}{n} \sum_{\sfr \in [\Lr]}  W^{\sfp}_{\sfr \sfc}  \| \breve{\Dvec}_{t, \sfr} \|^2 \geq \kappa \epsilon \Big)  + \Big(\frac{1}{n} \sum_{\sfr \in [\Lr]}  W^{\sfp}_{\sfr \sfc} \| \breve{\Dvec}_{s, \sfr} \|^2 \geq \kappa \epsilon \Big)  \leq   t^4   K K_{t-1}   \PC_{t-1}  e^{-\frac{1}{t^8 } \kappa \kappa_{t-1}   (\omega/\Lr)^{(\sfp-1)_{+}}  \pc_{t-1}  \e}. \nonumber
%t^4 K K_{t-1}  (\Lr \Lc)^{t}    \exp\Big\{\frac{-\kappa \kappa_{t-1} \Mr  \omega^{2t+ 1 + \max\{0, \sfp-1\}} \Lc^{2t} \e^2}{t^8 \Lr^{4t + \max\{0, \sfp-1\}} (\log M)^{2t} }\Big\},
%
\end{align}
The final inequality follows by \eqref{eq:breve_delta_bound}. %  \RV{similar to \eqref{eq:breve_delta_bound} except that $\e$ here instead of $\e^2$}
Finally terms $T_3, T_4$ can be shown to have the desired upper bounded by work similar to that in \eqref{eq:stepb_T3bound}.

%%%--------------------------------------------------------------------------------------------------------------------------------------------------------------------------------------------------------------

\noindent  \textbf{(d)} 
Let $0 \leq \tilde{s}, s \leq t$ where either $\tilde{s} = t$, $s = t$, or both $\tilde{s} = s = t$. Since $\madj^{s,\sfc}_{\sfr}  =  S^s_{\sfr \sfc} \sqrt{W_{\sfr \sfc}}  (\bvec^{s}_{\sfr} - \wvec_{\sfr})$,
\begin{align}
&P \Big( \frac{\Mr}{L}  \Big \lvert  \sum_{\sfr \in [\Lr]}  \Big( W^{\sfp - 1/2}_{\sfr \sfc}\frac{ (\bvec^{\tilde{s}}_{\sfr})^*\madj^{s, \sfc}_{\sfr}}{\Mr} - S^s_{\sfr \sfc} W^{\sfp}_{\sfr \sfc} \sigma^{\max(\tilde{s},s)}_{\sfr} \Big)  \Big \lvert  \geq \epsilon \Big) \label{eq:Btd_firstbound} \\
%
%&= P \Big( \frac{n}{L\Lr}  \Big \lvert  \sum_{\sfr \in [\Lr]}  \Big( W^{\sfp - 1/2}_{\sfr \sfc}\frac{ (\bvec^{\tilde{s}}_{\sfr})^*\madj^{s, \sfc}_{\sfr}}{\Mr} - S^s_{\sfr \sfc} W^{\sfp}_{\sfr \sfc} \sigma^{\max(\tilde{s},s)}_{\sfr} \Big)  \Big \lvert  \geq \epsilon \Big)  \\
%
&\leq P \Big( \frac{1}{\Lr}  \Big \lvert \sum_{\sfr \in [\Lr]}  W^{\sfp}_{\sfr \sfc} \Big(\frac{n}{L} S^s_{\sfr \sfc}\Big)  \Big[ \frac{(\bvec^{\tilde{s}}_{\sfr})^*\bvec^{s}_{\sfr}}{\Mr}   - \sigma^{\max(\tilde{s},s)}_{\sfr} \Big]  \Big \lvert  \geq \kappa \epsilon  \Big) + P \Big(   \frac{1}{\Lr}  \Big \lvert  \sum_{\sfr \in [\Lr]}  W^{\sfp}_{\sfr \sfc}  \Big(\frac{n}{L} S^s_{\sfr \sfc}\Big)  \Big[\frac{(\bvec^{\tilde{s}}_{\sfr}) ^*   \wvec_{\sfr}}{\Mr}\Big]
  \Big \lvert  \geq \kappa \epsilon \Big). \nonumber
\end{align}
Now the result follows from $\mathcal{B}_t(b)$ and $\mathcal{B}_t(c)$ using that $ (\frac{n}{L} S^s_{\sfr \sfc})  \in \Theta(1).$

%--------------------------------------------------------------------------------------------------------------------------------------------------------------------------------------------------------------

\noindent  \textbf{(e)} Since  $\madj^{s,\sfc}_{\sfr}  \overset{d}{=} S^s_{\sfr, \sfc} \sqrt{W_{\sfr \sfc}}  ( \bvec^s_{\sfr}  - \wvec_{\sfr})$, we have  
$(\madj^{s, \sfc}_{\sfr})^* \madj^{t, \sfc}_{\sfr} = S^s_{\sfr, \sfc} S^t_{\sfr, \sfc} W_{\sfr \sfc} [ (\bvec^s_{\sfr})^* \bvec^t_{\sfr} -  \wvec_{\sfr}^*  \bvec^s_{\sfr}  -  \wvec_{\sfr}^*  \bvec^t_{\sfr}   +   \norm{\wvec_{\sfr}}^2]$.
Therefore, using Lemma \ref{sums} and that $ \phi^t_{\sfr} = \sigma^2 + \sigma^t_{\sfr}$,  we obtain
\begin{align*}
&P \Big(\frac{n }{L^2} \Big\lvert \sum_{\sfr \in [\Lr]} W^{\sfp-1}_{\sfr \sfc} (\madj^{s, \sfc}_{\sfr})^* \madj^{t, \sfc}_{\sfr}   -\Mr   \sum_{\sfr \in [\Lr]} S^s_{\sfr, \sfc} \, S^t_{\sfr, \sfc} \, W^{\sfp}_{\sfr \sfc}  \, \phi^t_{\sfr}  \Big \lvert  \geq \epsilon \Big) \\
&= P \Big(\frac{n^2 }{\Lr L^2} \Big\lvert \sum_{\sfr \in [\Lr]}  S^s_{\sfr, \sfc} S^t_{\sfr, \sfc} W^{\sfp}_{\sfr \sfc} \Big[ \frac{(\bvec^s_{\sfr})^* \bvec^t_{\sfr}}{\Mr} -  \frac{ \wvec_{\sfr}^*  \bvec^t_{\sfr}}{\Mr}-  \frac{ \wvec_{\sfr}^*  \bvec^s_{\sfr}}{\Mr} +   \frac{\norm{\wvec_{\sfr}}^2}{\Mr}  -  (\sigma^2 + \sigma^t_{\sfr})\Big] \Big \lvert  \geq \epsilon \Big) \\
&\leq P \Big( \Big\lvert \frac{1}{\Lr} \sum_{\sfr \in [\Lr]} (\frac{n^2}{L^2} S^s_{\sfr, \sfc} S^t_{\sfr, \sfc})W^{\sfp}_{\sfr \sfc}  \Big[ \frac{(\bvec^s_{\sfr})^* \bvec^t_{\sfr}}{\Mr}  - \sigma^t_{\sfr} \Big] \Big \lvert  \geq   \frac{\epsilon}{4}  \Big) + P \Big( \Big\lvert \frac{1}{\Lr} \sum_{\sfr \in [\Lr]}   (\frac{n^2}{L^2} S^s_{\sfr, \sfc} S^t_{\sfr, \sfc})W^{\sfp}_{\sfr \sfc} \Big[\frac{\wvec_{\sfr}^*  \bvec^t_{\sfr} }{\Mr} \Big] \Big \lvert  \geq \frac{\epsilon}{4}  \Big) \\
& + P \Big( \Big\lvert \frac{1}{\Lr} \sum_{\sfr \in [\Lr]}   (\frac{n^2}{L^2} S^s_{\sfr, \sfc} S^t_{\sfr, \sfc}) W^{\sfp}_{\sfr \sfc}  \Big[\frac{\wvec_{\sfr}^*  \bvec^s_{\sfr}}{\Mr} \Big] \Big \lvert  \geq \frac{\epsilon}{4}  \Big) + P \Big( \Big\lvert\frac{1}{\Lr} \sum_{\sfr \in [\Lr]}   (\frac{n^2}{L^2} S^s_{\sfr, \sfc} S^t_{\sfr, \sfc}) W^{\sfp}_{\sfr \sfc}  \Big[   \frac{\norm{\wvec_{\sfr}}^2}{\Mr} - \sigma^2\Big] \Big \lvert  \geq \frac{\epsilon}{4} \Big).
\end{align*}
Now the result follows from Lemma \ref{lem:max_abs_normals}, $\mathcal{B}_t(b)$, and $\mathcal{B}_t(c)$ using that $ (\frac{n}{L} S^s_{\sfr \sfc})  \in \Theta(1)$.
%%%%--------------------------------------------------------------------------------------------------------------------------------------------------------------------------------------------------------------

\noindent  \textbf{(f)} We first prove \eqref{eq:Bf}.  Recall, $\alphavec^{t,\sfc} = \frac{n}{L^2}(\innerM^{\sfc}_t)^{-1} (\Madjmat^{\sfc}_t)^* \madj^{t,\sfc}$ where $\innerM^{\sfc}_t := \frac{n}{L^2} (\Madjmat^{\sfc}_t)^* \Madjmat^{\sfc}_t$, so 
 for $1 \leq k \leq t$, we have $\alpha^{t,\sfc}_{k-1} = \frac{n}{L^2} \sum_{i=1}^{t} [(\innerM^{\sfc}_t)^{-1}]_{k i} (\madj^{i-1, \sfc})^* \madj^{t,\sfc}$.  From the definition of $\hat{\alpha}^{t, \sfc}$ in \eqref{eq:hatalph_hatgam_def},
\begin{align*}
&P\Big(\lvert \alpha^{t,\sfc}_{k-1} - \hat{\alpha}^{t,\sfc}_{k-1}\lvert \geq \epsilon \Big) = P\Big(\Big \lvert \sum_{i=1}^{t}\Big[ \frac{n(\madj^{i-1, \sfc})^* \madj^{t,\sfc}}{L^2}[(\innerM^{\sfc}_t)^{-1}]_{k i}   - \frac{n \tau^{t}_{\sfc}}{L} [(\brCmat^{t, \sfc})^{-1}  ]_{k i}\Big] \Big \lvert \geq \e \Big) \\
&\overset{(a)}{\leq} \sum_{i=1}^{t} P\Big(\Big \lvert \frac{n(\madj^{i-1, \sfc})^* \madj^{t,\sfc}}{L^2}[(\innerM^{\sfc}_t)^{-1}]_{k i}   - \frac{n \tau^{t}_{\sfc}}{L} [(\brCmat^{t, \sfc})^{-1}  ]_{k i} \Big \lvert \geq \frac{\e}{t} \Big)\overset{(b)}{\leq}     t^5   K K_{t-1}   \PC_{t-1}   e^{-\frac{1}{t^{10}} \kappa \kappa_{t-1}     \pc_{t-1}  \e^2}.
%t^5 K K_{t-1}  (\Lr \Lc)^{t}    \exp\Big\{\frac{-\kappa \kappa_{t-1} \Mr  \omega^{2t+ 1} \Lc^{2t} \e^2}{t^{10} \Lr^{4t} (\log M)^{2t} }\Big\} .
%
\end{align*}
Step $(a)$ follows by Lemma \ref{sums} and $(b)$ by Lemmas \ref{products}, \ref{products_0} using $\mathcal{B}_{t}(e)$ with $\sfp =1$  and $\mathcal{B}_{t-1}(g)$ result \eqref{eq:Bg}.  Note that $ {n \tau^{t}_{\sfc}}/{L}$  and the absolute values of the non-zero entries of $(\brCmat^{t, \sfc})^{-1}$ are $\Theta(1)$.

Next we prove \eqref{eq:Bf1}.  First, note that 
$\| \madj^{t,\sfc}_{\perp} \|^2 = \| \madj^{t,\sfc} \|^2 - \| \madj^{t,\sfc}_{\parallel} \|^2 = \|\madj^{t,\sfc}\|^2 - \| \Madjmat^{\sfc}_t \alpha^{t,\sfc} \|^2$. 
Using the definition of $\tau^{t}_{\perp, \sfc}$ in \eqref{eq:sigperp_defs} and Lemma \ref{sums}, we  have
\begin{align}
P&\Big(\frac{n}{L} \Big \lvert \frac{1}{L} \| \madj^{t,\sfc}_{\perp} \|^2 - \tau^{t}_{\perp, \sfc} \Big \lvert \geq \epsilon\Big) = P\Big(\frac{n}{L} \Big \lvert \frac{1}{L}  \norm{\madj^{t,\sfc}}^2 - \frac{1}{L} \|  \Madjmat^{\sfc}_t \alphavec^{t,\sfc} \|^2  - \tau^{t}_{\sfc}  + \frac{(\tau^{t}_{\sfc})^2}{\tau^{t-1}_{\sfc}} \Big \lvert \geq \epsilon\Big) \nonumber \\
&\leq P\Big(\frac{n}{L}\Big \lvert \frac{1}{L} \norm{\madj^{t,\sfc}}^2 -  \tau^{t}_{\sfc}  \Big \lvert \geq \frac{\epsilon}{2} \Big) + P\Big(\frac{n}{L} \Big \lvert \frac{1}{L} \|  \Madjmat^{\sfc}_t \alphavec^{t,\sfc} \|^2 - \frac{(\tau^{t}_{\sfc})^2}{\tau^{t-1}_{\sfc}} \Big \lvert \geq \frac{\epsilon}{2}\Big).
\label{eq:Bgt1}
\end{align}
The first term has the desired bound by $\mathcal{B}_{t} (e)$ with $\sfp=1$.  For the second term, using $\madj^{t,\sfc}_{\parallel} = \Madjmat^{\sfc}_t \alphavec^{t,\sfc} =  \sum_{i=0}^{t-1} \alpha^{t,\sfc}_i \madj^{i,\sfc} $, we have 
$\|  \Madjmat^{\sfc}_t \alphavec^{t,\sfc} \|^2 = (\Madjmat^{\sfc}_t \alphavec^{t,\sfc})^*\Madjmat^{\sfc}_t \alphavec^{t,\sfc} =  \sum_{i=0}^{t-1} \alpha^{t,\sfc}_i (\madj^{t,\sfc}_{\parallel})^*\madj^{i,\sfc}   = \sum_{i=0}^{t-1} \alpha^{t,\sfc}_i (\madj^{i,\sfc} )^* \madj^{t,\sfc} .$
Hence, recalling the definition of $\hat{\alpha}^{t, \sfc}$ in \eqref{eq:hatalph_hatgam_def},
\begin{align*}
 &P\Big(\frac{n}{L} \Big \lvert  \frac{1}{L}  \|  \Madjmat^{\sfc}_t \alphavec^{t,\sfc} \|^2 -  \frac{(\tau^{t}_{\sfc})^2}{\tau^{t-1}_{\sfc}} \Big \lvert \geq \frac{\epsilon}{2}\Big) = P\Big(\frac{n}{L} \Big \lvert \sum_{i=0}^{t-1} \Big(\frac{1}{L} \alpha^{t,\sfc}_i  (\madj^{i,\sfc} )^* \madj^{t,\sfc} -  \hat{\alpha}^{t, \sfc}_i \tau^t_{\sfc} \Big)  \Big \lvert \geq \frac{\e}{2}\Big) \nonumber \\
& \leq \sum_{i=0}^{t-1}  P\Big(\frac{n}{L} \Big \lvert \frac{1}{L}  \alpha^{t,\sfc}_i (\madj^{i,\sfc} )^* \madj^{t,\sfc} -   \hat{\alpha}^{t, \sfc}_i \tau^t_{\sfc} \Big \lvert \geq \frac{\e}{2t} \Big)
\overset{(a)}{\leq}    t^6   K K_{t-1}   \PC_{t-1}   \exp\Big\{-\frac{1}{t^{12}} \kappa \kappa_{t-1}     \pc_{t-1}  \e^2\Big\}.
%t^6 K K_{t-1}  (\Lr \Lc)^{t}    \exp\Big\{\frac{-\kappa \kappa_{t-1} \Mr  \omega^{2t+ 1} \Lc^{2t} \e^2}{t^{12} \Lr^{4t} (\log M)^{2t} }\Big\}.
%
\end{align*}
Step $(a)$ is obtained using Lemma \ref{products} when $i=t-1$ and Lemma \ref{products_0} otherwise, along with the results $\mathcal{B}_{t} (e)$ with $\sfp =1$ and $\mathcal{B}_{t} (f)$ proved in \eqref{eq:Bf} above.  
%We note that  $ {n \tau^{t}_{\sfc}}/{L}$ are both $\Theta(1)$.

%-------------------------

\textbf{(g)} We first show \eqref{eq:Msing}.  For $0 \leq \tilde{s},s \leq t$, $[\innerM^{\sfc}_{t+1}]_{\tilde{s}+1,s+1}= \frac{n}{L^2}(\madj^{\tilde{s}, \sfc})^*\madj^{s, \sfc}$ and therefore, $[\textbf{M}^{\sfc}_{t+1} ]_{\tilde{s}+1,s+1}$  concentrates on $\frac{n}{L} \tau^{\max\{\tilde{s}, s\}}_{\sfc}$ by $\mc{B}_{t}(e)$ with $\sfp=1$.    By Fact \ref{fact:eig_proj}, if $\frac{n}{L^2}\norm{\madj^{s, \sfc}_{\perp}}^2 \geq \kappa >0$ for all $0 \leq s \leq t$, then $\innerM^{\sfc}_{t+1}$ is invertible.  Note from  $\mc{B}_{0}(f)- \mc{B}_{t}(f)$ that $\frac{n}{L^2}\norm{\madj^{s, \sfc}_{\perp}}^2$ concentrates on $\frac{n}{L}\tau^{s}_{\perp, \sfc}$, and $\frac{n}{L} \tau^{s}_{\perp, \sfc} > 0$ by Lemma \ref{lem:sigmatperp}. Choosing $\kappa = \frac{1}{2} \min \{2, \frac{n}{L}\tau^{0}_{\perp, \sfc}, \ldots, \frac{n}{L}\tau^{t}_{\perp, \sfc}\}$, we therefore have
\be
\begin{split}
P\Big(\innerM^{\sfc}_t \text{ singular}\Big)   \leq \sum_{s=0}^{t} P\Big(\frac{n}{L}\Big \lvert \frac{1}{L} \norm{\madj^{s, \sfc}_{\perp}}^2 - \tau^s_{\perp, \sfc} \Big \lvert \geq  \kappa \Big) &\leq    t^7   K K_{t-1}   \PC_{t-1}   e^{-\frac{1}{t^{12}} \kappa \kappa_{t-1}     \pc_{t-1}  \e^2},
%t^7 K K_{t-1}  (\Lr \Lc)^{t}    \exp\Big\{\frac{-\kappa \kappa_{t-1} \Mr  \omega^{2t+ 1} \Lc^{2t} }{t^{10} \Lr^{4t} (\log M)^{2t} }\Big\},
%
\end{split}
\label{eq:mr_perp_conc}
\ee
where the second inequality follows from  $\mc{B}_0 (f)-\mc{B}_{t} (f)$.%\RV{to be filled in}.

Next, we show \eqref{eq:Bg}.  We first note that each non-zero element of $(\brCmat^{t+1, \sfc})^{-1}$ is  $\Theta(1)$.  To see this, recall the definition of $\brCmat^{t+1, \sfc}$, from which it follows with work as in \eqref{eq:Cinverse_mat} that if $\brCmat^{t, \sfc}$ is invertible, by the block inversion formula we have
\be
 (\brCmat^{t+1, \sfc})^{-1} = \Big( \begin{array}{cc}
(\brCmat^{t, \sfc})^{-1} + \frac{L}{n} (\tau^{t}_{\perp, \sfc})^{-1} \hat{\alphavec}^{t,\sfc}   (\hat{\alphavec}^{t, c})^*  & - \frac{L}{n}(\tau^{t}_{\perp, \sfc})^{-1} \hat{\alphavec}^{t,\sfc}  \\
-\frac{L}{n} (\tau^{t}_{\perp, \sfc})^{-1}   (\hat{\alphavec}^{t, c})^* &  \frac{L}{n}(\tau^{t}_{\perp, \sfc})^{-1}  \end{array} \Big).
\label{eq:Cinverse_mat_new}
\ee
Each non-zero element of \eqref{eq:Cinverse_mat_new} is $ \in \Theta(1)$ since $\tau^s_{\sfc} \in \Theta(L/n)$, for $0 \leq s \leq T$, and $\tau^{t}_{\perp, \sfc}$ is bounded below (by Lemma \ref{lem:sigmatperp}).    
We can similarly represent $ [\innerM^{\sfc}_{t+1}]^{-1}$ by block inversion.  Noting that
\be
\innerM^{\sfc}_{t+1} = \frac{n}{L^2} \Big( \begin{array}{cc}
\frac{n}{L^2} \innerM^{\sfc}_{t} & (\Madjmat^{\sfc}_{t})^* \madj^{t, \sfc} \\
 ((\Madjmat^{\sfc}_{t})^* \madj^{t, \sfc})^* & \norm{\madj^{t, \sfc}}^2 \end{array} \Big),
 \label{eq:Mblocks}
 \ee
 if $\innerM^{\sfc}_{t}$ is invertible, by the block inversion formula we have
\be
 (\innerM^{\sfc}_{t+1})^{-1} = \Big( \begin{array}{cc}
(\innerM^{\sfc}_{t})^{-1} + \frac{n}{L^2} \| \madj^{t, \sfc}_{\perp} \|^{-2} \alphavec^{t, \sfc}(\alphavec^{t, \sfc})^* & -\frac{n}{L^2} \| \madj^{t, \sfc}_{\perp} \|^{-2}  \alphavec^{t, \sfc} \\
-\frac{n}{L^2}\| \madj^{t, \sfc}_{\perp} \|^{-2}  (\alphavec^{t, \sfc})^* &\frac{n}{L^2}\| \madj^{t, \sfc}_{\perp} \|^{-2}  \end{array} \Big),
\label{eq:Mt1_inverse0}
\ee
where we have used $\alphavec^{t, \sfc} = \frac{n}{L^2} (\innerM^{\sfc}_{t})^{-1} (\Madjmat^{\sfc}_{t})^* \madj^{t, \sfc}$ and $((\Madjmat^{\sfc}_{t})^* \madj^{t, \sfc})^* \alphavec^{t, \sfc} = (\madj^{t, \sfc})^* \madj^{t, \sfc}_{\parallel}$.  

In what follows we prove concentration for each of the elements in \eqref{eq:Mt1_inverse0} to the corresponding element of \eqref{eq:Cinverse_mat_new}. First, by $\mathcal{B}_{t} (f)$ and Lemma \ref{inverses},
\be
P\Big(\Big  \lvert \frac{L^2}{n \| \madj^{t,\sfc}_{\perp} \|^2} - \frac{L}{n}( \tau^{t}_{\perp, \sfc})^{-1} \Big \lvert \geq \e \Big)  \leq   t^6  K K_{t-1}  \PC_{t-1}    \exp\{-\frac{1}{t^{12}}  \kappa \kappa_{t-1} \pc_{t-1} \e^2\}.
%t^6 K K_{t-1}  (\Lr \Lc)^{t}    \exp\Big\{\frac{-\kappa \kappa_{t-1} \Mr  \omega^{2 t+ 1} \Lc^{2t} \e^2}{t^{12} \Lr^{4t} (\log M)^{2t} }\Big\},
\label{eq:ar_conc0_v2}
\ee
Next, consider the $i^{th}$ element of $-\frac{L^2}{n}  \| \madj^{t, \sfc}_{\perp} \|^{-2}  \alphavec^{t, \sfc} $.  For $i \in [t]$,
\begin{align}
 P \Big(\Big \lvert \frac{L^2 \alpha^{t, \sfc}_{i-1} }{n \| \madj^{t, \sfc}_{\perp}\|^2} - \frac{L \hat{\alpha}^{t, \sfc}_{i-1}}{n \tau^{t}_{\perp,  \sfc}}  \Big \lvert \geq \e\Big)  \leq  
t^6  K K_{t-1}  \PC_{t-1}    \exp\{-\frac{1}{t^{12}}  \kappa \kappa_{t-1} \pc_{t-1} \e^2\}.
 %t^6 K K_{t-1}  (\Lr \Lc)^{t}    \exp\Big\{\frac{-\kappa \kappa_{t-1} \Mr  \omega^{2 t+ 1} \Lc^{2t} \e^2}{t^{12} \Lr^{4t} (\log M)^{2t} }\Big\}. 
 %
 \label{eq:matlem_bound1}
\end{align}
The above follows from \eqref{eq:ar_conc0_v2}, Lemma \ref{products} for $i=t$ or Lemma \ref{products_0} when $i \in [t-1]$, and $\mc{B}_{t} (f)$.  Finally,  consider element $(i,j)$ of $(\innerM^{\sfc}_{t})^{-1} + (L^2 /n) \| \madj^{t, \sfc}_{\perp} \|^{-2} \alphavec^{t, \sfc}(\alphavec^{t, \sfc})^*$ for $i,j \in [t]$.
\begin{align}
P &\Big(\Big \lvert [(\innerM^{\sfc}_{t})^{-1}]_{i j} + \frac{L^2}{n} 
\| \madj^{t, \sfc}_{\perp} \|^{-2} \alpha^{t, \sfc}_{i-1} \alpha^{t, \sfc}_{j-1} - [(\brCmat^{t, \sfc})^{-1}]_{i j} -\frac{L}{n}(\tau^{t}_{\perp, \sfc})^{-2} \hat{\alpha}^{t, \sfc}_{i-1} \hat{\alpha}^{t, \sfc}_{j-1} \Big \lvert \geq \e\Big) \nonumber \\
&\overset{(a)}{\leq} P \Big(\Big \lvert [(\innerM^{\sfc}_{t})^{-1}]_{ij}  - [(\brCmat^{t, \sfc})^{-1}]_{i j} \Big \lvert \geq \frac{\e}{2}\Big) + P \Big(\Big \lvert \alpha^{t, \sfc}_{j-1} -  \hat{\alpha}^{t, \sfc}_{j-1} \Big \lvert \geq \frac{\e'}{2}\Big) + P \Big(\Big \lvert \frac{L^2 \alpha^{t, \sfc}_{i-1} }{n \| \madj^{t, \sfc}_{\perp} \|^{2}}  - \frac{ L \hat{\alpha}^{t, \sfc}_{i-1}}{n \tau^{t}_{\perp, \sfc}}  \Big \lvert \geq \frac{\e'}{2}\Big) \nonumber \\
&\overset{(b)}{\leq}  t^6  K K_{t-1}  \PC_{t-1}    \exp\{-\frac{1}{t^{12}}  \kappa \kappa_{t-1} \pc_{t-1} \e^2\}. \label{eq:Minverse_conc}
%t^6 K K_{t-1}  (\Lr \Lc)^{t}    \exp\Big\{\frac{-\kappa \kappa_{t-1} \Mr  \omega^{2 t+ 1} \Lc^{2t} \e^2}{t^{12} \Lr^{4t} (\log M)^{2t} }\Big\}.
%
\end{align}
Step $(a)$ follows from Lemma \ref{sums} and Lemma \ref{products}/\ref{products_0} with
$
\e' = \min (\sqrt{\e/3},  \, \frac{n \tau^{t}_{\perp, \sfc}}{ 3 L \hat{\alpha}^{t, \sfc}_{i-1}},  \, \frac{\e}{3  \hat{\alpha}^{t, \sfc}_{j-1}}).
$
Step $(b)$ follows from the inductive hypothesis $\mc{B}_{t-1}(g)$ (Eq.\ \eqref{eq:Bg}), $\mc{B}_{t}(f)$, and \eqref{eq:matlem_bound1}.

%%
%%%----------------------------------------------------------------------------------------------------------------------------------------
%%
%%

\subsection{Step 4: Showing $\mathcal{H}_{t+1}$ holds} \label{subsub:step4}

We prove the statements in $\mc{H}_{t+1}$ assuming that $\mc{B}_{0}, \ldots, \mc{B}_{t}$, and $\mc{H}_1, \ldots, \mc{H}_t$ hold due to the induction hypothesis.  We begin with a lemma that is used to prove $\mc{H}_{t+1} (a)$.  The lemma as well as other parts of $\mc{H}_{t+1}$ assume the invertibility of $\innerQ^{\sfr}_{1}, \ldots, \innerQ^{\sfr}_{t}$, but for the  sake of brevity, we do not explicitly specify the conditioning.

\begin{lem}
\label{lem:Qv_conc}
Let $\mathbf{v}^{\sfr, \sfc} := \frac{1}{L} (\Ymat_{t+1, \sfr} )^* \madj^{t, \sfc}_{\perp, \sfr}$ and $\innerQ^{\sfr}_{t+1} =  \frac{1}{L} (\Qadjmat^{\sfr}_{t+1})^* \Qadjmat^{\sfr}_{t+1}$.   If $\innerQ^{\sfr}_{1} , \ldots, \innerQ^{\sfr}_{t+1}$ are invertible, for $0 \leq j \leq t-1$,
\begin{align}
\label{eq:Hconc_j} 
&P\Big(\Big \lvert  \sum_{\sfr \in [\Lr]} \sqrt{W_{\sfr \sfc}} [(\innerQ^{\sfr}_{t+1})^{-1} \mathbf{v}^{\sfr, \sfc}]_{j+1} + \alpha^{t,\sfc}_{j} \Big \lvert \geq \e \Big) \leq   t^2 K K'_{t-1}   \PC_{t-1}   \exp\{-\frac{1}{t^4} \kappa \kappa'_{t-1}  \pc_{t-1}  \e^2\}  \\
%t^2 K K'_{t-1}  (\Lr \Lc)^{t}    \exp\Big\{\frac{-\kappa \kappa'_{t-1} \Mr  \omega^{2t+1} \Lc^{2t} \e^2}{t^4\Lr^{4t} (\log M)^{2t} }\Big\},
%
&P\Big(\Big \lvert \sum_{\sfr \in [\Lr]} \sqrt{W_{\sfr \sfc}} [(\innerQ^{\sfr}_{t+1})^{-1} \mathbf{v}^{\sfr, \sfc}]_{t+1}   - 1 \Big \lvert \geq \e \Big) \leq  t^2 K K'_{t-1}   \PC_{t-1}   \exp\{-\frac{1}{t^4} \kappa \kappa'_{t-1}  \pc_{t-1}  \e^2\}.
% t^2 K K'_{t-1}  (\Lr \Lc)^{t}    \exp\Big\{\frac{-\kappa \kappa'_{t-1} \Mr  \omega^{2t+1} \Lc^{2t} \e^2}{t^4\Lr^{4t} (\log M)^{2t} }\Big\},
%
\label{eq:Hconc_t}
\end{align}
\end{lem}
\begin{proof}
By induction hypothesis $\mathcal{H}_t (g)$, element-wise $(\innerQ^{\sfr}_{t+1})^{-1}$  concentrates to $(\tCmat^{t+1, \sfr})^{-1}$, so before proving \eqref{eq:Hconc_j} and  \eqref{eq:Hconc_t}, we prove intermediary results about element-wise concentration of $\mathbf{v}^{\sfr,\sfc}$.  For $0\leq k \leq t$, define
\be
\begin{split}
\mathbb{E}_{k+1, \sfr}  &=  \frac{\Mr  \sqrt{W_{\sfr \sfc}}}{L} \Big(S^t_{\sfr \sfc} \sigma^{t}_{\sfr}  -  \hat{\alpha}^{t,\sfc}_{t-1}  S^{t-1}_{\sfr \sfc}  \sigma^{\max(k,t-1)}_{\sfr} \Big) =  \frac{\Mr \tau^t_{\sfc}  \sqrt{W_{\sfr \sfc}}}{L}  \Big(\frac{\sigma^{t}_{\sfr}}{\phi^t_{\sfr}} - \frac{\sigma^{\max(k,t-1)}_{\sfr}}{\phi^{t-1}_{\sfr}}\Big).
\label{eq:Ekr_def}
\end{split}
\ee
For $\textsf{B} > 0$, a universal constant, we will prove that
\begin{align}
P\Big(\Big \lvert   \sum_{\sfr \in [\Lr]}  \sqrt{W_{\sfr \sfc}} ( \tv^{\sfr,\sfc}_{k+1} -  \mathbb{E}_{k+1, \sfr}) \Big \lvert \geq \e \Big) &\leq  t K K'_{t-1}   \PC_{t-1}   \exp\{-\frac{1}{t^2} \kappa \kappa'_{t-1}  \pc_{t-1}  \e^2\},
\label{eq:term1} \\
P\Big(  \sum_{\sfr \in [\Lr]} \sqrt{W_{\sfr \sfc}} \abs{\tv^{\sfr,\sfc}_{k+1}} \geq \textsf{B} \Big) &\leq
 t K K'_{t-1}   \PC_{t-1}   \exp\{-\frac{1}{t^2} \kappa \kappa'_{t-1} \pc_{t-1}  \}.
\label{eq:term2}
\end{align}

We prove the lemma using \eqref{eq:term1} and \eqref{eq:term2}, and then prove  \eqref{eq:term1}. We first claim that
\begin{align}
\label{eq:E_simplification_2_v1}
& \sum_{\sfr \in [\Lr]}  \sqrt{W_{\sfr \sfc}}  \sum_{k=1}^{t+1} [(\tCmat^{t+1, \sfr})^{-1}]_{(j+1), k} \, \mathbb{E}_{k, \sfr} = -\hat{\alpha}^{t,\sfc}_{t-1}, \qquad  0 \leq j \leq (t-1), \\
&\sum_{\sfr \in [\Lr]}  \sqrt{W_{\sfr \sfc}}  \sum_{k=1}^{t+1} [(\tCmat^{t+1, \sfr})^{-1}]_{(t+1), k} \, \mathbb{E}_{k, \sfr} = 1.
\label{eq:E_simplification_v1}
\end{align}
To show \eqref{eq:E_simplification_v1}, noting from \eqref{eq:Ekr_def} that $\mathbb{E}_{1, \sfr} = \ldots = \mathbb{E}_{t, \sfr}$, we have
\begin{align}
&\sum_{\sfr \in [\Lr]}  \sqrt{W_{\sfr \sfc}}  \sum_{k=1}^{t+1} [(\tCmat^{t+1, \sfr})^{-1}]_{(t+1) k} \mathbb{E}_{k, \sfr} = \sum_{\sfr \in [\Lr]}  \sqrt{W_{\sfr \sfc}}  \Big[ \mathbb{E}_{t+1, \sfr}  [(\tCmat^{t+1, \sfr})^{-1}]_{(t+1) (t+1)} + \mathbb{E}_{t, \sfr} \sum_{k=1}^{t} [(\tCmat^{t+1, \sfr})^{-1}]_{(t+1) k}\Big] \nonumber \\
&\overset{(a)}{=} \sum_{\sfr \in [\Lr]} \frac{ \sqrt{W_{\sfr \sfc}} }{\sigma^{t}_{\perp, \sfr}} \Big[ \mathbb{E}_{t+1, \sfr} -\mathbb{E}_{t, \sfr}   \hat{\gamma}^{t, \sfr}_{t-1}  \Big]=  \sum_{\sfr \in [\Lr]}  \frac{ \sqrt{W_{\sfr \sfc}} }{\sigma^{t}_{\perp, \sfr}} \Big[ \mathbb{E}_{t+1, \sfr} - \mathbb{E}_{t, \sfr}  \hat{\gamma}^{t, \sfr}_{t-1} \Big] \overset{(b)}{=} \frac{\Mr \tau^t_{\sfc}}{L}  \sum_{\sfr \in [\Lr]}  \frac{W_{\sfr \sfc} }{ \phi^t_{\sfr}} = 1.  \label{eq:E_simplification}
\end{align}
In the above, step $(a)$ is obtained as follows using the definition of $\tCmat^{t+1, \sfr}$ in \eqref{eq:Ct_def}.  Using block inversion as we did for $\brCmat^{t, \sfc}$ in  \eqref{eq:Cinverse_mat}, it is straightforward to show that if $\tCmat^{t, \sfr}$ is invertible,  
\be
 (\tCmat^{t+1, \sfr})^{-1} = \Big( \begin{array}{cc}
(\tCmat^{t, \sfr})^{-1} +  (\sigma^{t}_{\perp, \sfr})^{-1} \hat{\gammavec}^{t,\sfr}   (\hat{\gammavec}^{t, \sfr})^*  & - (\sigma^{t}_{\perp, \sfr})^{-1} \hat{\gammavec}^{t,\sfr}  \\
-(\sigma^{t}_{\perp, \sfr})^{-1}   (\hat{\gammavec}^{t, \sfr})^* &  (\sigma^{t}_{\perp, \sfr})^{-1}  \end{array} \Big).
\label{eq:tildeCinverse_mat}
\ee
Step  $(a)$ then follows since $\hat{\gamma}^{t, \sfr}_{0}  = \ldots = \hat{\gamma}^{t, \sfr}_{t-2} =0$.
Step $(b)$ of \eqref{eq:E_simplification} then follows by using the definitions of $\mathbb{E}_{t+1, \sfr}, \mathbb{E}_{t, \sfr}$ in \eqref{eq:Ekr_def} and  of $\hat{\gamma}^{t, \sfr}_{t-1}$ in \eqref{eq:hatalph_hatgam_def}. The result in \eqref{eq:E_simplification_2_v1} can be shown similarly, using
that for $0 \leq j \leq (t-1)$, 
$\sum_{k=1}^{t} [(\tCmat^{t, \sfr})^{-1}]_{(j+1) k} =  [(\tCmat^{t, \sfr})^{-1}  (1, \ldots, 1)^*]_{j+1} =(\sigma^t_{\sfr})^{-1} \hat{\gamma}^{t,\sfr}_{j},$ where the last equality follows from the definition of $\hat{\gamma}^{t,\sfr}$ in \eqref{eq:hatalph_hatgam_def}.

We now use these results to prove \eqref{eq:Hconc_t}. Using \eqref{eq:E_simplification_v1} and Lemma \ref{sums}, the LHS of \eqref{eq:Hconc_t} can be bounded as follows:
\begin{align}
&P\Big(\Big \lvert \sum_{\sfr \in [\Lr]} \sqrt{W_{\sfr \sfc}}  [(\innerQ^{\sfr}_{t+1})^{-1} \mathbf{v}^{\sfr, \sfc}]_{t+1}   - 1 \Big \lvert \geq \e \Big) = P\Big(\Big \lvert \sum_{\sfr \in [\Lr]} \sqrt{W_{\sfr \sfc}} \sum_{k=1}^{t+1} [(\innerQ^{\sfr}_{t+1})^{-1}]_{(t+1) k} \tv^{\sfr,\sfc}_{k}   - 1\Big \lvert \geq \e \Big) \nonumber \\
&= P\Big(\Big \lvert \sum_{\sfr \in [\Lr]}  \sqrt{W_{\sfr \sfc}}  \sum_{k=1}^{t+1}   \Big[ [(\innerQ^{\sfr}_{t+1})^{-1}]_{(t+1) k} \tv^{\sfr,\sfc}_{k} - [(\tCmat^{t+1, \sfr})^{-1}]_{(t+1) k} \, \mathbb{E}_{k, \sfr} \Big]\Big \lvert  \geq \e \Big) \nonumber \\
&\leq  \sum_{k=1}^{t+1}  P\Big(\Big \lvert \sum_{\sfr \in [\Lr]} \sqrt{W_{\sfr \sfc}} [(\tCmat^{t+1, \sfr})^{-1}]_{(t+1) k}    [ \tv^{\sfr,\sfc}_{k} - \mathbb{E}_{k, \sfr}]\Big \lvert  \geq \frac{\e}{2(t+1)} \Big) \label{eq:toprove_num2}  \\
& \qquad +  \sum_{k=1}^{t+1}  P\Big(\Big \lvert \sum_{\sfr \in [\Lr]}  \sqrt{W_{\sfr \sfc}} \tv^{\sfr,\sfc}_{k} \Big[[(\innerQ^{\sfr}_{t+1})^{-1}]_{(t+1) k} -  [(\tCmat^{t+1, \sfr})^{-1}]_{(t+1) k}\Big]\Big \lvert  \geq \frac{\e}{2(t+1)} \Big). \nonumber
\end{align}
Label the terms on the RHS of \eqref{eq:toprove_num2} as $T_1$ and $T_2$.  Note that term $T_1$ has the desired upper bound using \eqref{eq:term1} and the fact that the non-zero elements of $(\tCmat^{t+1, \sfr})^{-1} \in \Theta(1)$.   Now consider  $T_2$.
\begin{align}
&T_2 \leq  \sum_{k=1}^{t+1}  P\Big(\sum_{\sfr \in [\Lr]}  \sqrt{W_{\sfr \sfc}} \abs{\tv^{\sfr,\sfc}_{k}}  \Big \lvert [(\innerQ^{\sfr}_{t+1})^{-1}]_{(t+1) k} -  [(\tCmat^{t+1, \sfr})^{-1}]_{(t+1) k} \Big \lvert  \geq \frac{\e}{2(t+1)} \Big) \nonumber \\
&\leq \sum_{k=1}^{t+1}  \Big[P\Big( \cup_{\sfr \in [\Lr]}  \Big \lvert [(\innerQ^{\sfr}_{t+1})^{-1}]_{(t+1) k} -  [(\tCmat^{t+1, \sfr})^{-1}]_{(t+1) k} \Big \lvert \leq \frac{\kappa \e}{2t \textsf{B}}  \Big) + P\Big( \sum_{\sfr \in [\Lr]}  \sqrt{W_{\sfr \sfc}}  \abs{\tv^{\sfr,\sfc}_{k}}  \geq  \textsf{B} \Big) \Big]   \nonumber \\
& \leq \sum_{k=1}^{t+1}   \Big[  \sum_{\sfr \in [\Lr]} P\Big(  \Big \lvert [(\innerQ^{\sfr}_{t+1})^{-1}]_{(t+1) k} -[(\tCmat^{t+1, \sfr})^{-1}]_{(t+1) k}  \Big \lvert \geq \frac{\kappa \e}{2t \textsf{B}} \Big)
+ P\Big( \sum_{\sfr \in [\Lr]}  \sqrt{W_{\sfr \sfc}}  \abs{\tv^{\sfr,\sfc}_{k}}  \geq  \textsf{B} \Big) \Big]  \nonumber \\
& \leq  t \Lr K K'_{t-1}   \PC'_{t-2}   \exp\Big\{-\frac{ \kappa \kappa'_{t-1} (\omega/\Lr) ^2\pc'_{t-2}  \e^2}{t^2(\log M)^2}\Big\}  +  t^2 K K'_{t-1}   \PC_{t-1}   \exp\{-\frac{1}{t^2} \kappa \kappa'_{t-1}  \pc_{t-1} \}. \label{eq:Qv_termT2_bound}
\end{align}
The upper bound for the two terms of \eqref{eq:Qv_termT2_bound} follows by \eqref{eq:term2} and $\mathcal{H}_{t} (g)$.  The bound in \eqref{eq:Hconc_t} follows since $\Lr  \PC'_{t-2} =  \PC_{t-1}$ and  $ (\omega/\Lr) ^2\pc'_{t-2}  (\log M)^2 = \pc_{t-1}.$
%%

%%%%%%%%%%%%%%%%%%%%%%%%%%%%%%%%%%%%%%%%%%%%%

We now show the bound in \eqref{eq:Hconc_j} for $0 \leq j \leq t-1$.   By Lemma \ref{sums},
\begin{align*}
&P\Big(\Big \lvert  \sum_{\sfr \in [\Lr]} \sqrt{W_{\sfr \sfc}} [(\innerQ^{\sfr}_{t+1})^{-1} \mathbf{v}^{\sfr, \sfc}]_{j+1} + \alpha^{t,\sfc}_{j} \Big \lvert \geq \e \Big) \\
&\leq P\Big(\lvert \alpha^{t, \sfc}_{j} - \hat{\alpha}^{t, \sfc}_{j} \lvert  \geq \frac{\e}{2} \Big) + P\Big(\Big \lvert  \sum_{\sfr \in [\Lr]} \sqrt{W_{\sfr \sfc}} [(\innerQ^{\sfr}_{t+1})^{-1} \mathbf{v}^{\sfr, \sfc}]_{j+1} + \hat{\alpha}^{t,\sfc}_{j} \Big \lvert \geq \frac{\e}{2} \Big).
\end{align*}
The first term of the above is upper bounded with $\mathcal{B}_t (f)$.  Using \eqref{eq:E_simplification_2_v1},  the second term above is
\begin{align}
&P\Big(\Big \lvert  \sum_{\sfr \in [\Lr]} \sqrt{W_{\sfr \sfc}} [(\innerQ^{\sfr}_{t+1})^{-1} \mathbf{v}^{\sfr, \sfc}]_{j+1} + \hat{\alpha}^{t,\sfc}_{j} \Big \lvert \geq \frac{\e}{2} \Big) \nonumber  \\
&= P\Big(\Big \lvert  \sum_{\sfr \in [\Lr]} \sqrt{W_{\sfr \sfc}} \sum_{k=1}^{t+1}\Big( [(\innerQ^{\sfr}_{t+1})^{-1}]_{(j+1), k} \tv^{\sfr, \sfc}_{k}-[(\tCmat^{t+1, \sfr})^{-1}]_{(j+1) k} \, \mathbb{E}_{k, \sfr} \Big) \Big \lvert \geq \frac{\e}{2} \Big).
\label{eq:QV_toprove_num3}
\end{align}
Now we can bound the term in \eqref{eq:QV_toprove_num3} with work similar to that used  for  \eqref{eq:toprove_num2}.

%%%%%%%%%%%%%%%%%%%%%%%%%%%%%%%%%%%%%%%%%%%%%

To complete the proof, it remains to show \eqref{eq:term1}, \eqref{eq:term2}. Using the definition of $\Ymat_{t+1}$ in \eqref{eq:Ytr_Xtc}, 
\be
\sqrt{W_{\sfr \sfc}} \tv^{\sfr,\sfc}_{k+1} = \frac{\sqrt{W_{\sfr \sfc}}  (\bvec^{k}_{\sfr} + \upsilon^{k}_\sfr  \brmvec^{k-1}_{\sfr})^*\, \madj^{t, \sfc}_{\perp, \sfr}}{L} =\Big(\frac{\sqrt{W_{\sfr \sfc}} \bvec^{k}_{\sfr}}{L} + \frac{ \upsilon^{k}_\sfr  \madj^{k-1, \sfc}_{\sfr} }{L S^{k-1}_{\sfr \sfc}}\Big)^* \, (\madj^{t, \sfc}_{\sfr} - \sum_{i=0}^{t-1} \alpha^{t,\sfc}_i \madj_{\sfr}^{i, \sfc}),
\label{eq:vk_def2}
\ee
where we have used $ \upsilon^{k}_\sfr= \sigma^k_\sfr/\phi^{k-1}_\sfr$ and $ \madj^{t, \sfc}_{\perp, \sfr} =  \madj^{t, \sfc}_{\sfr} - \madj^{t, \sfc}_{\parallel, \sfr} = \madj^{t, \sfc}_{\sfr} - \sum_{i=0}^{t-1} \alpha^{t,\sfc}_i \madj_{\sfr}^{i, \sfc}$. 

 We first prove \eqref{eq:term1}.  Using \eqref{eq:vk_def2}, we have the following bound for the probability in \eqref{eq:term1}:
\begin{align}
&P\Big( \Big  \lvert  \sum_{\sfr \in [\Lr]} \sqrt{W_{\sfr \sfc}}  ( \tv^{\sfr,\sfc}_{k+1} -  \mathbb{E}_{k+1, \sfr} )\Big \lvert \geq \e \Big)  \nonumber \\ 
&\overset{(a)}{\leq}  P\Big ( \Big \lvert \sum_{\sfr \in [\Lr]} \frac{ \sqrt{W_{\sfr \sfc}} (\bvec^{k}_{\sfr})^* (\madj^{t, \sfc}_{\sfr} - \sum_{i=0}^{t-1} \alpha^{t,\sfc}_i \madj_{\sfr}^{i, \sfc}) }{L} -   \frac{\Mr}{L} \sum_{\sfr \in [\Lr]}  W_{\sfr \sfc}  (S^t_{\sfr \sfc}  \sigma^{t}_{\sfr}  - \hat{\alpha}^{t,\sfc}_{t-1} S^{t-1}_{\sfr \sfc}    \sigma^{\max(k,t-1)}_{\sfr} ) \Big \lvert \geq \frac{\e}{2} \Big ) \nonumber \\
& \quad + P\Big ( \Big  \lvert \sum_{\sfr \in [\Lr]}  \frac{ \upsilon^{k}_{\sfr} (  \madj^{k-1, \sfc}_{\sfr})^*\, (\madj^{t, \sfc}_{\sfr} - \sum_{i=0}^{t-1} \alpha^{t,\sfc}_i \madj_{\sfr}^{i, \sfc}) }{L S^{k-1}_{\sfr \sfc}}  \Big \lvert \geq \frac{\e}{2} \Big ). \label{eq:Qv_conc_eq1}
\end{align}
The terms on the RHS of \eqref{eq:Qv_conc_eq1} are bounded via the inductive hypotheses $\mathcal{B}_t (d)$ and $\mathcal{B}_t (e)$ with $\sfp = 1$ and $\mathcal{B}_t (f)$, similarly to the first term on the RHS of \eqref{eq:v_conc_eq1} (see \eqref{eq:v_conc_eq1_bound1}).

%%%%%%%%%%%%%%%%%%%%%%%%%%%%%%%%%%%%%%%%%%%%%

We now prove \eqref{eq:term2} using \eqref{eq:vk_def} and $\upsilon^{k}_{\sfr}/S^{k-1}_{\sfr \sfc} = \sigma^{k}_{\sfr}/\tau^{k-1}_{\sfc}$. Recall that $\tau^{k-1}_{\sfc} \in \Theta(L/n)$, then
\begin{align}
&P\Big(  \sum_{\sfr \in [\Lr]} \sqrt{W_{\sfr \sfc}} \abs{v^{\sfr,\sfc}_{k+1}} \geq \textsf{B} \Big) =  P\Big(  \sum_{\sfr \in [\Lr]} \Big \lvert \Big(\frac{\sqrt{W_{\sfr \sfc}} \bvec^{k}_{\sfr}}{L} + \frac{ \upsilon^{k}_r  \madj^{k-1, \sfc}_{\sfr} }{L S^{k-1}_{\sfr \sfc}}\Big)^* \, (\madj^{t, \sfc}_{\sfr} - \sum_{i=0}^{t-1} \alpha^{t,\sfc}_i \madj_{\sfr}^{i, \sfc}) \Big \lvert \geq \textsf{B} \Big) \nonumber \\
&\leq  P\Big ( \sum_{\sfr \in [\Lr]} \Big \lvert   \frac{ 1  }{L } \sqrt{W_{\sfr \sfc}} ( \bvec^{k}_{\sfr})^*\, \madj^{t, \sfc}_{\sfr} \Big \lvert \geq \frac{\textsf{B}}{4} \Big ) +   P\Big ( \sum_{i=0}^{t-1}  \abs{ \alpha^{t,\sfc}_i }  \sum_{\sfr \in [\Lr]}  \Big \lvert \frac{ 1}{L }\sqrt{W_{\sfr \sfc}} ( \bvec^{k}_{\sfr})^*\,\madj_{\sfr}^{i, \sfc}   \Big \lvert \geq \frac{\textsf{B}}{4} \Big ) \nonumber \\
& \quad + P\Big (  \frac{n}{L^2} \sum_{\sfr \in [\Lr]} \Big  \lvert   ( \madj^{k-1, \sfc}_{\sfr})^*\,  \madj_{\sfr}^{i, \sfc} \Big \lvert \geq \frac{\textsf{B}}{4} \Big ) + P\Big ( \frac{n}{L^2}  \sum_{i=0}^{t-1}  \abs{\alpha^{t,\sfc}_i } \sum_{\sfr \in [\Lr]}  \Big  \lvert  ( \madj^{k-1, \sfc}_{\sfr})^*\,  \madj_{\sfr}^{i, \sfc} \Big \lvert \geq \frac{\textsf{B}}{4} \Big). \label{eq:Qv_conc_eqA}
\end{align}
 We now provide upper bounds for the terms in \eqref{eq:Qv_conc_eqA}, labelled $T_1 - T_4$, by taking 
\be
\textsf{B} = 4 \max\Big\{%\frac{n}{L} \tau^{0}_{\sfc} + 1, 
(\frac{n}{L} \tau^{0}_{\sfc} + 1)(\hat{\alpha}^{t,\sfc}_{t-1} + 1), \, 
(\frac{n}{L} \tau^{0}_{\sfc} + 1)^{1/2}(\frac{1}{\Lr} \sum_{\sfr \in [\Lr]} W_{\sfr \sfc} \sigma^t_{\sfr} + 1)^{1/2}\Big\}.
\ee   
First consider $T_3$. For any $0 \leq j \leq t$,  using Cauchy-Schwarz 
$\sum_{\sfr \in [\Lr]}  \lvert  (\madj^{k-1, \sfc}_{\sfr})^*\, \madj^{j, \sfc}_{\sfr}  \lvert \leq   \| \madj^{k-1, \sfc} \|  \| \madj^{j, \sfc} \|$, and  therefore using the fact that $\textsf{B}^2 \geq 16 (\frac{n}{L} \tau^{0}_{\sfc} + 1)^2 \geq  16 (\frac{n}{L} \tau^{k-1}_{\sfc} + 1) (\frac{n}{L} \tau^j_{\sfc} + 1)$, we have
\begin{align}
T_3 &  \leq P\Big ( \frac{n }{L^2 }  \| \madj^{k-1, \sfc} \|  \norm{\madj^{j, \sfc}} \geq \frac{\textsf{B}}{4} \Big ) %&= P\Big ( \frac{n }{L^2 }  \| \madj^{k-1, \sfc} \|^2 \cdot  \frac{n }{L^2 } \norm{\madj^{j, \sfc}}^2 \geq \frac{\upsilon^2}{16} \Big ) \\
 \leq  P\Big ( \frac{n }{L^2 } \| \madj^{k-1, \sfc} \|^2 \geq \frac{n}{L} \tau^{k-1}_{\sfc} + 1 \Big ) + P\Big (\frac{n }{L^2 }\norm{\madj^{j, \sfc}}^2 \geq \frac{n}{L} \tau^j_{\sfc} + 1 \Big )  \nonumber \\
& \leq P\Big ( \frac{n }{L}  \abs{ \frac{1}{L} \| \madj^{k-1, \sfc} \|^2 - \tau^{k-1}_{\sfc}} \geq  1 \Big ) + P\Big ( \frac{n }{L}  \abs{ \frac{1}{L}\| \madj^{j, \sfc} \|^2 - \tau^{j}_{\sfc}} \geq 1 \Big ).
\label{eq:T3_bound}
\end{align}
This has the desired bound by $\mathcal{B}_t (e)$ with $\sfp =1$.  Next we bound  $T_4$ as follows using  \eqref{eq:T3_bound} and the fact that $\textsf{B} > 4(\frac{n}{L} \tau^{0}_{\sfc} + 1)(\hat{\alpha}^{t,\sfc}_{t-1} + 1) > 4(\hat{\alpha}^{t,\sfc}_{t-1} + 1)\sqrt{(\frac{n}{L} \tau^{k-1}_{\sfc} + 1)(\frac{n}{L} \tau^j_{\sfc} + 1)}$:
\begin{align*}
&T_4 \leq P\Big ( \frac{n }{L^2 }  \| \madj^{k-1, \sfc} \| \norm{\madj^{j, \sfc}}  \sum_{i=0}^{t-1}  \abs{\alpha^{t,\sfc}_i }  \geq \frac{\textsf{B}}{4} \Big ) \\
& \leq P\Big (\frac{n }{L^2 }  \| \madj^{k-1, \sfc} \| \norm{\madj^{j, \sfc}}  \geq \sqrt{(\frac{n}{L} \tau^{k-1}_{\sfc} + 1)(\frac{n}{L} \tau^j_{\sfc} + 1)}\Big )  + P\Big (  \sum_{i=0}^{t-1}  \abs{\alpha^{t,\sfc}_i } \geq   \sum_{i=0}^{t-1}  \abs{\hat{\alpha}^{t,\sfc}_i } + 1\Big ). 
\end{align*}
The first term in the above can be bounded as in \eqref{eq:T3_bound} and the second as follows by $\mathcal{B}_t (f)$:
\begin{align}
&P\Big (  \sum_{i=0}^{t-1}  \abs{\alpha^{t,\sfc}_i } \geq   \sum_{i=0}^{t-1}  \abs{\hat{\alpha}^{t,\sfc}_i } + 1\Big )\leq  \sum_{i=0}^{t-1}  P\Big ( \lvert \alpha^{t,\sfc}_i - \hat{\alpha}^{t,\sfc}_i \lvert \geq  \frac{1}{t} \Big ) \leq t K K'_{t-1}   \PC_{t-1}   e^{-\frac{1}{t^2} \kappa \kappa'_{t-1}  \pc_{t-1}}.
\label{eq:T2alpha_bound}
\end{align}

Now we study term $T_1$.  Using Cauchy-Schwarz, for any $0 \leq j \leq t$,
\begin{align}
 \sum_{\sfr \in [\Lr]}  \frac{ 1  }{L } \sqrt{W_{\sfr \sfc}} \lvert   ( \bvec^{k}_{\sfr})^*\, \madj^{j, \sfc}_{\sfr} \lvert \leq  \sum_{\sfr \in [\Lr]}  \sqrt{\frac{W_{\sfr \sfc}}{n}}  \| \bvec^{k}_{\sfr}\| \cdot  \frac{ \sqrt{n} }{L }   \norm{\madj^{j, \sfc}_{\sfr}}  \leq   \sqrt{\frac{1}{n}\sum_{\sfr \in [\Lr]} W_{\sfr \sfc} \| \bvec^{k}_{\sfr}\|^2} \cdot  \frac{ \sqrt{n}}{L}  \norm{\madj^{j, \sfc}} .
\label{eq:bmadj_bound}
\end{align}
Using the above in term $T_1$ we have,
\begin{align*}
&T_1 \leq  P\Big ( \frac{1}{n}\sum_{\sfr \in [\Lr]} W_{\sfr \sfc} \| \bvec^{k}_{\sfr}\|^2 \cdot  \frac{n}{L^2}  \norm{\madj^{j, \sfc}}^2 \geq \frac{\textsf{B}^2}{16} \Big )  \\
& \stackrel{(a)}{\leq} P\Big ( \frac{1}{n}\sum_{\sfr \in [\Lr]} W_{\sfr \sfc} \| \bvec^{k}_{\sfr}\|^2 \geq \frac{1}{\Lr} \sum_{\sfr \in [\Lr]} W_{\sfr \sfc} \sigma^t_{\sfr} + 1\Big )  +P\Big ( \frac{n}{L^2}  \norm{\madj^{j, \sfc}}^2 \geq \frac{n}{L} \tau^{j}_{\sfc} + 1\Big ) \\
& \leq P\Big ( \Big \lvert \frac{1}{\Lr}\sum_{\sfr \in [\Lr]} W_{\sfr \sfc} \Big( \frac{1}{\Mr}\| \bvec^{k}_{\sfr}\|^2 - \sigma^t_{\sfr}\Big) \Big \lvert  \geq 1\Big )  +P\Big (\frac{n}{L} \Big \lvert \frac{1}{L}  \norm{\madj^{j, \sfc}}^2 -  \tau^{j}_{\sfc} \Big \lvert \geq 1\Big )\\
&\stackrel{(b)}{\leq} K K'_{t-1}   \PC_{t-1}   \exp\{- \kappa \kappa'_{t-1}  \pc_{t-1}\} + K K'_{t-1}   \PC_{t-1}   \exp\{-\kappa \kappa'_{t-1}  \pc_{t-1}\}.
\end{align*}
Step $(a)$ uses $\frac{\textsf{B}^2}{16} \geq (\frac{n}{L} \tau^{j}_{\sfc} + 1)(\frac{1}{\Lr} \sum_{\sfr \in [\Lr]} W_{\sfr \sfc} \sigma^t_{\sfr} + 1)$ and $(b)$ uses $\mathcal{B}_t (c)$ and $\mathcal{B}_t (e)$ both with $\sfp =1$.

Finally, using \eqref{eq:bmadj_bound}, we bound $T_2$:
\begin{align*}
&T_2 \leq P\Big ( \sqrt{\frac{1}{n}\sum_{\sfr} W_{\sfr \sfc} \| \bvec^{k}_{\sfr}\|^2} \cdot  \frac{ \sqrt{n}}{L}  \norm{\madj^{j, \sfc}} \sum_{i=0}^{t-1}  \abs{\alpha^{t,\sfc}_i }  \geq \frac{\textsf{B}}{4} \Big ) \\
& \leq P\Big ( \frac{1}{n}\sum_{\sfr } W_{\sfr \sfc} \| \bvec^{k}_{\sfr}\|^2 \cdot  \frac{n}{L^2}  \norm{\madj^{j, \sfc}}^2 \geq (\frac{1}{\Lr} \sum_{\sfr} W_{\sfr \sfc} \sigma^t_{\sfr} + 1) (\frac{n}{L} \tau^{j}_{\sfc} + 1)\Big )  + P\Big (  \sum_{i=0}^{t-1}  \abs{\alpha^{t,\sfc}_i } \geq   \sum_{i=0}^{t-1}  \abs{\hat{\alpha}^{t,\sfc}_i } + 1\Big ). 
\end{align*}
Now the first term in the above can be bounded as the work in $T_1$ and the second term as in \eqref{eq:T2alpha_bound}. This completes the proof of the lemma.
\end{proof}

%%--------------------------------------------------------

\noindent  \textbf{(a)}  
 To show the upper bound in \eqref{eq:Ha}, recall the definition of $\Dvec_{t+1,t, \sfc}$ from  Lemma \ref{lem:hb_cond} Eq.\ \eqref{eq:Dt1t}. First, using arguments similar to \eqref{eq:finalrep_deltaterm1}, we can show for a single entry $i \in \sfc$,
\[\sum_{\sfr' \in [\Lr]}  \frac{\| \madj^{t, \sfc}_{\perp, \sfr'} \|}{\sqrt{L}}    \, [\proj^{\parallel}_{\Qadjmat^{\sfr'}_{t+1}} \Z^{\sfr'}_{t}]_{i \in \sfc}   \overset{d}{=} Z\sqrt{\sum_{j=0}^{t} \sum_{\sfr' \in [\Lr]}  \frac{\|  \madj^{t, \sfc}_{\perp, \sfr'} \|^2 ([\qadj_{\perp, \sfc}^{j, \sfr'}]_i)^2}{L \| \qadj_{\perp}^{j, \sfr'} \|^2}}, \qquad Z \sim \normal(0,1). \]
Therefore if we consider the maximum squared element for a section $\ell \in \sfc$, we have
\begin{align}
 \max_{i \in sec(\ell); \ell \in \sfc} \Big\{ \Big(\sum_{\sfr' \in [\Lr]}  \frac{\| \madj^{t, \sfc}_{\perp, \sfr'} \|}{\sqrt{L}}    \, [\proj^{\parallel}_{\Qadjmat^{\sfr'}_{t+1}} \Z^{\sfr'}_{t}]_{i \in \sfc}\Big)^2 \Big\}   = &Z^2 \sum_{j=0}^{t} \sum_{\sfr' \in [\Lr]}  
 \frac{\| \madj^{t, \sfc}_{\perp, \sfr'} \|^2}{L\| \qadj_{\perp}^{j, \sfr'} \|^2} \max_{i \in sec(\ell)}  ([\qadj_{\perp, \sfc}^{j, \sfr'}]_i)^2  
\label{eq:final_Delta_T2bound}
\end{align}
Now we simplify the final three terms of $\Dvec_{t+1,t, \sfc}$ in \eqref{eq:Dt1t}, using the Lemma \ref{lem:Qv_conc}  notation $\innerQ^{\sfr}_{t+1}, \mathbf{v}^{\sfr \sfc}$.
\begin{align*}
&\sum_{\sfr'} \Qadjmat^{\sfr'}_{t+1, \sfc}  (\innerQ^{\sfr'}_{t+1})^{-1} \mathbf{v}^{\sfr', \sfc} \, + \,  \sum_{i=0}^{t-1}\alpha^{t,\sfc}_{i} \brqvec^i_{\sfc}    \,- \, \brqvec^t_{\sfc} = \sum_{j=1}^{t+1} \sum_{\sfr'}  \qadj_{\sfc}^{j-1, \sfr'} [(\innerQ^{\sfr'}_{t+1})^{-1} \mathbf{v}^{\sfr', \sfc}]_{j} \, + \,  \sum_{i=0}^{t-1}\alpha^{t,\sfc}_{i} \brqvec^i_{\sfc}    \,- \, \brqvec^t_{\sfc} \\
&= \sum_{j=1}^{t} \brqvec_{\sfc}^{j-1} \Big[ \sum_{\sfr'} \sqrt{W_{\sfr' \sfc}} [(\innerQ^{\sfr'}_{t+1})^{-1} \mathbf{v}^{\sfr', \sfc}]_{j} + \alpha^{t,\sfc}_{j-1}  \Big]   +  \brqvec^t_{\sfc} \Big[ \sum_{\sfr'} \sqrt{W_{\sfr' \sfc}} [(\innerQ^{\sfr'}_{t+1})^{-1} \mathbf{v}^{\sfr', \sfc}]_{t+1}   - 1\Big].
\end{align*}

Therefore for section $\ell \in \sfc$, using the triangle inequality, Lemma \ref{lem:squaredsums}, the fact that $  \max_{j \in sec(\ell)} \abs{[\brqvec^s_{\sfc}]_j}^2 \leq 4$ for $0 \leq s \leq t $, and the definition of $\Dvec_{t+1,t, \sfc}$ from  \eqref{eq:Dt1t}, we have 
\begin{align}
&\frac{1}{2(t+1)} \max_{j \in sec(\ell)} \abs{[\Dvec_{t+1,t, \sfc}]_j}^2    \leq \sum_{i=0}^{t-2} \lvert \alpha^{t,\sfc}_{i} - \hat{\alpha}^{t,\sfc}_{i}\lvert^2   \max_{j \in sec(\ell)} \abs{[\hvec^{i+1}_{\sfc}]_j}^2   \nonumber \\
& + \, \Big[\frac{\| \madj^{t, \sfc}_{\perp} \| }{\sqrt{L}}   - \sqrt{\tau^{t}_{\perp, \sfc}} \Big]^2   \max_{j \in sec(\ell)} \abs{[\Z_{t,\sfc}]_j}^2 + Z^2 \sum_{j=0}^{t} \sum_{\sfr' \in [\Lr]}  \frac{\|  \madj^{t, \sfc}_{\perp, \sfr'} \|^2}{L \| \qadj_{\perp}^{j, \sfr'}\|^2} \max_{i \in sec(\ell)}  ([\qadj_{\perp, \sfc}^{j, \sfr'}]_i)^2 \nonumber \\
&  + \, 4 \sum_{j=1}^{t} \Big[ \sum_{\sfr' \in [\Lr]} \sqrt{W_{\sfr' \sfc}} [(\innerQ^{\sfr'}_{t+1})^{-1}  \mathbf{v}^{\sfr', \sfc}]_{j} + \alpha^{t,\sfc}_{j-1} \Big]^2   + 4 \Big[ \sum_{\sfr' \in [\Lr]} \sqrt{W_{\sfr' \sfc}} [(\innerQ^{\sfr'}_{t+1})^{-1} \mathbf{v}^{\sfr', \sfc}]_{t+1}   - 1 \Big]^2.
 \label{eq:normDelt1t0}
\end{align}

Using  \eqref{eq:normDelt1t0} and  Lemma \ref{sums}, we  have the following bound, where $\tilde{\e}_t = \frac{\e}{4(t+1)^2}$,
\begin{align}
& P  \Big(\frac{1}{L} \sum_{\sfc \in [\Lc]} \sum_{\ell \in \sfc}  \max_{j \in sec(\ell)}  W^{2 \sfq}_{\sfr \sfc}  \abs{[\Dvec_{t+1,t, \sfc}]_{j}}^2 \geq \epsilon \Big)  \nonumber\\
& \overset{(a)}{\leq}   \sum_{i=0}^{t-1}  P \Big( \frac{1}{L} \sum_{\sfc \in [\Lc]}   W^{2 \sfq}_{\sfr \sfc}  \Big \lvert \alpha^{t,\sfc}_{i} - \hat{\alpha}^{t,\sfc}_{i} \Big \lvert^2  \sum_{\ell \in \sfc} \max_{j \in sec(\ell)} \abs{[ \hvec_{\sfc}^{i+1}]_j}^2  \geq \tilde{\e}_t \Big) \nonumber \\
& + P \Big( \frac{1}{L}  \sum_{\sfc \in [\Lc]} \frac{L}{n}  W^{2 \sfq}_{\sfr \sfc}   \Big \lvert \sqrt{\frac{n \| \madj^{t, \sfc}_{\perp} \|^2}{L^2}}  - \sqrt{\frac{n\tau^t_{\perp, c}}{L}} \Big \lvert^2 \sum_{\ell \in \sfc} \max_{j \in sec(\ell)}  \abs{ [\Z_{t, \sfc}]_j}^2  \geq \tilde{\e}_t \Big) \nonumber \\
& + P \Big(\frac{Z^2}{L}  \sum_{\sfc \in [\Lc]}   W^{2 \sfq}_{\sfr \sfc}  \sum_{j=0}^{t} \sum_{\sfr' \in [\Lr]} \  \frac{\| \madj^{t, \sfc}_{\perp, \sfr'}\|^2}{L 
\| \qadj_{\perp}^{j, \sfr'} \|^2} \sum_{\ell \in \sfc}  \max_{i \in sec(\ell)}  ([\qadj_{\perp, \sfc}^{j, \sfr'}]_i)^2 \geq \tilde{\e}_t \Big) \nonumber \\
& + P \Big( \frac{4}{\Lc} \sum_{\sfc \in [\Lc]}  W^{2 \sfq}_{\sfr \sfc}  \Big \lvert \sum_{\sfr' \in [\Lr]} \sqrt{W_{\sfr' \sfc}} [(\innerQ^{\sfr'}_{t+1})^{-1} \mathbf{v}^{\sfr', \sfc}]_{t}   - 1\Big \lvert^2 \geq  \tilde{\e}_t \Big) \nonumber  \\
&+ \sum_{i=1}^{t} P \Big(  \frac{4}{\Lc} \sum_{\sfc \in [\Lc]}  W^{2 \sfq}_{\sfr \sfc}   \Big \lvert  \sum_{\sfr' \in [\Lr]} \sqrt{W_{\sfr' \sfc}} [(\innerQ^{\sfr'}_{t+1})^{-1} \mathbf{v}^{\sfr', \sfc}]_{i} + \alpha^{t,\sfc}_{i-1} \Big \lvert^2 \geq\tilde{\e}_t \Big) .
\label{eq:sum_delt1t2_bnd1}
\end{align}
Label the terms on the RHS of \eqref{eq:sum_delt1t2_bnd1} as $T_1- T_5$.  We show that each has the desired upper bound.

%%__________________________________________________________

First consider $T_1$.  Let $\kappa_0$ be defined via $(\max_{\sfc' \in [\Lc]}   W^{2 \sfq}_{\sfr \sfc'})  \tau^t_{\sfc} \log M = \kappa_0 (\Lr/\omega)^{2\sfq}$, noting that $\kappa_0 = \Theta(1)$.  Then,
\begin{align*}
&T_1 \leq  \sum_{i=0}^{t-1} \Big[ P \Big( \cap_{\sfc \in [\Lc]}  \Big\{   \lvert \alpha^{t,\sfc}_{i} - \hat{\alpha}^{t,\sfc}_{i} \lvert^2    \leq \frac{\tilde{\e}_t}{(6 \kappa_0 (\Lr/\omega)^{2\sfq} + \e)} \Big \}   \Big)  \\
&\qquad +  P \Big(  \frac{1}{L}  \sum_{\sfc \in [\Lc]}  \sum_{\ell \in \sfc} W^{2 \sfq}_{\sfr \sfc}  \max_{j \in sec(\ell)}  ([\hvec_{\sfc}^{i+1}]_{j})^2 \leq 6 (\max_{\sfc' \in [\Lc]}   W^{2 \sfq}_{\sfr \sfc'})  \tau^t_{\sfc} \log M + \e  \Big)\Big] \\
& \leq  \sum_{i=0}^{t-1} \Big[ \sum_{\sfc \in [\Lc]}  P \Big(   \lvert \alpha^{t,\sfc}_{i} - \hat{\alpha}^{t,\sfc}_{i} \lvert^2    \geq \frac{\tilde{\e}_t}{(6 \kappa_0 (\Lr/\omega)^{2\sfq} + \e)}  \Big) \\
& \qquad +   P\Big( \frac{1}{L}  \sum_{\sfc \in [\Lc]}  \sum_{\ell \in \sfc} W^{2 \sfq}_{\sfr \sfc} \max_{j \in sec(\ell)}  ([\hvec_{\sfc}^{i+1}]_{j})^2 \geq 6(\max_{\sfc' \in [\Lc]}   W^{2 \sfq}_{\sfr \sfc'})  \tau^t_{\sfc} \log M + \e \Big) \Big]. \\
&\leq t \Lc K K'_{t-1}   \PC_{t-1} \exp\{-\frac{1}{t^2}\kappa \kappa'_{t-1}  (\omega/\Lr)^{2 \sfq}  \pc_{t-1}  \e^2 \} +   t K K'_{t-1}   \PC'_{t-2}   \exp\{- \kappa \kappa'_{t-1} (\omega/\Lr)^{2\sfv} \pc'_{t-2}  \e^2\},
\end{align*}
The final inequality follows by $\mathcal{B}_{t} (f)$ and $\mathcal{H}_{t} (e)$.

Term $T_2$ of \eqref{eq:sum_delt1t2_bnd1} has the desired upper bound by work similar to that used to bound the corresponding term in the $\mathcal{H}_1 (a)$ step in equations \eqref{eq:H1F_def} - \eqref{eq:H1T1_bound}.

Next consider  term $T_5$ in \eqref{eq:sum_delt1t2_bnd1}.  Using the union bound,
\be
\begin{split} 
T_5 &\leq  \sum_{i=1}^{t} \sum_{\sfc \in [\Lc]} P \Big(   \Big \lvert  \sum_{\sfr' \in [\Lr]} \sqrt{W_{\sfr' \sfc}} [(\innerQ^{\sfr'}_{t+1})^{-1} \mathbf{v}^{\sfr', \sfc}]_{i} + \alpha^{t,\sfc}_{i-1} \Big \lvert \geq. \sqrt{\frac{\tilde{\e}_t}{4(\sum_{\sfc' \in [\Lc]} W^{2 \sfq}_{\sfr \sfc'}/\Lc)}} \Big) \\
& \leq    t^3 \Lc K K'_{t-1}   \PC_{t-1}   \exp\Big\{\frac{-1}{t^6} \kappa \kappa'_{t-1}  (\omega/\Lr)^{(2\sfq-1)_+}\pc_{t-1}  \e \Big\} ,
\label{eq:Ha_T5}
\end{split}
\ee
where the second inequality follows from  Lemma \ref{lem:Qv_conc} and \eqref{eq:Wrc_avgs}, which gives $\sum_{\sfc' \in [\Lc]} W^{2 \sfq}_{\sfr \sfc'}/\Lc \leq \kappa (\Lr/\omega)^{(2 \sfq-1)_+}$  for $\sfq \in \{ 0,1\}$.  The bound for $T_4$ is obtained similarly to $T_5$.

Finally, we bound the third term $T_3$ in \eqref{eq:sum_delt1t2_bnd1}.  %  \CR{I changed the $T_3$ work slightly...}. 
For $T_3$, if all three of the following events hold: % for $\sfc \in [\Lc], \sfr' \in [\Lr]$, and $0\leq j \leq t$:
\begin{align*}
&\cap_{\sfc \in [\Lc]} \cap_{\sfr' \in [\Lr]} \cap_{j=0}^t \Big\{ \sum_{\ell \in \sfc} \max_{i \in sec(\ell)}  \frac{([\qadj_{\perp, \sfc}^{j, \sfr'}]_i)^2}{\| \qadj_{\perp}^{j, \sfr'} \|^2} \leq \frac{\kappa t \Lr}{\Lc \omega}\Big\},  \quad   \quad \Big\{Z^2 \leq \frac{\Mr \omega  (\Lr/\omega)^{2 \sfq} \, \tilde{\e}_t}{t^2}\Big\},
 \\
 & 
\text{ and } \Big\{ \frac{(\omega/\Lr)^{2\sfq}}{\Lc} \sum_{\sfc \in [\Lc]}  \frac{n  W^{2 \sfq}_{\sfr \sfc}   \|  \madj^{t, \sfc}_{\perp} \|^2}{L^2} \leq    \frac{n (\omega/\Lr)^{2\sfq}}{L  \Lc} \sum_{\sfc \in [\Lc]}  W^{2 \sfq}_{\sfr \sfc}  \tau^{t}_{\perp, \sfc} + 1 \Big\}, 
\end{align*}
then
\begin{align*}
&\frac{Z^2}{L}  \sum_{\sfc \in [\Lc]}   W^{2 \sfq}_{\sfr \sfc}  \sum_{j=0}^{t} \sum_{\sfr' \in [\Lr]} \  \frac{\| \madj^{t, \sfc}_{\perp, \sfr'} \|^2}{L \| \qadj_{\perp}^{j, \sfr'} \|^2} \sum_{\ell \in \sfc} \max_{i \in sec(\ell)}  ([\qadj_{\perp, \sfc}^{j, \sfr'}]_i)^2  \leq   \kappa' \tilde{\e}_t,
\end{align*}
where we have used that $ \frac{n (\omega/\Lr)^{2\sfq}}{L \Lc} \sum_{\sfc \in [\Lc]}  W^{2 \sfq}_{\sfr \sfc}  \tau^{t}_{\perp, \sfc}\leq \kappa''$ using \eqref{eq:Wrc_avgs}.
Here, $\kappa, \kappa', $ and $\kappa''$ are suitable universal positive constants. 
Therefore,
\begin{align}
T_3 &\leq  \sum_{\sfc \in [\Lc]}  \sum_{j=0}^{t} \sum_{\sfr' \in [\Lr]} P\Big(\sum_{\ell \in \sfc} \max_{i \in sec(\ell)}  \frac{([\qadj_{\perp, \sfc}^{j, \sfr'}]_i)^2}{\| \qadj_{\perp}^{j, \sfr'} \|^2} \geq \frac{\kappa t \Lr}{\Lc \omega}\Big) + P\Big(Z^2 \geq   \frac{\Mr \omega (\omega/\Lr)^{2 \sfq} \tilde{\e}_t}{t^2} \Big)  \nonumber \\
& \ +P\Big( \frac{(\omega/\Lr)^{2\sfq}}{ \Lc} \sum_{\sfc \in [\Lc]}  \frac{n  W^{2 \sfq}_{\sfr \sfc}   \|  \madj^{t, \sfc}_{\perp} \|^2}{L^2} \geq    \frac{n (\omega/\Lr)^{2\sfq}}{L\Lc} \sum_{\sfc \in [\Lc]}  W^{2 \sfq}_{\sfr \sfc}  \tau^{t}_{\perp, \sfc} + 1 \Big).
\label{eq:DeltaH_T3_bound}
\end{align}
We consider each of the terms of \eqref{eq:DeltaH_T3_bound} separately. For the first term, using Lemma \ref{lem:squaredsums} and the fact that $  \max_{j \in sec(\ell)} \abs{[\brqvec^s_{\sfc}]_j}^2 \leq 4$ for $0 \leq s \leq t $,  for any $\sfc \in [\Lc], \sfr' \in [\Lr]$ and $0 \leq j \leq t$,
\begin{align*}
& \sum_{\ell \in \sfc}  \max_{i \in sec(\ell)}  ([\qadj_{\perp, \sfc}^{j, \sfr'}]_i)^2 =  W_{\sfr' \sfc}    \sum_{\ell \in \sfc} \max_{i \in sec(\ell)}  ([\brqvec_{\sfc}^{j}]_i - \sum_{k=0}^{j-1} \gamma^{j,\sfr'}_k [\brqvec^{k}_{\sfc}]_i)^2 \\
& \leq  2 t  W_{\sfr' \sfc}   \sum_{\ell \in \sfc}  \max_{i \in sec(\ell)} \Big\{  ([\brqvec_{\sfc}^{j}]_i)^2 + \sum_{k=0}^{j-1} (\gamma^{j,\sfr'}_k [\brqvec^{k}_{\sfc}]_i)^2 \Big\}% \leq \frac{4t  (1 + \sum_{k=0}^{j-1} (\gamma^{j,\sfr'}_k)^2) W_{\sfr' \sfc} L}{\Lc} 
\leq \frac{4P t  (1 + \sum_{k=0}^{j-1} (\gamma^{j,\sfr'}_k)^2)  L \Lr}{\Lc \omega} ,
\end{align*}
Therefore letting $\kappa:= \max_j \kappa_j$ with $\kappa_j :=  4P (2 + (\hat{\gamma}^{j,\sfr'}_{j-1})^2)((\sigma^{j+1}_{\perp, \sfr'})^{-1} + 1)$ we have,
\be
\begin{split}
& \sum_{\sfc \in [\Lc]}  \sum_{j=0}^{t} \sum_{\sfr' \in [\Lr]} P\Big(\sum_{\ell \in \sfc}  \max_{i \in sec(\ell)}  \frac{([\qadj_{\perp, \sfc}^{j, \sfr'}]_i)^2}{\| \qadj_{\perp}^{j, \sfr'} \|^2} \geq \frac{\kappa t \Lr}{\Lc \omega}\Big)  \\
& \leq \sum_{\sfc \in [\Lc]}  \sum_{j=0}^{t} \sum_{\sfr' \in [\Lr]}   P\Big( \frac{  (1 + \sum_{k=0}^{j-1} (\gamma^{j,\sfr'}_k)^2)L}{\| \qadj_{\perp}^{j, \sfr} \|^2} \geq (2 + (\hat{\gamma}^{j,\sfr'}_{j-1})^2)\Big(\frac{1}{\sigma^{j+1}_{\perp, \sfr'}} + 1\Big) \Big) \\
&\leq \sum_{\sfc \in [\Lc]} \sum_{j=0}^{t}  \sum_{\sfr' \in [\Lr]}  \Big[ P\Big(  1 + \sum_{k=0}^{j-1} (\gamma^{j,\sfr'}_k)^2 \geq 2 + (\hat{\gamma}^{j,\sfr'}_{j-1})^2 \Big) +  P\Big( \frac{L}{\| \qadj_{\perp}^{j, \sfr'} \|^2} \geq  \frac{1}{\sigma^{j+1}_{\perp, \sfr'}} + 1 \Big) \Big]. \label{eq:DeltaH_T3_bound_2}
\end{split}
\ee
For the second term on the RHS of \eqref{eq:DeltaH_T3_bound_2}, by Lemma \ref{inverses} and $\mathcal{H}_t (f)$,
\begin{align*}
 \sum_{\sfc \in [\Lc]}  \sum_{j=0}^{t}  \sum_{\sfr \in [\Lr]}  P\Big( \frac{L}{\| \qadj_{\perp}^{j, \sfr} \|^2} \geq \frac{1}{\sigma^{j+1}_{\perp, \sfr}} + 1 \Big) & \leq  \sum_{\sfc \in [\Lc]}  \sum_{j=0}^{t}  \sum_{\sfr \in [\Lr]}  P\Big(\Big \lvert \frac{L}{\| \qadj_{\perp}^{j, \sfr} \|^2} - \frac{1}{\sigma^{j+1}_{\perp, \sfr}} \Big \lvert \geq 1\Big)\\
& \leq  t \Lr \Lc K K'_{t-1}   \PC'_{t-2}   \exp\Big\{-\frac{ \kappa \kappa'_{t-1} (\omega/\Lr)^{2} \pc'_{t-2} }{(\log M)^2}\Big\}.
%t K K'_{t-1}  (\Lr \Lc)^{t} \Lc  \exp\Big\{\frac{-\kappa \kappa'_{t-1} \Mr   \omega^{t + 1}  \Lc^{3t-2}}{\Lr^{4t- 2} (\log M)^{2t}}\Big\}.
%
\end{align*}
For the first term on the RHS of \eqref{eq:DeltaH_T3_bound_2}, using Lemma \ref{sums}, the fact that $ \hat{ \gamma}^{t,\sfr}_{k} = 0$ for $0 \leq k \leq t-2$, along with Lemma \ref{powers}, and $\mathcal{H}_t (f)$,
\begin{align*}
&  \sum_{\sfc \in [\Lc]} \sum_{j=0}^{t}  \sum_{\sfr \in [\Lr]}   P\Big(  1 + \sum_{k=0}^{j-1} (\gamma^{j,\sfr}_k)^2 \geq 2 + (\hat{\gamma}^{j,\sfr}_{j-1})^2 \Big)  \leq  \sum_{\sfc \in [\Lc]}\sum_{j=0}^{t}  \sum_{\sfr \in [\Lr]}   P\Big(   \Big \lvert \sum_{k=0}^{j-1} (\gamma^{j,\sfr}_k)^2 - (\hat{\gamma}^{j,\sfr}_{j-1})^2 \Big \lvert  \geq 1 \Big)   \\
&\leq  \sum_{\sfc \in [\Lc]} \sum_{j=0}^{t}  \sum_{\sfr \in [\Lr]}   \sum_{k=0}^{j-1} P\Big( \Big \lvert (\gamma^{j,\sfr}_k)^2 - (\hat{\gamma}^{j,\sfr}_{j})^2 \Big \lvert \geq \frac{1}{t} \Big)  \leq  t \Lr \Lc K K'_{t-1}   \PC'_{t-2}   \exp\Big\{ \frac{- \kappa \kappa'_{t-1} (\omega/\Lr)^{2} \pc'_{t-2} }{t^4(\log M)^2}\Big\}.
%t^2 K K'_{t-1}  (\Lr \Lc)^{t} \Lc  \exp\Big\{\frac{-\kappa \kappa'_{t-1} \Mr   \omega^{2t + 1}  \Lc^{2t-2}}{t^2\Lr^{4t- 2} (\log M)^{2t}}\Big\}.
%
\end{align*}
 The second term of \eqref{eq:DeltaH_T3_bound} is upper bounded by 
 $K \exp\{- {\kappa \Mr \omega (\omega/\Lr)^{2 \sfq} \e^2}/{t^4}\}$ using Lemma \ref{lem:max_abs_normals}. Finally, we bound the second term of \eqref{eq:DeltaH_T3_bound}. 
\begin{align*}
&P\Big( \frac{(\omega/\Lr)^{2\sfq}}{ \Lc} \sum_{\sfc \in [\Lc]}  \frac{n  W^{2 \sfq}_{\sfr \sfc}   \|  \madj^{t, \sfc}_{\perp} \|^2}{L^2} \geq    \frac{n (\omega/\Lr)^{2\sfq}}{L \Lc} \sum_{\sfc \in [\Lc]}  W^{2 \sfq}_{\sfr \sfc}  \tau^{t}_{\perp, \sfc} + 1 \Big)\\
& \leq P\Big( \frac{n (\omega/\Lr)^{2\sfq}}{L  \Lc}  \Big\lvert \sum_{\sfc \in [\Lc]}  W^{2 \sfq}_{\sfr \sfc}  \Big(\frac{   \|  \madj^{t, \sfc}_{\perp} \|^2}{L} -  \tau^{t}_{\perp, \sfc}  \Big)\Big\lvert  \geq  1 \Big) \overset{(a)}{\leq} P\Big(  \frac{n }{L \Lc}   \sum_{\sfc \in [\Lc]}  \Big \lvert  \frac{   \|  \madj^{t, \sfc}_{\perp} \|^2}{L} -  \tau^{t}_{\perp, \sfc} \Big \lvert \geq 1 \Big)  \\
&\overset{(b)}{\leq}  \sum_{\sfc \in [\Lc]}  P\Big( \frac{n}{ L} \Big \lvert  \frac{ \|  \madj^{t, \sfc}_{\perp} \|^2}{L} -   \tau^{t}_{\perp, \sfc}  \Big \lvert \geq 1 \Big)  \overset{(c)}{\leq}   
  \Lc K K'_{t-1}   \PC_{t-1} \exp\{-\kappa \kappa'_{t-1}    \pc_{t-1}   \} .
%K K'_{t-1}  (\Lr \Lc)^{t}   \Lc \exp\Big\{\frac{-\kappa \kappa'_{t-1} \Mr  \omega^{2t+ 1} \Lc^{2t} \e^2}{ \Lr^{4t} (\log M)^{2t} }\Big\}. 
%
\end{align*}
In the above, step $(a)$ follows since $(\omega W_{\sfr \sfc}/\Lr)^{2 \sfq} \leq ((1- \rho) P)^{2 \sfq} \leq \kappa$.  Step $(b)$ uses  Lemma \ref{sums} and step $(c)$ from $\mathcal{B}_t (f)$.  Note that $\sum_{\sfc' \in [\Lc]} W_{\sfr \sfc'}/\Lr \leq \kappa$.

We finally note that the desired upper bound in \eqref{eq:Ha} follows since 
$\PC'_{t-2} \leq  \PC'_{t-1}$ and $\Lr \Lc \PC'_{t-2} = \Lc \PC_{t-1} =  \PC'_{t-1}$ along with the fact that $(\omega/\Lr)^{(2\sfq-1)_+}\pc_{t-1}$, $(\omega/\Lr)^{2 \sfv} \pc'_{t-2}$, and $(\omega/\Lr)^{2} \pc'_{t-2} /(\log M)^2$ are all lower bounded by $(\omega/\Lr)^{2 \sfv} \pc_{t-1} = (\omega/\Lr)^{2 \sfv} \pc'_{t-1}$.
%%%%%%%%%%

\noindent  \textbf{(b)} Using the conditional distribution of $\hvec^{t+1}$ from  Lemma \ref{lem:hb_cond} Eq.\ \eqref{eq:Ha_dist}, we have
\begin{align*}
&P\Big(\frac{1}{L}  \Big \lvert \sum_{\sfc \in [\Lc]} \sqrt{W_{\sfr \sfc}} (\hvec_{\sfc}^{t+1})^* \qadj^{0, \sfr}_{\sfc} \Big \lvert \geq \epsilon \Big) \nonumber \\
& = P\Big(\frac{1}{L} \Big \lvert  \sum_{\sfc \in [\Lc]}  \sqrt{W_{\sfr \sfc}} \Big(\frac{\tau^t_{\sfc} (\hvec_{\sfc}^{t})^*  \qadj^{0, \sfr}_{\sfc}}{\tau^{t-1}_{\sfc}}  + \sqrt{ \tau^{t}_{\perp, \sfc}} \,\Z_{t, \sfc}^* \,  \qadj^{0, \sfr}_{\sfc} + \Dvec_{t+1,t,\sfc}^* \,  \qadj^{0, \sfr}_{\sfc} \Big) \Big \lvert \geq \epsilon \Big) \nonumber \\
& \leq P\Big(\frac{1}{L}  \Big \lvert  \sum_{\sfc \in [\Lc]} \frac{ \sqrt{W_{\sfr \sfc}} \tau^t_{\sfc}(\hvec_{\sfc}^{t})^*  \qadj^{0, \sfr}_{\sfc} }{\tau^{t-1}_{\sfc}}  \Big  \lvert \geq \frac{\epsilon}{2} \Big) + P\Big(\frac{1}{L}  \Big \lvert  \sum_{\sfc \in [\Lc]}  \sqrt{W_{\sfr \sfc}} \Big[ \sqrt{ \tau^{t}_{\perp, \sfc}} \, \Z_{t, \sfc}^* \,  \qadj^{0, \sfr}_{\sfc} +  \Dvec_{t+1,t,\sfc}^* \,  \qadj^{0, \sfr}_{\sfc} \Big] \Big \lvert \geq \frac{\epsilon}{2} \Big). %\label{eq:Ht1b_split}
\end{align*}
The bound for first term above follows from the induction hypothesis $\mathcal{H}_t (b)$.
%$K K_{t-1}  (\Lr \Lc)^{t-1} \Lc  \exp\Big\{\frac{-\kappa \kappa_{t-1} \Mr  \omega^{2t + 1} \Lc^{2(t-1)} \e^2 }{\Lr^{4(t-1) + 2} (\log M)^{2t-2}}\Big\}$. 
The second term has the desired upper bound using work as in the proof of $\mathcal{H}_1(b)$ (Eqs.\ \eqref{eq:H1b_split} - \eqref{eq:H1b_T2}).
%

%%%

\noindent  \textbf{(c)} We show \eqref{eq:Hc_adj} when $0 \leq s \leq t+1$.
For brevity,  we write $\eta^{s-1}_{\sfc}(\betavec_0 - \hvec^s)$ to mean $\eta^{s-1}_{\sfc}(\betavec_{0, {\sfc}} - \hvec^s_{\sfc})$, noting that $\eta^{s-1}_{\sfc}(\cdot)$ depends only on the elements of its input in column block ${\sfc}$.  Using the conditional distribution of $\hvec^{t+1}$ in Lemma \ref{lem:ideal_cond_dist} Eq.\ \eqref{eq:htil_rep}, we have
\be
\begin{split}
&P \Big(  \Big \lvert \sum_{\sfc \in [\Lc]}  \Big[  \frac{W^{\sfq}_{\sfr \sfc}(\brqvec^{s}_{\sfc})^* \brqvec^{t+1}_{\sfc} }{L}   - \frac{W^{\sfq}_{\sfr \sfc} \psi^{t+1}_{\sfc}}{\Lc}  \Big] \Big \lvert  \geq \epsilon \Big) \\
&= P\Big( \frac{1}{\Lc}\Big \lvert\sum_{\sfc \in [\Lc]}  W^{\sfq}_{\sfr \sfc}  \Big[\frac{\Lc}{L}  ( \eta^{s-1}_{\sfc}(\betavec_0 -  \tilde{\hvec}^{s} - \tilde{\Dvec}_{s}) - \betavec_{0, \sfc})^*( \eta^t_{\sfc}(\betavec_0 - \tilde{\hvec}^{t+1} -  \tilde{\Dvec}_{t+1}) - \betavec_{0, \sfc})  - \psi^{t+1}_{\sfc}  \Big] \Big \lvert \geq \epsilon \Big) \\
& \leq P\Big( \frac{1}{L}\Big \lvert\sum_{\sfc \in [\Lc]}  W^{\sfq}_{\sfr \sfc}  ( \eta^{s-1}_{\sfc}(\betavec_0 -  \tilde{\hvec}^{s} - \tilde{\Dvec}_{s}) - \betavec_{0, \sfc})^*\Big[  \eta^t_{\sfc}(\betavec_0 - \tilde{\hvec}^{t+1} -  \tilde{\Dvec}_{t+1}) -  \eta^t_{\sfc}(\betavec_0 - \tilde{\hvec}^{t+1})  \Big] \Big \lvert \geq \frac{\epsilon}{3} \Big)\\
&+ P\Big( \frac{1}{L}\Big \lvert\sum_{\sfc \in [\Lc]} W^{\sfq}_{\sfr \sfc} \Big[\eta^{s-1}_{\sfc}(\betavec_0 -  \tilde{\hvec}^{s} - \tilde{\Dvec}_{s}) - \eta^{s-1}_{\sfc}(\betavec_0 -  \tilde{\hvec}^{s})  \Big] ^*( \eta^t_{\sfc}(\betavec_0 - \tilde{\hvec}^{t+1}) - \betavec_{0, \sfc}) \Big \lvert \geq \frac{\epsilon}{3} \Big) \\
& + P\Big( \frac{1}{\Lc}\Big \lvert\sum_{\sfc \in [\Lc]} W^{\sfq}_{\sfr \sfc}  \Big[\frac{\Lc}{L}  ( \eta^{s-1}_{\sfc}(\betavec_0 -  \tilde{\hvec}^{s}) - \betavec_{0, \sfc})^*( \eta^t_{\sfc}(\betavec_0 - \tilde{\hvec}^{t+1}) - \betavec_{0, \sfc})  - \psi^{t+1}_{\sfc}  \Big] \Big \lvert \geq \frac{\epsilon}{3} \Big).
\label{eq:Ht1_c_eq1}
\end{split}
\ee
Label the three terms on the RHS of \eqref{eq:Ht1_c_eq1} as $T_1-T_3$ and provide bounds for each.  First note the following bound for $\tilde{\Dvec}_{s, \sfc}$ (defined in Lemma \ref{lem:ideal_cond_dist}) that will be used repeatedly. For 
$1 \leq s \leq t+1$,
\begin{align}
&P\Big(\frac{1}{L} \sum_{\sfc \in [\Lc]} \sum_{\ell \in \sfc} W^{\sfq}_{\sfr \sfc} \max_{j \in sec(\ell)}  \lvert [\tilde{\Dvec}_{s, \sfc}]_j \lvert \geq\e\Big) = P\Big(\frac{1}{L} \sum_{\sfc \in [\Lc]} \sum_{\ell \in \sfc} W^{\sfq}_{\sfr \sfc}\max_{j \in sec(\ell)} \Big \lvert   \sum_{i=0}^{s-1} \Big(\frac{\tau^{s-1}_{\sfc}}{\tau^{i}_{\sfc}}\Big) [\Dvec_{i+1,i, \sfc}]_j \Big \lvert \geq \e \Big) \nonumber \\
&\overset{(a)}{\leq}  \hspace{-1pt}  \sum_{i=0}^{s-1} \hspace{-1pt}  P\Big(\frac{1}{L}  \sum_{\sfc \in [\Lc]} \sum_{\ell \in \sfc} W^{2 \sfq}_{\sfr \sfc}  \hspace{-1pt}  \max_{j \in sec(\ell)}  \hspace{-1pt}  \abs{[\Dvec_{t+1,t, \sfc}]_{j}}^2  \hspace{-1pt}  \geq \hspace{-1pt}  \frac{\e^2}{s^2} \Big)   \hspace{-1pt} \overset{(b)}{\leq}    \hspace{-1pt} t^4 K K'_{t-1}   \PC_{t-1}' e^{-\frac{1}{t^8}\kappa \kappa'_{t-1} (\omega/\Lr)^{2\sfq}  \pc_{t-1}'  \e^2 }. \label{eq:maxtildedelta0} 
%t^4 K K'_{t-1}  (\Lr \Lc)^{t} \Lc   \exp\Big\{\frac{-\kappa \kappa'_{t-1} \Mr  \omega^{2t+1 + 2\sfq} \Lc^{2t} \e^2}{t^8 \Lr^{4t + 2 \sfq} (\log M)^{2t} }\Big\}.
%
\end{align}
In the above, step $(a)$ follows by using the triangle inequality, and then Cauchy-Schwarz, noting that ${\tau^{s-1}_{\sfc}}/{\tau^{i}_{\sfc}} \leq 1$ for $i \leq s-1$. % and ${\tau^{s-1}_{\sfc}}/{\tau^{i}_{\sfc}} \in \Theta(1)$, 
Step $(b)$ follows from $\mc{H}_{t+1}(a)$ and $s \leq (t+1)$.

Consider $T_1$, the first term on the RHS of \eqref{eq:Ht1_c_eq1}.  Since $\max_j \lvert \eta^{s-1}_{j}(\cdot) - \beta_{0, j}\lvert \leq 2$, we have
\begin{align}
& T_1 \leq P\Big( \frac{2}{L}\sum_{\sfc \in [\Lc]} \sum_{\ell \in \sfc} \sum_{j \in \ell} W^{\sfq}_{\sfr \sfc} \Big \lvert  \eta^t_{j}(\betavec_0 - \tilde{\hvec}^{t+1} -  \tilde{\Dvec}_{t+1}) -  \eta^t_{j}(\betavec_0 - \tilde{\hvec}^{t+1})  \Big \lvert \geq \kappa \epsilon \Big) \nonumber \\
&\overset{(a)}{\leq}  P\Big( \frac{4}{L} \sum_{\sfc \in [\Lc]} \frac{W^{\sfq}_{\sfr \sfc}}{\tau_{\sfc}^{t}}  \sum_{\ell \in \sfc} \max_{j \in sec({\ell})} \lvert [\tilde{\Dvec}_{t+1, \sfc}]_j \lvert \geq \kappa \e \Big) \overset{(b)}{\leq} P\Big( \frac{1}{L} \sum_{\sfc \in [\Lc]} \sum_{\ell \in \sfc} W^{\sfq}_{\sfr \sfc}   \max_{j \in sec({\ell})} \lvert [\tilde{\Dvec}_{t+1, \sfc}]_j \lvert \geq \frac{\kappa \e }{\log M} \Big) \nonumber \\
&\overset{(c)}{\leq}  t^4 K K'_{t-1}   \PC_{t-1}'  \exp\Big\{-\frac{\kappa \kappa'_{t-1} (\omega/\Lr)^{2\sfq}  \pc_{t-1}'  \e^2}{t^8(\log M)^2}\Big \}. \label{eq:T1_new} %t^4 K K'_{t-1}  (\Lr \Lc)^{t} \Lc   \exp\Big\{\frac{-\kappa \kappa'_{t-1} (\Mr \omega)  \, (\omega/\Lr)^{2t + 2\sfq}  \, \e^2}{t^8  (\log M)^{2t} }\Big\}. 
\end{align}
Step $(a)$ follows by Lemma \ref{lem:BC9} applied to each section, $(b)$ by $\tau^t_{\sfc} = \Theta(L/n) = \Theta(1/\log M)$, and $(c)$ by \eqref{eq:maxtildedelta0}.  The bound for $T_1$ in \eqref{eq:T1_new} also holds for $T_2$ of \eqref{eq:Ht1_c_eq1} and is shown similarly.

Finally, consider term $T_3$ of \eqref{eq:Ht1_c_eq1}.  Recalling from Lemma  \ref{lem:ideal_cond_dist} that  
$\tilde{\hvec}^{s}_{\sfc}  \stackrel{d}{=} \sqrt{ \tau^{s-1}_{\sfc} }\tilde{\Z}_{s-1, {\sfc}}$ where $\tilde{\Z}_{s-1, {\sfc}} \sim \mathcal{N}(0, \mathbb{I}_{\Mc})$.  Then, (dropping the $\sfc$ subscript when used inside the $\eta^t_{\sfc}(\cdot)$ function),
\begin{align*}
&T_3   =  P\Big( \frac{1}{L}\Big \lvert \sum_{\sfc \in [\Lc]} \sum_{\ell  \in \sfc} W^{\sfq}_{\sfr \sfc} \Big[ ( \eta^{s-1}_{\ell}(\betavec_0 -  \sqrt{ \tau^{s-1}_{\sfc} }\tilde{\Z}_{s-1}) - \beta_{0, \ell})^*( \eta^t_{\ell}(\betavec_0 - \sqrt{ \tau^t_{\sfc} }\tilde{\Z}_{t}) - \beta_{0, \ell})  -  \psi^{t+1}_{\sfc}  \Big] \Big \lvert \geq \kappa \e \Big),  \nonumber \\
& \stackrel{(a)}{\leq} 2 \exp\Big\{\frac{-\kappa L^2 \e^2}{(\sum_{\sfc \in [\Lc]}   \sum_{\ell \in \sfc}   W^{2\sfq}_{\sfr \sfc} )}\Big\}  = 2 \exp\Big\{\frac{-\kappa L \e^2}{(\sum_{\sfc \in [\Lc]}    W^{2\sfq}_{\sfr \sfc} /\Lc)}\Big\} \stackrel{(b)}{\leq} 2 \exp\Big\{\frac{-\kappa n (\omega/\Mr)^{(2 \sfq -1)_+}\e^2}{\log M}\Big\},
%
%\label{eq:Hct2_new}
\end{align*}
where $(a)$ is obtained using Hoeffding's inequality (Lemma \ref{lem:hoeff_lem}) and $(b)$ uses $n R = L \log M$. Verifying that the expectation of the random variable in $T_3$ is $\psi_\sfc^{t+1}$ is similar to \eqref{eq:Hoeffdings2}.

\noindent  \textbf{(d)}  Using the representation of $\hvec^{s+1}$ in Lemma \ref{lem:ideal_cond_dist},  we have
\be
\begin{split}
&(\hvec^{\tilde{s}+1}_{\sfc})^* \qadj^{s+1, \sfr}_{\sfc} = \sqrt{W_{\sfr \sfc}} (\hvec^{\tilde{s}+1}_{\sfc})^* [\eta^s_{\sfc}(\betavec_0 - \hvec^{s+1}) - \betavec_{0, \sfc}]  \\
%= \sqrt{W_{\sfr \sfc}} (\hvec^{\tilde{s}+1}_{\sfc})^* \eta^s_{\sfc}(\betavec_0 - \hvec^{s+1}) + (\hvec^{\tilde{s}+1}_{\sfc})^*\qadj^{0, \sfr}_{\sfc}\\
& = \sqrt{W_{\sfr \sfc}} (\tilde{\hvec}^{\tilde{s} + 1}_{\sfc} + \tilde{\Dvec}_{\tilde{s} + 1, \sfc})^* \eta^s_{\sfc}(\betavec_0 - \tilde{\hvec}^{s+1} - \tilde{\Dvec}_{s+1}) + (\hvec^{\tilde{s}+1}_{\sfc})^*\qadj^{0, \sfr}_{\sfc}.
\label{eq:Hdt1}
\end{split}
\ee
Using \eqref{eq:Hdt1} and Lemma \ref{sums}, 
\be
\begin{split}
&P\Big(  \Big \lvert  \sum_{\sfc \in [\Lc]} \sqrt{W_{\sfr \sfc}}  \Big[  \frac{(\hvec^{\tilde{s}+1}_{\sfc})^*\qadj^{s+1, \sfr}_{\sfc}}{L} +  \frac{ \sqrt{W_{\sfr \sfc}} \, \psi^{s+1}_{\sfc} \tau^{\max(\tilde{s}, s)}_{\sfc}}{\Lc \tau^{s}_{\sfc}} \Big] \Big \lvert  \geq \epsilon\Big) \\
 &\leq    P\Big(  \Big \lvert  \sum_{\sfc \in [\Lc]}  W_{\sfr \sfc} \Big[\frac{(\tilde{\hvec}^{\tilde{s} + 1}_{\sfc} + \tilde{\Dvec}_{\tilde{s} + 1, \sfc})^* \eta^s_{\sfc}(\betavec_0 - \tilde{\hvec}^{s+1} - \tilde{\Dvec}_{s+1})}{L} +  \frac{\psi^{s+1}_{\sfc} \tau^{\max(\tilde{s}, s)}_{\sfc}}{\Lc \tau^{s}_{\sfc}} \Big] \Big \lvert  \geq \frac{\epsilon}{2} \Big) \\
 &\qquad + P\Big(  \Big \lvert  \sum_{\sfc \in [\Lc]} \sqrt{W_{\sfr \sfc}} \Big[ \frac{ (\hvec^{\tilde{s}+1}_{\sfc})^*\qadj^{0, \sfr}_{\sfc}}{L}  \Big] \Big \lvert  \geq \frac{\epsilon}{2}\Big).
  \label{eq:Hdt1a}
\end{split}
\ee
The second term on the RHS of \eqref{eq:Hdt1a} has the desired upper bound by $\mathcal{H}_{t+1} (b)$.  We now bound the first term of \eqref{eq:Hdt1a}, denoted by $T_1$.  Using Lemma \ref{sums},
\begin{align}
T_1 &\leq P\Big(\frac{1}{L} \Big \lvert  \sum_{\sfc \in [\Lc]}  W_{\sfr \sfc} \Big[(\tilde{\hvec}^{\tilde{s} + 1}_{\sfc})^* \eta^s_{\sfc}(\betavec_0 - \tilde{\hvec}^{s+1}) +  \frac{L \psi^{s+1}_{\sfc} \tau^{\max(\tilde{s}, s)}_{\sfc}}{\Lc \tau^{s}_{\sfc}} \Big] \Big \lvert  \geq \e/6 \Big)  \nonumber\\
& +  P\Big( \frac{1}{L}  \sum_{\sfc \in [\Lc]}   W_{\sfr \sfc} \Big \lvert  (\tilde{\Dvec}_{\tilde{s} + 1, \sfc})^* \eta^s_{\sfc}(\betavec_0 - \tilde{\hvec}^{s+1} - \tilde{\Dvec}_{s+1}) \Big \lvert  \geq  \e/6 \Big)   \nonumber \\
&  +  P\Big( \frac{1}{L}  \sum_{\sfc \in [\Lc]}  W_{\sfr \sfc} \Big \lvert (\tilde{\hvec}^{\tilde{s} + 1}_{\sfc})^* [\eta^s_{\sfc}(\betavec_0 - \tilde{\hvec}^{s+1} - \tilde{\Dvec}_{s+1}) - \eta^s_{\sfc}(\betavec_0 - \tilde{\hvec}^{s+1} )] \Big \lvert  \geq \e/6 \Big).
\label{eq:Hdt13_b}
\end{align}
Label the terms  of \eqref{eq:Hdt13_b} as $T_{1a}, T_{1b}, T_{1c}$, and we bound each separately.  First consider $T_{1a}$,
\begin{align}
& \sqrt{\tau^{\tilde{s}}_{\sfc}} \, \mathbb{E}_{\tilde{\Z}_{\tilde{s}}, \tilde{\Z}_{s}}\Big\{(\tilde{\Z}_{\tilde{s}, \sfc})^* \eta^s_{\sfc}(\betavec_0 - \sqrt{\tau^{s}_{\sfc}} \tilde{\Z}_{s})\Big\}  \overset{(a)}{=}  \sqrt{\tau^{\tilde{s}}_{\sfc}} \, \mathbb{E}_{\tilde{\Z}_{\tilde{s}}, \tilde{\Z}_{s}, \betavec_0}\Big\{(\tilde{\Z}_{\tilde{s}, \sfc})^* \eta^s_{\sfc}(\betavec_0 - \sqrt{\tau^{s}_{\sfc}} \tilde{\Z}_{s})\Big\}    \nonumber \\
& \quad  \overset{(b)}{=} \frac{\tau^{\max\{\tilde{s}, s\}}_{\sfc}}{\tau^{s}_{\sfc}} \Big[ \mathbb{E}_{\tilde{\Z}_{s}, \betavec}\Big\{\| \eta^s_{\sfc}(\betavec - \sqrt{\tau^{s}_{\sfc}} \tilde{\Z}_{s}) \|^2\Big\}  - \frac{L}{\Lc}\Big] \overset{(c)}{=}  \frac{-L \psi^{s+1}_{\sfc} \tau^{\max\{\tilde{s}, s\}}_{\sfc}}{\Lc\tau^{s}_{\sfc}}.
\label{eq:Hdt9a}
\end{align}
Step $(a)$ follows for each $\betavec_0 \in \mathcal{B}_{M,L}$ because of the uniform distribution of the non-zero entry in each section of $\betavec_0$ over the $M$ possible locations and the entry-wise i.i.d.\ distributions of $\tilde{\Z}_{\tilde{s}}$ and $\tilde{\Z}_{s}$, step $(b)$ by Stein's Lemma (see  \cite[p.1491, Eqs.\ (102)--(104)]{rush2017} for details), and step $(c)$ from Lemma \ref{lem:expect_etar_etas}.  Now recall from Lemma \ref{lem:ideal_cond_dist} that  $\tilde{\hvec}^{s+1}_{\sfc}  \stackrel{d}{=} \sqrt{ \tau^{s}_{\sfc} }\tilde{\Z}_{s, {\sfc}}$ where $\tilde{\Z}_{s, {\sfc}} \sim \mathcal{N}(0, \mathbb{I}_{\Mc})$ with $\expec\{[\tilde{\Z}_{\tilde{s}, {\sfc}}]_j [\tilde{\Z}_{s, {\sfc}}]_j\} = \sqrt{{\tau^{\max\{s, \tilde{s}\}}_{\sfc}}/{\tau^{\min\{s, \tilde{s}\}}_{\sfc}}}$, for $j \in [\Mc]$. Using this representation of of $\tilde{\hvec}^s$ and  \eqref{eq:Hdt9a}, the first term in \eqref{eq:Hdt13_b} can be  bounded using Lemma \ref{lem:Hd_convergence}:
\ben
\begin{split}
T_{1a} &= P\Big(\frac{1}{L}  \Big \lvert  \sum_{\sfc \in [\Lc]} \sum_{\ell \in \sfc}  \sqrt{\tau^{\tilde{s}}_{\sfc}} W_{\sfr \sfc}  \Big[ (\tilde{\Z}_{\tilde{s}, \ell})^* \eta^s_{\ell}(\betavec_0 - \sqrt{ \tau^s_{\sfc} }\tilde{\Z}_{s}) - \mathbb{E}_{\tilde{\Z}_{\tilde{s}}, \tilde{\Z}_{s}, \betavec_0}\{(\tilde{\Z}_{\tilde{s}, \ell})^* \eta^s_{\ell}(\betavec_0 - \sqrt{\tau^{s}_{\sfc}} \tilde{\Z}_{s})\} \Big] \Big \lvert  \geq  \e/6 \Big) \\
&\leq  \exp\{-\kappa_1 L (\omega/\Lr) \e^2 \} = \exp\{-\kappa \Mr \omega  \e^2/ \log M\}.
\end{split}
\een
For  $T_{1b}$, we recall that $\sum_{j \in sec(\ell)} \eta^s_{j}( \betavec_0 - \tilde{\hvec}^{s+1} - \tilde{\Dvec}_{s+1} ) = 1$ and use \eqref{eq:maxtildedelta0} with $\sfq =1$ to obtain
\ben
\begin{split}
&T_{1b} \leq  P\Big( \frac{1}{L} \sum_{\sfc \in [\Lc]} \sum_{\ell \in \sfc} W_{\sfr \sfc} \Big \lvert  (\tilde{\Dvec}_{\tilde{s} + 1, \ell})^* \eta^s_{\ell}(\betavec_0 - \tilde{\hvec}^{s+1} - \tilde{\Dvec}_{s+1}) \Big \lvert  \geq  \epsilon/6 \Big) \\
& \leq  P\Big( \frac{1}{L} \sum_{\sfc \in [\Lc]} \sum_{\ell \in \sfc} W_{\sfr \sfc} \max_{j \in sec(\ell)} \abs{[\tilde{\Dvec}_{\tilde{s} + 1, \sfc}]_{j}}  \geq \kappa \epsilon \Big)  \leq    t^4 K K'_{t-1}   \PC_{t-1}'  \exp\{-\frac{1}{t^8} \kappa \kappa'_{t-1} (\omega/\Lr)^{2}  \pc_{t-1}'  \e^2 \} .
\end{split}
\een
Finally for $T_{1c}$, using Lemma \ref{lem:ideal_cond_dist} Eq.\ \eqref{eq:htilde_def} to write $\tilde{\hvec}^{\tilde{s} + 1}_{\sfc} =  \sqrt{\tau^{\tilde{s}}_{\sfc}} \tilde{\Z}_{\tilde{s}, \sfc}$, we can prove the bound as in the $\mc{H}_1(d)$ step in \eqref{eq:H1c_T3_eq1} - \eqref{eq:H1c_T3_eq2}, appealing to Lemma \ref{lem:BC9}, Lemma \ref{lem:max_abs_normals}, and \eqref{eq:maxtildedelta0}.

%%%%%%

\noindent  \textbf{(e)}  Using the conditional distribution of $\hvec^{t+1}$ from Lemma \ref{lem:ideal_cond_dist} Eq.\ \eqref{eq:htil_rep} and Eq.\ \eqref{eq:htilde_def} we write,
%
%\be
%
$\hvec^{t + 1}_{\sfc} \overset{d}{=} \tilde{\hvec}^{t + 1}_{\sfc} + \tilde{\Dvec}_{t+1, \sfc} \overset{d}{=}  \sqrt{\tau^{t}_{\sfc}} \tilde{\Z}_{t, \sfc} + \tilde{\Dvec}_{t+1, \sfc}.$
%
%\label{eq:htilde_equality}
%\ee
%
Then we have as in the $\mc{H}_1(e)$ proof, namely \eqref{eq:H1e_eq1},
\begin{align*}
&P\Big( \frac{1}{L}  \sum_{\sfc \in [\Lc]}  \sum_{\ell \in \sfc} W^{2\sfq}_{\sfr \sfc}  \max_{j \in sec(\ell)}  ([\hvec_{\sfc}^{t+1}]_{j})^2 \geq 6 (\max_{\sfc} W^{2\sfq}_{\sfr \sfc}) \tau^t_{\sfc} \log M + \e \Big) \leq \\
& P\Big(  \frac{1}{L}  \sum_{\sfc \in [\Lc]}  \sum_{\ell \in \sfc} W^{2\sfq}_{\sfr \sfc}  \max_{j \in sec(\ell)}   ([\tilde{\Z}_{t, \sfc}]_{j})^2 \geq 3  (\max_{\sfc} W^{2\sfq}_{\sfr \sfc}) \log M \Big) \hspace{-3pt} + \hspace{-3pt}  P\Big(  \frac{1}{L}  \sum_{\sfc \in [\Lc]}  \sum_{\ell \in \sfc} W^{2\sfq}_{\sfr \sfc}  \max_{j \in sec(\ell)}   ([\tilde{\Dvec}_{t+1, \sfc}]_{j})^2 \geq \frac{\e}{2} \Big).
\end{align*}
Then the desired bound follows by Lemma \ref{lem:max_abs_normals}  and \eqref{eq:maxtildedelta0}.

%%%%

\textbf{(f)} The proof of \eqref{eq:Hf} follows similarly to the corresponding $\mc{B}_t (f)$ proof of \eqref{eq:Bf} using  inductive hypothesis $\mathcal{H}_{t} (g)$ result \eqref{eq:Hg} and $\mathcal{H}_{t+1} (c)$.   The proof of \eqref{eq:Hf1} follows similarly to the corresponding $\mc{B}_t (f)$ proof of \eqref{eq:Bf1} using $\mathcal{H}_{t+1} (c)$ and $\mathcal{H}_{t+1} (f)$ result \eqref{eq:Hf}.

%%%%%

\textbf{(g)}  The proof of \eqref{eq:Qsing} follows similarly to the corresponding $\mc{B}_t (g)$ proof of \eqref{eq:Msing} using Fact \ref{fact:eig_proj} and inductive hypotheses $\mc{H}_1 (g) -\mc{H}_{t+1} (g)$.   Result \eqref{eq:Hg} follows similarly to the corresponding $\mc{B}_t (g)$ proof of \eqref{eq:Bf}: using block inversion we can represent $(\tCmat^{t+2, \sfr})^{-1}$ and $(\innerQ^{\sfr}_{t+2} )^{-1} $.  From the representation of $(\tCmat^{t+2, \sfr})^{-1}$, we can see that each non-zero element is $\Theta(1)$, and then we can show the element-wise concentration with inductive hypothesis $\mc{H}_{t}(g)$ and $\mathcal{H}_{t+1} (f)$.

\newpage
\appendices

% !TEX root =  sc_sparc_journal_paper_full.tex

\section{Proofs of bounds on state evolution parameters}

%%%%%%%
\subsection{Proof of Lemma \ref{lem:psi_nonasymp}} \label{app:psi_nonasymp}

 Recall that $\psi_\sfc^{t+1} = 1 - \mc{E}(\tau_\sfc^t)$ where $\mc{E}(\tau_\sfc^t)$ is defined  in \eqref{eq:E_tau}.   It is shown in \cite[Appendix A]{RushV19} that for sufficiently large $\M$ and any $\delta\in(0,\frac{1}{2})$,
\be
\mc{E}(\tau_\sfc^t) \geq \left(1- \frac{\M^{-k \delta^2}}{\delta\sqrt{\log{M}}}   \right) \indic\{\nu_\sfc^t > 2+\delta\}    \qquad \sfc \in [\Lc].
\label{eq:Etau_bound}
\ee
The upper bound on $\psi_\sfc^t$ directly follows from the lower bound in \eqref{eq:Etau_bound}.

We now lower bound $\psi_\sfc^t$ by obtaining an upper bound for $\mc{E}(\tau_\sfc^t)$. We will use the following concentration inequality the maximum of  $M$ i.i.d. standard Gaussian random variables  $U_1, \ldots, U_M$.  For any $\e \in (0,1)$,
\be P\Big( \max_{1 \leq j \leq M} U_j < \sqrt{2 \ln M}(1-\e) \Big) \leq \exp\Big( \frac{-\kappa  M^{\e(2-\e)}}{\sqrt{\ln M}} \Big), 
\label{eq:max_std_Gaussians}
\ee
where $\kappa >0$ is a universal constant.

Consider a column block $\sfc$ such that $\nu_\sfc^t \leq 2- \delt$. Recall from \eqref{eq:nu_c^t} that $\nu_\sfc^t = (\tau_\sfc^t \ln M)^{-1}$ and due to the assumption on the base matrix,  $\kappa_1 \leq \nu_\sfc^t \leq \kappa_2$ for some strictly positive constants $\kappa_1, \kappa_2$. Then, with positive constants $a, a^{\prime} \in (0,1)$ to be specified later,  we have 
\begin{align}
& \mc{E}(\tau_\sfc^t)   = \expec \left[  \frac{1}{1 +   
 e^{- \nu_\sfc^t \ln M }\sum_{j=2}^{\M} e^{U_j \sqrt{\nu_\sfc^t \ln M} } \cdot e^{-U_1 \sqrt{\nu_\sfc^t \ln M} }  } \right] \nonumber \\
& \leq P\left( \max_{2 \leq j \leq M} U_j   <  \sqrt{2 \ln M}(1 - a\delt) \right) \cdot 1  \nonumber \\
& \ +   P\left( \max_{2 \leq j \leq M} U_j   \geq   \sqrt{2 \ln M}(1 - a\delt) \right) 
 \expec \left(  {1 + \exp\left(-\sqrt{\nu_\sfc^t \ln M} \, U_1  \right)  M^{-\nu_\sfc^t} \exp(\sqrt{2\nu_\sfc^t}(1-a\delt) \ln M ) } \right)^{-1} \nonumber  \\
 & \stackrel{(a)}{\leq} \exp\left( \frac{-\kappa  M^{a \delt(2-a \delt)}}{\sqrt{\ln M}} \right) +  
 \expec \left( { 1+ \exp\left(-\sqrt{\nu_\sfc^t \ln M} \, U_1  \right) M^{\sqrt{2\nu_\sfc^t}(1-a\delt) -\nu_\sfc^t} } \right)^{-1}  \nonumber \\
 & \stackrel{(b)}{\leq}  \exp\left( \frac{-\kappa  M^{a \delt(2-a \delt)}}{\sqrt{\ln M}} \right) +  P\left(  U_1  \leq \sqrt{2 \ln M}a^{\prime} \delt \right) \cdot   \frac{1}{1 +  M^{\sqrt{2\nu_\sfc^t}(1- (a + a^{\prime})\delt) -\nu_\sfc^t} } + P\left(  U_1  > \sqrt{2 \ln M}a^{\prime} \delt \right) \nonumber \\
  &  \stackrel{(c)}{\leq}    \exp\left( \frac{-\kappa  M^{a \delt(2-a \delt)}}{\sqrt{\ln M}} \right)  
  + M^{-\left(\sqrt{2(2 - \delt)}(1 - (a + a^{\prime}) \delt) - 2 + \delt \right)} + M^{-(a^{\prime})^2 \delt^2}. \label{eq:elub}
\end{align}
In the above sequence of inequalities, step $(a)$ uses \eqref{eq:max_std_Gaussians}; step $(b)$ holds because
 \[  \exp\left(-\sqrt{\nu_\sfc^t \ln M} \, U_1  \right)  \geq M^{-\sqrt{2 \nu_\sfc^t} a' \delt} \quad  \text{ when } \quad U_1  \leq  \sqrt{2 \ln M}a^{\prime} \delt.  \]
 In step $(c)$, the last term is obtained using a standard Gaussian tail bound for $U_1$ (Lemma \ref{lem:normalconc}); the second  term is obtained by noting that  $\sqrt{2\nu_\sfc^t}(1- (a + a^{\prime})\delt) -\nu_\sfc^t$ is a concave function of $\nu_\sfc^t$ which (for $a, a^{\prime}$ specified below) takes its minimum value in $\nu_\sfc^t \in [\kappa_1, 2-\delt]$ at the endpoint $\nu_\sfc^t = (2-\delt)$.   
 
 Using the inequality $(1-x)^{1/2} \geq 1 - \frac{5x}{8}$ for $x \in [0,1/2]$, we have the following lower bound for the exponent of the second term in \eqref{eq:elub}:
\begin{align*}
\sqrt{2(2 - \delt)}(1 - (a + a^{\prime}) \delt) -2 +\delt & \geq 2 \Big(1 - \frac{5\delt}{16} \Big)
 \left( 1 - (a + a^{\prime})\delt \right) -2 + \delt \\
& = \delt \Big( \frac{3}{8} - 2(a + a^{\prime}) \Big) + \frac{5}{8}(a + a^{\prime})\delt^2.
\end{align*}
Taking $a=1/64$ and $a^{\prime} = 5/32$, we get the desired upper bound.

%%%%%%%
\subsection{Proof of Proposition \ref{prop:nonasymp_se}} \label{app:nonasymp_se}

Using the definition in \eqref{eq:W_rc}, the state evolution equations  \eqref{eq:se_phi}--\eqref{eq:tau_ct_def}  for the $(\omega, \Lambda, \rho)$ base matrix  are as follows. With $\psi_\sfc^0 =1$ for $\sfc \in [\Lc]$, for $t \geq 0$:
\begin{align}
\phi_\sfr^t &= \sigma^2 \left[ 1 + \vth \, \snr \,\left(
\frac{1-\rho}{\omega} \sum_{\sfc= \underline{\sfc}_\sfr }^{\overline{\sfc}_\sfr} \psi_\sfc^t 
\,+\,\frac{\rho}{\Lambda-1} \sum_{\sfc\in[\Lambda]\setminus \{\underline{\sfc}_\sfr,\ldots, \overline{\sfc}_\sfr \}} \psi_\sfc^t 
\right)\right], \quad \sfr \in [ \Lambda + \omega -1], \label{eq:se_phi_omega_rho} \\
\nu_\sfc^t & =  \frac{1}{\tau_\sfc^t \ln M} = \frac{1}{R} \left[\frac{(1-\rho)\cdot P}{\omega}\sum_{\sfr=\sfc}^{\sfc+\omega-1} \frac{1}{\phi_\sfr^t} 
\, + \, \frac{\rho \cdot P}{\Lambda-1}\sum_{\sfr\in [\Lambda+\omega-1]\setminus\{\sfc,\ldots,\sfc+\omega-1\}} \frac{1}{\phi_\sfr^t}\right], \label{eq:nu_c^t_def_omega_rho} \\
\psi_\sfc^{t+1} &= 1 - \mathcal{E}(\tau_\sfc^t). \label{eq:se_psi_omega_rho}
\end{align}
Here $\mathcal{E}(\tau_\sfc^t)$ is defined in \eqref{eq:E_tau}, and  $\underline{\sfc}_\sfr, \overline{\sfc}_\sfr $  are defined  \eqref{eq:c_r}. 

Since the variables $\psi_\sfc^t$ for $\sfc\in[\Lambda]$ and $t\geq0$ are symmetric about the center column index, i.e. $\psi_{\sfc}^t =  \psi_{\Lambda - \sfc + 1}^t$ for $\sfc\leq\left\lceil \frac{\Lambda}{2} \right\rceil$, we carry out the analysis for  $\sfc\leq\left\lceil \frac{\Lambda}{2} \right\rceil$; the result for the other half then holds by symmetry. We will upper bound $\psi_\sfc^t$ using Lemma \ref{lem:psi_nonasymp}.
Using \eqref{eq:psi_c_bound}, for  the first iteration  we will have $\psi_\sfc^1 \leq f_{M, \delta}$ for indices $\sfc$ for which $\{ \nu_\sfc^0 > 2 + \delta \}$. Letting $F_\sfc^0 \ := \ \nu_\sfc^0 \, R$, this corresponds to finding indices $\sfc$ for which $F_\sfc^0 > (2 + \delta)R$. We now obtain a lower bound on $F_\sfc^0$ for indices $\sfc < \omega$. 

Using \eqref{eq:nu_c^t_def_omega_rho} we have
\begin{align}
F_\sfc^0 
&= \frac{(1-\rho) P}{\omega}\sum_{\sfr=\sfc}^{\sfc+\omega-1} \frac{1}{\phi_\sfr^0} 
\, + \, \frac{\rho P}{\Lambda-1}\sum_{\sfr\in [\Lambda+\omega-1]\setminus \{\sfc,\ldots,\sfc+\omega-1\}} \frac{1}{\phi_\sfr^0} \label{eq:Fc_def} \\
&\geq \frac{(1-\rho) P}{\omega}\sum_{\sfr=\sfc}^{\sfc+\omega-1} \frac{1}{\phi_\sfr^0} \nonumber \\
&\stackrel{\text{(i)}}{=} \frac{(1-\rho)\,\snr}{\omega} \left(
\sum_{\sfr=\sfc}^{\omega-1} \frac{1}{1 + \frac{(1-\rho)\,\vth \,\snr}{\omega}\,\sfr + \frac{\rho\,\vth\, \snr}{\Lambda-1} \,(\Lambda-\sfr)} 
+ \frac{\sfc}{1+(1-\rho)\,\vth \,\snr + \rho\,\vth\,\snr \, \frac{\Lambda-\omega}{\Lambda-1}}\right) \nonumber \\
&\geq \frac{(1-\rho)\,\snr}{\omega} 
\sum_{\sfr=\sfc}^{\omega-1} \frac{1}{1 + \rho\,\vth\, \snr + (1-\rho)\,\vth \,\snr \, \frac{\sfr}{\omega}} 
+ \frac{\sfc}{\omega}\,\frac{(1-\rho)\,\snr}{1+ \vth \,\snr} \nonumber \\
&\stackrel{\text{(ii)}}{\geq} \frac{1}{\vth} \left(\ln(1+\vth\,\snr) - \ln\Big(1+\rho\,\vth\,\snr + (1-\rho)\,\vth\,\snr\,\frac{\sfc}{\omega} \Big)\right)
+ \frac{\sfc}{\omega}\,\frac{(1-\rho)\,\snr}{1+ \vth \,\snr} \nonumber \\
&\stackrel{\text{(iii)}}{\geq} \frac{1}{\vth} \ln(1+\vth\,\snr) - \rho\,\snr - (1-\rho)\,\snr\,\frac{\sfc}{\omega}
+ \frac{\sfc}{\omega}\,\frac{(1-\rho)\,\snr}{1+ \vth \,\snr} \nonumber \\
%
%&= \frac{1}{\vth} \ln(1+\vth\,\snr) - \rho\, \snr - \frac{\sfc}{\omega}\,\frac{(1-\rho)\,\vth\, \snr^2}{1+\vth\,\snr} \nonumber \\
%
&\geq \frac{1}{\vth} \ln(1+\vth\,\snr) - \rho\, \snr - \frac{\sfc}{\omega}\,\frac{\vth\, \snr^2}{1+\vth\,\snr}, \nonumber
\end{align}
where the labelled steps are obtained as follows: (i) using the expression for $\phi_\sfr^t$ in \eqref{eq:se_phi_omega_rho} and the fact that $\sfc < \omega$,  (ii) using a definite integral to lower bound the left Riemann sum of a positive decreasing function: 
\be
\frac{1}{\omega} \sum_{\sfr=\sfc}^{\omega-1} \frac{1}{ \frac{1 + \rho\,\vth\, \snr}{(1-\rho)\,\vth \,\snr}  + \, \frac{\sfr}{\omega}}   \geq \int_{\sfc/\omega}^1 \,  \frac{1}{ \frac{1 + \rho\,\vth\, \snr}{(1-\rho)\,\vth \,\snr}  + \, x} \, \textrm{d} x.
\label{eq:Reimann_bnd}
 \ee
Inequality (iii) is obtained using $\ln (1+x) \leq x$. 

Therefore, the condition $F_\sfc^0 > (2 + \delta)R$ will be satisfied if
\begin{align}\label{eq:Fc_lb_condition}
\frac{1}{\vth} \ln(1+\vth\,\snr) - \rho\, \snr - \frac{\sfc}{\omega}\,\frac{\vth\, \snr^2}{1+\vth\,\snr}
\,> \, (2+\delta)R.
\end{align}
Rearranging \eqref{eq:Fc_lb_condition} gives
\be\label{eq:c_omega_cond}
\frac{\sfc}{\omega} \, < \,
2\left(1+\frac{1}{\vth\,\snr} \right) \frac{1}{\snr} \left[\frac{1}{2\vth} \ln\left(1+\vth\,\snr \right) 
- \frac{\rho\,\snr}{2}  - R - \frac{\delta R}{2}\right].
\ee
Note that the RHS of \eqref{eq:c_omega_cond} is smaller than or equal to 1 if $R \geq \frac{(1-\rho)\snr}{(2+\delta)(1+\vth\,\snr)}$.
Using $\rho \leq  \frac{\Delta}{3\snr}$ and $\delta <  \frac{\Delta}{2R}$, the sufficient condition in \eqref{eq:c_omega_cond} for $\psi_\sfc^1 \leq f_{M, \delta}$ can be weakened to $\sfc \leq g$ where $g$ is defined in \eqref{eq:g_def}. Note that the condition \eqref{eq:omega_thresh} on $\omega$ guarantees that $g > 1$.

Notice from \eqref{eq:Fc_def}  that $F_\sfc^0$ is decreasing in $\sfc$ for $\sfc\in [1, \omega]$ and is then constant for  $\sfc\in [\,\omega, \ceil{\frac{\Lambda}{2}}]$. Therefore, for any $\delta \in (0,1)$, if  $F_\sfc^0>(2+\delta)R$ is satisfied for $\sfc=\omega$ then $\psi_\sfc^1 \leq f_{M,\delta}$ for all $\sfc\in[\Lambda]$.
By using a similar analysis as above for lower bounding $F_\sfc^0$, one can show that a sufficient condition for $F_\omega^0>(2+\delta)R$ is $R < \frac{(1-\rho)\snr}{(2+\delta)(1+\vth\,\snr)}$.

Next we consider subsequent iterations $t>1$. Assume towards induction that 
\be 
\psi_\sfc^t \leq f_{M,\delta}, \ \text{ for }  \sfc \leq g_t, 
\label{eq:ind_assump}
\ee 
where $g_t \geq  t g $. We will prove that \eqref{eq:ind_assump} implies  $\psi_\sfc^{t+1} \leq f_{M,\delta}$ for $\sfc \leq g_t + g$. We prove the result for $g_t \geq \omega$, with the other case being similar.  We wish to find column indices $\sfc\in(g_t,g_t+\omega)$ for which $\psi_\sfc^{t+1} \leq f_{M,\delta}$, or equivalently $F_\sfc^t > (2+\delta)R$. For brevity, we will use the shorthand $f := f_{M,\delta}$.  Using the induction assumption \eqref{eq:ind_assump} in \eqref{eq:se_phi_omega_rho}, we deduce 
\begin{align}\label{eq:phi_over_sigma_2}
	 \frac{\phi^t_\sfr}{\sigma^2} \leq 
 	\begin{cases}  
		1+ f\,(1-\rho)\,\vth\,\snr \, \frac{\sfr}{\omega} + f\, \rho\,\vth\,\snr \,\frac{g_t-\sfr}{\Lambda-1} + \rho\,\vth\,\snr \,\frac{\Lambda-g_t}{\Lambda-1}, & 1 \leq \sfr \leq \omega,\\
		1+ f\,(1-\rho)\,\vth\,\snr + f\, \rho\,\vth\,\snr \,\frac{g_t-\omega}{\Lambda-1} + \rho\,\vth\,\snr \,\frac{\Lambda-g_t}{\Lambda-1}, &  \omega \leq \sfr \leq g_t, \\
		1+\frac{(1-\rho)\,\vth\,\snr}{\omega}\left[f (\omega - (\sfr-g_t)) + (\sfr-g_t) \right] + f\, \rho\,\vth\,\snr \,\frac{\sfr-\omega}{\Lambda-1} + \rho\,\vth\,\snr \,\frac{\Lambda-\sfr}{\Lambda-1} ,  &  g_t \leq \sfr < g_t + \omega, \\
		1+ (1-\rho)\,\vth\,\snr + f\,\rho\,\vth\,\snr \,\frac{\sfr-\omega}{\Lambda-1} + \rho\,\vth\,\snr \,\frac{\Lambda-\sfr}{\Lambda-1}, & \sfr \geq g_t + \omega.
	\end{cases}  
\end{align}
For $M$ sufficiently large (i.e., $f$ sufficiently small), noting that $g_t \geq \omega$ we can simplify \eqref{eq:phi_over_sigma_2} to
\begin{align}\label{eq:phi_over_sigma_3}
	 \frac{\phi^t_\sfr}{\sigma^2} \leq 
 	\begin{cases}  
		1+ f\,(1-\rho)\,\vth\,\snr \, \frac{\sfr}{\omega} + \rho\,\vth\,\snr \,\frac{\Lambda-\sfr}{\Lambda-1}, & 1 \leq \sfr \leq \omega,\\
		1+ f\,(1-\rho)\,\vth\,\snr + \rho\,\vth\,\snr \,\frac{\Lambda-\omega}{\Lambda-1}, &  \omega \leq \sfr \leq g_t, \\
		1+\frac{(1-\rho)\,\vth\,\snr}{\omega}\left[f (\omega - (\sfr-g_t)) + (\sfr-g_t) \right] + \rho\,\vth\,\snr \,\frac{\Lambda-\omega}{\Lambda-1},  &  g_t \leq \sfr < g_t + \omega, \\
		1+ (1-\rho)\,\vth\,\snr + \rho\,\vth\,\snr \,\frac{\Lambda-\omega}{\Lambda-1}, & \sfr \geq g_t + \omega.
	\end{cases}  
\end{align}

We now obtain a lower bound on $F_\sfc^t := \nu^t_\sfc R$ for $g_t < \sfc < g_t+\omega$. Using \eqref{eq:nu_c^t_def_omega_rho} we have
\begin{align}
F_\sfc^t &\geq \frac{(1-\rho) P}{\omega}\sum_{\sfr=\sfc}^{\sfc+\omega-1} \frac{1}{\phi_\sfr^t} \nonumber   \\
&\stackrel{\text{(i)}}{\geq} \frac{(1-\rho)\,\snr}{\omega} 
\sum_{\sfr=\sfc}^{g_t+\omega-1} \frac{1}{1 + \frac{(1-\rho)\,\vth \,\snr}{\omega}\,[f(\omega-(r-g_t))+(r-g_t)] + \frac{\rho\,\vth\, \snr}{\Lambda-1} \,(\Lambda-\omega)}  \nonumber   \\
& \hspace{2in} + \frac{\sfc-g_t}{\omega} \, \frac{(1-\rho)\,\snr}{1+(1-\rho)\,\vth \,\snr + \rho\,\vth\,\snr \, \frac{\Lambda-\omega}{\Lambda-1}}  \nonumber    \\
&\geq \frac{(1-\rho)\,\snr}{\omega} 
\sum_{\sfr=\sfc}^{g_t+\omega-1} \frac{1}{1 + \rho\,\vth\, \snr + f\,(1-\rho)\,\vth\,\snr +  (1-f)(1-\rho)\,\vth \,\snr\, \frac{(\sfr-g_t)}{\omega}} 
+ \frac{\sfc-g_t}{\omega} \, \frac{(1-\rho)\,\snr}{1+\vth\,\snr} \nonumber    \\
&\geq \frac{(1-\rho)\,\snr}{\omega} 
\sum_{\sfr'=\sfc-g^t}^{\omega-1} \frac{1}{1 + \rho\,\vth\, \snr + f\,(1-\rho)\,\vth\,\snr + (1-f)(1-\rho)\,\vth \,\snr \, \frac{\sfr'}{\omega}} 
+ \frac{\sfc-g_t}{\omega} \, \frac{(1-\rho)\,\snr}{1+\vth\,\snr} \nonumber   \\
&\stackrel{\text{(ii)}}{\geq} \frac{1}{1-f}\, \frac{1}{\vth} \left(\ln(1+\vth\,\snr) - \ln\Big(1 + \rho\,\vth\, \snr + f\,(1-\rho)\,\vth\,\snr + (1-f)(1-\rho)\,\vth \,\snr\,\frac{\sfc-g_t}{\omega}\Big)\right) \nonumber  \\
&\qquad + \frac{\sfc -g_t}{\omega}\,\frac{(1-\rho)\,\snr}{1+ \vth \,\snr} \nonumber \\
&\stackrel{\text{(iii)}}{\geq} \frac{1}{\vth} \ln(1+\vth\,\snr) - \rho\,\snr - f\,(1-\rho)\,\snr - (1-\rho)\,\snr\,\frac{\sfc -g_t}{\omega}
+ \frac{\sfc-g_t}{\omega}\,\frac{(1-\rho)\,\snr}{1+ \vth \,\snr} \nonumber \\
%
%&= \frac{1}{\vth} \ln(1+\vth\,\snr) - \rho\, \snr - f\,(1-\rho)\,\snr - \frac{\sfc-g_t}{\omega}\,\frac{(1-\rho)\,\vth\, \snr^2}{1+\vth\,\snr}\\
%
& \geq \frac{1}{\vth} \ln(1+\vth\,\snr) - \rho\, \snr - f\,\snr - \frac{\sfc-g_t}{\omega}\,\frac{\vth\, \snr^2}{1+\vth\,\snr},
\label{eq:Fct_lb} 
\end{align}
where the labelled steps are obtained as follows: (i) using the bounds for $\phi_\sfr^t$ given in \eqref{eq:phi_over_sigma_3}, (ii) using a definite integral to lower bound the left Riemann sum of a decreasing function, similar to \eqref{eq:Reimann_bnd}, and (iii) using the inequalities $\ln (1+x) \leq x$ and $\frac{1}{1-f} \geq 1$.

Recall from Lemma \ref{lem:psi_nonasymp} that $\psi_\sfc^{t+1} \leq f_{M,\delta}$ if $F_\sfc^t > (2 + \delta)R$. From \eqref{eq:Fct_lb}, this condition will be satisfied if
\begin{align}\label{eq:Fct_lb_condition}
\frac{1}{\vth} \ln(1+\vth\,\snr) - \rho\, \snr - f\,\snr - \frac{\sfc-g_t}{\omega}\,\frac{\vth\, \snr^2}{1+\vth\,\snr}
\,> \, (2+\delta)R.
\end{align}
Rearranging \eqref{eq:Fct_lb_condition} gives
\be\label{eq:c_gt_cond}
\frac{\sfc-g_t}{\omega}   < 
2\left(1+\frac{1}{\vth\,\snr} \right) \frac{1}{\snr} \left[\frac{1}{2\vth} \ln\left(1+\vth\,\snr \right) 
- \frac{\rho\,\snr}{2} - \frac{f\,\snr}{2} - R - \frac{\delta R}{2}\right]. 
\ee
Note that the RHS of \eqref{eq:c_gt_cond} is smaller than or equal to 1 if $R \geq \frac{(1-\rho)\snr}{(2+\delta)(1+\vth\,\snr)}$.
Using  $\rho \leq \frac{\Delta}{3\snr}$ and $\delta < \frac{\Delta}{2R}$ in \eqref{eq:c_gt_cond}, we obtain that a sufficient condition for $\psi_\sfc^{t+1} \leq f_{M,\delta}$ is
\be
\label{eq:c_gt_cond2}
\frac{\sfc-g_t}{\omega}  <  2\left(1+\frac{1}{\vth\,\snr} \right) \frac{1}{\snr}\left[ \frac{7 \Delta}{12} - \frac{f \snr}{2} \right]
\ee
 For $M$ sufficiently large, $ f < \Delta/(6 \snr)$.  Thus we conclude from \eqref{eq:c_gt_cond2} that $\psi_\sfc^{t+1} \leq f_{M,\delta}$ for $\sfc -g_t \leq g$, and hence for $\sfc \leq (t+1)g$ (since $g_t \geq tg$). 
%

%%%%%%%%%%%%

\subsection{Proof of Lemma \ref{lem:sigmatperp}} \label{app:sigmatperp}

We will prove the lemma by showing that the following inequalities hold $1 \leq t <T$:
\begin{align}
 \sigma_\sfr^{t-1} - \sigma_\sfr^{t}  \geq C_1  \left(\frac{\omega}{\Lambda}\right),  \label{eq:sig_gap} \\ 
\nu_\sfc^t -\nu_\sfc^{t-1} \geq C_2 \left( \frac{\omega}{\Lambda}\right), \label{eq:nu_gap}
\end{align}
where we recall that $\nu_\sfc^t = \frac{1}{\tau_\sfc^t \ln M}$.  Here the constant $C_1$ is defined in \eqref{eq:diff_t1_proof} below, and $C_2 = P C_1/(R \sigma^4)$. From these inequalities it follows that 
\begin{align}
\sigma_{\perp, \sfr}^t & = \frac{\sigma_\sfr^t}{\sigma_\sfr^{t-1}}\left( \sigma_\sfr^{t-1} - \sigma_\sfr^{t} \right) \geq \frac{C_1^2}{\vth P}\left(\frac{\omega}{\Lambda}\right)^2, \label{eq:sig_perp_bnd} \\
\frac{n}{L} \tau_{\perp, \sfc}^t   & =  \frac{(\nu_\sfc^{t-1} - \nu_\sfc^t)}{R (\nu_\sfc^t)^2} \geq \frac{R C_2}{\snr^2}\left(\frac{\omega}{\Lambda}\right).  \label{eq:tau_perp_bnd}
\end{align}
The inequality in \eqref{eq:sig_perp_bnd} is obtained   using $ \sigma_\sfr^{t} \geq C_1\omega/\Lambda$ and the fact that 
$\sigma_\sfr^{t-1} \leq \vth P$ (from \eqref{eq:se_phi_omega_rho} since $\psi_{\sfc}^{t-1} \leq 1$). The inequality in \eqref{eq:tau_perp_bnd} is obtained by using  $\nu_\sfc^t \geq \snr/R$ (from \eqref{eq:nu_c^t_def_omega_rho} since $\phi_\sfr^t \geq \sigma^2$).

We will prove  \eqref{eq:sig_gap} via induction. The result \eqref{eq:nu_gap} can then be obtained from \eqref{eq:sig_gap} as follows. 
Using \eqref{eq:nu_c^t_def_omega_rho}, we write 
\begin{align*}
\nu_\sfc^{t} -\nu_\sfc^{t-1} = \frac{1}{R} \left[\frac{(1-\rho)\cdot P}{\omega}\sum_{\sfr=\sfc}^{\sfc+\omega-1}\left( \frac{1}{\phi_\sfr^{t}} -   \frac{1}{\phi_\sfr^{t-1}}  \right) 
\, + \, \frac{\rho \cdot P}{\Lambda-1}\sum_{\sfr\in [\Lambda+\omega-1]\setminus\{\sfc,\ldots,\sfc+\omega-1\}} \left( \frac{1}{\phi_\sfr^t} -   \frac{1}{\phi_\sfr^{t-1}}  \right)  \right].
\end{align*}
Using  $\phi^{t-1}_\sfr \geq \phi^t_\sfr \geq \sigma^2$ and
   $(\phi_\sfr^{t-1} -  \phi_\sfr^t) = (\sigma_\sfr^{t-1} -  \sigma_\sfr^t) \geq C_1 \omega /\Lambda$ by the induction assumption,  the above equation yields
\be
\nu_\sfc^{t} -\nu_\sfc^{t-1} \geq \frac{PC_1}{R \sigma^4} \frac{ \omega}{\Lambda} = \frac{C_2 \omega}{\Lambda}.
\ee
Next  we show the lower bound for $ \sigma_\sfr^{t} - \sigma_\sfr^{t+1}$ in \eqref{eq:sig_gap}.  For $t=1$, noting that 
$\psi_\sfc=1$ for $\sfc \in [\Lc]$, we have 
\be
\sigma_\sfc^{1} -\sigma_\sfc^{0} =  \frac{1}{\Lc}\sum_{c=1}^{\Lc}W_{\sfr \sfc} ( 1 - \psi_\sfc^1).
\ee
From Proposition \ref{prop:nonasymp_se}, we know that $ \psi_\sfc^1 \leq f_{M, \delta}$, for $1 \leq \sfc \leq g$ and
$\Lambda - g + 1 \leq \sfc  \leq \Lambda$.  Since $W_{\sfr \sfc}  \geq \rho P \cdot \frac{\Lambda + \omega -1}{\Lambda-1} $ and $\Lc =\Lambda$, we have
\be
\sigma_\sfc^{1} -\sigma_\sfc^{0} \geq \frac{1}{\Lambda} 2 g \rho P \,  \frac{\Lambda + \omega -1}{\Lambda -1} (1 - f_{M, \delta} )\geq
\underbrace{ \frac{ \rho P (1 + \vth \, \snr)\Delta }{\vth \, \snr^2}}_{C_1}  \, \frac{\omega}{\Lambda-1},
\label{eq:diff_t1_proof}
\ee
where the last inequality holds for $M$ sufficiently large. Assume towards induction that \eqref{eq:sig_gap} holds for $ \sigma_\sfr^{t-1} - \sigma_\sfr^{t}$, for $1 \leq t \leq (T-2)$. From \eqref{eq:se_phi_omega_rho} we have
\begin{align}
\label{eq:sig_diff_lb}
 \sigma_\sfr^{t} - \sigma_\sfr^{t+1}  &=  \vth \, P \,
\left( \frac{1-\rho}{\omega} \sum_{\sfc= \underline{\sfc}_\sfr }^{\overline{\sfc}_\sfr} (\psi_\sfc^{t}  - \psi_\sfc^{t+1})
\,+\,\frac{\rho}{\Lambda-1} \sum_{\sfc\in[\Lambda]\setminus \{\underline{\sfc}_\sfr,\ldots, \overline{\sfc}_\sfr \}} 
(\psi_\sfc^{t}  - \psi_\sfc^{t+1}) \right).
\end{align}

For $\delta, \delt \in (0, \, \min\{ \frac{1}{2}, \, \frac{\Delta}{2R} \})$, define the set
%be such that $\delta +  \delt <  \frac{C_2 \omega}{2\Lambda}$. 
\be
\mc{I}_t :=  \left\{  \sfc \in [\Lambda]: \ \nu_\sfc^{t-1} \leq 2- \delt \text{ and }  \nu_\sfc^{t} > 2+ \delta \right\}
\ee

For each $\sfc  \in \mc{I}_t$,  from Lemma \ref{lem:psi_nonasymp} and Proposition \ref{prop:nonasymp_se} we have
\be
\label{eq:psic_tt1}
\psi_\sfc^t \geq 1 - M^{-k_1 \tilde{\delta}^2}, \qquad  \psi_\sfc^{t+1} \leq  \frac{\M^{-k \delta^2}}{\delta\sqrt{\log{M}}}, \qquad \sfc \in \mc{I}_t.
\ee

We prove below that for sufficiently small choices of $\delta, \delt$, we have $\abs{\mc{I}_t} \geq 2(g-2)$, where $g$ is defined in \eqref{eq:g_def}. Using this and \eqref{eq:psic_tt1} in \eqref{eq:sig_diff_lb},
we conclude that 
 \be
 \begin{split}
 \sigma_\sfr^{t} - \sigma_\sfr^{t+1}   & \geq  \frac{\vth  P  \rho}{\Lambda-1} 2(g-2) \left( 1 -   M^{-k_1 \tilde{\delta}^2} -   \frac{\M^{-k \delta^2}}{\delta\sqrt{\log{M}}} \right) \\ 
  & \stackrel{(a)}{\geq}   \frac{ P  \rho \, g}{\Lambda-1} \, = \,  \underbrace{\frac{ P  \rho (1 + \vth \, \snr)\Delta }{\vth \, \snr^2}}_{C_1} \cdot \frac{\omega}{\Lambda-1},
 \end{split}
 \ee
where the inequality $(a)$ holds for $M$ sufficiently large. It remains to show that  $\abs{\mc{I}_t} \geq 2(g-2)$ for suitably chosen $\delta, \delt$. Let $\sfc_* $ denote the largest index $\sfc \leq  \Lambda/2$ such that $\nu_{\sfc}^{t-1} > 2$. That is, 
\be
\nu_{\sfc}^{t-1} >  2 \ \text{ for } \  \sfc \leq \sfc_*, \qquad  \nu_{\sfc_*+1}^{t-1}  \leq 2 , \qquad \text{ and } \  \nu_{\sfc}^{t-1}  < 2 \  \text{ for }  \ 
\sfc_* +2 \leq \sfc \leq \lfloor \Lambda/2 \rfloor.
\ee
Let $\delta =  \min\{ \frac{(\nu_{\sfc_*}^{t-1} -2)}{2}, \, \frac{\Delta}{2R}, \, \frac{1}{2}\}$ and $\delt = \min\{ 
\frac{(2-  \nu_{\sfc_*+2}^{t-1})}{2}, \, \frac{1}{3} \}$. For these choices, Lemma \ref{lem:psi_nonasymp} guarantees that
\begin{align}
\psi_\sfc^t  < f_{M, \delta} := \frac{\M^{-k \delta^2}}{\delta\sqrt{\log{M}}} \, \text{ for } \, \sfc \leq \sfc_*,  \quad \text{ and } \quad
\psi_\sfc^t > 1-   M^{-k_1 \tilde{\delta}^2} \, \text{ for } \, \sfc \geq \sfc_*+2,  
\end{align}
Furthermore, Proposition \ref{prop:nonasymp_se} guarantees that  $\psi_\sfc^{t+1} \leq f_{M, \delta}$ for column indices 
$\sfc_* +2 \leq \sfc \leq \sfc_* +g$. Therefore, all these indices belong to the set $\mc{I}_t$. Therefore, $\mc{I}_t$ contains  at least $(g-2)$ column indices $\sfc \leq \Lambda/2$, and by symmetry,  contains an equal number of indices $\Lambda/2 < \sfc  \leq \Lambda$. This completes the proof of the lemma.

\section{Concentration lemmas} \label{app:conc_lemmas}

In the following $\e >0$ is assumed to be a generic constant, with additional conditions specified whenever needed.  %The proofs of  Lemmas \ref{sums}--\ref{lem:max_abs_normals} can be found in \cite{RushV16}.
%
%%%%%State Hoeffding's inequality

\begin{applem}[Hoeffding's inequality  {\cite[Thm.\ 2.8]{BLMconcbook}}]
\label{lem:hoeff_lem}
If $X_1, \ldots, X_n$ are independent, bounded random variables such that $a_i \leq X_i \leq b_i$, then for $\nu = 2[\sum_{i} (b_i -a_i)^2]^{-1}$
\begin{align*}
P\Big( \frac{1}{n}\sum_{i=1}^n (X_i - \expec X_i) \geq \e \Big) &\leq e^{ -\nu n^2 \e^2}, \quad P\Big( \Big \lvert \frac{1}{n}\sum_{i=1}^n (X_i - \expec X_i) \Big \lvert \geq \e \Big) \leq 2e^{ -\nu n^2 \e^2}.
\end{align*}
\end{applem}

%
%%%%%%%%%%%
\begin{applem}[Concentration of sums]
\label{sums}
If random variables $X_1, \ldots, X_M$ satisfy $P(\abs{X_i} \geq \e) \leq e^{-r_i \e^2}$ for $1 \leq i \leq M$ with $r_i$ indicating the `rate' of concentration of random variable $X_i$, then 
\ben
P\Big(\Big \vert \sum_{i=1}^M X_i \Big \vert \geq \e\Big) \leq \sum_{i=1}^M P\Big(|X_i| \geq \frac{\e}{M}\Big) \leq M e^{- (\min_i r_i) \e^2/M^2}.
\een
Moreover, for constants $\kappa_1, \ldots, \kappa_M > 0$,
\ben
P\Big(\Big \vert  \sum_{i=1}^M \kappa_i X_i \Big \vert  \geq \e\Big) \leq \sum_{i=1}^M P\Big(|X_i| \geq \frac{\e}{\sum_{j=1}^M \kappa_{j}}\Big) \leq M e^{- (\min_i r_i) \e^2/(\sum_{j=1}^M \kappa_{j})^2}.
\een
\end{applem}

%%%%%%

\begin{applem}[Concentration of Products {\cite[Lemma 15]{rush2018finite}}]
\label{products} 
For random  variables $X,Y$ and non-zero constants $c_X, c_Y$, if
$
P(  | X- c_X |  \geq \e ) \leq K e^{- r_X \e^2}$ and $P( | Y- c_Y  |  \geq \e ) \leq K e^{-r_Y \e^2},
$
then the probability  $P(  | XY - c_Xc_Y  |  \geq \e )$ is bounded by 
\begin{align*}  
&  P(  | X- c_X  |  \geq \min\Big \{ \sqrt{\frac{\e}{3}}, \frac{\e}{3 c_Y} \Big \} )  +  
P( | Y- c_Y  |  \geq \min \Big \{ \sqrt{\frac{\e}{3}}, \frac{\e}{3 c_X} \Big \} )  \leq 2K \exp\Big\{\frac{-\min\{r_X, r_Y\} \e^2}{9\max(1, c_X^2, c_Y^2)}\Big\}.
\end{align*}
\end{applem}

\begin{applem}[Concentration of Products]
\label{products_0} 
For random variables $X,Y,$ and constant $c_X \neq 0$, if
$
P( | X- c_X  |  \geq \e) \leq K e^{- r_X \e^2}$ and $P(  | Y  |  \geq \e) \leq K e^{-r_Y \e^2},
$
then,
\begin{align*}  
& P(  | XY  |  \geq \sqrt{\e}) \leq P( \abs{X - c_X}   \geq \sqrt{\e} ) +  P\Big( \abs{Y}  \geq 
\frac{\e}{2 \max\{1, |c_X|\} } \Big)  
\leq 
2 K \exp\Big\{- \frac{\e^2 \min\{r_Y, r_X \}}{4 \max\{1, c_X^{2}\}}\Big\}.
\end{align*}
\end{applem}

%\begin{proof}
%\begin{align*}
%&P(  | XY |  \geq \e ) \leq P( (\abs{X - c_X} + c_X) \abs{Y}  \geq \e )  \overset{(a)}{\leq} P( \abs{X - c_X}   \geq \sqrt{\e} ) +  P( \abs{Y}  \geq \frac{\e}{2}\min\{1, c_X^{-1}\} ) \\
%&\leq K \exp\{- r_X \e\} + K \exp\Big\{- \frac{r_Y \e^2}{4} \min\{1, c_X^{-2}\}\Big\} \leq 2 K \exp\Big\{- \frac{\e^2}{4} \min\{r_Y, r_X \}  \min\{1, c_X^{-2}\}\Big\}.
%\end{align*}
%In the above, step $(a)$ uses the fact that if $P( \abs{X - c_X}   < \sqrt{\e} )$ and $ P( \abs{Y}  < \frac{\e}{2}\min\{1, c_X^{-1}\} )$ then 
%\[(\abs{X - c_X} + c_X) \abs{Y} < (\sqrt{\epsilon} + c_X) \frac{\e}{2}\min\{1, c_X^{-1}\} \leq \frac{\e}{2} + \frac{\e}{2}  \leq \e.\]
%\end{proof}

%%%%%%%%%%%%

\begin{applem}[Concentration of square roots, {\cite[Lemma 16]{rush2018finite}}]
\label{sqroots}
Let $c \neq 0$. Then
\ben
\text{If } P( \lvert X_n^2 - c^2 \lvert \geq \epsilon ) \leq e^{-\kappa n \epsilon^2},
\text{ then }
P ( \lvert \abs{X_n} - \abs{c}  \lvert \geq \epsilon) \leq e^{-\kappa n \abs{c}^2 \epsilon^2}.
\een
\end{applem}
%%

%%%%%%%%%%%
%
\begin{applem}[Concentration of powers, {\cite[Lemma 17]{rush2018finite}}]
\label{powers}
Assume $c \neq 0$ and $0 < \e \leq 1$.  Then for any integer $k \geq 2$,
\ben
\text{if } P( \lvert X_n - c\lvert \geq \epsilon) \leq e^{-\kappa n \epsilon^2},
\text{ then }
P( \lvert X_n^k - c^k  \lvert \geq \epsilon) \leq e^{ {-\kappa n \e^2}/[{(1+\abs{c})^k -\abs{c}^k}]^2}.
\een
\end{applem}
%

%%%%%%%%%%%%%%%%%%
\begin{applem}[Concentration of Scalar Inverses, {\cite[Lemma 18]{rush2018finite}}]
\label{inverses} Assume $c \neq 0$ and $0<\e <1$. 
\ben
\text{If } P( \lvert X_n - c \lvert \geq \epsilon ) \leq e^{-\kappa n \epsilon^2},
\text{ then }
P( \lvert X_n^{-1} - c^{-1} \lvert \geq \epsilon ) \leq 2 e^{-n \kappa \e^2 c^2 \min\{c^2, 1\}/4}.
\een
\end{applem}

%%%%%%

\begin{applem}
\label{lem:normalconc}
For a standard Gaussian random variable $Z$ and  $\e > 0$,
$P( \abs{Z} \geq \e ) \leq 2e^{-\frac{1}{2}\e^2}$.
\end{applem}

%%%%
\begin{applem}
\label{lem:max_abs_normals}
Let $Z_1, Z_2, \ldots, Z_N$ and $\tilde{Z}_1, \tilde{Z}_2, \ldots, \tilde{Z}_N$  be i.i.d.\ standard Gaussian random variables and  $0 \leq \e \leq 1$ and $\sigma_1, \sigma_2, \ldots, \sigma_N$ be positive constants.  Let $\sigma_{\max} = \max(\sigma_1, \sigma_2, \ldots, \sigma_N)$.  Then the following concentration results hold for $\e \in (0,1]$.
\begin{align}
&P\Big(\Big \lvert \frac{1}{N}\sum_{i=1}^N \sigma_i (Z_i^2 -  1) \Big \lvert \geq  \e \Big) \leq 2\exp\Big\{\frac{-N \e^2}{4  \max\{\sum_{i=1}^N 2\sigma_i^2/N, \, \sigma_{\max}, \, 1 \}}\Big\},  \label{res2} \\
& P\Big(\Big \lvert \frac{1}{N}\sum_{i=1}^N \sigma_i Z_i \tilde{Z}_i \Big \lvert \geq \e \Big) \leq 2\exp\Big\{\frac{-N \e^2}{4  \max\{\sum_{i=1}^N \sigma_i^2/N, \sigma_{\max}, \, 1 \}}\Big\},  \label{res3} \\
&P\Big(\frac{1}{L}\sum_{\ell = 1}^L \sigma_{\ell} \max_{j \in sec(\ell)}  Z_{j}^2 \geq  3\sigma_{\max} \log M + \e \Big) \leq \exp\Big\{ \frac{-L }{5 } \Big(2\e + \log \frac{M}{70} \Big) \Big\}. \label{res1} 
\end{align} 
\end{applem}
%%%%

\begin{proof}
Recall that a random variable $X$ is sub-exponential with non-negative parameters $(\nu, b)$ if
\be
\expec[\exp(\lambda(X - \expec X))] \leq \exp(\nu^2 \lambda^2/2), \quad \text{ for all } \abs{\lambda} < {1}/{b}.
\ee
Furthermore, if $X$ is sub-exponential with  parameters $(\nu, b)$, then \cite[Proposition 2.9]{wainwright2019high}
\be
P(\abs{X- \expec X} \geq t) \leq 
\begin{cases}
2\, e^{- \frac{t^2}{2 \nu^2}} \quad & \text{ for } 0 < t \leq \frac{\nu^2}{b}, \\ 
2 \, e^{-\frac{t}{b}} \quad & \text{ for } t > \frac{\nu^2}{b}.
\end{cases}
\label{eq:exp_tail_bound}
\ee
Taking $X = \sum_{i=1}^N \sigma_i Z_i^2$, we will show that $X$ is sub-exponential with 
$
\nu^2 = 4 \sum_{i=1}^N \sigma_i^2$ and $b = \max\{ 4 \sigma_{\max}, \, 1 \},$
from which it follows that (plugging in $t = N \e$ in \eqref{eq:exp_tail_bound})
\ben
P\Big(\Big \lvert \frac{1}{N} \sum_{i=1}^N \sigma_i (Z_i^2- 1) \Big \lvert \geq  \e \Big) \leq 
\begin{cases}
2\, \exp\Big\{- \frac{N \e^2}{ \frac{8}{N}\sum_{i=1}^N \sigma_i^2}\Big\} \quad & \text{ for } 0 <  \e \leq \frac{\frac{4}{N} \sum_{i=1}^N \sigma_i^2}{\max\{ 4 \sigma_{\max}, \, 1 \}}, \\ 
2 \,  \exp\Big\{-\frac{N \e}{\max\{ 4 \sigma_{\max}, \, 1 \}}\Big\} \quad & \text{ for }  \e > \frac{\frac{4}{N} \sum_{i=1}^N \sigma_i^2}{\max\{ 4 \sigma_{\max}, \, 1 \}}.
\end{cases}
\een
To finish the proof of \eqref{res2}, we now show that $X = \sum_{i=1}^N \sigma_i Z_i^2$ has sub-exponential parameters given above.  Using the moment generating function of a chi-squared random variable, we have
\ben
\expec[\exp(\lambda X)] = \prod_{i=1}^n \frac{1}{\sqrt{1- 2\lambda \sigma_i}}, \quad \text{ for } \lambda< \frac{1}{2 \sigma_{\max}}.
\een
Therefore, for  $\lambda< \min \{ \frac{1}{4 \sigma_{\max}}, \, 1 \}$, we show the desired sub-exponential parameters as follows:
\begin{align}
\expec[\exp(\lambda(X - \expec X))] 
&  = \exp\Big( - \frac{1}{2} \sum_{i=1}^N \ln(1 - 2 \lambda \sigma_i) \, - \, 
\lambda \sum_{i=1}^N \sigma_i  \Big)  \nonumber \\
&\overset{(a)}{ \leq} \exp\Big(  \sum_{i=1}^N \Big[ \lambda \sigma_i + \frac{(\lambda \sigma_i)^2}{1 - 2 \lambda \sigma_i} \Big]
- \lambda \sum_{i=1}^N \sigma_i     \Big) 
\overset{(b)}{ \leq} \exp\Big(   {2 \lambda^2} \sum_{\i=1}^N \sigma_i^2 \Big).
\label{eq:MGF_bound2}
\end{align}
In step $(a)$, we use that $-\log(1-u) \leq u + \frac{u^2}{2(1-u)}$ for $u \in [0,1)$ and step $(b)$ holds since
%%
%\[ \sum_{i=1}^N \Big[ \lambda \sigma_i + \frac{(\lambda \sigma_i)^2}{1 - 2 \lambda \sigma_i} \Big]
%- \lambda \sum_{i=1}^N \sigma_i    =  \lambda^2 \sum_{i=1}^N \frac{\sigma_i^2}{1 - 2 \lambda \sigma_i} \leq 2   \lambda^2 \sum_{i=1}^N \sigma_i^2,  \]
%%
for $\lambda<  \frac{1}{4 \sigma_{\max}}$, we have $1 - 2 \lambda \sigma_i \geq 1 - \sigma_i/(2 \sigma_{\max}) \geq 1/2.$

We prove \eqref{res3} similarly.  For $X = \sum_{i=1}^N \sigma_i Z_i \tilde{Z}_i$,  the moment generating function is
\be
\expec[\exp(\lambda X)] = \prod_{i=1}^n (1- \lambda^2 \sigma_i^2)^{-1/2}, \quad \text{ for } \lambda^2 \sigma_i^2 < 1.
\ee
Using this and steps similar to \eqref{eq:MGF_bound2}, we can show that $X$ is sub-exponential with parameters 
$\nu^2 = 2 \sum_{i=1}^N \sigma_i^2$ and $b =  \max\{ \sigma_{\max} \sqrt{{3}/{2}} , \, 1\}.$
 Then, using $t = N \e$ in  \eqref{eq:exp_tail_bound}, we obtain
\ben
P\Big(\Big \lvert \frac{1}{N} \sum_{i=1}^N \sigma_i Z_i \tilde{Z}_i \Big \lvert \geq  \e \Big) \leq 
\begin{cases}
2\, \exp\Big\{- \frac{N \e^2}{\frac{4}{N} \sum_{i=1}^N \sigma_i^2} \Big\} \quad & \text{ for } 0 < \e \leq \frac{\frac{2}{N} \sum_{i=1}^N \sigma_i^2}{  \max\{ \sigma_{\max} \sqrt{{3}/{2}} , \, 1\}.}, \\ 
2 \,  \exp\Big\{-\frac{N \e}{  \max\{ \sigma_{\max} \sqrt{{3}/{2}} , \, 1\}. } \Big\} \quad & \text{ for } \e > \frac{\frac{2}{N}  \sum_{i=1}^N \sigma_i^2}{  \max\{ \sigma_{\max} \sqrt{{3}/{2}} , \, 1\}. }.
\end{cases}
\een
This proves \eqref{res3}. The inequality \eqref{res1} is shown in \cite[Lemma 16]{RushV19}. 
\end{proof}

%%%%%
%%%%%

\section{Other useful lemmas}  \label{app:useful_lemmas}

\begin{applem} \cite[Lemma 8]{bayati2011}
Let $\mathbf{v}_1, \ldots, \mathbf{v}_t$ be a sequence of vectors in $\mathbb{R}^n$ such that  for $i \in [t]$,  $\frac{1}{n} \|  \mathbf{v}_i - \proj^{\parallel}_{i-1}(\mathbf{v}_i) \|^2 \geq c$, where $c$ is a positive constant that does not depend on $n$, and $\proj^{\parallel}_{i-1}$ is the orthogonal projection onto the span of $\mathbf{v}_1, \ldots, \mathbf{v}_{i-1}$. Then the matrix $\mathbf{C} \in \mathbb{R}^{t \times t}$ with $C_{ij} = \mathbf{v}^*_i \mathbf{v}_j / n$ has minimum eigenvalue $\lambda_{\min} \geq c'_t$, where $c'_t$ is a  positive constant (not depending on $n$).
\label{fact:eig_proj}
\end{applem}

\begin{applem}
For any scalars $a_1, ..., a_t$ and positive integer $m$, we have  $(\abs{a _1} + \ldots + \abs{a_t} )^m \leq t^{m-1} \sum_{i=1}^t \abs{a_i}^m$.
Consequently, for any vectors $\mathbf{u}_1, \ldots, \mathbf{u}_t \in \mathbb{R}^N$, $\norm{\sum_{k=1}^t \mathbf{u}_k}^2 \leq t \sum_{k=1}^t \norm{\mathbf{u}_k}^2$.
\label{lem:squaredsums}
\end{applem}

%%%%

%
\begin{applem}[Stein's lemma]
For zero-mean jointly Gaussian random variables $Z_1, Z_2$, and any function $f:\mathbb{R} \to \mathbb{R}$ for which $\expec[Z_1 f(Z_2)]$ and $\expec[f'(Z_2)]$  both exist, we have $\expec[Z_1 f(Z_2)] = \expec[Z_1Z_2] \expec[f'(Z_2)]$.
\label{lem:stein}
\end{applem}
%

%%%%%%%
%%%%%%

\begin{applem} \label{lem:expect_etar_etas}
 Let $\Z_{\tilde{s}, \sfc}, \Z_{s, \sfc} \in \reals^{\Mc}$ be random vectors such that the pairs
 % \sim \mathcal{N}(0, \mathbb{I}_{\Mc})$ such that the pairs 
$ (Z_{\tilde{s}, i}, Z_{s,i}), \ i \in [M_C]$, are i.i.d.\  bivariate Gaussian with covariance $\mathbb{E}[Z_{\tilde{s},i} Z_{s,i}] = (\tau^s_{\sfc}/\tau^{\tilde{s}}_{\sfc})$.   Then for $0 \leq \tilde{s} \leq s \leq T$,
\begin{align}
&\frac{\Lc}{L}\expec\{ [\eta^{\tilde{s}}_{\sfc}(\betavec_{0, \sfc} - \sqrt{\tau^{\tilde{s}}_{\sfc}} \Z_{\tilde{s}, \sfc})]^* [\eta^{s}_{\sfc}(\betavec_{0, \sfc} - \sqrt{\tau^s_{\sfc}} \Z_{s, \sfc}) ] \}  =  (1 - \psi^{\tilde{s}+1}_{\sfc}), \label{eq:etar_etas} \\ 
&\frac{\Lc}{L} \mathbb{E}\{[\eta^{\tilde{s}}_{\sfc}(\betavec_{0, \sfc} - \sqrt{\tau^{\tilde{s}}_{\sfc}} \Z_{\tilde{s}, \sfc}) - \betavec_{0, \sfc}]^*  [\eta^{s}_{\sfc}(\betavec_{0, \sfc} - \sqrt{\tau^s_{\sfc}} \Z_{s, \sfc}) - \betavec_{0, \sfc}] \} = \psi^{s+1}_{\sfc}. \label{eq:full}
\end{align}
\end{applem}

\begin{proof}
We will use the following fact, adapted from \cite[Proposition 1]{rush2017}:
\be
 \expec\Big\{   \betavec_{0, \sfc} ^{*}\eta^{s}_{\sfc}\Big(\betavec_{0, \sfc} - \sqrt{\tau^{s}_{\sfc}} \Z_{s, \sfc}\Big)\Big \} = \frac{L}{\Lc}  (1 - \psi^{s+1}_{\sfc}), \quad \text{ for } 0 \leq s < T.
 \label{eq:beta_etar}
 \ee

 Let $\mathbf{u}^{\tilde{s}} = \betavec_{0, \sfc} - \sqrt{\tau^{\tilde{s}}_{\sfc}} \Z_{\tilde{s}, \sfc}$ and $\mathbf{u}^s = \betavec_{0, \sfc} - \sqrt{\tau^s_{\sfc}} \Z_{s, \sfc}$. Then,  $\eta^{\tilde{s}}_{\sfc}(\betavec_{0, \sfc} - \sqrt{\tau^{\tilde{s}}_{\sfc}} \Z_{\tilde{s}, \sfc}) = \expec[\betavec_{0, \sfc} \mid \mathbf{u}^{\tilde{s}}]$ and $\eta^s_{\sfc}(\betavec_{0, \sfc} - \sqrt{\tau^s_{\sfc}} \Z_{s, \sfc}) = \expec[\betavec_{0, \sfc} \mid \mathbf{u}^s]$, and therefore, for   $\tilde{s} \leq s$,
\begin{align*}
\expec\{ [\eta^{\tilde{s}}_{\sfc}(\betavec_{0, \sfc} - \sqrt{\tau^{\tilde{s}}_{\sfc}} \Z_{\tilde{s}, \sfc})]^* [\eta^s_{\sfc}(\betavec_{0, \sfc} - \sqrt{\tau^s_{\sfc}} \Z_{s, \sfc}) ] \} &  = \expec\{ [\expec[\betavec_{0, \sfc} \mid \mathbf{u}^{\tilde{s}}]]^* \expec[\betavec_{0, \sfc} \mid \mathbf{u}^s] \} \\ 
&  \overset{(a)}{=} \expec\{ [\expec[\betavec_{0, \sfc} \mid \mathbf{u}^{\tilde{s}}]]^* [\expec[\betavec_{0, \sfc} \mid \mathbf{u}^s, \mathbf{u}^{\tilde{s}}] - \betavec_{0, \sfc} + \betavec_{0, \sfc}] \}  \\
&\overset{(b)}{=} \expec\{ [\expec[\betavec_{0, \sfc} \mid \mathbf{u}^{\tilde{s}}]]^* \betavec_{0, \sfc} \} \overset{(c)}{=} \frac{L}{\Lc}  (1 - \psi^{\tilde{s}+1}_{\sfc}). 
\end{align*}
In the above, step $(a)$ holds since $\expec[\betavec_{0, \sfc} \mid \mathbf{u}^s, \mathbf{u}^r]= \expec[\betavec_{0, \sfc} \mid \mathbf{u}^s]$, which can be shown using steps similar to those in \cite[Lemma 22]{RushV19}. Step $(b)$ follows from the orthogonality property of conditional expectation:  $ \expec\{(\expec[ \betavec_{0, \sfc} | \mathbf{u}^s , \mathbf{u}^{\tilde{s}}] - \betavec_{0, \sfc})^* \expec[ \betavec_{0, \sfc} | \mathbf{u}^{\tilde{s}} ] \} =0$ due to the orthogonality principle, and step $(c)$ by \eqref{eq:beta_etar}.  The result \eqref{eq:full} follows from \eqref{eq:etar_etas} and \eqref{eq:beta_etar}, noting that $\norm{\betavec_{0, \sfc}}^2= {L}/{L_C}$. 
\end{proof}

\begin{applem}{}
\label{lem:BC9}
For the function $\eta^t:\mathbb{R}^{ML} \rightarrow \mathbb{R}^{ML}$ defined \eqref{eq:eta_function},  $\mathbf{s}, \Dvec \in \mathbb{R}^{ML}$, and  $sec(\ell) \in \sfc$,
\ben
\sum_{i \in sec(\ell)} \abs{\eta^t_i(\mathbf{s}) - \eta^t_i(\mathbf{s} + \Dvec)}  \leq  (2/\tau^t_{\sfc}) \max_{i \in sec(\ell)} \abs{\Delta_i}. %\quad \text{ and } \quad \sum_{i \in \sfc} \abs{\eta^t_i(s) - \eta^t_i(\mathbf{s} + \Dvec)}  \leq  2\sqrt{\frac{M_C}{M}} \cdot \frac{M_R}{\tau^t_c} \norm{\Delta} .
\een
\end{applem}

\begin{proof}
From the multivariate version of Taylor's theorem, for any $i \in [ML]$ and for some $\kappa \in (0,1),$
\be
 \eta^t_i(\mathbf{s} + \Dvec) =  \eta^t_i(\mathbf{s}) + \Delta^T \nabla \eta^t_i(\mathbf{s} + \kappa \Dvec).
 \label{eq:eta_taylor}
 \ee
For  $i \in sec(\ell)$, as $\eta^t_i$ depends only on the subset of its input also belonging to section $\ell$,  using \eqref{eq:eta_taylor},
\ben
\begin{split}
&\sum_{i \in sec(\ell)} \abs{\eta^t_i(\mathbf{s}) - \eta^t_i(\mathbf{s} + \Dvec)} =  \sum_{i \in sec(\ell) } \Big \lvert \sum_{ j \in sec(\ell)} \Delta_j \frac{\partial}{\partial s_j} \eta^t_i(\mathbf{s}+ \kappa \Dvec) \Big \lvert \\
&\overset{(a)}{\leq} \frac{1}{\tau^t_{\sfc}}  \sum_{i \in sec(\ell) } \Big \lvert \Delta_i \eta^t_i(\mathbf{s} + \kappa \Dvec)  \Big \lvert +  \frac{1}{\tau^t_{\sfc}}  \sum_{i \in sec(\ell) }  \Big \lvert \eta^t_i(\mathbf{s}+ \kappa \Dvec)  \sum_{ j \in sec(\ell)}  \Delta_j  \eta^t_j(\mathbf{s}+ \kappa \Dvec) \Big \lvert    \overset{(b)}{\leq}  \frac{2}{\tau^t_{\sfc}} \max_{i \in \text{sec}(\ell) } \abs{\Delta_i}, % \leq \frac{2 M_R L}{\tau_c L_C}  \max_{i \in \sfc } \abs{\Delta_i},
\end{split}
\een
where inequality $(a)$ uses the fact that for $i,j \in [ML]$,
\[\frac{\partial}{\partial s_j} \eta^t_i(\mathbf{s}) = \frac{\eta^t_i(\mathbf{s})}{\tau^t_{\sfc}}  [ \mathbf{1}\{j = i\} - \eta^t_j(\mathbf{s})]\mathbf{1}\{i, j \in sec(\ell), \ell \in \sfc\}.\] 
Inequality $(b)$ uses the fact that $\sum_{j\in sec(\ell)}  | \eta^t_j(\mathbf{s}+ \kappa \Dvec) | = \sum_{j\in sec(\ell)} \eta^t_j(\mathbf{s}+ \kappa \Dvec) =1$. 
\end{proof}

%%%%%

\begin{applem}
Let $\mc{W}$  be a $d$-dimensional subspace of $\reals^n$ for $d \leq n$ and let $\bZ \sim \normal(0, \iden_n)$ be a standard Gaussian random vector. Let $(\wvec_{1}, ..., \wvec_{d})$ be an orthonormal basis of $\mc{W}$ with $\norm{\wvec_{i}}^2 = 1$ for
$i \in [d]$, and let  $\proj_\mc{W}^{\parallel}$ denote the orthogonal projection operator onto $\mc{W}$.  Then for $\mathbf{D} = [ \wvec_1\mid \ldots \mid \wvec_d]$, we have $ \proj^{\parallel}_\mc{W} \bZ  \overset{d}{=}  \mathbf{D} \tilde{\Z}$ where $\tilde{\Z} \sim \normal(0, \iden_d)$ is independent of 
$\mathbf{D}$. 
\label{lem:gauss_p0}
\end{applem}
%%%%%

\begin{applem}[$\mathcal{H} (d)$ concentration]
\label{lem:Hd_convergence}
Let $\Z \sim \mathcal{N}(0, \mathbb{I}_{ML})$ and $\tilde{\Z} \sim \mathcal{N}(0, \mathbb{I}_{ML})$ such that $(Z_i, \tilde{Z}_i)$ are i.i.d.\ bivariate Gaussian, for $1 \leq i \leq ML$.  For $\ell \in [L]$, let $\mathbf{Y}_{\ell} =\Z_{\ell}^*\eta_{\ell}(\betavec_0 - \sqrt{\tau} \tilde{\Z})$. Then, for a universal positive constant, $\kappa$, and $ \lambda_{\sfc} \in \Theta(1)$ for each $\sfc \in [\Lc]$,
\ben
P\Big( \frac{1}{L} \Big \lvert \sum_{\sfc \in [\Lc]}  \frac{\lambda_{\sfc}  W_{\sfr \sfc} }{\sqrt{\log M}}    \sum_{\ell \in \sfc}  ( \mathbf{Y}_{\ell} - \mathbb{E}[\mathbf{Y}_{\ell}])\Big \lvert \geq \e \Big) \leq \exp\{-\kappa_1 L (\omega/\Lr) \e^2 \}.
\een
\end{applem}

\begin{proof}
The proof is along the same lines as that of Lemma 20 in \cite{RushV19}, and is hence omitted.
\end{proof}

\bibliographystyle{ieeetr}
{\small{
\begin{spacing}{0.9}
\bibliography{sc_sparcs}
\end{spacing}
}}

\end{document}